\documentclass[10pt]{article}
\usepackage{graphicx} % Required for inserting images
\usepackage{macros}
\usepackage{authblk}
\usepackage{setspace}

\title{Gauge origami and quiver W-algebras III: Donaldson--Thomas $qq$-characters }
\author[$\diamondsuit$]{Taro Kimura}
\author[$\spadesuit$]{Go Noshita}
\affil[$\diamondsuit$]{Institut de Mathématiques de Bourgogne, Université de Bourgogne, CNRS, Dijon, France\footnote{Unité Mixte de Recherche (UMR 5584) commune au CNRS et à l'Université de Bourgogne}}
\affil[$\spadesuit$]{Department of Physics, The University of Tokyo, Tokyo, Japan}

\date{ }

\begin{document}

\begin{titlepage}
%\maketitle
\vspace*{1cm}
\vskip 12mm
\begin{center}
    {\LARGE \bf{Gauge origami and quiver W-algebras III: Donaldson--Thomas $qq$-characters}\par}          
\vskip 3cm    
\begin{center}
{\Large Taro Kimura${}^{\diamondsuit}$ and Go Noshita${}^{\spadesuit}$}
\end{center}
\vskip 2cm 
{\it $\diamondsuit$ Institut de Mathématiques de Bourgogne,}
{\it Université de Bourgogne, CNRS, Dijon, France }\\
{\it $\spadesuit$ Department of Physics, The University of Tokyo, Tokyo, Japan}
\end{center}
\begin{center}
{E-mail: \href{mailto:taro.kimura@u-bourgogne.fr}{\texttt{taro.kimura@u-bourgogne.fr}}, \\ 
 \href{mailto:noshita-go969@g.ecc.u-tokyo.ac.jp}{\texttt{noshita-go969@g.ecc.u-tokyo.ac.jp}}}
\end{center}
\vskip 2cm

\begin{abstract}
We further develop the BPS/CFT correspondence between quiver W-algebras/$qq$-characters and partition functions of gauge origami. We introduce $qq$-characters associated with multi-dimensional partitions with nontrivial boundary conditions which we call Donaldson--Thomas (DT) $qq$-characters. They are operator versions of the equivariant DT vertices of toric Calabi--Yau three and four-folds. Moreover, we revisit the construction of the D8 $qq$-characters with no boundary conditions and give a quantum algebraic derivation of the sign rules of the magnificent four partition function. We also show that under the proper sign rules, the D6 and D8 $qq$-characters with no boundary conditions all commute with each other and discuss its physical interpretation.

\end{abstract}

\end{titlepage}

\begin{spacing}{1.3}
\setcounter{tocdepth}{2}
\tableofcontents
\end{spacing}

\newpage
\section{Introduction}

\noindent Gauge origami, introduced by Nekrasov \cite{Nekrasov:2015wsu,Nekrasov:2016qym,Nekrasov:2016ydq,Nekrasov:2017rqy,Nekrasov:2017gzb} (see also \cite{Szabo:2022zyn,Kanno:2020ybd}) is a generalization of gauge theory in the sense that it is defined on multiple, generally intersecting space-time components. In string theory, such kind of setup is obtained as the low energy limit of a collection of multiple possibly intersecting D-branes. Even for such cases, with suitable numbers of supersymmetries and the power of supersymmetric localization, the infinite-dimensional path integral is reduced to a finite-dimensional integral over the moduli space of Bogomolny--Prasad--Sommerfield (BPS) configurations which eventually gives the \textit{gauge origami partition functions}.

The simplest example is the gauge origami obtained from type IIA string theory\footnote{Of course, one can consider setups like $\mathcal{C}\times \mathbb{C}^{4}$, where $\mathcal{C}=\mathbb{R}^{2},\mathbb{T}^{2}$ and get rational and elliptic versions of the partition functions.} on $\mathbb{R}\times \mathbb{S}^{1}\times \mathbb{C}^{4}$ with $\D(2p)$-branes wrapping $\mathbb{S}^{1}\times \Sigma$, where $\Sigma$ are holomorphic cycles in $\mathbb{C}^{4}$. The D0-branes wrapping~$\mathbb{S}^{1}$ and probing the $\D(2p)$-branes will play roles of instantons and the partition function is the Witten index \cite{Witten:1982df} of the supersymmetric quantum mechanics of these D0-brane worldvolume theories. Depending on $p=2,3,4$, the setup is called spiked instantons \cite{Nekrasov:2016gud,Nekrasov:2016qym}, tetrahedron instantons \cite{Pomoni:2021hkn,Pomoni:2023nlf}, and magnificent four \cite{Nekrasov:2017cih,Nekrasov:2018xsb}, respectively, and one can evaluate their partition functions explicitly.

Our specific interest lies in the generalization of the gauge origami setup to general toric Calabi--Yau four-folds (CY4). To do this, one can replace $\mathbb{C}^{4}$ to general toric CY4 and consider D-branes wrapping holomorphic cycles inside the toric CY4. One strong motivation to study these kind of setups is to get a complete understanding of the geometric engineering of quantum field theories related to Calabi--Yau four-folds and compute the associated physical quantities explicitly. See for example \cite{Leung:1997tw,Diaconescu:1998ua,Gukov:1999ya,Intriligator:2012ue,Jockers:2016bwi} for early works and \cite{Najjar:2023hee,Sangiovanni:2024nfz,Moleti:2024skd} for recent discussions.

Basically, there are two approaches to study such quantities. The first method is to use the quiver structure of toric CY4 and study the associated Witten indices of them, which we call the \textit{quiver formalism}. Attempts and examples on this direction are done in \cite{Kimura:2023bxy,Bao:2024ygr,Nekrasov:2012xe,Nekrasov:2013xda,Nekrasov:2015wsu,Jeong:2018qpc,Jeong:2021rll,Szabo:2024lcp, Bonelli:2020gku,Szabo:2023ixw}. The second method is to decompose the toric CY4 into $\mathbb{C}^{4}$ patches and consider the \textit{gluing} of them. This construction is related to the topological vertex methods \cite{Aganagic:2003db,Iqbal:2007ii,Taki:2007dh,Awata:2008ed,Nekrasov:2014nea,Iqbal:2003ds} so we call them the \textit{vertex formalism}. See also~\cite{Nekrasov:2003vi}. Such directions are studied in \cite{Nekrasov:2023nai,Piazzalunga:2023qik,Cao:2019tnw,Cao:2019tvv,Monavari:2022rtf,Bae:2022pif,Bae:2024bpx}. The relation between these two methods for the CY 4-fold setup is not so clear for the moment (see \cite{Yamazaki:2010fz} for discussion on three-folds for example) but we expect that new physical and mathematical phenomena will definitely occur here. Thus, conducting research from both perspectives is of utmost importance.

An interesting property of the gauge origami partition functions is the existence of non-perturbative Dyson--Schwinger equations associated with the symmetries of adding and removing BPS particles \cite{Nekrasov:2013xda,Nekrasov:2015wsu,Kim:2016qqs,Bourgine:2016vsq,Kimura:2015rgi}. Such kind of symmetries imply the existence of a quantum algebraic structure which leads to a duality that is now called the BPS/CFT correspondence \cite{Alday:2009aq,Nekrasov:2015wsu,Awata:2009ur,Awata:2010yy} (see \cite{LeFloch:2020uop} for a review). The $qq$-characters or the quiver W-algebras\footnote{Strictly speaking, the quiver W-algebras are vertex operator versions of the $qq$-characters, but we will not distinguish them from each other.} are physical observables characterizing the non-perturbative Dyson--Schwinger equations, and quantum toroidal algebras\footnote{One can also consider rational and elliptic variants of the story and affine Yangians and elliptic analogues will appear, but in this paper, we will basically focus on the trigonometric story.} \cite{Ginzburg,FFJMM1,Feigin2011,Feigin2011plane} are the infinite dimensional algebras that play the roles of creation and annihilation of the BPS particles. Various studies have been done in this direction and for example see \cite{DIMreview} and the references there. 

Construction of the $qq$-characters/quiver W-algebras associated with the gauge origami system of $\mathbb{C}^{4}$ was performed in a unified way in our previous work \cite{Kimura:2023bxy}. There, we introduced $qq$-characters associated with each $\D(2p)$-branes wrapping non-compact subspaces $\mathbb{C},\mathbb{C}^{2},\mathbb{C}^{3},\mathbb{C}^{4}$ and showed that the compositions of them reproduce the gauge origami partition functions, which establishes the BPS/CFT correspondence. Moreover, we gave conjectures for generalizations to toric Calabi--Yau four-folds on both directions, the quiver formalism and the vertex formalism, and generally called them BPS $qq$-characters. Since the former method is related with the quiver structure of toric CY4 and quiver Yangian \cite{Li:2020rij,Galakhov:2020vyb,Galakhov:2021xum} and its generalizations \cite{Galakhov:2021vbo,Noshita:2021ldl,Noshita:2021dgj}, we call the $qq$-characters appearing there the BPS quiver $qq$-characters or BPS quiver W-algebras. Some examples already appeared in \cite{Kimura:2023bxy} and a full description will be discussed in a future work \cite{Kimura-Noshita2}. For the latter method, the associated $qq$-characters were dubbed as \emph{webs of BPS $qq$-characters}, whose name comes from the fact that the appearing $qq$-characters are associated with brane webs and topological vertices. Although, the concept was introduced, the explicit construction of such $qq$-characters is yet to be done.

The goal of this paper is to initiate explicit studies on the constructions of the webs of BPS $qq$-characters. We will introduce a new type of $qq$-characters which we call the Donaldson--Thomas $qq$-characters. The D-brane $qq$-characters introduced in \cite{Kimura:2023bxy} have monomial terms associated with multi-dimensional partitions. For the D4 $qq$-characters, they are 2d partitions (Young diagrams), for the D6 $qq$-characters, they are 3d partitions (plane partitions), and for the D8 $qq$-characters, they are 4d partitions (solid partitions). Combinatorially, the DT $qq$-characters are $qq$-characters whose monomial terms correspond to multi-dimensional partitions with nontrivial asymptotic boundary conditions. Mathematically, they are vertex operator lift ups of equivariant DT vertices. Namely, the vacuum expectation value of the DT $qq$-characters reproduces the equivariant K-theoretic and elliptic DT vertices. Given such DT $qq$-characters, one would like to \emph{glue}\footnote{A different gluing method of Y-algebras \cite{Gaiotto:2017euk} and their $q$-deformations \cite{Harada:2021xnm,Kojima2019,Kojima2021} was introduced in \cite{Prochazka:2017qum, Harada:2020woh}. See also a recent paper discussing how to define $q$-deformations of $\mathcal{N}=2$ SCA \cite{Awata:2024jzk}. We expect that both of the constructions are related by string duality but for the moment it is not clear yet.} them and reproduce the DT partition function of toric CY4. For the moment, such gluing procedure is technically difficult and a complete description will be done in a future work \cite{Kimura-Noshita1}. Namely, in this paper, we will focus only on the equivariant DT vertex contribution.

The organization and summary of this paper is as follows. In section~\ref{sec:partitionfunct}, we review the contour integral formulas and partition functions where multidimensional partitions with nontrivial asymptotic boundary conditions appear, which are equivalent to the equivariant K-theoretic vertex of toric Calabi--Yau three and four-folds. We then move on to the free field realizations of the contour integral formulas in section~\ref{sec:freefieldvertexop}. We will see there that co-dimensional one boundary conditions (leg boundary conditions) are related with the D2-vertex operators, co-dimensional two boundary conditions (surface boundary conditions) are related with the D4-vertex operators, and co-dimensional three boundary conditions (hypersurface boundary conditions) are related with the D6 vertex operators. The free field realizations discussed here gives a systematic way to define highest weights which will be essential for the construction of DT $qq$-characters. Moreover, the free field realization of the contour integrals introduced here gives the contour integral expression of the DT $qq$-character. In this section, we will also discuss on how to dynamically generate the boundary conditions from the vertex operator viewpoint.

In section~\ref{sec:D4_qq} and \ref{sec:D6_qq}, we derive the DT $qq$-characters of the D4 and D6 setup by using the commutativity with screening charges. At the end we will see that the commutativity will uniquely determine the D4 and D6 DT $qq$-characters similar to how it was done in \cite{Kimura:2023bxy}. We will also see that the vacuum expectation value reproduces the partition functions in section~\ref{sec:partitionfunct}, which establishes the BPS/CFT correspondence. 

Section~\ref{sec:D8_qq} explores the generalization to the D8 setup. We first revisit the D8 $qq$-characters with no boundary conditions introduced in \cite{Kimura:2023bxy} by studying the infinite products of D6 $qq$-characters and give a quantum algebraic proof of the sign issues in section~\ref{sec:fusionD6toD8}. This is a new aspect which was not treated properly in our previous work. At the end, we will see that the infinite products will reproduce the magnificent four partition function \emph{including} the sign rules. We will also see that the D8 $qq$-characters with trivial boundary conditions defined in this way all commute with each other in section~\ref{sec:D8qqsignrule}. Moreover, the commutativity of them determines the sign rules uniquely. We will also show that one can obtain the D6 $qq$-characters from the D8 $qq$-character by tuning the parameter controlling the distance between the D8 and anti D8-branes. Therefore, at the end, we will see that the D6 and D8 $qq$-characters with trivial boundary conditions all commute with each other. We also show that this phenomenon is related with the fact that the magnificent four partition function and tetrahedron instanton partition functions do not depend on the Coulomb branch parameters and have a nice plethystic exponential formula. In sections~\ref{sec:D8qqlegboundary}, \ref{sec:D8qqsurfaceboundary}, \ref{sec:D8qqhypersurfaceboundary}, we explicitly construct the D8 DT $qq$-characters with nontrivial boundary conditions. Finally, in section~\ref{sec:conclusion}, we give a conclusion and discuss on future directions.

\section{Partition functions}\label{sec:partitionfunct}
In this section, we review the partition functions when the multi-dimensional partitions have nontrivial boundary conditions and have infinite number of boxes. We follow the notation of \cite{Kimura:2023bxy,Nekrasov:2023nai} and a summary is given in Appendix~\ref{app:notations}.

\subsection{D4 partition functions}\label{sec:D4partitionfunction}
\paragraph{Spiked instantons and finite Young diagrams}
Let us review the partition functions arising from the spiked instanton setup \cite{Nekrasov:2016gud,Nekrasov:2016qym,Nekrasov:2016ydq,Nekrasov:2015wsu}. The spiked instanton partition function comes from the following characters (see Appendix~\ref{app:notations} for the notations)
\bea\label{eq:D4chexpand}
\mathbf{V}=\frac{\mathbf{Y}^{\vee}\bfY}{\bfP_{\four}},\quad \bfY=\sum_{A\in\six}\bfP_{\bar{A}}\bfY_{A},\quad \bfY_{A}=\bfN_{A}-\bfP_{A}\bfK_{A},\quad \bfN_{A}=\sum_{\alpha=1}^{n_{A}}v_{A,\alpha},\quad \bfK_{A}=\sum_{I=1}^{k_{A}}x_{A,I}
\eea
for $A\in\six$. Here $\bfN_{A}$ represents the character of the framing bundle coming from the $\D4_{A}$-branes and $\bfK_{A}$ represents the character of the instanton bundle coming from the D0-branes attached to the $\D4_{A}$-branes. For later use, we also introduce the character of the total instanton bundle as $\bfK=\sum_{A\in\six}\bfK_{A}=\sum_{I=1}^{k}x_{I}$ and then obtain
\bea
 \bfY=\sum_{A\in\six}\bfP_{\bar{A}}\bfN_{A}-\bfK,\quad k=\sum_{A\in\six}k_{A}.
\eea
Expanding it, we have
\bea
\mathbf{V}&=\mathbf{V}_{\text{pert.}}+\mathbf{V}_{\text{inst.}},\quad\mathbf{V}_{\text{pert.}}=\sum_{A,B\in\six}\bfN_{A}^{\vee}\frac{\bfP_{\bar{A}}^{\vee}\bfP_{\bar{B}}}{\bfP_{\four}}\bfN_{B},\\
\mathbf{V}_{\text{inst.}}&=-\sum_{A\in\six}\bfP_{\bar{A}}^{\vee}\bfN_{A}^{\vee}\bfK-\sum_{A\in\six}\bfP_{\bar{A}}\bfN_{A}\bfK^{\vee}+\bfP_{\four}\bfK^{\vee}\bfK\\
&=\sum_{A,B\in\six}\left(-\bfN_{A}^{\vee}\bfP_{\bar{A}}^{\vee}\bfK_{B}-\bfK_{A}^{\vee}\bfP_{\bar{B}}\bfN_{B}+\bfP_{\four}\bfK_{A}^{\vee}\bfK_{B}\right).
\eea
To obtain the instanton partition function, we need to extract the square root part, which is identified with the (minus of) tangent bundle of the corresponding moduli space,
\beq
\mathbf{V}_{\text{inst.}}=\mathbf{v}_{\text{inst.}}+\mathbf{v}^{\vee}_{\text{inst.}},\quad 
\mathbf{v}_{\text{inst.}}=-\sum_{A\in\six}\bfP_{\bar{A}}^{\vee}\bfN_{A}^{\vee}\bfK+\sqrt{\bfP_{\four}\bfK^{\vee}\bfK},
\eeq
where we denoted the square root part of $\bfP_{\four}\bfK^{\vee}\bfK$ as $\sqrt{\bfP_{\four}\bfK^{\vee}\bfK}$. This square root part depends on which D4-brane the instanton is attached to. For the explicit formula of this square root part, see for example \cite[eq.~(3.5.5)]{Kimura:2023bxy}.

The contour integral formula is schematically given as
\beq
\mathcal{Z}_{k}^{\D4}=\frac{1}{k!}\oint\prod_{I=1}^{k}\frac{dx_{I}}{2\pi\iota x_{I}}\mathbb{I}'\left[\mathbf{v}_{\text{inst.}}\right],
\eeq
where $\mathbb{I}'$ means we extracted the divergent collision terms and $\iota=\sqrt{-1}$.\footnote{The contour integral formula is given by taking the index of the character $\mathbf{v}_{\text{inst.}}$. The collision term is understood as $\oint\prod_{I=1}^{k}\frac{dx_{I}}{2\pi\iota x_{I}}$:
\bea
\mathcal{Z}_{k}^{\D4}=\mathbb{I}[\mathbf{v}_{\text{inst.}}]=\frac{1}{k!}\oint\prod_{I=1}^{k}\frac{dx_{I}}{2\pi\iota x_{I}} \mathbb{I}'[\mathbf{v}_{\text{inst.}}],
\eea
where $\mathbb{I}'$ is understood as $\mathbb{I}[\mathbf{v}_{\text{inst.}}-k]$ and the unmovable terms (see Appendix~\ref{app:notations} for the definitions) are omitted.} The contour integral formula will then be
\bea
\mathcal{Z}_{k}^{\D4}=\frac{\mathcal{G}^{k}}{k!}\oint \prod_{I=1}^{k}\frac{dx_{I}}{2\pi\iota x_{I}}\prod_{A\in\six}\prod_{\alpha=1}^{n_{A}}\prod_{I=1}^{k}\mathscr{S}_{\bar{A}}\left(\frac{v_{A,\alpha}}{x_{I}}\right)\prod_{I<J}\mathcal{A}_{\mathbb{C}^{4}}\left(\frac{x_{I}}{x_{J}}\right)^{-1}
\eea
where
\bea
\mathcal{G}=-q_{4}\frac{(1-q_{12})(1-q_{13})(1-q_{23})}{(1-q_{1})(1-q_{2})(1-q_{3})(1-q_{4})}.
\eea
Note that the factor $\mathcal{G}$ is invariant under permutation of the equivariant parameters $q_{1,2,3,4}$.

After evaluating the contour integral formula, the poles are classified by \textit{finite} 2d partitions (Young diagrams) and we obtain 
\bea
\mathbf{K}_{A}|_{\vec{\lambda}_{A}}=\sum_{\alpha=1}^{n_{A}}\sum_{\Abox\in\lambda_{A}^{(\alpha)}}\chi_{A,v_{A,\alpha}}(\Bbox),\quad A\in\six,
\eea
where
\bea
\chi_{A,x}(\Bbox)=xq_{a}^{i-1}q_{b}^{j-1},\quad A=(ab)\in\six.
\eea
The partition function will then be a sum of arbitrary finite Young diagrams:
\bea
\mathcal{Z}^{\D4}_{\text{inst.}}=\sum_{\underline{\vec{\lambda}}}\mathfrak{q}^{|\underline{\vec{\lambda}}|}\mathcal{Z}^{\D4}_{\text{spk.inst.}}[\underline{\vec{v}},\underline{\vec{\lambda}}].
\eea

\paragraph{Young diagrams with infinite size}We denote a Young diagram with infinite size as $\tilde{\lambda}$. The Young diagram $\tilde{\lambda}$ will look like
\begin{equation}\label{eq:2dpartition}
        \adjustbox{valign=c}{\begin{tikzpicture}[scale=0.7]
    
        \fill[white!50!gray] (0,0)--(0,4.3)--(0.6,4.3)--(0.6,0)--(0,0);
        \fill[white!50!gray] (0,0)--(4.2,0)--(4.2,0.8)--(0,0.8)--(0,0);
        \fill[white!50!gray] (0,0)--(0.6,0)--(0.6,0.8)--(0,0.8)--(0,0);
        \draw [decorate,decoration = {mirror,brace}] (0.02,-0.05)--(0.58,-0.05);
        \draw [decorate,decoration = {brace}] (-0.05,0.02)--(-0.05,0.78);
        \node [left] at (-0.05,0.4){$k$};
        \node[below] at (0.3,-0.05){$l$};
        
        \draw[->] (0,-0.5)--(0,4.4);
        \draw[->] (-0.5,0)--(4.3,0);
        \node[right] at (4.3,0){$1$};
        \node[above] at (0,4.4) {$2$};

        \draw (0.2,4.28)--(0.2,0);
        \draw (0.4,4.28)--(0.4,0);
        \draw (0.6,4.28)--(0.6,0);

        \draw (0,0.2)--(4.2,0.2);
        \draw (0,0.4)--(4.2,0.4);
        \draw (0,0.6)--(4.2,0.6);
        \draw (0,0.8)--(4.2,0.8);
        \draw[thick] (0.6,0.8)--(3.8,0.8);%--(3.8,1.4)--(3.4,1.4)--(3.4,2.0)--(3.0,2.0)--(3.0,2.4)--(2.4,2.4);
        \draw [thick] (3.8,0.8)--+(0,0.6)--+(-0.6,0.6)--+(-0.6,1.2)--+(-1.2,1.2)--+(-1.2,1.8)--+(-1.8,1.8)--+(-1.8,2.4)--+(-2.4,2.4)--+(-2.4,3.0)--+(-3.2,3.0);
        \fill[red] (0.6,0.8)--(0.8,0.8)--(0.8,1.0)--(0.6,1.0)--(0.6,0.8);
        \draw[thick] (0.6,3.8)--(0.6,0.8);
        \draw[thick] (1.8,1.8)--(3.5,3.5);
        \node[scale=1.5,thick,right] at (3.5,3.5){$\lambda_{\reg}$};
        \draw (0.1,3.4)--(-0.4,3.8);
        \node[scale=1.5,left] at (-0.4,3.8){$\lambda_{\bd}$};
        \end{tikzpicture}
        }
\end{equation}
Namely, we have nontrivial boundary conditions at the two legs of the Young diagram. The conditions are written as
\bea\label{eq:2dpartitionbd}
\tilde{\lambda}_{i}=\infty,\quad (i=1,\ldots, l),\quad \tilde{\lambda}^{\rmT}_{j}=\infty,\quad (j=1,\ldots, k),
\eea
where $\tilde{\lambda}_{i}$ denotes the number of boxes in the 2-axis and $\tilde{\lambda}^{\rmT}$ is the transpose of it. We can decompose $\tilde{\lambda}$ into two parts $\lambda_{\bd}$ and $\lambda_{\reg}$ as \eqref{eq:2dpartition}. The part $\lambda_{\bd}$ is the infinite size part of the Young diagram determined by two parameters $k,l\in\mathbb{Z}_{\geq 0}$ and it will be called the boundary contributions. Obviously, once the boundary condition is fixed, the finite size part $\lambda_{\reg}$ determines the possible Young diagrams with infinite size obeying the conditions \eqref{eq:2dpartitionbd} and it is just a Young diagram with finite size whose origin is shifted from $(1,1)$ to $(l+1,k+1)$.

\paragraph{D4 partition functions with boundary conditions}To consider partition functions where the appearing Young diagrams have nontrivial boundary conditions, we can first formally decompose the character as
\bea
\bfK=\bfK^{\text{bd}}+\bfK^{\text{reg}},\quad \bfK_{A}=\bfK_{A}^{\text{bd}}+\bfK_{A}^{\text{reg}},\quad \bfK^{\text{bd,reg}}=\sum_{A\in\six}\bfK^{\text{bd,reg}}_{A}.
\eea
%where 
%\bea
%\bfK^{\reg}=\sum_{I=1}^{k}x_{I},\quad \bfK_{A}^{\text{bd}}=\sum_{\alpha=1}^{n_{A}}\frac{(1-q_{A}^{(k_{A,\alpha},l_{A,\alpha})})}{\bfP_{A}}v_{A,\alpha},\quad q_{A}^{(k_{A,\alpha},l_{A,\alpha})}=q_{b}^{k_{A,\alpha}}q_{a}^{l_{A,\alpha}},
%\eea
%and $k_{A,\alpha},l_{A,\alpha}\in\mathbb{Z}_{\geq 0},\,\, A=(ab)\in\six$. 
Inserting this expansion, we have
\bea
\bfV&=\bfV_{\text{pert.}}+\bfV_{\text{bd.}}+\bfV_{\text{inst.}},\quad \bfV_{\text{bd.}}=-\sum_{A\in\six}\bfP_{\bar{A}}^{\vee}\bfN_{A}^{\vee}\bfK^{\bd}-\sum_{A\in\six}\bfP_{\bar{A}}\bfN_{A}\bfK^{\bd}+\bfP_{\four}\bfK^{\bd\vee}\bfK^{\bd},\\
\bfV_{\text{inst.}}&=-\sum_{A\in\six}\bfP_{\bar{A}}^{\vee}\bfN_{A}^{\vee}\bfK^{\reg}-\sum_{A\in\six}\bfP_{\bar{A}}\bfN_{A}\bfK^{\reg\vee}+\bfP_{\four}\bfK^{\bd\vee}\bfK^{\reg}+\bfP_{\four}\bfK^{\reg\vee}\bfK^{\bd}+\bfP_{\four}\bfK^{\reg\vee}\bfK^{\reg}.
\eea
The square root part is 
\bea
\mathbf{v}_{\text{inst.}}&=-\sum_{A\in\six}\bfP_{\bar{A}}^{\vee}\bfN_{A}^{\vee}\bfK^{\reg}+\bfP_{\four}\bfK^{\bd\vee}\bfK^{\reg}+\sqrt{\bfP_{\four}\bfK^{\reg\vee}\bfK^{\reg}}\\
&=-\sum_{A\in\six}\bfP_{\bar{A}}^{\vee}\bfN_{A}^{\reg\vee}\bfK^{\reg}+\sqrt{\bfP_{\four}\bfK^{\reg\vee}\bfK^{\reg}}.
\eea
%where
%\bea
%\bfN_{A}^{\reg}=\sum_{A\in\six}\sum_{\alpha=1}^{n_{A}}q_{b}^{k_{A,\alpha}}q_{a}^{l_{A,\alpha}}v_{A,\alpha}.
%\eea
%The contour integral formula is then given as
%\bea
%\mathcal{Z}_{k}^{\D4}\propto \frac{1}{k!}\oint \prod_{I=1}^{k}\frac{dx_{I}}{2\pi\iota x_{I}}\prod_{A\in\six}\prod_{\alpha=1}^{n_{A}}\prod_{I=1}^{k}\mathscr{S}_{\bar{A}}\left(\frac{q_{A}^{(k_{A,\alpha},l_{A,\alpha})}v_{A,\alpha}}{x_{I}}\right)\prod_{I<J}\mathcal{A}_{\mathbb{C}^{4}}\left(\frac{x_{I}}{x_{J}}\right)^{-1}.
%\eea

Let us focus on when there is only one $\D4_{12}$-brane:
\bea
\mathbf{v}_{\text{inst.}}=-\bfP_{34}^{\vee}\bfN_{12}^{\vee}\bfK^{\reg}_{12}+\bfP_{\four}\bfK^{\bd\vee}\bfK^{\reg}+\bfP_{123}^{\vee}\bfK_{12}^{\reg\vee}\bfK_{12}^{\reg}
\eea
where
\bea
\bfN_{12}=x,\quad \bfK^{\bd}_{12}=\sum_{\Abox\in\lambda_{\bd}}\chi_{12,x}(\Bbox),\quad \bfK^{\reg}_{12}=\sum_{I=1}^{k}x_{I}.
\eea
The character $\bfK_{12}^{\bd}$ can is written explicitly as
\bea
\bfK_{12}^{\bd}&=\sum_{i=1}^{\infty}\sum_{j=1}^{k}xq_{1}^{i-1}q_{2}^{j-1}+\sum_{j=1}^{\infty}\sum_{i=1}^{l}xq_{2}^{k}q_{2}^{j-1}q_{1}^{i-1}.
\eea
The boundary contribution can be further computed as
\bea\label{eq:D4boundarych}
\bfK_{12}^{\bd}=\frac{(1-q_{2}^{k}q_{1}^{l})x}{(1-q_{1})(1-q_{2})}.
\eea
Inserting this to the character, we have
\bea
\mathbf{v}_{\text{inst.}}&=-\bfP_{34}^{\vee}\bfN_{12}^{\vee}\bfK_{12}^{\reg}+\bfP_{\four}\left(\frac{(1-q_{2}^{k}q_{1}^{l})x}{\bfP_{12}}\right)^{\vee}\bfK^{\reg}+\bfP_{123}^{\vee}\bfK_{12}^{\reg\vee}\bfK_{12}^{\reg}\\
&=-\bfP_{34}^{\vee}\bfN_{12}^{\reg\vee}\bfK_{12}^{\reg}+\bfP_{123}^{\vee}\bfK_{12}^{\reg\vee}\bfK_{12}^{\reg}
\eea
where $\bfN_{12}^{\reg}=q_{2}^{k}q_{1}^{l}x$. The contour integral formula is then written as
\bea
\mathcal{Z}_{k}= \frac{\mathcal{G}^{k}}{k!}\oint\prod_{I=1}^{k}\frac{dx_{I}}{2\pi\iota x_{I}}\prod_{I=1}^{k}\mathscr{S}_{34}\left(\frac{q_{2}^{k}q_{1}^{l}x}{x_{I}}\right)\prod_{I<J}\mathcal{A}_{\mathbb{C}^{4}}\left(\frac{x_{I}}{x_{J}}\right)^{-1}.
\eea
The nontrivial boundary condition only shifts the Coulomb branch parameter $x\rightarrow q_{2}^{k}q_{1}^{l}x$ and thus the poles are simply classified by finite Young diagrams $\lambda_{\reg}$ as
\bea
\bfK^{\reg}_{12}|_{\lambda^{\reg}}=\sum_{\Abox\in\lambda_{\reg}}\chi_{12,q_{2}^{k}q_{1}^{l}x}(\Bbox)
\eea
 which is also consistent with \eqref{eq:2dpartitionbd}. The partition function is just 
 \bea\label{eq:D4-partition}
\mathcal{Z}^{\D4}_{\text{inst.}}=\sum_{\lambda_{\reg}}\mathfrak{q}^{|\lambda_{\reg}|}\widetilde{\mathcal{Z}}^{\D4}_{12}[\lambda_{\reg}],\quad\widetilde{\mathcal{Z}}^{\D4}_{12}[\lambda_{\text{reg}}]&=\mathbb{I}\left[-\bfP_{34}^{\vee}\bfN_{12}^{\reg\vee}\bfK_{12}^{\reg}+\bfP_{123}^{\vee}\bfK_{12}^{\reg\vee}\bfK_{12}^{\reg}\right].
 \eea
Generalizations to the spiked instanton case is straightforward.

\subsection{D6 partition functions}\label{sec:D6partitionfunction}

\paragraph{Tetrahedron instantons and finite plane partitions}Let us review the partition function of the tetrahedron instantons \cite{Pomoni:2021hkn, Pomoni:2023nlf,Fasola:2023ypx}. The partition function comes from the following character
\bea
\mathbf{V}=\frac{\bfY^{\vee}\bfY}{\bfP_{\four}},\quad \bfY=\sum_{a\in\four}\bfP_{a}\bfY_{\bar{a}},\quad \bfY_{\bar{a}}=\bfN_{\bar{a}}-\bfP_{\bar{a}}\bfK_{\bar{a}},\quad a\in\four,
\eea
where
\bea
\bfN_{\bar{a}}=\sum_{\alpha=1}^{n_{\bar{a}}}v_{\bar{a},\alpha},\quad \bfK_{\bar{a}}=\sum_{I=1}^{k_{\bar{a}}}x_{\bar{a},I},\quad a\in\four.
\eea
Similar to the spiked instanton case, we introduce the character of the total instanton bundle as
\bea
\bfK=\sum_{a\in\four}\bfK_{\bar{a}}=\sum_{I=1}^{k}x_{I},\quad \bfY=\sum_{a\in\four}\bfP_{a}\bfN_{\bar{a}}-\bfK,\quad k=\sum_{a\in\four}k_{\bar{a}}.
\eea
Expanding $\mathbf{V}$, we have
\bea
\mathbf{V}&=\bfV_{\text{pert.}}+\mathbf{V}_{\text{inst.}},\quad \bfV_{\text{pert.}}=\sum_{a,b\in\four}\bfN_{\bar{a}}^{\vee}\frac{\bfP^{\vee}_{a}\bfP_{b}}{\bfP_{\four}}\bfN_{\bar{b}},\\
\mathbf{V}_{\text{inst.}}&=-\sum_{a\in\four}\bfP_{a}^{\vee}\bfN_{\bar{a}}^{\vee}\bfK-\sum_{a\in\four}\bfP_{a}\bfN_{\bar{a}}\bfK^{\vee}+\bfP_{\four}\bfK^{\vee}\bfK\\
&=\sum_{a,b\in\four}\left(-\bfN_{\bar{a}}^{\vee}\bfP_{a}^{\vee}\bfK_{\bar{b}}-\bfK_{\bar{a}}^{\vee}\bfP_{b}\bfN_{\bar{b}}+\bfP_{\four}\bfK_{\bar{a}}^{\vee}\bfK_{\bar{b}}\right).
\eea
The instanton partition function comes from the square root part as
\bea
\mathbf{V}_{\text{inst.}}&=\mathbf{v}_{\text{inst.}}+\mathbf{v}^{\vee}_{\text{inst.}},\quad 
\mathbf{v}_{\text{inst.}}=-\sum_{a\in\four}\bfP_{a}^{\vee}\bfN_{\bar{a}}^{\vee}\bfK+\sqrt{\bfP_{\four}\bfK^{\vee}\bfK}.
\eea
See \cite[eq.~(3.4.2)]{Kimura:2023bxy} for the explicit choice of the square root part $\sqrt{\bfP_{\four}\bfK^{\vee}\bfK}$. The contour integral formula is then 
\bea
\mathcal{Z}_{k}^{\D6}= \frac{\mathcal{G}^{k}}{k!}\oint \prod_{I=1}^{k}\frac{dx_{I}}{2\pi\iota x_{I}}\prod_{a\in\four}\prod_{I=1}^{k}\mathscr{V}_{a}\left(\frac{v_{\bar{a},\alpha}}{x_{I}}\right)\prod_{I<J}\mathcal{A}_{\mathbb{C}^{4}}\left(\frac{x_{I}}{x_{J}}\right)^{-1}.
\eea
The poles are classified by a \textit{finite} plane partition, which is a sequence of non-negative integers obeying
\bea\label{eq:planepartitioncond}
\pi=\{\pi_{i,j}\in\mathbb{Z}_{\geq 0}\},\quad \pi_{i,j}\geq \pi_{i+1,j},\quad \pi_{i,j}\geq \pi_{i,j+1}.
\eea
We then have
\bea
\bfK_{\bar{a}}|_{\vec{\pi}_{\bar{a}}}=\sum_{\alpha=1}^{n_{\bar{a}}}\sum_{\scube\in\pi_{\bar{a}}^{(\alpha)}}\chi_{\bar{a},v_{\bar{a},\alpha}}(\cube),\quad \chi_{\bar{a},x}(\cube)=xq_{b}^{i-1}q_{c}^{j-1}q_{d}^{k-1},\quad \bar{a}=(bcd)
\eea
which gives the tetrahedron instanton partition function
\bea
\mathcal{Z}^{\D6}_{\text{inst.}}=\sum_{\underline{\vec{\pi}}}\mathfrak{q}^{|\underline{\vec{\pi}}|}\mathcal{Z}^{\D6}_{\text{tet.inst.}}[\underline{\vec{v}},\underline{\vec{\pi}}],\quad \mathcal{Z}^{\D6}_{\text{tet.inst.}}[\underline{\vec{v}},\underline{\vec{\pi}}]=\mathbb{I}[\mathbf{v}_{\text{inst.}}|_{\underline{\vec{\pi}}}].
\eea

\paragraph{Plane partitions with infinite size} We denote a plane partition with infinite size as $\tilde{\pi}$. We have the following two cases for such kind of plane partitions:
\bea\label{eq:D6planepartition}
\adjustbox{valign=c}{\includegraphics[width=6cm]{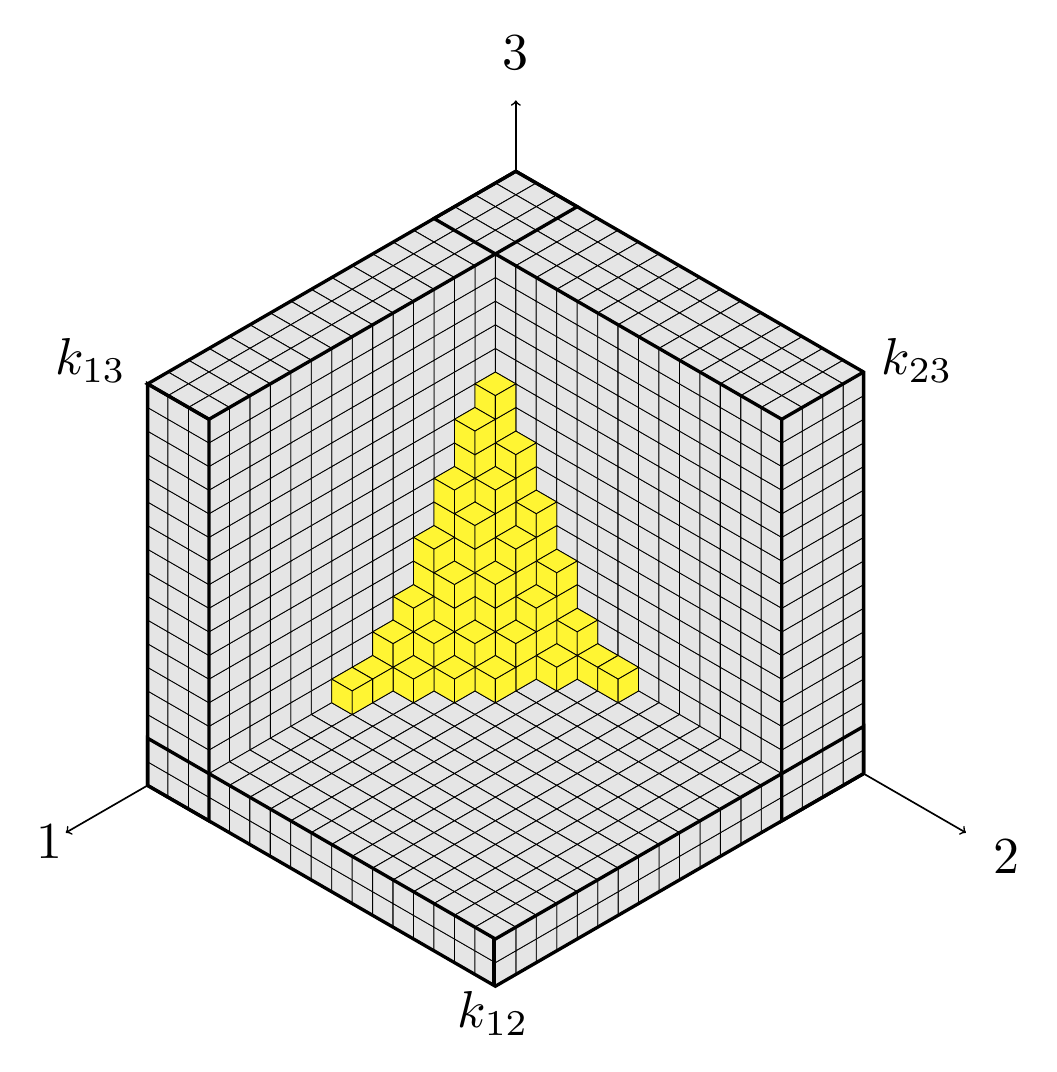}}\quad \adjustbox{valign=c}{\includegraphics[width=6cm]{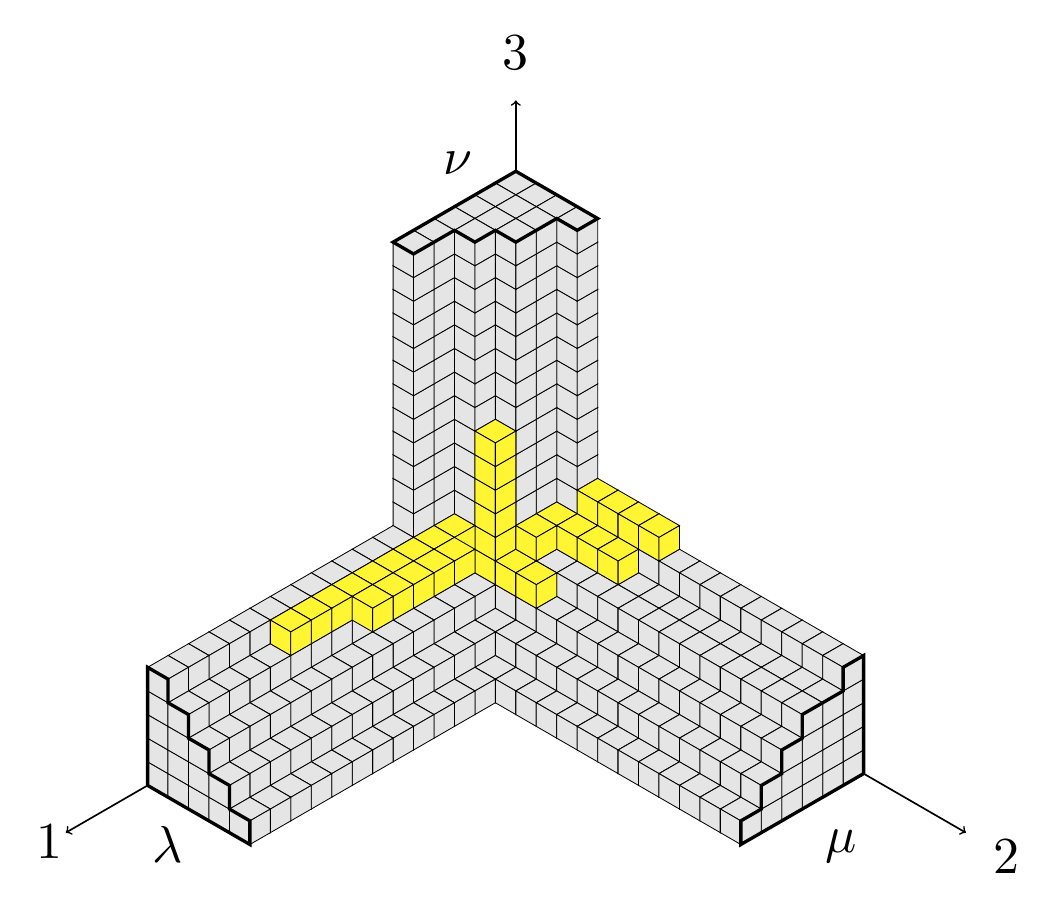}}
\eea
The left plane partition is a plane partition where we have asymptotic 1d partitions denoted as $k_{13},k_{23},k_{12}$, while the right plane partition is a plane partition where we have asymptotic Young diagrams denoted as $\lambda,\mu,\nu$. Generally, we can also combine these two boundary conditions, but in this paper we will only focus when the boundary conditions are either the above cases though the generalizations are straightforward. Note that the plane partition $\tilde{\pi}$ is the set of both the gray boxes and the yellow boxes. We can decompose the plane partition $\tilde{\pi}$ into two parts: the gray boxes and the yellow boxes. We denote them $\pi_{\bd},\pi_{\reg}$ respectively. 

The boundary conditions of the left plane partition are written as
\bea\label{eq:planesurfaceboundcond}
%&\tilde{\pi}_{i,j}\geq \tilde{\pi}_{i+1,j},\quad \tilde{\pi}_{i,j}\geq \tilde{\pi}_{i,j+1},\\
&\tilde{\pi}_{i,j}=k_{12},\quad i\rightarrow \infty,\,\,j\rightarrow \infty,\\
&\tilde{\pi}_{i,j}=\infty,\quad i=1,\ldots,k_{23},\,\,j\rightarrow \infty,\\
&\tilde{\pi}_{i,j}=\infty,\quad i\rightarrow \infty,\,\,j=1,\ldots,k_{13}.
\eea
On the other hand, the boundary conditions of the right plane partition are written as
\bea\label{eq:planelegboundcond}
%&\tilde{\pi}_{i,j}\geq \tilde{\pi}_{i+1,j},\quad \tilde{\pi}_{i,j}\geq \tilde{\pi}_{i,j+1},\\
&\tilde{\pi}_{i,j}=\infty,\quad  i=1,2,\ldots,\,\,\, j=1,\ldots,\nu_{i},\\
&\tilde{\pi}_{i,j}=\lambda_{j},\quad i\rightarrow \infty,\\
&\tilde{\pi}_{i,j}=\mu_{i}^{\rmT},\quad j\rightarrow \infty.
\eea
The $\tilde{\pi}_{i,j}$ represents the number of boxes in the third direction and $(i,j)\in \mathbb{Z}_{>0}$ is the coordinate of the 12-plane. We also chose a natural orientation for the boundary Young diagrams. 

Given a plane partition $\tilde{\pi}$, we may consider boxes which can be added to or removed from the plane partition without breaking the plane partition condition \eqref{eq:planepartitioncond} and the boundary conditions \eqref{eq:planesurfaceboundcond}, \eqref{eq:planelegboundcond}. The set of these addable boxes and removable boxes are represented by $A(\tilde{\pi}),R(\tilde{\pi})$ respectively. Since the boxes in $\pi_{\bd}$ are fixed depending on the boundary conditions, the addable/removable boxes are yellow boxes. To make the notation simple, we write $A(\pi_{\reg})\coloneqq A(\tilde{\pi}),R(\pi_{\reg})\coloneqq R(\tilde{\pi})$ but one has to be careful that whether a box is addable or removable depends on the whole plane partition~$\tilde{\pi}$.

Let $\mathcal{PP}$ denote the set of arbitrary finite plane partitions. We also denote $\mathcal{PP}_{k_{23},k_{13},k_{12}}$ the possible plane partitions of the left hand side with fixed $k_{23},k_{13},k_{12}$. Obviously, the set of the yellow boxes of the left hand side is just a normal plane partition with finite size but the origin of the plane partition is only shifted. Therefore, we have the following equivalency
\bea
\mathcal{PP}_{k_{23},k_{13},k_{12}}=\mathcal{PP}.
\eea
For the right hand side, we denote the set of all possible plane partitions with fixed boundary Young diagrams $\lambda,\mu,\nu$ as $\mathcal{PP}_{\lambda\mu\nu}$. For later use, let $\mathcal{B}_{\lambda\mu\nu}$ denote the set of boxes of $\pi_{\bd}$.

\paragraph{D6 partition functions with boundary conditions}
Similar to the spiked instanton case, we formally decompose the character $\bfK$ into
\bea
\bfK=\bfK^{\bd}+\bfK^{\reg},\quad \bfK_{\bar{a}}=\bfK_{\bar{a}}^{\bd}+\bfK_{\bar{a}}^{\reg},\quad \bfK^{\bd,\reg}=\sum_{a\in\four}\bfK_{\bar{a}}^{\bd,\reg}.
\eea
We then have
\bea
\bfV&=\bfV_{\text{pert.}}+\bfV_{\text{bd.}}+\bfV_{\text{inst.}},\quad \bfV_{\text{bd.}}=-\sum_{a\in\four}\bfP_{a}^{\vee}\bfN_{\bar{a}}^{\vee}\bfK^{\bd}-\sum_{a\in\four}\bfP_{a}\bfN_{\bar{a}}\bfK^{\bd\vee}+\bfP_{\four}\bfK^{\bd\vee}\bfK^{\bd},\\
\bfV_{\text{inst.}}&=-\sum_{a\in\four}\bfP_{a}^{\vee}\bfN_{\bar{a}}^{\vee}\bfK^{\reg}-\sum_{a\in\four}\bfP_{a}\bfN_{\bar{a}}\bfK^{\reg\vee}+\bfP_{\four}\bfK^{\bd\vee}\bfK^{\reg}+\bfP_{\four}\bfK^{\reg\vee}\bfK^{\bd}+\bfP_{\four}\bfK^{\reg\vee}\bfK^{\reg}.
\eea
The square roots of $\mathbf{V}_{\bd.}$ and $\mathbf{V}_{\text{inst.}}$ are
\bea\label{eq:D6boundarycharactergeneral}
\mathbf{v}_{\bd.}&=-\sum_{a\in\four}\bfP_{a}^{\vee}\bfN_{\bar{a}}^{\vee}\bfK^{\bd}+\sqrt{\bfP_{\four}\bfK^{\bd\vee}\bfK^{\bd}},\\
\mathbf{v}_{\text{inst.}}&=-\sum_{a\in\four}\bfP_{a}^{\vee}\bfN_{\bar{a}}^{\vee}\bfK^{\reg}+\bfP_{\four}\bfK^{\bd\vee}\bfK^{\reg}+\sqrt{\bfP_{\four}\bfK^{\reg\vee}\bfK^{\reg}}.
\eea
The tetrahedron instanton partition function is then obtained by taking the index of $\mathbf{v}_{\text{inst.}}$. In this paper, we will only focus on the contributions coming from the $\mathbf{v}_{\text{inst.}}$ part. To explicitly write down the partition function, we need to specify the explicit boundary contributions $\bfK^{\bd}$. The contour integral formula is formally given as
\bea
\mathcal{Z}_{k}=\frac{1}{k!}\oint \prod_{I=1}^{k}\frac{dx_{I}}{2\pi\iota x_{I}}\mathbb{I}'[\mathbf{v}_{\text{inst.}}].
\eea

Let us focus when there is only one D6$_{123}$-brane. The partition function is given as
\bea
\mathbf{v}_{\text{inst.}}=-\bfP_{4}^{\vee}\bfN_{\bar{4}}^{\vee}\bfK_{\bar{4}}^{\reg}+\bfP_{\four}\bfK_{\bar{4}}^{\bd\vee}\bfK_{\bar{4}}^{\reg}+\bfP_{123}^{\vee}\bfK_{\bar{4}}^{\reg\vee}\bfK_{\bar{4}}^{\reg}
\eea
where
\bea
\bfN_{\bar{4}}=x,\quad \bfK_{\bar{4}}^{\bd}=\sum_{\scube\in\pi_{\bd}}\chi_{\bar{4},x}(\cube).
\eea
When $\pi_{\bd}=\emptyset$ and $\bfK^{\bd}=0$, we simply obtain the U(1) partition function of the 7d theory on the $\D6_{123}$-brane compactified on the circle $\mathbb{S}^1$. The boundary contributions $\bfK_{\bar{4}}^{\bd}$ comes from the gray boxes in \eqref{eq:D6planepartition}. After deriving the contour integral formula and evaluating the residues, $\mathbf{K}^{\reg}_{\bar{4}}$ will eventually be the character of $\pi_{\reg}$, which is the set of yellow boxes in \eqref{eq:D6planepartition}. The partition function is then defined as
\bea
\mathcal{Z}=\sum_{\pi_{\reg}}\mathfrak{q}^{|\pi_{\reg}|}\,\mathbb{I}[\mathbf{v}_{\text{inst.}}].
\eea
Note here that the topological term counts the boxes of the \textit{yellow boxes} but not the boxes of the full plane partition $\tilde{\pi}$. Let us show this explicitly for the two boundary conditions.

\begin{comment}
The character $\bfK_{\bar{4}}$ can be decomposed into two parts coming from the gray boxes and yellow boxes denoted as $\bfK_{\bar{4}}^{\bd},\bfK_{\bar{4}}^{\reg}$, respectively:
\bea
&\bfK_{\bar{4}}=\bfK_{\bar{4}}^{\bd}+\bfK_{\bar{4}}^{\reg}.
%&\bfK_{\bar{4}}^{\bd}=\sum_{\scube\in \pi_{\bd}}\chi_{\bar{4},x}(\cube),\quad \bfK_{\bar{4}}^{\reg}=\sum_{\scube\in\pi_{\reg}}\chi_{\bar{4},x}(\cube)
\eea
The total character is then decomposed as
\bea
\mathbf{v}&=\mathbf{v}_{\bd}+\mathbf{v}_{\reg},\\
\mathbf{v}_{\bd}&=-\bfP_{4}^{\vee}\bfN_{\bar{4}}^{\vee}\bfK_{\bar{4}}^{\bd}+\bfP_{123}^{\vee}\bfK_{\bar{4}}^{\bd\vee}\bfK_{\bar{4}}^{\bd}\\
\mathbf{v}_{\reg}&=-\bfP_{4}^{\vee}\bfN_{\bar{4}}^{\vee}\bfK_{\bar{4}}^{\reg}+\bfP_{123}^{\vee}(\bfK_{\bar{4}}^{\bd\vee}\bfK_{\bar{4}}^{\reg}+\bfK_{\bar{4}}^{\reg\vee}\bfK_{\bar{4}}^{\bd})+\bfP_{123}^{\vee}\bfK_{\bar{4}}^{\reg\vee}\bfK_{\bar{4}}^{\reg}.
\eea
We define the full partition function as
\bea
\mathcal{Z}=\sum_{\pi_{\reg}}\mathfrak{q}^{|\pi_{\reg}|}\,\,\mathbb{I}[\mathbf{v}_{\reg}].
\eea
Note here that the topological term counts the boxes of the  \textit{yellow boxes} but not the boxes of the full plane partition $\tilde{\pi}$.
\end{comment}

\paragraph{Surface boundary condition} We call the boundary condition when the asymptotics of the three surfaces $12,23,31$-planes are specified, the \textit{surface boundary conditions}. In this case, we can explicitly compute the characters $\bfK_{\bar{4}}^{\bd},\bfK_{\bar{4}}^{\reg}$. For the boundary contributions, the result is
\bea
\bfK_{\bar{4}}^{\bd}=\frac{x}{\bfP_{123}}(1-q_{1}^{k_{23}}q_{2}^{k_{13}}q_{3}^{k_{12}}).
\eea
We then have
\bea
\mathbf{v}_{\text{inst.}}=-\bfP_{4}^{\vee}\bfN_{\bar{4}}^{\reg\vee}\bfK_{\bar{4}}^{\reg}+\bfP_{123}^{\vee}\bfK_{\bar{4}}^{\reg}\bfK_{\bar{4}}^{\reg},\quad \bfN_{\bar{4}}^{\reg}=xq_{1}^{k_{23}}q_{2}^{k_{13}}q_{3}^{k_{12}}.
\eea
The computation here is similar to what happened in the D4 case of section~\ref{sec:D4partitionfunction}. 

The contour integral is then given as
\bea\label{eq:D6surfacecontourintegral}
\mathcal{Z}_{k}=\frac{\mathcal{G}^{k}}{k!}\oint \frac{dx_{I}}{2\pi\iota x_{I}}\prod_{I=1}^{k}\mathscr{V}_{4}\left(\frac{xq_{1}^{k_{23}}q_{2}^{k_{13}}q_{3}^{k_{12}}}{x_{I}}\right)\prod_{I<J}\mathcal{A}_{\mathbb{C}^{4}}\left(\frac{x_{I}}{x_{J}}\right)^{-1}
\eea
which is just the same contour integral for the usual D6 partition function but only the Coulomb branch parameter $x$ is shifted to $xq_{1}^{k_{23}}q_{2}^{k_{13}}q_{3}^{k_{12}}$. The poles are classified with finite plane partitions which the origin at $xq_{1}^{k_{23}}q_{2}^{k_{13}}q_{3}^{k_{12}}$. Namely, the yellow boxes in the left figure of \eqref{eq:D6planepartition} contributes to $\bfK_{\bar{4}}^{\reg}$ for this case:
\bea
\bfK_{\bar{4}}^{\reg}=\sum_{\scube\in\pi_{\reg}}\chi_{\bar{4},q_{1}^{k_{23}}q_{2}^{k_{13}}q_{3}^{k_{12}}x}(\cube).
\eea
Therefore, the partition function is then given as
\bea
\mathcal{Z}&=\sum_{\pi_{\reg}\in\mathcal{PP}}\mathfrak{q}^{|\pi_{\reg}|}\widetilde{\mathcal{Z}}_{\bar{4}}^{\D6}[\pi_{\reg}],\quad \widetilde{\mathcal{Z}}^{\D6}_{\bar{4}}[\pi_{\reg}]=\mathbb{I}\left[-\bfP_{4}^{\vee}\bfK_{\bar{4}}^{\reg}\bfN_{\bar{4}}^{\reg\vee}+\bfP_{123}^{\vee}\bfK_{\bar{4}}^{\reg\vee}\bfK_{\bar{4}}^{\reg}\right]
\eea
and $\widetilde{\mathcal{Z}}^{\D6}_{\bar{4}}[\pi_{\reg}]$ is just the same 7d U(1) partition function.

\paragraph{Leg boundary conditions}We call the boundary condition when the asymptotics of the three axes are specified, the \textit{leg boundary conditions}. In this case,
\bea
\bfK_{\bar{4}}^{\bd}&=\sum_{\scube\in\pi_{\bd}}\chi_{\bar{4},x}(\cube)\eqqcolon\bfN_{\bar{4},\lambda\mu\nu},%\\
%\bfK_{\bar{4}}^{\reg}&=\sum_{\scube \in\pi_{\reg}}\chi_{\bar{4},x}(\cube)\coloneqq \bfK_{\pi_{\reg}}
\eea
where we denoted the asymptotic Young diagrams explicitly. %Note that compared to the previous case, the coordinates of the boxes are assigned from the origin of the plane partition $\tilde{\pi}$. 
Writing down the explicit formula for this character is difficult because there are boxes at the intersection of the three legs. Let us introduce the following set of boxes
\bea
\mathcal{B}_{\lambda}&=\left\{(a,b,c)\mid a=1,\ldots,\infty,\,b=1,\ldots, \ell(\lambda),\,c=1,\ldots,\lambda_{b}\right\},\\
\mathcal{B}_{\mu}&=\left\{ (a,b,c) \mid b=1,\ldots,\infty,\,a=1,\ldots, \mu_{c},\,c=1,\ldots,\ell(\nu)  \right\},\\
\mathcal{B}_{\nu}&=\left\{(a,b,c) \mid c=1,\ldots,\infty,\, a=1,\ldots, \ell(\nu),\, b=1,\ldots, \nu_{a}   \right\}
\eea
%\bea
%\mathcal{B}_{\lambda}&=\left\{xq_{1}^{a-1}q_{2}^{b-1}q_{3}^{c-1}\mid a=1,\ldots,\infty,\,b=1,\ldots, \ell(\lambda),\,c=1,\ldots,\lambda_{b}\right\},\\
%\mathcal{B}_{\mu}&=\left\{ xq_{1}^{a-1}q_{2}^{b-1}q_{3}^{c-1} \mid b=1,\ldots,\infty,\,a=1,\ldots, \mu_{c},\,c=1,\ldots,\ell(\nu)  \right\},\\
%\mathcal{B}_{\nu}&=\left\{ xq_{1}^{a-1}q_{2}^{b-1}q_{3}^{c-1} \mid c=1,\ldots,\infty,\, a=1,\ldots, \ell(\nu),\, b=1,\ldots, \nu_{a}   \right\}
%\eea
where we identified the boxes with the $q$-coordinates. We also introduce
\bea
\mathcal{B}_{\lambda\cap \mu}=\mathcal{B}_{\lambda}\cap \mathcal{B}_{\mu},\quad \mathcal{B}_{\lambda\cap \nu}=\mathcal{B}_{\lambda}\cap \mathcal{B}_{\nu},\quad \mathcal{B}_{\mu\cap \nu}=\mathcal{B}_{\mu}\cap \mathcal{B}_{\nu},\quad \mathcal{B}_{\lambda\cap\mu\cap\nu}=\mathcal{B}_{\lambda}\cap\mathcal{B}_{\mu}\cap \mathcal{B}_{\nu}
\eea
and define
\bea
\mathcal{B}_{\lambda\mu\nu}\coloneqq (\mathcal{B}_{\lambda}+\mathcal{B}_{\mu}+\mathcal{B}_{\nu})- (\mathcal{B}_{\lambda\cap \mu}+\mathcal{B}_{\lambda\cap \nu}+\mathcal{B}_{\mu\cap \nu})+\mathcal{B}_{\lambda\cap\mu\cap\nu}.
\eea

\begin{itemize}
    \item One-leg boundary condition: Focusing on $\bfN_{\emptyset\emptyset\mu}$, we have
    \bea
    \bfN_{\bar{4},\emptyset\emptyset\nu}=\sum_{\scube\in\mathcal{B}_{\nu}}\chi_{\bar{4},x}(\cube)=\frac{\sum_{\Abox\in\nu}\chi_{12,x}(\Bbox)}{1-q_{3}}=\frac{\sum_{i=1}^{\ell(\nu)}\sum_{j=1}^{\nu_{i}}xq_{1}^{i-1}q_{2}^{j-1}}{1-q_{3}}.
    \eea
    \item Two-legs boundary condition: Focusing on $\bfN_{\lambda\mu\emptyset}$, we have
    \bea\label{eq:D6twolegscharacter}
    \bfN_{\bar{4},\lambda\mu\emptyset}&=\sum_{\scube\in\mathcal{B}_{\lambda}}\chi_{\bar{4},x}(\cube)+\sum_{\scube\in\mathcal{B}_{\mu}}\chi_{\bar{4},x}(\cube)-\sum_{\scube\in\mathcal{B}_{\lambda\cap\mu}}\chi_{\bar{4},x}(\cube)\\
    &=\frac{\sum_{k=1}^{\infty}\sum_{j=1}^{\lambda_{k}^{\rmT}}xq_{2}^{j-1}q_{3}^{k-1}}{1-q_{1}}+\frac{\sum_{k=1}^{\infty}\sum_{i=1}^{\mu_{k}}xq_{3}^{k-1}q_{1}^{i-1}}{1-q_{2}}-\sum_{k=1}^{\infty}\frac{(1-q_{1}^{\mu_{k}})(1-q_{2}^{\lambda^{\rmT}_{k}})}{\bfP_{12}}q_{3}^{k-1}x\\
    &=\frac{1}{\bfP_{12}}\sum_{k=1}^{\infty}q_{3}^{k-1}\left(1-q_{1}^{\mu_{k}}q_{2}^{\lambda^{\rmT}_{k}}\right)x.
    \eea
    %\eqref{eq:planelegboundcond}
    \item Three-legs boundary condition: We have
    \bea
    \bfN_{\bar{4},\lambda\mu\nu}=\sum_{\scube\in\mathcal{B}_{\lambda\mu\nu}}\chi_{\bar{4},x}(\cube).
    \eea
    For this case, writing down the explicit form is difficult.
\end{itemize}
For all cases, $\bfP_{123}^{\vee}\bfN_{\bar{4},\lambda\mu\nu}$ is a character with finite terms. This is because for large integers $N_{1,2,3}\gg 1$, we can write $\bfN_{\bar{4},\lambda\mu\nu}$ as
\bea
\bfN_{\bar{4},\lambda\mu\nu}=\sum_{\scube\in\pi_{\bd}\langle N_{1},N_{2},N_{3}\rangle}\chi_{\bar{4},x}(\cube)+\frac{\sum_{\Abox\in\lambda}\chi_{23,x}(\Bbox)}{1-q_{1}}q_{1}^{N_{1}}+\frac{\sum_{\Abox\in\mu}\chi_{13,x}(\Bbox)}{1-q_{2}}q_{2}^{N_{2}}+\frac{\sum_{\Abox\in\nu}\chi_{12,x}(\Bbox)}{1-q_{3}}q_{3}^{N_{3}}
    \eea
    where we denoted $\pi_{\bd}\langle N_{1},N_{2},N_{3}\rangle$ the set of boxes of $\pi_{\bd}\cap[1,N_{1}]\times [1,N_{2}]\times [1,N_{3}]$. We then have
    \bea
\bfP_{123}^{\vee}\bfN_{\bar{4},\lambda\mu\nu}&=\bfP_{123}^{\vee}\sum_{\scube\in\pi_{\bd}\langle N_{1},N_{2},N_{3}\rangle}\chi_{\bar{4},x}(\cube)-\bfP_{23}^{\vee}\sum_{\Abox\in\lambda}\chi_{23,x}(\Bbox)q_{1}^{N_{1}-1}\\
&-\bfP_{13}^{\vee}\sum_{\Abox\in\mu}\chi_{13,x}(\Bbox)q_{2}^{N_{2}-1}-\bfP_{12}^{\vee}\sum_{\Abox\in\nu}\chi_{12,x}(\Bbox)q_{3}^{N_{3}-1}.
\eea
All terms are now regularized properly and we only have finite number of terms.

The contour integral formulas are then given as the following, where we regularized the infinite terms properly.
\begin{itemize}

\item One-leg boundary condition:
\bea\label{eq:D6onelegcontourintegral}
\mathcal{Z}_{k}=\frac{\mathcal{G}^{k}}{k!} \oint \prod_{I=1}^{k}\frac{dx_{I}}{2\pi\iota x_{I}} \prod_{I<J}\mathcal{A}_{\mathbb{C}^{4}}\left(\frac{x_{I}}{x_{J}}\right)^{-1}\prod_{I=1}^{k}\mathscr{V}_{4}\left(\frac{x}{x_{I}}\right) \prod_{I=1}^{k}\prod_{\Abox \in\nu} g_{\bar{3}}\left(\frac{\chi_{12,x}(\Bbox)}{x_{I}}\right)^{-1}
\eea

\item Two-legs boundary condition:
\bea\label{eq:D6twolegscontourintegral}
\mathcal{Z}_{k}&=\frac{\mathcal{G}^{k}}{k!}\oint\prod_{I=1}^{k}\frac{dx_{I}}{2\pi\iota x_{I}}\prod_{I<J}\mathcal{A}_{\mathbb{C}^{4}}\left(\frac{x_{I}}{x_{J}}\right)^{-1}\prod_{I=1}^{k}\mathscr{V}_{4}\left(\frac{x}{x_{I}}\right)\prod_{I=1}^{k}\prod_{i=1}^{\max(\ell(\mu),\ell(\lambda^{\rmT}))}\mathscr{S}_{34}\left(\frac{q_{3}^{i-1}q_{1}^{\mu_{i}}q_{2}^{\lambda_{i}^{\rmT}}x}{x_{I}}\right)\mathscr{S}_{34}\left(\frac{q_{3}^{i-1}x}{x_{I}}\right)^{-1}%\\
%&\oint \prod_{I=1}^{k}\frac{dx_{I}}{2\pi\iota x_{I}}  \prod_{I>J}\mathcal{A}_{\mathbb{C}^{4}}\left(\frac{x_{I}}{x_{J}}\right)^{-1}\prod_{I=1}^{k}\mathscr{V}_{4}\left(\frac{x}{x_{I}}\right)\prod_{I=1}^{k}\prod_{l=1}^{\ell(\lambda^{\rmT})}\prod_{j=1}^{\lambda_{i}^{\rmT}}g_{\bar{1}}\left(\frac{q_{2}^{j-1}q_{3}^{i-1}x}{x_{I}}\right)^{-1}\\
%&\qquad \times \prod_{I=1}^{k}\prod_{i=1}^{\ell(\mu)}\prod_{j=1}^{\mu_{i}}g_{\bar{2}}\left(\frac{q_{3}^{i-1}q_{1}^{j-1}q_{2}^{\lambda_{i}^{\rmT}}x}{x_{I}}\right)^{-1}
\eea

\item Three-legs boundary condition:
\bea\label{eq:D6threelegscontourintegral}
\mathcal{Z}_{k}&=\frac{\mathcal{G}^{k}}{k!}\oint \prod_{I=1}^{k}\frac{dx_{I}}{2\pi\iota x_{I}} \prod_{I<J}\mathcal{A}_{\mathbb{C}^{4}}\left(\frac{x_{I}}{x_{J}}\right)^{-1} \prod_{I=1}^{k}\mathscr{V}_{4}\left(\frac{x}{x_{I}}\right) \\
&\qquad \times \prod_{I=1}^{k}\left(\prod_{\Abox \in\lambda}g_{\bar{1}}\left(\frac{\chi_{23,x}(\Bbox)}{x_{I}}\right)^{-1}\prod_{\Abox \in\mu}g_{\bar{2}}\left(\frac{\chi_{13,x}(\Bbox)}{x_{I}}\right)^{-1}\prod_{\Abox \in\nu}g_{\bar{3}}\left(\frac{\chi_{12,x}(\Bbox)}{x_{I}}\right)^{-1}\right)\\
&\qquad \times \prod_{I=1}^{k}\prod_{\scube \in \mathcal{S}}\mathcal{A}_{\mathbb{C}^{4}}\left(\frac{\chi_{\bar{4},x}(\cube)}{x_{I}}\right)
\eea 
where $\mathcal{S}_{\lambda\mu\nu}=-\mathcal{B}_{\lambda\cap\mu\cap\nu}+(\mathcal{B}_{\lambda\cap \mu}+\mathcal{B}_{\lambda\cap \nu}+\mathcal{B}_{\mu\cap \nu})$ is a finite set of boxes.
\end{itemize}
After evaluating the residues, the D6 partition function is given as
\bea\label{eq:D6legpartitionfunction}
\mathcal{Z}&=\sum_{\pi_{\reg}\in\mathcal{PP}_{\lambda\mu\nu}}\mathfrak{q}^{|\pi_{\reg}|}\widetilde{\mathcal{Z}}^{\D6}_{\bar{4};\lambda\mu\nu}[\pi_{\reg}],\\
\widetilde{\mathcal{Z}}^{\D6}_{\bar{4};\lambda\mu\nu}[\pi_{\reg}]&=\mathbb{I}\left[-\bfP_{4}^{\vee}\bfN_{\bar{4}}^{\vee}\bfK_{\pi_{\reg}}+\bfP_{\four}\bfN_{\lambda\mu\nu}^{\vee}\bfK_{\pi_{\reg}}+\bfP_{123}^{\vee}\bfK_{\pi_{\reg}}^{\vee}\bfK_{\pi_{\reg}}\right].
\eea

\begin{remark}
    We note that the partition function $\widetilde{\mathcal{Z}}^{\D6}_{\bar{4};\lambda\mu\nu}[\pi]$ defined above slightly differs from the \textit{vertex term} defined in \cite{Nekrasov:2014nea}. The difference comes from the definition of $\bfN_{\bar{4},\lambda\mu\nu}$. The boundary contribution of \cite{Nekrasov:2014nea} is defined as
    \bea
\bfN^{\text{NO}}_{\bar{4},\lambda\mu\nu}=\frac{\sum_{\Abox\in\lambda}\chi_{23,x}(\Bbox)}{1-q_{1}}+\frac{\sum_{\Abox\in\mu}\chi_{13,x}(\Bbox)}{1-q_{2}}+\frac{\sum_{\Abox\in\nu}\chi_{12,x}(\Bbox)}{1-q_{3}}
    \eea
    which is symmetric in the three legs. When there is only one leg, we have $\bfN_{\bar{4},\emptyset\emptyset\nu}^{\text{NO}}=\bfN_{\bar{4},\emptyset\emptyset\nu}$. The difference comes from the boxes living at the intersection of the three legs. For example, for two legs we have
    \bea
    \bfN_{\bar{4},\lambda\mu\emptyset}^{\text{NO}}-\bfN_{\bar{4},\lambda\mu\emptyset}=\frac{1}{\bfP_{12}}\sum_{i=1}^{\infty}q_{3}^{i-1}(1-q_{1}^{\mu_{i}})(1-q_{2}^{\lambda_{i}^{\rmT}})x.
    \eea
    The difference will appear as overall factors only depending on the boundary conditions and thus it is not important if we are interested in the instanton contributions coming from each vertex term. However, they will play roles when we study the gluings of the vertex terms to obtain the DT invariants of toric Calabi--Yau three-folds. Depending on the definition of this vertex contribution, we need to define a proper \textit{edge} contribution. After this procedure, the result will not depend on the above decomposition.

    In the later sections, we will see that the $qq$-characters reproduce the partition functions obtained by using our definition $\bfN_{\bar{4},\lambda\mu\nu}$. Gluings of these $qq$-characters are dubbed as the \textbf{web of BPS $qq$-characters} and will be discussed in a future work \cite{Kimura-Noshita1}.
\end{remark}

\paragraph{Tetrahedron instantons}
Generally, using the quadrality of $q_{1,2,3,4}$, we can define partition functions for theories on $\mathbb{C}^{3}_{\bar{a}}\times \mathbb{S}^{1}\,\,(a\in\four)$. We denote them such as
\bea
\mathcal{Z}=\begin{dcases}
    \sum_{\pi_{\reg}\in\mathcal{PP}}\mathfrak{q}^{|\pi_{\reg}|}\widetilde{\mathcal{Z}}_{\bar{a}}^{\D6}[\pi_{\reg}]\qquad \text{surface boundary cond.}\\
    \sum_{\pi_{\reg}\in\mathcal{PP}_{\lambda\mu\nu}}\mathfrak{q}^{|\pi_{\reg}|}\widetilde{\mathcal{Z}}^{\D6}_{\bar{a};\lambda\mu\nu}[\pi_{\reg}]\qquad \text{leg boundary cond.}
\end{dcases}
\eea
If $\bar{a}=bcd\,\,(b<c<d)$, the three Young diagrams $\lambda,\mu,\nu$ are assigned to the legs $q_{b},q_{c},q_{d}$, respectively. From now on, when discussing the finite part of the plane partition (the set of yellow boxes), we omit the subscript ``$\reg$" and simply denote it as $\pi$. 

We furthermore can also consider the tetrahedron instanton generalization of this setup by doing the same procedure with the character given in \eqref{eq:D6boundarycharactergeneral}, so we omit the explicit expressions (see \cite{Nekrasov:2023nai} for example).

\subsection{D8 partition functions}\label{sec:D8partitionfunction}
\paragraph{Magnificent four and finite solid partitions}
Let us review the magnificent four setup. Consider the following character
\bea
\mathbf{V}=\frac{\mathbf{Y}^{\vee}\mathbf{Y}}{\bfP_{\four}},\quad \bfY=\bfN-\bfK,\quad \bfN=\sum_{\alpha=1}^{n}(1-K_{\alpha})v_{\alpha},\quad \bfK=\sum_{I=1}^{k}x_{I}.
\eea
The parameters $\{K_{\alpha}\}_{\alpha=1}^{n}$ represent the distance between each pair of $\D8\tbar\overline{\D8}$ branes. Expanding $\mathbf{V}$, we have 
\bea
\mathbf{V}=\mathbf{V}_{\text{pert.}}+\mathbf{V}_{\text{inst.}},\quad \mathbf{V}_{\text{pert.}}=\frac{\mathbf{N}^{\vee}\mathbf{N}}{\bfP_{\four}},\quad \mathbf{V}_{\text{inst.}}=-\bfN^{\vee}\bfK-\bfK^{\vee}\bfN+\bfP_{\four}\bfK^{\vee}\bfK.
\eea
The instanton partition function comes from the square root part as
\bea
\mathbf{V}_{\text{inst.}}=\mathbf{v}_{\text{inst.}}+\mathbf{v}_{\text{inst.}}^{\vee},\quad \mathbf{v}_{\text{inst.}}=-\bfN^{\vee}\bfK+\bfP_{123}^{\vee}\bfK^{\vee}\bfK
\eea
where we chose a specific square root part of $\bfP_{\four}\bfK^{\vee}\bfK$. The contour integral formula is then given as
\bea
\mathcal{Z}_{k}^{\D8}= \frac{\mathcal{G}^{k}}{k!}\oint \prod_{I=1}^{k}\frac{dx_{I}}{2\pi\iota x_{I}}\prod_{I=1}^{k}\prod_{\alpha=1}^{n}\frac{1-K_{\alpha}v_{\alpha}/x_{I}}{1-v_{\alpha}/x_{I}}\prod_{I<J}\mathcal{A}_{\mathbb{C}^{4}}\left(\frac{x_{I}}{x_{J}}\right)^{-1}
\eea
The poles are then classified by \textit{finite} solid partitions:
\bea
\bfK|_{\vec{\rho}}=\sum_{\alpha=1}^{n}\sum_{\shcube\in\rho^{(\alpha)}}\chi_{\four,v_{\alpha}}(\hcube),\quad \chi_{\four,v}(\hcube)=vq_{1}^{i-1}q_{2}^{j-1}q_{3}^{k-1}q_{4}^{l-1},
\eea
where we denoted $\rho^{(\alpha)}$ as solid partitions obeying the condition
\bea\label{eq:solidpartitionmelting}
\rho^{(\alpha)}=\{\rho^{(\alpha)}_{i,j,k}\in\mathbb{Z}_{\geq 0}\},\quad \rho^{(\alpha)}_{i,j,k}\geq \rho^{(\alpha)}_{i+1,j,k},\quad \rho^{(\alpha)}_{i,j,k}\geq \rho^{(\alpha)}_{i,j+1,k},\quad \rho^{(\alpha)}_{i,j,k}\geq \rho^{(\alpha)}_{i,j,k+1}.
\eea
The partition function is then given as
\bea
\mathcal{Z}_{\text{inst.}}^{\D8}=\sum_{\vec{\rho}}\mathfrak{q}^{|\vec{\rho}|}(-1)^{\sigma_{4}(\vec{\rho})}\mathcal{Z}_{\four;4}^{\D8}[\vec{v},\vec{\rho};\vec{K}],\quad \mathcal{Z}_{\four;4}^{\D8}[\vec{v},\vec{\rho};\vec{K}]=\mathbb{I}[\mathbf{v}_{\text{inst.}}|_{\vec{\rho}}], 
\eea
where $(-1)^{\sigma_{4}(\vec{\rho})}\in\mathbb{Z}_{2}$ is a sign factor determined by the solid partitions. The subindex $4$ means that we chose the $4$ direction to be a special direction determining the square root part.\footnote{ Instead of the sign rule \eqref{eq:signfactor}, we may choose a different sign rule using the quadrality as
\bea
\sigma_{i}(\rho)=\#\{(x_{1},x_{2},x_{3},x_{4})\in\rho\mid x_{j}=x_{k}=x_{l}<x_{i}\}.
\eea
Similarly, we can define $\mathcal{Z}^{\D8}_{\four;a}[\rho;K]$ by using a different square root part $\bfP_{\bar{a}}^{\vee}\bfK^{\vee}\bfK$ in the definition of the character. Actually, the partition function itself does not depend on the choice of the square root part as long as the sign factor is chosen properly (see for example \cite[Thm.~2.8]{Monavari:2022rtf} for the proof):
\bea
(-1)^{\sigma_{a}(\rho)}\mathcal{Z}^{\D8}_{\bar{4};a}[\rho;K]=(-1)^{\sigma_{b}(\rho)}\mathcal{Z}_{\four;b}^{\D8}[\rho;K].
\eea
}

Focusing on the rank one $\U(1|1)$ case, we have
\bea\label{eq:D8U1partitionfunction}
\mathcal{Z}=\sum_{\rho\in\mathcal{SP}}\mathfrak{q}^{|\rho|}(-1)^{\sigma_{4}(\rho)}\mathcal{Z}_{\four;4}^{\D8}[\rho;K],\quad \mathcal{Z}^{\D8}_{\four;4}[\rho\,;K]&=\prod_{\shcube\in\rho}\frac{(1-Kx/\chi_{\four,x}(\shcube))}{(1-x/\chi_{\four,x}(\shcube))}\prod_{\shcube,\shcube'\in\rho}g_{\bar{4}}\left(\frac{\chi_{\four,x}(\hcube)}{\chi_{\four,x}(\hcube')}\right)^{-1},
\eea
where $\mathcal{SP}$ denotes the set of solid partitions and the sign factor we use in this paper is defined as\footnote{The sign rule here is the same with the one used in \cite{Cao:2019tvv,Monavari:2022rtf}. A different sign rule was proposed in \cite{Nekrasov:2018xsb} which is different from the one used in this paper. This comes from the fact that the square root used here is different from the one in \cite{Nekrasov:2018xsb}.
}
\bea\label{eq:signfactor}
    \sigma_{4}(\rho)=\#\left\{(i,i,i,j)\in\rho\mid i< j\right\}
\eea
Under this notation, actually one can show that the partition function has a nice plethystic exponential expression \cite{Nekrasov:2017cih,Nekrasov:2018xsb} as
\bea
\mathcal{Z}=\PE\left[\frac{-q_{4}\prod_{i=1}^{3}(1-q_{i4}^{-1})}{\prod_{a\in\four}(1-q_{a})}\frac{1-K^{-1}}{(1-\mathfrak{q})(1-K^{-1}\mathfrak{q}^{-1})}\right],
\eea
where the plethystic exponential operator $\PE$ is defined as
\bea
\PE[f(x_{1},\ldots,x_{n})]=\exp\left(\sum_{\ell=1}^{\infty}\frac{1}{\ell}f(x_{1}^{\ell},\ldots,x_{n}^{\ell})\right).
\eea

Following the previous D4, D6 partition functions, we will generalize the story to the case when we have nontrivial boundary conditions at the asymptotics of the solid partition.

%%%%%%%%%%%%%%%%%%%%%%%%%%%%%%%%%%%%%%%%%%%%%%%%%%%%%%%%%%%%%%%%%%%%%%%%%%%%%%%%%%%%%%%%%%%%%%%%%%%%%%%%%%%%%%%%%%%%%%%%%%%%%%%%
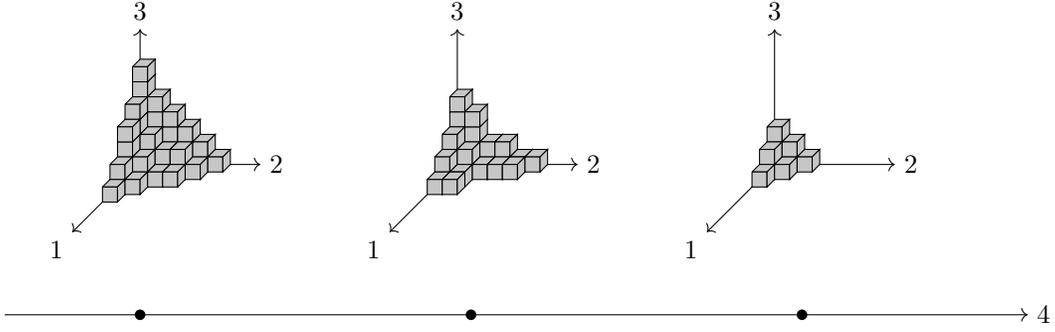
\begin{figure}
    \centering
    \begin{tikzpicture}[scale=0.4]
    \draw[->] (-4,-5)--(30,-5);
    \node[right] at (30,-5){$4$};
    \fill (0.5,-5) circle (5pt);
    \fill (11.5,-5) circle (5pt);
    \fill (22.5,-5) circle (5pt);
    
    \begin{scope}[scale=0.5]
        \draw[->] (1,0)--(9,0);
        \node[right] at (9,0){$2$};
\draw[->] (1,0)--(1,9);
\node[above] at (1,9){$3$};
\draw[->] (1,0)--(-3.5,-4.5);
\node[below left] at (-3.5,-4.5) {$1$};

 \foreach \z in {1,2,3,4,5,6,7} \diagcube{1}{1}{\z};
 \foreach \z in {1,2,3,4,5} \diagcube{2}{1}{\z};
 \foreach \z in {1,2,3,4} \diagcube{3}{1}{\z};
 \foreach \z in {1,2} \diagcube{4}{1}{\z};
\foreach \z in {1} \diagcube{5}{1}{\z};

\foreach \z in {1,2,3,4,5} \diagcube{1}{2}{\z};
 \foreach \z in {1,2,3} \diagcube{2}{2}{\z};
 \foreach \z in {1,2} \diagcube{3}{2}{\z};
 \foreach \z in {1} \diagcube{4}{2}{\z};

\foreach \z in {1,2,3,4} \diagcube{1}{3}{\z};
 \foreach \z in {1,2} \diagcube{2}{3}{\z};
 \foreach \z in {1} \diagcube{3}{3}{\z};

\foreach \z in {1,2,3} \diagcube{1}{4}{\z};
 \foreach \z in {1,2} \diagcube{2}{4}{\z};
 \foreach \z in {1} \diagcube{3}{4}{\z};

\foreach \z in {1,2} \diagcube{1}{5}{\z};
 \foreach \z in {1} \diagcube{2}{5}{\z};

\foreach \z in {1} \diagcube{1}{6}{\z};
\end{scope}
\begin{scope}[xshift=600,scale=0.5]
        \draw[->] (1,0)--(9,0);
        \node[right] at (9,0){$2$};
\draw[->] (1,0)--(1,9);
\node[above] at (1,9){$3$};
\draw[->] (1,0)--(-3.5,-4.5);
\node[below left] at (-3.5,-4.5) {$1$};

 \foreach \z in {1,2,3} \diagcube{1}{1}{\z};
 \foreach \z in {1,2} \diagcube{2}{1}{\z};
 \foreach \z in {1} \diagcube{3}{1}{\z};

\foreach \z in {1,2} \diagcube{1}{2}{\z};
 \foreach \z in {1} \diagcube{2}{2}{\z};
 
\foreach \z in {1} \diagcube{1}{3}{\z};
 
\end{scope}

\begin{scope}[xshift=300,scale=0.5]
        \draw[->] (1,0)--(9,0);
        \node[right] at (9,0){$2$};
\draw[->] (1,0)--(1,9);
\node[above] at (1,9){$3$};
\draw[->] (1,0)--(-3.5,-4.5);
\node[below left] at (-3.5,-4.5) {$1$};

 \foreach \z in {1,2,3,4,5} \diagcube{1}{1}{\z};
 \foreach \z in {1,2,3} \diagcube{2}{1}{\z};
 \foreach \z in {1,2} \diagcube{3}{1}{\z};
 \foreach \z in {1} \diagcube{4}{1}{\z};

\foreach \z in {1,2,3,4} \diagcube{1}{2}{\z};
 \foreach \z in {1,2} \diagcube{2}{2}{\z};
 \foreach \z in {1} \diagcube{3}{2}{\z};
 \foreach \z in {1} \diagcube{4}{2}{\z};

\foreach \z in {1,2} \diagcube{1}{3}{\z};
 \foreach \z in {1} \diagcube{2}{3}{\z};

\foreach \z in {1,2} \diagcube{1}{4}{\z};
 \foreach \z in {1} \diagcube{2}{4}{\z};

\foreach \z in {1} \diagcube{1}{5}{\z};
 \foreach \z in {1} \diagcube{2}{5}{\z};

\foreach \z in {1} \diagcube{1}{6}{\z};

\end{scope}

\end{tikzpicture}
    \caption{Decomposition of solid partitions. The solid partition extends in the 4-axis and for each layer orthogonal to the 4-axis, we have plane partitions. The plane partitions get smaller along the 4-axis which comes from the condition  \eqref{eq:solidpartitionmelting} of the solid partition.}
    \label{fig:solid-partition}
\end{figure}
\paragraph{Solid partitions with infinite size}To visualize the solid partition, we decompose it into non-increasing sequences of plane partitions% (see \cite[Fig.~5]{Kimura:2023bxy} for example)
\bea
\rho=(\Pi^{(1)},\Pi^{(2)},\ldots,),\quad \Pi^{(i)}\succeq \Pi^{(i+1)}
\eea
where $\Pi^{(i)}$ are plane partitions. We thus have the condition
\bea
(i,j,k)\in \Pi^{(i+1)}\Rightarrow (i,j,k)\in\Pi^{(i)}.
\eea
This is the $(1,3)$-type description (see \cite[Sec.~2]{Kimura:2023bxy} for a summary). Depending on which axis we take for the decomposition, we have four possible descriptions. In this paper, we fix the 4-direction to be this direction.

The above discussion is even true when the solid partition has nontrivial boundary conditions. We denote such solid partition as $\tilde{\rho}$. Let us study the possible nontrivial boundary conditions. To simplify the figures, we use the following figures to denote the boundary conditions of the plane partitions:
\bea\label{eq:rod-surface-figure}
\adjustbox{valign=c}{
\begin{tikzpicture}[scale=0.2]
%\cube{1}{1}{1};
% horizontal y axis (the origin is at (1,0))
\draw[->] (1,0)--(13,0);
\node[right] at (13,0){$2$};
\draw[->] (1,0)--(1,9);
\node[above] at (1,9){$3$};
\draw[->] (1,0)--(-5.5,-6.5);
\node[below left] at (-5.5,-6.5){$1$};
%\diagcube{1}{1}{1};
%\zstick{1}{1}{1};
%\ystick{1}{2}{1};
%\xstickvar{1}{1}{1}{3};

%\begin{scope}[xscale=0.5,yscale=0.5]
\zstickvar{1}{1}{1}{7};
\zstickvar{2}{1}{1}{7};

\ystickvar{1}{2}{1}{9};
\ystickvar{1}{2}{2}{9};
\ystickvar{1}{2}{3}{9};
\ystickvar{2}{2}{1}{9};
\ystickvar{2}{2}{2}{9};

\xstickvar{3}{1}{1}{4};
\xstickvar{3}{1}{2}{4};
\xstickvar{3}{1}{3}{4};
\xstickvar{3}{2}{1}{4};
\xstickvar{3}{2}{2}{4};
\xstickvar{3}{3}{1}{4};

%\end{scope}

%\diagcube{2}{1}{1};

%\diagcube{1}{2}{1};
%\diagcube{1}{1}{2};
\end{tikzpicture}}\qquad 
\adjustbox{valign=c}{
\begin{tikzpicture}[scale=0.2]
%\cube{1}{1}{1};
% horizontal y axis (the origin is at (1,0))
\draw[->] (1,0)--(13,0);
\draw[->] (1,0)--(1,9);
\draw[->] (1,0)--(-5.5,-6.5);
\node[right] at (13,0){$2$};
\node[above] at (1,9){$3$};
\node[below left] at (-5.5,-6.5){$1$};

\zxsurfacevar{1}{1}{1}{7}{6};
\zxsurfacevar{1}{2}{1}{7}{6};
\zxsurfacevar{1}{3}{1}{7}{6};

\xysurfacevar{1}{4}{1}{6}{8};
\xysurfacevar{1}{4}{2}{6}{8};

\yzsurfacevar{1}{4}{3}{8}{5};
\yzsurfacevar{1}{4}{3}{8}{5};

\end{tikzpicture}}
\eea
The leg-boundary conditions in the plane partition are colored in blue and the surface boundary conditions in the plane partition are colored in green. Note that infinite number of boxes extending in one direction becomes the blue \textit{rod} configuration. On the other hand, infinite number of boxes extending in two directions become the green \textit{surface} configuration.

We have the following three possible boundary conditions depending on how the infinite number of boxes extends. We call the boundary condition when infinite number of boxes extends in an one-dimensional way the leg boundary condition. When infinite number of boxes extends in a two (three)-dimensional way, we call it the (hyper)surface boundary condition. Generally we can have multiple boundary conditions as long as they are compatible and obey the solid partition condition \eqref{eq:solidpartitionmelting}. In this paper, we only focus on the case when there are only leg, surface, hypersurface boundary conditions but generalizations are straightforward.
\begin{itemize}
    \item Leg boundary conditions: This is a configuration where we have asymptotic plane partitions denoted as $\pi_{1,2,3,4}$ in the four legs of the solid partition. Namely, the solid partition $\tilde{\rho}$ obeys the boundary conditions:
    \bea
    &\tilde{\rho}_{i,j,k} =\infty,\quad (i,j,k)\in \pi_{4},\\
    &\pi_{1}=\{\tilde{\rho}_{\infty,j,k}\},\quad \pi_{2}=\{\tilde{\rho}_{i,\infty,k}\},\quad \pi_{3}=\{\tilde{\rho}_{i,j,\infty}\}.
    \eea
    We denote the set of possible solid partitions with the leg boundary conditions $\pi_{1,2,3,4}$ as $\mathcal{SP}_{\pi_{1}\pi_{2}\pi_{3}\pi_{4}}$. 
    
In the $(1,3)$-type description, the asymptotic plane partitions of the legs $1,2,3$ are visualized as
    \bea\label{eq:fig-solidpartitionleg123}
    \begin{tikzpicture}[scale=0.3]
    \draw[->] (-4,-5)--(30,-5);
    \node[right] at (30,-5){$4$};
    \fill (0.5,-5) circle (5pt);
    \fill (11.5,-5) circle (5pt);
    \fill (22.5,-5) circle (5pt);
    
    \begin{scope}[scale=0.45]
    \draw[->] (1,0)--(13,0);
\node[right] at (13,0){$2$};
\draw[->] (1,0)--(1,9);
\node[above] at (1,9){$3$};
\draw[->] (1,0)--(-5.5,-6.5);
\node[below left] at (-5.5,-6.5){$1$};
%\diagcube{1}{1}{1};
%\zstick{1}{1}{1};
%\ystick{1}{2}{1};
%\xstickvar{1}{1}{1}{3};

%\begin{scope}[xscale=0.5,yscale=0.5]
\zstickvar{1}{1}{1}{7};
\zstickvar{2}{1}{1}{7};

\ystickvar{1}{2}{1}{9};
\ystickvar{1}{2}{2}{9};
\ystickvar{1}{2}{3}{9};
\ystickvar{2}{2}{1}{9};
\ystickvar{2}{2}{2}{9};

\xstickvar{3}{1}{1}{4};
\xstickvar{3}{1}{2}{4};
\xstickvar{3}{1}{3}{4};
\xstickvar{3}{2}{1}{4};
\xstickvar{3}{2}{2}{4};
\xstickvar{3}{3}{1}{4};

\end{scope}
\begin{scope}[xshift=300,scale=0.45]
        \draw[->] (1,0)--(13,0);
\node[right] at (13,0){$2$};
\draw[->] (1,0)--(1,9);
\node[above] at (1,9){$3$};
\draw[->] (1,0)--(-5.5,-6.5);
\node[below left] at (-5.5,-6.5){$1$};
%\diagcube{1}{1}{1};
%\zstick{1}{1}{1};
%\ystick{1}{2}{1};
%\xstickvar{1}{1}{1}{3};

%\begin{scope}[xscale=0.5,yscale=0.5]
\zstickvar{1}{1}{1}{7};
\zstickvar{2}{1}{1}{7};

\ystickvar{1}{2}{1}{9};
\ystickvar{1}{2}{2}{9};
\ystickvar{1}{2}{3}{9};
\ystickvar{2}{2}{1}{9};
%\ystickvar{2}{2}{2}{9};

\xstickvar{3}{1}{1}{4};
\xstickvar{3}{1}{2}{4};
%\xstickvar{3}{1}{3}{4};
\xstickvar{3}{2}{1}{4};
%\xstickvar{3}{2}{2}{4};
%\xstickvar{3}{3}{1}{4};

\end{scope}

\begin{scope}[xshift=600,scale=0.45]
        \draw[->] (1,0)--(13,0);
\node[right] at (13,0){$2$};
\draw[->] (1,0)--(1,9);
\node[above] at (1,9){$3$};
\draw[->] (1,0)--(-5.5,-6.5);
\node[below left] at (-5.5,-6.5){$1$};
%\diagcube{1}{1}{1};
%\zstick{1}{1}{1};
%\ystick{1}{2}{1};
%\xstickvar{1}{1}{1}{3};

%\begin{scope}[xscale=0.5,yscale=0.5]
\zstickvar{1}{1}{1}{7};
%\zstickvar{2}{1}{1}{7};

\ystickvar{1}{2}{1}{9};
\ystickvar{1}{2}{2}{9};
%\ystickvar{1}{2}{3}{9};
%\ystickvar{2}{2}{1}{9};
%\ystickvar{2}{2}{2}{9};

\xstickvar{2}{1}{1}{5};
\xstickvar{2}{1}{2}{5};
%\xstickvar{3}{1}{3}{4};
%\xstickvar{3}{2}{1}{4};
%\xstickvar{3}{2}{2}{4};
%\xstickvar{3}{3}{1}{4};

\end{scope}
\end{tikzpicture}
    \eea
Namely, for each layer orthogonal to the 4-axis, we have a plane partition with asymptotic Young diagrams and the asymptotic Young diagrams obey the non-increasing condition along the 4th direction. 

For the case when we have an asymptotic plane partition at the leg $4$, it is visualize as
\bea\label{eq:fig-solidpartitionleg4}
    \begin{tikzpicture}[scale=0.33]
    \draw[->] (-4,-5)--(30,-5);
    \node[right] at (30,-5){$4$};
    \fill (0.5,-5) circle (5pt);
    \fill (11.5,-5) circle (5pt);
    \fill (22.5,-5) circle (5pt);
    
    \begin{scope}[scale=0.5]
        \draw[->] (1,0)--(9,0);
        \node[right] at (9,0){$2$};
\draw[->] (1,0)--(1,9);
\node[above] at (1,9){$3$};
\draw[->] (1,0)--(-3.5,-4.5);
\node[below left] at (-3.5,-4.5) {$1$};

 \foreach \z in {1,2,3,4,5,6,7} \diagcubecolor{1}{1}{\z};
 \foreach \z in {1,2,3,4,5} \diagcubecolor{2}{1}{\z};
 \foreach \z in {1,2,3,4} \diagcubecolor{3}{1}{\z};
 \foreach \z in {1,2} \diagcubecolor{4}{1}{\z};
\foreach \z in {1} \diagcubecolor{5}{1}{\z};

\foreach \z in {1,2,3,4,5} \diagcubecolor{1}{2}{\z};
 \foreach \z in {1,2,3} \diagcubecolor{2}{2}{\z};
 \foreach \z in {1,2} \diagcubecolor{3}{2}{\z};
 \foreach \z in {1} \diagcubecolor{4}{2}{\z};

\foreach \z in {1,2,3,4} \diagcubecolor{1}{3}{\z};
 \foreach \z in {1,2} \diagcubecolor{2}{3}{\z};
 \foreach \z in {1} \diagcubecolor{3}{3}{\z};

\foreach \z in {1,2,3} \diagcubecolor{1}{4}{\z};
 \foreach \z in {1,2} \diagcubecolor{2}{4}{\z};
 \foreach \z in {1} \diagcubecolor{3}{4}{\z};

\foreach \z in {1,2} \diagcubecolor{1}{5}{\z};
 \foreach \z in {1} \diagcubecolor{2}{5}{\z};

\foreach \z in {1} \diagcubecolor{1}{6}{\z};
\end{scope}
\begin{scope}[xshift=600,scale=0.5]
        \draw[->] (1,0)--(9,0);
        \node[right] at (9,0){$2$};
\draw[->] (1,0)--(1,9);
\node[above] at (1,9){$3$};
\draw[->] (1,0)--(-3.5,-4.5);
\node[below left] at (-3.5,-4.5) {$1$};

 \foreach \z in {1,2,3,4,5,6,7} \diagcubecolor{1}{1}{\z};
 \foreach \z in {1,2,3,4,5} \diagcubecolor{2}{1}{\z};
 \foreach \z in {1,2,3,4} \diagcubecolor{3}{1}{\z};
 \foreach \z in {1,2} \diagcubecolor{4}{1}{\z};
\foreach \z in {1} \diagcubecolor{5}{1}{\z};

\foreach \z in {1,2,3,4,5} \diagcubecolor{1}{2}{\z};
 \foreach \z in {1,2,3} \diagcubecolor{2}{2}{\z};
 \foreach \z in {1,2} \diagcubecolor{3}{2}{\z};
 \foreach \z in {1} \diagcubecolor{4}{2}{\z};

\foreach \z in {1,2,3,4} \diagcubecolor{1}{3}{\z};
 \foreach \z in {1,2} \diagcubecolor{2}{3}{\z};
 \foreach \z in {1} \diagcubecolor{3}{3}{\z};

\foreach \z in {1,2,3} \diagcubecolor{1}{4}{\z};
 \foreach \z in {1,2} \diagcubecolor{2}{4}{\z};
 \foreach \z in {1} \diagcubecolor{3}{4}{\z};

\foreach \z in {1,2} \diagcubecolor{1}{5}{\z};
 \foreach \z in {1} \diagcubecolor{2}{5}{\z};

\foreach \z in {1} \diagcubecolor{1}{6}{\z};
\end{scope}

\begin{scope}[xshift=300,scale=0.5]
        \draw[->] (1,0)--(9,0);
        \node[right] at (9,0){$2$};
\draw[->] (1,0)--(1,9);
\node[above] at (1,9){$3$};
\draw[->] (1,0)--(-3.5,-4.5);
\node[below left] at (-3.5,-4.5) {$1$};

 \foreach \z in {1,2,3,4,5,6,7} \diagcubecolor{1}{1}{\z};
 \foreach \z in {1,2,3,4,5} \diagcubecolor{2}{1}{\z};
 \foreach \z in {1,2,3,4} \diagcubecolor{3}{1}{\z};
 \foreach \z in {1,2} \diagcubecolor{4}{1}{\z};
\foreach \z in {1} \diagcubecolor{5}{1}{\z};

\foreach \z in {1,2,3,4,5} \diagcubecolor{1}{2}{\z};
 \foreach \z in {1,2,3} \diagcubecolor{2}{2}{\z};
 \foreach \z in {1,2} \diagcubecolor{3}{2}{\z};
 \foreach \z in {1} \diagcubecolor{4}{2}{\z};

\foreach \z in {1,2,3,4} \diagcubecolor{1}{3}{\z};
 \foreach \z in {1,2} \diagcubecolor{2}{3}{\z};
 \foreach \z in {1} \diagcubecolor{3}{3}{\z};

\foreach \z in {1,2,3} \diagcubecolor{1}{4}{\z};
 \foreach \z in {1,2} \diagcubecolor{2}{4}{\z};
 \foreach \z in {1} \diagcubecolor{3}{4}{\z};

\foreach \z in {1,2} \diagcubecolor{1}{5}{\z};
 \foreach \z in {1} \diagcubecolor{2}{5}{\z};

\foreach \z in {1} \diagcubecolor{1}{6}{\z};
\end{scope}
\end{tikzpicture}
\eea
Namely, we have a fixed plane partition for each layer and it is extending in the 4-direction semi-infinitely. For this type, the boxes corresponding to the boundary condition is colored in orange. Combining the above two cases, the most general configuration when the boundaries have four generic asymptotic plane partitions is decomposed into non-increasing plane partitions with asymptotic Young diagrams: 
\bea\label{eq:fig-solidpartitionleg1234}
 \begin{tikzpicture}[scale=0.33]
    \draw[->] (-4,-5)--(30,-5);
    \node[right] at (30,-5){$4$};
    \fill (0.5,-5) circle (5pt);
    \fill (11.5,-5) circle (5pt);
    \fill (22.5,-5) circle (5pt);
    
    \begin{scope}[scale=0.45]
        \draw[->] (1,0)--(12,0);
        \node[right] at (12,0){$2$};
\draw[->] (1,0)--(1,12);
\node[above] at (1,12){$3$};
\draw[->] (1,0)--(-5.5,-6.5);
\node[below left] at (-5.5,-6.5) {$1$};

%\ystickvar{1}{7}{1}{3};
%\ystickvar{1}{2}{2}{9};
%\ystickvar{1}{2}{3}{9};
%\ystickvar{2}{2}{1}{9};
%\ystickvar{2}{2}{2}{9};

%\xstickvar{3}{1}{1}{4};
%\xstickvar{3}{1}{2}{4};
%\xstickvar{3}{1}{3}{4};
%\xstickvar{3}{2}{1}{4};
%\xstickvar{3}{2}{2}{4};
%\xstickvar{3}{3}{1}{4};

 \foreach \z in {1,2,3,4,5,6,7} \diagcubecolor{1}{1}{\z};
 \zstickvar{1}{1}{8}{3};
 \foreach \z in {1,2,3,4,5} \diagcubecolor{2}{1}{\z};
\zstickvar{2}{1}{6}{5};
 
 \foreach \z in {1,2,3,4} \diagcubecolor{3}{1}{\z};
 \foreach \z in {1,2} \diagcubecolor{4}{1}{\z};
\foreach \z in {1} \diagcubecolor{5}{1}{\z};
\xstickvar{6}{1}{1}{3};
\xstickvar{5}{1}{2}{3.5};
\xstickvar{4}{1}{3}{4};

\foreach \z in {1,2,3,4,5} \diagcubecolor{1}{2}{\z};
 \foreach \z in {1,2,3} \diagcubecolor{2}{2}{\z};
 \foreach \z in {1,2} \diagcubecolor{3}{2}{\z};
 \foreach \z in {1} \diagcubecolor{4}{2}{\z};
\xstickvar{5}{2}{1}{3.5};
\xstickvar{4}{2}{2}{4};

\foreach \z in {1,2,3,4} \diagcubecolor{1}{3}{\z};
 \foreach \z in {1,2} \diagcubecolor{2}{3}{\z};
 \foreach \z in {1} \diagcubecolor{3}{3}{\z};
%\xstickvar{3}{3}{1}{4};

\foreach \z in {1,2,3} \diagcubecolor{1}{4}{\z};
 \foreach \z in {1,2} \diagcubecolor{2}{4}{\z};
 \foreach \z in {1} \diagcubecolor{3}{4}{\z};

\foreach \z in {1,2} \diagcubecolor{1}{5}{\z};
 \foreach \z in {1} \diagcubecolor{2}{5}{\z};

\foreach \z in {1} \diagcubecolor{1}{6}{\z};

\ystickvar{1}{7}{1}{3};
\ystickvar{1}{6}{2}{4};
\ystickvar{1}{5}{3}{5};
\ystickvar{2}{6}{1}{4};
\ystickvar{2}{5}{2}{5};

\end{scope}
\begin{scope}[xshift=300,scale=0.45]
             \draw[->] (1,0)--(12,0);
        \node[right] at (12,0){$2$};
\draw[->] (1,0)--(1,12);
\node[above] at (1,12){$3$};
\draw[->] (1,0)--(-5.5,-6.5);
\node[below left] at (-5.5,-6.5) {$1$};

 \foreach \z in {1,2,3,4,5,6,7} \diagcubecolor{1}{1}{\z};
 \zstickvar{1}{1}{8}{3};
 \foreach \z in {1,2,3,4,5} \diagcubecolor{2}{1}{\z};
\zstickvar{2}{1}{6}{5};
 
 \foreach \z in {1,2,3,4} \diagcubecolor{3}{1}{\z};
 \foreach \z in {1,2} \diagcubecolor{4}{1}{\z};
\foreach \z in {1} \diagcubecolor{5}{1}{\z};
\xstickvar{6}{1}{1}{3};
\xstickvar{5}{1}{2}{3.5};
%\xstickvar{4}{1}{3}{4};

\foreach \z in {1,2,3,4,5} \diagcubecolor{1}{2}{\z};
 \foreach \z in {1,2,3} \diagcubecolor{2}{2}{\z};
 \foreach \z in {1,2} \diagcubecolor{3}{2}{\z};
 \foreach \z in {1} \diagcubecolor{4}{2}{\z};
\xstickvar{5}{2}{1}{3.5};
%\xstickvar{4}{2}{2}{4};

\foreach \z in {1,2,3,4} \diagcubecolor{1}{3}{\z};
 \foreach \z in {1,2} \diagcubecolor{2}{3}{\z};
 \foreach \z in {1} \diagcubecolor{3}{3}{\z};
%\xstickvar{3}{3}{1}{4};

\foreach \z in {1,2,3} \diagcubecolor{1}{4}{\z};
 \foreach \z in {1,2} \diagcubecolor{2}{4}{\z};
 \foreach \z in {1} \diagcubecolor{3}{4}{\z};

\foreach \z in {1,2} \diagcubecolor{1}{5}{\z};
 \foreach \z in {1} \diagcubecolor{2}{5}{\z};

\foreach \z in {1} \diagcubecolor{1}{6}{\z};

\ystickvar{1}{7}{1}{3};
\ystickvar{1}{6}{2}{4};
\ystickvar{1}{5}{3}{5};
\ystickvar{2}{6}{1}{4};
%\ystickvar{2}{5}{2}{5};

\end{scope}
\begin{scope}[xshift=600,scale=0.45]
            \draw[->] (1,0)--(12,0);
        \node[right] at (12,0){$2$};
\draw[->] (1,0)--(1,12);
\node[above] at (1,12){$3$};
\draw[->] (1,0)--(-5.5,-6.5);
\node[below left] at (-5.5,-6.5) {$1$};

 \foreach \z in {1,2,3,4,5,6,7} \diagcubecolor{1}{1}{\z};
 \zstickvar{1}{1}{8}{3};
 \foreach \z in {1,2,3,4,5} \diagcubecolor{2}{1}{\z};
%\zstickvar{2}{1}{6}{5};
 
 \foreach \z in {1,2,3,4} \diagcubecolor{3}{1}{\z};
 \foreach \z in {1,2} \diagcubecolor{4}{1}{\z};
\foreach \z in {1} \diagcubecolor{5}{1}{\z};
\xstickvar{6}{1}{1}{3};
\xstickvar{5}{1}{2}{3.5};
%\xstickvar{4}{1}{3}{4};

\foreach \z in {1,2,3,4,5} \diagcubecolor{1}{2}{\z};
 \foreach \z in {1,2,3} \diagcubecolor{2}{2}{\z};
 \foreach \z in {1,2} \diagcubecolor{3}{2}{\z};
 \foreach \z in {1} \diagcubecolor{4}{2}{\z};
%\xstickvar{5}{2}{1}{3.5};
%\xstickvar{4}{2}{2}{4};

\foreach \z in {1,2,3,4} \diagcubecolor{1}{3}{\z};
 \foreach \z in {1,2} \diagcubecolor{2}{3}{\z};
 \foreach \z in {1} \diagcubecolor{3}{3}{\z};
%\xstickvar{3}{3}{1}{4};

\foreach \z in {1,2,3} \diagcubecolor{1}{4}{\z};
 \foreach \z in {1,2} \diagcubecolor{2}{4}{\z};
 \foreach \z in {1} \diagcubecolor{3}{4}{\z};

\foreach \z in {1,2} \diagcubecolor{1}{5}{\z};
 \foreach \z in {1} \diagcubecolor{2}{5}{\z};

\foreach \z in {1} \diagcubecolor{1}{6}{\z};

\ystickvar{1}{7}{1}{3};
\ystickvar{1}{6}{2}{4};
%\ystickvar{1}{5}{3}{5};
%\ystickvar{2}{6}{1}{4};
%\ystickvar{2}{5}{2}{5};

\end{scope}

\end{tikzpicture}
\eea

\item Surface boundary conditions: This is a configuration when we have asymptotic Young diagrams denoted as $\lambda_{12,13,14,23,24,34}$ for the six surfaces of the solid partition. Namely, the solid partition $\tilde{\rho}$ obeys the boundary conditions:
\bea
&\lambda_{12}=\{\tilde{\rho}_{\infty,\infty,k}\mid k=1,\ldots,\infty\},\quad \lambda_{23}=\{\tilde{\rho}_{i,\infty,\infty}\mid i=1,\ldots,\infty\}\quad \lambda_{13}=\{\tilde{\rho}_{\infty,j,\infty}\mid j=1,\ldots,\infty\},\\
&\tilde{\rho}_{\infty,j,k}=\infty,\,\,(j,k)\in\lambda_{14},\quad \tilde{\rho}_{i,\infty,k}=\infty,\,\,(i,k)\in\lambda_{24},\quad \tilde{\rho}_{i,j,\infty}=\infty,\,\,(i,j\in\lambda_{34})
\eea
The Young diagram $\lambda_{ab}\,(ab\in\six)$ extends in the $(a,b)$-direction semi-infinitely. We denote the set of possible solid plane partitions with surface boundary conditions $\lambda_{A},\,(A\in\six)$ as $\mathcal{SP}_{\{\lambda_{A}\}_{A\in\six}}$.

In the $(1,3)$-type description, the boundary Young diagrams $\lambda_{14,24,34}$ are visualized as
 \bea\label{eq:fig-solidpartitionsurface1}
    \begin{tikzpicture}[scale=0.33]
    \draw[->] (-4,-5)--(30,-5);
    \node[right] at (30,-5){$4$};
    \fill (0.5,-5) circle (5pt);
    \fill (11.5,-5) circle (5pt);
    \fill (22.5,-5) circle (5pt);
    
    \begin{scope}[scale=0.45]
    \draw[->] (1,0)--(13,0);
\node[right] at (13,0){$2$};
\draw[->] (1,0)--(1,9);
\node[above] at (1,9){$3$};
\draw[->] (1,0)--(-5.5,-6.5);
\node[below left] at (-5.5,-6.5){$1$};
%\diagcube{1}{1}{1};
%\zstick{1}{1}{1};
%\ystick{1}{2}{1};
%\xstickvar{1}{1}{1}{3};

%\begin{scope}[xscale=0.5,yscale=0.5]
\zstickvar{1}{1}{1}{7};
\zstickvar{2}{1}{1}{7};

\ystickvar{1}{2}{1}{9};
\ystickvar{1}{2}{2}{9};
\ystickvar{1}{2}{3}{9};
\ystickvar{2}{2}{1}{9};
\ystickvar{2}{2}{2}{9};

\xstickvar{3}{1}{1}{4};
\xstickvar{3}{1}{2}{4};
\xstickvar{3}{1}{3}{4};
\xstickvar{3}{2}{1}{4};
\xstickvar{3}{2}{2}{4};
\xstickvar{3}{3}{1}{4};

\end{scope}
\begin{scope}[xshift=300,scale=0.45]
        \draw[->] (1,0)--(13,0);
\node[right] at (13,0){$2$};
\draw[->] (1,0)--(1,9);
\node[above] at (1,9){$3$};
\draw[->] (1,0)--(-5.5,-6.5);
\node[below left] at (-5.5,-6.5){$1$};
%\diagcube{1}{1}{1};
%\zstick{1}{1}{1};
%\ystick{1}{2}{1};
%\xstickvar{1}{1}{1}{3};

%\begin{scope}[xscale=0.5,yscale=0.5]
\zstickvar{1}{1}{1}{7};
\zstickvar{2}{1}{1}{7};

\ystickvar{1}{2}{1}{9};
\ystickvar{1}{2}{2}{9};
\ystickvar{1}{2}{3}{9};
\ystickvar{2}{2}{1}{9};
\ystickvar{2}{2}{2}{9};

\xstickvar{3}{1}{1}{4};
\xstickvar{3}{1}{2}{4};
\xstickvar{3}{1}{3}{4};
\xstickvar{3}{2}{1}{4};
\xstickvar{3}{2}{2}{4};
\xstickvar{3}{3}{1}{4};

\end{scope}

\begin{scope}[xshift=600,scale=0.45]
\draw[->] (1,0)--(13,0);
\node[right] at (13,0){$2$};
\draw[->] (1,0)--(1,9);
\node[above] at (1,9){$3$};
\draw[->] (1,0)--(-5.5,-6.5);
\node[below left] at (-5.5,-6.5){$1$};
%\diagcube{1}{1}{1};
%\zstick{1}{1}{1};
%\ystick{1}{2}{1};
%\xstickvar{1}{1}{1}{3};

%\begin{scope}[xscale=0.5,yscale=0.5]
\zstickvar{1}{1}{1}{7};
\zstickvar{2}{1}{1}{7};

\ystickvar{1}{2}{1}{9};
\ystickvar{1}{2}{2}{9};
\ystickvar{1}{2}{3}{9};
\ystickvar{2}{2}{1}{9};
\ystickvar{2}{2}{2}{9};

\xstickvar{3}{1}{1}{4};
\xstickvar{3}{1}{2}{4};
\xstickvar{3}{1}{3}{4};
\xstickvar{3}{2}{1}{4};
\xstickvar{3}{2}{2}{4};
\xstickvar{3}{3}{1}{4};

\end{scope}
\end{tikzpicture}
    \eea
    where for each layer, we have a plane partition with asymptotic Young diagrams $\lambda_{i4}\,(i=1,2,3)$ for each leg. In this case, the asymptotic Young diagrams are the same for every layer.

The other boundary Young diagrams $\lambda_{12,13,23}$ are visualized as
\bea\label{eq:fig-solidpartitionsurface2}
\begin{tikzpicture}[scale=0.35]
    \draw[->] (-4,-5)--(30,-5);
    \node[right] at (30,-5){$4$};
    \fill (0.5,-5) circle (5pt);
    \fill (11.5,-5) circle (5pt);
    \fill (22.5,-5) circle (5pt);
    
    \begin{scope}[scale=0.4]
        \draw[->] (1,0)--(13,0);
\draw[->] (1,0)--(1,9);
\draw[->] (1,0)--(-5.5,-6.5);
\node[right] at (13,0){$2$};
\node[above] at (1,9){$3$};
\node[below left] at (-5.5,-6.5){$1$};

\zxsurfacevar{1}{1}{1}{7}{6};
\zxsurfacevar{1}{2}{1}{7}{6};
\zxsurfacevar{1}{3}{1}{7}{6};

\xysurfacevar{1}{4}{1}{6}{8};
\xysurfacevar{1}{4}{2}{6}{8};

\yzsurfacevar{1}{4}{3}{8}{5};
\yzsurfacevar{2}{4}{3}{8}{5};
\end{scope}
\begin{scope}[xshift=300,scale=0.4]
\draw[->] (1,0)--(13,0);
\draw[->] (1,0)--(1,9);
\draw[->] (1,0)--(-5.5,-6.5);
\node[right] at (13,0){$2$};
\node[above] at (1,9){$3$};
\node[below left] at (-5.5,-6.5){$1$};

\zxsurfacevar{1}{1}{1}{7}{6};
\zxsurfacevar{1}{2}{1}{7}{6};
%\zxsurfacevar{1}{3}{1}{7}{6};

\xysurfacevar{1}{3}{1}{6}{9};
%\xysurfacevar{1}{4}{2}{6}{8};

\yzsurfacevar{1}{3}{2}{9}{6};
%\yzsurfacevar{2}{4}{3}{8}{5}; 
\end{scope}

\begin{scope}[xshift=600,scale=0.4]
  \draw[->] (1,0)--(13,0);
\draw[->] (1,0)--(1,9);
\draw[->] (1,0)--(-5.5,-6.5);
\node[right] at (13,0){$2$};
\node[above] at (1,9){$3$};
\node[below left] at (-5.5,-6.5){$1$};

\zxsurfacevar{1}{1}{1}{7}{6};
%\zxsurfacevar{1}{2}{1}{7}{6};
%\zxsurfacevar{1}{3}{1}{7}{6};

\xysurfacevar{1}{2}{1}{6}{10};
%\xysurfacevar{1}{4}{2}{6}{8};

%\yzsurfacevar{1}{4}{3}{8}{5};
%\yzsurfacevar{2}{4}{3}{8}{5};
\end{scope}
\end{tikzpicture}
\eea
where we have surfaces boundary conditions for each layer. The number of surfaces obey the non-increasing condition according the 4-direction. For example, for each layer we have $k^{(i)}_{23},k^{(i)}_{13},k_{12}^{(i)}\in\mathbb{Z}_{\geq 0}$ with $k_{12,13,23}^{(i)}\geq k_{12,13,23}^{(i+1)}$ and 
\bea\label{eq:D8surfacedecomp}
\lambda_{12}=\{k_{12}^{(i)}\mid i=1,\ldots, \infty\},\quad \lambda_{13}=\{k_{13}^{(i)}\mid i=1,\ldots, \infty\},\quad \lambda_{23}=\{k_{23}^{(i)}\mid i=1,\ldots,\infty\}
\eea

\item Hypersurface boundary conditions: This is a configuration when we have asymptotic 1d partitions denoted as $k_{234,134,124,123}=k_{\bar{1},\bar{2},\bar{3},\bar{4}}\in\mathbb{Z}_{\geq 0}$ for the four hypersurfaces. The solid partition $\tilde{\rho}$ obeys the boundary conditions
\bea
\rho_{i,\infty,\infty}=\infty\,\,(i=1,\ldots,k_{234}),&\quad \rho_{\infty,j,\infty}=\infty\,\,(j=1,\ldots,k_{134}),\\ \rho_{\infty,\infty,k}=\infty\,\,(i=1,\ldots,k_{124}),&\quad  \rho_{\infty\infty\infty}=k_{123}.
\eea
We denote the set of solid partitions with hypersurface boundary conditions as $\mathcal{SP}_{k_{\bar{1}},k_{\bar{2}},k_{\bar{3}},k_{\bar{4}}}$. Actually, one can see that after shifting the origin of the solid partition, we have $\mathcal{SP}_{k_{\bar{1}},k_{\bar{2}},k_{\bar{3}},k_{\bar{4}}}=\mathcal{SP}$.

In the $(1,3)$-type description, the 1d partitions $k_{234,134,124}$ are visualized as
\bea\label{eq:fig-solidpartitionhypersurface123}
\begin{tikzpicture}[scale=0.35]
    \draw[->] (-4,-5)--(30,-5);
    \node[right] at (30,-5){$4$};
    \fill (0.5,-5) circle (5pt);
    \fill (11.5,-5) circle (5pt);
    \fill (22.5,-5) circle (5pt);
    
    \begin{scope}[scale=0.4]
        \draw[->] (1,0)--(13,0);
\draw[->] (1,0)--(1,9);
\draw[->] (1,0)--(-5.5,-6.5);
\node[right] at (13,0){$2$};
\node[above] at (1,9){$3$};
\node[below left] at (-5.5,-6.5){$1$};

\zxsurfacevar{1}{1}{1}{7}{6};
\zxsurfacevar{1}{2}{1}{7}{6};
\zxsurfacevar{1}{3}{1}{7}{6};

\xysurfacevar{1}{4}{1}{6}{8};
\xysurfacevar{1}{4}{2}{6}{8};

\yzsurfacevar{1}{4}{3}{8}{5};
\yzsurfacevar{2}{4}{3}{8}{5};
\end{scope}
\begin{scope}[xshift=300,scale=0.4]
 \draw[->] (1,0)--(13,0);
\draw[->] (1,0)--(1,9);
\draw[->] (1,0)--(-5.5,-6.5);
\node[right] at (13,0){$2$};
\node[above] at (1,9){$3$};
\node[below left] at (-5.5,-6.5){$1$};

\zxsurfacevar{1}{1}{1}{7}{6};
\zxsurfacevar{1}{2}{1}{7}{6};
\zxsurfacevar{1}{3}{1}{7}{6};

\xysurfacevar{1}{4}{1}{6}{8};
\xysurfacevar{1}{4}{2}{6}{8};

\yzsurfacevar{1}{4}{3}{8}{5};
\yzsurfacevar{2}{4}{3}{8}{5};
%\yzsurfacevar{2}{4}{3}{8}{5}; 
\end{scope}

\begin{scope}[xshift=600,scale=0.4]
 \draw[->] (1,0)--(13,0);
\draw[->] (1,0)--(1,9);
\draw[->] (1,0)--(-5.5,-6.5);
\node[right] at (13,0){$2$};
\node[above] at (1,9){$3$};
\node[below left] at (-5.5,-6.5){$1$};

\zxsurfacevar{1}{1}{1}{7}{6};
\zxsurfacevar{1}{2}{1}{7}{6};
\zxsurfacevar{1}{3}{1}{7}{6};

\xysurfacevar{1}{4}{1}{6}{8};
\xysurfacevar{1}{4}{2}{6}{8};

\yzsurfacevar{1}{4}{3}{8}{5};
\yzsurfacevar{2}{4}{3}{8}{5};
%\yzsurfacevar{2}{4}{3}{8}{5};
\end{scope}
\end{tikzpicture}
\eea
where $k_{234,134,124}$ is the number of $23,13,12$-surfaces and we have the same number for each layer orthogonal to the 4-direction.

The 1d partition $k_{123}$ is visualized as
\bea\label{eq:fig-solidpartitionhypersurface4}
\begin{tikzpicture}[scale=0.33]
    \draw[->] (-4,-5)--(30,-5);
    \node[right] at (30,-5){$4$};
    \fill (0.5,-5) circle (5pt);
    \fill (11.5,-5) circle (5pt);
    \fill (22.5,-5) circle (5pt);
    
\begin{scope}[scale=0.4]
        \draw[->] (1,0)--(12,0);
\draw[->] (1,0)--(1,9);
\draw[->] (1,0)--(-5.5,-6.5);
\node[right] at (12,0){$2$};
\node[above] at (1,9){$3$};
\node[below left] at (-5.5,-6.5){$1$};

\xyzsurfacevar{1}{1}{1}{5}{9}{6};
\end{scope}
\begin{scope}[xshift=300,scale=0.4]
 \draw[->] (1,0)--(12,0);
\draw[->] (1,0)--(1,9);
\draw[->] (1,0)--(-5.5,-6.5);
\node[right] at (12,0){$2$};
\node[above] at (1,9){$3$};
\node[below left] at (-5.5,-6.5){$1$};

\xyzsurfacevar{1}{1}{1}{5}{9}{6};
%\yzsurfacevar{2}{4}{3}{8}{5}; 
\end{scope}

\begin{scope}[xshift=600,scale=0.4]
 \draw[->] (1,0)--(12,0);
\draw[->] (1,0)--(1,9);
\draw[->] (1,0)--(-5.5,-6.5);
\node[right] at (12,0){$2$};
\node[above] at (1,9){$3$};
\node[below left] at (-5.5,-6.5){$1$};

%\xyzsurfacevar{1}{1}{1}{1}{1}{1};
\end{scope}
\end{tikzpicture}
\eea
where up to the $k_{123}$-th layer, we have a plane partition spanning the whole 123-plane.

\end{itemize}

\paragraph{D8 partition functions with boundary conditions}
Similar to the spiked instanton and tetrahedron instanton cases, we formally decompose the character $\bfK$ into $\bfK=\bfK^{\bd}+\bfK^{\reg}$
and then have
\bea
\mathbf{V}&=\bfV_{\text{pert.}}+\bfV_{\bd.}+\bfV_{\text{inst.}},\quad\bfV_{\bd.}=-\bfN^{\vee}\bfK^{\bd}-\bfN\bfK^{\bd\vee}+\bfP_{\four}\bfK^{\bd\vee}\bfK,\\
\bfV_{\text{inst.}}&=-\bfN^{\vee}\bfK^{\reg}-\bfN\bfK^{\reg\vee}+\bfP_{\four}\left(\bfK^{\bd\vee}\bfK^{\reg}+\bfK^{\bd}\bfK^{\reg\vee}\right)+\bfP_{123}^{\vee}\bfK^{\reg\vee}\bfK^{\reg}.
\eea
We choose the following square roots
\bea
\mathbf{v}_{\bd.}&=-\bfN^{\vee}\bfK^{\bd}+\sqrt{\bfP_{\four}\bfK^{\bd\vee}\bfK^{\bd}},\quad \mathbf{v}_{\text{inst.}}=-\bfN^{\vee}\bfK^{\reg}+\bfP_{\four}\bfK^{\bd\vee}\bfK^{\reg}+\bfP_{123}^{\vee}\bfK^{\reg\vee}\bfK^{\reg}
\eea
where we omit the explicit formula for the nontrivial square root part for the boundary contributions. The instanton partition function then comes from the contour integral formula
\bea
\mathcal{Z}_{k}=\frac{\mathcal{G}^{k}}{k!}\oint \prod_{I=1}^{k}\frac{dx_{I}}{2\pi\iota x_{I}}\mathbb{I}'[\mathbf{v}_{\text{inst.}}].
\eea
The complete formula needs to be studied case by case by giving $\bfK^{\bd}$ explicitly
\bea
\bfK^{\bd}=\sum_{\shcube\in\rho_{\bd}}\chi_{\four,x}(\hcube)
\eea
where $\rho_{\bd}$ is the boundary contributions, but generally we have the following structure. After evaluating the residues, $\bfK^{\reg}$ will be 
\bea
\bfK^{\reg}=\sum_{\shcube\in\rho_{\reg}}\chi_{\four,x}(\hcube),
\eea
where $\rho_{\reg}$ denotes the set of possible boxes that can be added to a given boundary conditions which was classified above. Namely, we have an infinite size solid partition $\tilde{\rho}$ obeying the solid partition function with nontrivial boundary contributions $\rho_{\bd}$ and $\rho_{\reg}$ is the set of boxes not included in the boundary contributions. The partition function is then defined as
\bea
\mathcal{Z}=\sum_{\rho_{\reg}}\mathfrak{q}^{|\rho_{\reg}|}(-1)^{\sigma_{4}(\rho_{\reg})}\mathbb{I}[\mathbf{v}|_{\rho_{\reg}}],
\eea
where we introduced the sign factors. We conjecture the sign rules to be
\bea\label{eq:D8boundary_signrule}
\sigma_{4}(\rho_{\reg})=\#\left\{(i,i,i,j)\in\rho_{\reg}\mid  i< j\right\}.
\eea
 Namely, only the boxes of the $\rho_{\reg}$ but not $\rho_{\bd}$ contribute to the signs.\footnote{Strictly speaking, sign factors depending on the boundary conditions will appear when considering the gluings of these \textit{vertex terms}, but they are only overall factors from the vertex perspective.} A derivation of this sign rule\footnote{This sign rule is similar to the one used in \cite[Conj.~2.11]{Monavari:2022rtf}. See \cite[Thm.~5.16]{Nekrasov:2023nai} also for another description of the sign factor.} in our formalism will be given in section~\ref{sec:D8qqlegboundary}, \ref{sec:D8qqsurfaceboundary}.

%For the moment, we do not know the explicit formula for the sign factor \remred{refer}. 

\paragraph{Leg boundary conditions}
We introduce the following set of boxes
\bea
\mathcal{B}_{a,\pi_{a}}=\left\{(x_{1},x_{2},x_{3},x_{4})\mid x_{a}=1,\ldots \infty,\quad (x_{b},x_{c},x_{d})\in\pi_{a},\,\,(b,c,d\neq a)\right\}
\eea
for $a\in\four$, where $\pi_{a}$ are finite plane partitions. We also define
\bea
\mathcal{B}_{\pi_{1}\pi_{2}\pi_{3}\pi_{4}}\coloneqq \sum_{a\in\four}\mathcal{B}_{a,\pi_{a}}-\sum_{ab\in\six}\mathcal{B}_{ab,\pi_{a}\cap\pi_{b}}+\sum_{\overline{(abc)}\in\four}\mathcal{B}_{abc,\pi_{a}\cap\pi_{b}\cap\pi_{c}}-\mathcal{B}_{\four,\pi_{1}\cap\pi_{2}\cap\pi_{3}\cap\pi_{4}}
\eea
where $\mathcal{B}_{ab,\pi_{a}\cap\pi_{b}}=\mathcal{B}_{a,\pi_{a}}\cap \mathcal{B}_{b,\pi_{b}},\,\,\mathcal{B}_{abc,\pi_{a}\cap \pi_{b}\cap\pi_{c}}=\bigcap\limits_{i=a,b,c}\mathcal{B}_{i,\pi_{a}} $ and $\mathcal{B}_{\four,\cap_{a\in\four}\pi_{a}}=\bigcap\limits_{a\in\four}\mathcal{B}_{a,\pi_{a}}$. The contributions coming from the leg boundaries are
\bea
\bfK^{\bd}=\sum_{\shcube\in\mathcal{B}_{\pi_{1}\pi_{2}\pi_{3}\pi_{4}}}\chi_{\four,x}(\hcube)=\sum_{a\in\four}\sum_{\scube\in\pi_{a}}\frac{\chi_{\bar{a},x}(\cube)}{1-q_{a}}-\sum_{\shcube\in\mathcal{S}_{\pi_{1}\pi_{2}\pi_{3}\pi_{4}}}\chi_{\four,x}(\hcube)\eqqcolon \bfN_{\pi_{1}\pi_{2}\pi_{3}\pi_{4}}
\eea
where $\mathcal{S}_{\pi_{1}\pi_{2}\pi_{3}\pi_{4}}=\sum_{a\in\four}\mathcal{B}_{a,\pi^{(a)}}-\mathcal{B}_{\pi^{(1)}\pi^{(2)}\pi^{(3)}\pi^{(4)}}$ is a finite set.

The contour integral formula is then explicitly written as
\bea\label{eq:D8contourintegral-leg}
\mathcal{Z}_{k}&= \frac{\mathcal{G}^{k}}{k!}\oint \prod_{I=1}^{k}\frac{dx_{I}}{2\pi\iota x_{I}}\prod_{I=1}^{k}\frac{1-Kx/x_{I}}{1-x/x_{I}}\prod_{I<J}\mathcal{A}_{\mathbb{C}^{4}}\left(\frac{x_{I}}{x_{J}}\right)^{-1}\\
&\qquad \times \prod_{I=1}^{k}\prod_{a\in\four}\prod_{\scube\in\pi_{a}}g_{\bar{a}}\left(\frac{\chi_{\bar{a},x}(\cube)}{x_{I}}\right)^{-1}\prod_{I=1}^{k}\prod_{\shcube\in\mathcal{S}_{\pi_{1}\pi_{2}\pi_{3}\pi_{4}}}\mathcal{A}_{\mathbb{C}^{4}}\left(\frac{\chi_{\four,x}(\hcube)}{x_{I}}\right).
\eea
For the case when we have only one leg, the set $\mathcal{S}_{\emptyset\emptyset\emptyset\pi_{4}}$ is empty and thus the contour integral formula is simplified as
\bea
\mathcal{Z}_{k}=\frac{\mathcal{G}^{k}}{k!} \oint \prod_{I=1}^{k}\frac{dx_{I}}{2\pi\iota x_{I}}\prod_{I=1}^{k}\frac{1-Kx/x_{I}}{1-x/x_{I}}\prod_{I<J}\mathcal{A}_{\mathbb{C}^{4}}\left(\frac{x_{I}}{x_{J}}\right)^{-1}\prod_{I=1}^{k}\prod_{\scube\in\pi_{4}}g_{\bar{4}}\left(\frac{\chi_{\bar{4},x}(\cube)}{x_{I}}\right)^{-1}.
\eea

After evaluating the residues, the D8 partition function is given as
\bea\label{eq:D8legpartition}
&\mathcal{Z}=\sum_{\rho_{\reg}\in\mathcal{SP}_{\pi_{1}\pi_{2}\pi_{3}\pi_{4}}}\mathfrak{q}^{|\rho_{\reg}|}(-1)^{\sigma_{4}(\rho_{\reg})}\mathcal{Z}^{\D8}_{\four;4;\pi_{1}\pi_{2}\pi_{3}\pi_{4}}[\rho_{\reg},K],\\
&\mathcal{Z}^{\D8}_{\four;4;\pi_{1}\pi_{2}\pi_{3}\pi_{4}}[\rho_{\reg},K]=\mathbb{I}\left[-\bfN^{\vee}\bfK^{\reg}+\bfP_{\four}\bfN_{\pi_{1}\pi_{2}\pi_{3}\pi_{4}}^{\vee}\bfK^{\reg}+\bfP_{123}^{\vee}\bfK^{\reg\vee}\bfK^{\reg}\right].
\eea
Note here that we are identifying the elements $\rho_{\reg}\in\mathcal{SP}_{\pi_{1}\pi_{2}\pi_{3}\pi_{4}}$ with the set of boxes \textit{not} included in the boundary plane partitions.

\paragraph{Surface boundary conditions}
Following the previous discussion, we introduce the following set of boxes
\bea
\mathcal{B}_{A,\lambda_{A}}=\{(x_{1},x_{2},x_{3},x_{4})\mid x_{a,b}=1,\ldots,\infty\,\,\,(a,b\in A),\quad (x_{c},x_{d})\in\lambda_{A}\,\,\,(c,d\in\bar{A})\}
\eea
for $A\in\six$, where $\lambda_{A}$ are finite Young diagrams. Namely, we have six Young diagrams extending infinitely in the two directions in $A\in\six$ (see \eqref{eq:fig-solidpartitionsurface1} and \eqref{eq:fig-solidpartitionsurface2}). The set of boxes included in the boundaries is given as
\bea
\mathcal{B}_{\{\lambda_{A}\}_{A\in\six}}=\sum_{A\in\six}\mathcal{B}_{A,\lambda_{A}}-\mathcal{S}_{\{\lambda_{A}\}_{A\in\six}},
\eea
where $\mathcal{S}_{\{\lambda_{A}\}}$ is a finite set. The explicit formula of $\mathcal{S}_{\{\lambda_{A}\}}$ can be written but it is complicated so we do not write it here. Roughly speaking, the set $\sum_{A\in\six}\mathcal{B}_{A,\lambda_{A}}$ has contribution of boxes with double counting coming from the intersection of the six possible surfaces and the set $\mathcal{S}_{\{\lambda_{A}\}}$ removes such double counting. The surface boundary contribution is then given as 
\bea
\bfK^{\bd}=\sum_{\shcube\in\mathcal{B}_{\{\lambda_{A}\}}}\chi_{\four,x}(\hcube)=\sum_{A\in\six}\sum_{\Abox\in\lambda_{A}}\frac{\chi_{A,x}(\Bbox)}{\bfP_{A}}-\sum_{\shcube\in\mathcal{S}_{\{\lambda_{A}\}}}\chi_{\four,x}(\hcube)\eqqcolon \bfN_{\{\lambda_{A}\}_{A\in\six}}
\eea
The contour integral formula is then given as
\bea\label{eq:D8contourintegral-surface}
\mathcal{Z}_{k}=\frac{\mathcal{G}^{k}}{k!} \oint \prod_{I=1}^{k}\frac{dx_{I}}{2\pi\iota x_{I}}\prod_{I=1}^{k}\frac{1-Kx/x_{I}}{1-x/x_{I}}\prod_{I<J}\mathcal{A}_{\mathbb{C}^{4}}\left(\frac{x_{I}}{x_{J}}\right)^{-1}\prod_{I=1}^{k}\prod_{A\in\six}\prod_{\Abox\in\lambda_{A}}\mathscr{S}_{\bar{A}}\left(\frac{\chi_{A,x}(\Bbox)}{x_{I}}\right)^{-1}\prod_{I=1}^{k}\prod_{\shcube\in\mathcal{S}_{\{\lambda_{A}\}}}\mathcal{A}_{\mathbb{C}^{4}}\left(\frac{\chi_{\four,x}(\hcube)}{x_{I}}\right).
\eea
For the case, when we have only one surface, say $\lambda_{12}$, the set $\mathcal{S}_{\{\lambda_{A}\}}$ will be empty and the contour integral formula is simplified as
\bea
\mathcal{Z}_{k}=\frac{\mathcal{G}^{k}}{k!} \oint \prod_{I=1}^{k}\frac{dx_{I}}{2\pi\iota x_{I}}\prod_{I=1}^{k}\frac{1-Kx/x_{I}}{1-x/x_{I}}\prod_{I<J}\mathcal{A}_{\mathbb{C}^{4}}\left(\frac{x_{I}}{x_{J}}\right)^{-1}\prod_{I=1}^{k}\prod_{\Abox\in\lambda_{12}}\mathscr{S}_{34}\left(\frac{\chi_{12,x}(\Bbox)}{x_{I}}\right)^{-1}.
\eea
The poles are then classified by the position of boxes possible to add to the solid partition with boundary conditions and the partition function is given as
\bea\label{eq:D8surfacepartition}
&\mathcal{Z}=\sum_{\rho_{\reg}\in\mathcal{SP}_{\{\lambda_{A}\}}}\mathfrak{q}^{|\rho_{\reg}|}(-1)^{\sigma_{4}(\rho_{\reg})}\mathcal{Z}_{\four;4;\{\lambda_{A}\}}^{\D8}[\rho_{\reg},K],\\
&\mathcal{Z}_{\four;4;\{\lambda_{A}\}}^{\D8}[\rho_{\reg},K]=\mathbb{I}\left[-\bfN^{\vee}\bfK^{\reg}+\bfP_{\four}\bfN_{\{\lambda_{A}\}}^{\vee}\bfK^{\reg}+\bfP_{123}^{\vee}\bfK^{\reg\vee}\bfK^{\reg}\right].
\eea
Note here that we are identifying the elements $\rho_{\reg}\in\mathcal{SP}_{\{\lambda_{A}\}}$ with the set of boxes \textit{not} included in the boundary Young diagrams.

\paragraph{Hypersurface boundary conditions}
Let us next consider the hypersurface boundary conditions. In this case, we can explicitly compute the character $\bfK^{\bd}$. For example, consider the situation when we have the $k_{\bar{4}}$ hypersurfaces spanning the $123$-plane. The boundary contributions are given as
\bea
\bfK^{\bd}=\sum_{l=1}^{k_{\bar{4}}}\sum_{i,j,k=1}^{\infty}xq_{1}^{i-1}q_{2}^{j-1}q_{3}^{k-1}q_{4}^{l-1}=\frac{1}{\bfP_{123}}\sum_{l=1}^{k_{\bar{4}}}xq_{4}^{l-1}=\frac{x}{\bfP_{\four}}(1-q_{4}^{k_{\bar{4}}}).
\eea
The general situation when we have $k_{\bar{a}}$ hypersurfaces for the $\bar{a}$-plane, the boundary contributions are computed as 
\bea
\bfK^{\bd}=\frac{x}{\bfP_{\four}}(1-q_{1}^{k_{\bar{1}}}q_{2}^{k_{\bar{2}}}q_{3}^{k_{\bar{3}}}q_{4}^{k_{\bar{4}}})
\eea
where the computation is similar to the leg boundary conditions of the Young diagram and the surface boundary conditions of the plane partition. We then have
\bea
\mathbf{v}_{\text{inst.}}&=-\bfN^{\vee}\bfK^{\reg}+\bfP_{\four}\left(\frac{x(1-q_{1}^{k_{\bar{1}}}q_{2}^{k_{\bar{2}}}q_{3}^{k_{\bar{3}}}q_{4}^{k_{\bar{4}}})}{\bfP_{\four}}\right)^{\vee}\bfK^{\reg}+\bfP_{123}^{\vee}\bfK^{\reg\vee}\bfK^{\reg}\\
&=-\bfN^{\reg\vee}\bfK^{\reg}+\bfP_{123}^{\vee}\bfK^{\reg\vee}\bfK^{\reg}
\eea
where $\bfN^{\reg}=(q_{1}^{k_{\bar{1}}}q_{2}^{k_{\bar{2}}}q_{3}^{k_{\bar{3}}}q_{4}^{k_{\bar{4}}}-K)x$. The contour integral formula is then given as
\bea\label{eq:D8contourintegral-hypersurface}
\mathcal{Z}_{k}=\frac{\mathcal{G}^{k}}{k!} \oint \prod_{I=1}^{k}\frac{dx_{I}}{2\pi\iota x_{I}}\prod_{I=1}^{k}\frac{1-Kx/x_{I}}{1-q_{1}^{k_{\bar{1}}}q_{2}^{k_{\bar{2}}}q_{3}^{k_{\bar{3}}}q_{4}^{k_{\bar{4}}}x/x_{I}}\prod_{I<J}\mathcal{A}_{\mathbb{C}^{4}}\left(\frac{x_{I}}{x_{J}}\right)^{-1}.
\eea
The poles will be classified by a finite solid partition whose origin is shifted from $x$ to $q_{1}^{k_{\bar{1}}}q_{2}^{k_{\bar{2}}}q_{3}^{k_{\bar{3}}}q_{4}^{k_{\bar{4}}}x$. Effectively, the parameter $K$ is modified to $q_{1}^{-k_{\bar{1}}}q_{2}^{-k_{\bar{2}}}q_{3}^{-k_{\bar{3}}}q_{4}^{-k_{\bar{4}}}K$:
\bea\label{eq:D8hypersurfacepartition}
\mathcal{Z}=\sum_{\rho\in\mathcal{SP}}\mathfrak{q}^{|\rho|}(-1)^{\sigma_{4}(\rho)}\mathcal{Z}^{\D8}_{\four;4}[\rho\,,\,\,q_{1}^{-k_{\bar{1}}}q_{2}^{-k_{\bar{2}}}q_{3}^{-k_{\bar{3}}}q_{4}^{-k_{\bar{4}}}K].
\eea

%\paragraph{Generic boundary conditions}

%\paragraph{Comparison with DT 4-vertex}

\begin{remark}
Similar to the D6 setup, we note that the partition functions $\mathcal{Z}^{\D8}_{\four;4;\pi_{1}\pi_{2}\pi_{3}\pi_{4}}[\rho,K]$ and $\mathcal{Z}^{\D8}_{\four;4;\{\lambda_{A}\}}[\rho,K]$ introduced here differs with the one used in \cite{Nekrasov:2023nai,Cao:2019tnw,Cao:2019tvv,Monavari:2022rtf,Bae:2022pif,Bae:2024bpx} up to boundary contributions. The one used there is defined in a symmetric way as
\bea
\bfN^{\text{DT}4}_{\pi_{1}\pi_{2}\pi_{3}\pi_{4}}=\sum_{a\in\four}\sum_{\scube\in\pi_{a}}\frac{\chi_{\bar{a},x}(\cube)}{1-q_{a}},\quad \bfN^{\text{DT}4}_{\{\lambda_{A}\}}=\sum_{A\in\six}\sum_{\Abox\in\lambda_{A}}\frac{\chi_{A,x}(\Bbox)}{\bfP_{A}}.
\eea
The difference comes from the contributions at the intersection of the boundary contributions which are nonessential when considering the vertex contributions.
\end{remark}

%%%%%%%%%%%%%%%%%%%%%%%%%%%%%%%%%%%%%%%%%%%%%%%%%%%%%%%%%%%%%%%%%%%%%%%%%%%%%%%%%%%%%%%%%%%%%%%%%%%%%%%%%%%%%%%%%%%%%%%%%%%%%%%%%%%%%%%%%
\section{Free field realizations and vertex operators}\label{sec:freefieldvertexop}
%%%%%%%%%%%%%%%%%%%%%%%%%%%%%%%%%%%%%%%%%%%%%%%%%%%%%%%%%%%%%%%%%%%%%%%%%%%%%%%%%%%%%%%%%%%%%%%%%%%%%%%%%%%%%%%%%%%%%%%%%%%%%%%%%%%%%%%%%
In this section, we introduce vertex operators which reproduce the contour integral formulas where partitions with nontrivial boundary conditions appear. Let us first review the vertex operators introduced in \cite{Kimura:2023bxy}.

\begin{definition}[\cite{Kimura:2023bxy}]\label{def:vertex-op}
    We introduce the following vertex operators:
    \bea
    \mathsf{A}(x)&=\mathsf{a}_{0}(x):\exp\left(\sum_{n\neq 0}\mathsf{a}_{n}x^{-n}\right):,\quad  \mathsf{S}_{a}(x)=\mathsf{s}_{a,0}(x):\exp\left(\sum_{n\neq 0}\mathsf{s}_{a,n}x^{-n}\right):,\\
    \mathsf{X}_{A}(x)&=\mathsf{x}_{A,0}(x):\exp\left(\sum_{n\neq 0}\mathsf{x}_{A,n}x^{-n}\right):,\quad 
    \mathsf{W}_{\bar{a}}(x)=\mathsf{w}_{\bar{a},0}(x):\exp\left(\sum_{n\neq 0}\mathsf{w}_{\bar{a},n}x^{-n}\right):,\\
   \mathsf{Z}(x)&=\mathsf{z}_{0}(x):\exp\left(\sum_{n\neq 0}\mathsf{z}_{n}x^{-n}\right):,\quad [\mathsf{a}_{n},\mathsf{a}_{m}]=-\frac{1}{n}\mathbf{P}_{\four}^{[n]}\delta_{n+m,0},\quad \mathsf{a}_{n}=\begin{dcases}
       \bfP^{[-n]}_{a}\mathsf{s}_{a,n},\\
       \bfP^{[-n]}_{A}\mathsf{x}_{A,n},\\
       \bfP^{[-n]}_{\bar{a}}\mathsf{w}_{\bar{a},n},\\
       \bfP_{\four}^{[-n]}\mathsf{z}_{n}.
   \end{dcases}
    \eea
for $a\in\four,A\in\six$. The zero-modes $\mathsf{a}_{0}(x),\mathsf{s}_{a,0}(x),\mathsf{x}_{A,0}(x),\mathsf{w}_{\bar{a},0}(x),\mathsf{z}_{0}(x)$ are given in Appendix~\ref{app:zero-modes}.
\end{definition}
Physically, $\mathsf{A}(x),\mathsf{S}_{a}(x),\mathsf{X}_{A}(x),\mathsf{W}_{\bar{a}}(x),\mathsf{Z}(x)$ correspond with the D0, D2, D4, D6, and D8-branes. For later use, we also introduce the following vertex operator
\bea
\mathsf{Z}(K,x)={:\frac{\mathsf{Z}(x)}{\mathsf{Z}(Kx)}:}=\widetilde{\mathsf{z}}^{K}_{0}(x):\exp\left(\sum_{n\neq 0}\widetilde{\mathsf{z}}^{K}_{n}x^{-n}\right):,\quad \widetilde{\mathsf{z}}^{K}_{n}=(1-K^{-n})\mathsf{z}_{n}.
\eea
Physically, this corresponds with the $\U(1|1)$ $\D8\tbar\overline{\D8}$ magnificent four system where we need the antibranes so that the system is stabilized with the background flux~\cite{Witten:2000mf}.

Under the explicit zero-modes given in Appendix~\ref{app:zero-modes}, some of the operator product formulas when the arising factors are rational functions are given as follows.
\begin{proposition}[\cite{Kimura:2023bxy}]
The operator products of the operators $\mathsf{A}(x),\mathsf{S}_{\bar{a}}(x),\mathsf{X}_{A}(x)\,(A\in\six),\mathsf{W}_{\bar{a}}(x)\,(a\in\four),\mathsf{Z}(K,x)$ are
\bea\label{eq:contractions}
\mathsf{A}(x)\mathsf{S}_{a}(x')=g_{\bar{a}}\left(x'/x\right)^{-1}: \mathsf{A}(x)\mathsf{S}_{a}(x'):,&\quad \mathsf{S}_{a}(x')\mathsf{A}(x)=g_{\bar{a}}(q_{a}x/x'):\mathsf{A}(x)\mathsf{S}_{a}(x'):,\\
\mathsf{A}(x)\mathsf{X}_{A}(x')=\mathscr{S}_{\bar{A}}(x'/x)^{-1}:\mathsf{A}(x)\mathsf{X}_{A}(x'):,&\quad \mathsf{X}_{A}(x')\mathsf{A}(x)=\mathscr{S}_{\bar{A}}(q_{A}x/x')^{-1}:\mathsf{X}_{A}(x')\mathsf{X}_{A}(x):, \\
    \mathsf{A}(x)\mathsf{W}_{\bar{a}}(x')=\mathscr{V}_{a}\left(x'/x\right)^{-1}:\mathsf{A}(x)\mathsf{W}_{\bar{a}}(x'):,&\quad 
    \mathsf{W}_{\bar{a}}(x')\mathsf{A}(x)=q_{a}^{-1}\mathscr{V}_{a}(q_{a}^{-1}x/x'):\mathsf{W}_{\bar{a}}(x')\mathsf{A}(x):,\\
    \mathsf{S}_{a}(x)\mathsf{S}_{b}(x')=\mathscr{S}_{\overline{ab}}(q_{a}x'/x):\mathsf{S}_{a}(x)\mathsf{S}_{b}(x'):,&\quad 
        \mathsf{S}_{b}(x')\mathsf{S}_{a}(x)=\mathscr{S}_{\overline{ab}}(q_{b}x/x'):\mathsf{S}_{a}(x)\mathsf{S}_{b}(x'):,\\
        \mathsf{X}_{A}(x)\mathsf{S}_{c}(x')=\mathscr{V}_{d}\left(q_{A}x'/x\right)^{-1} : \mathsf{X}_{A}(x)\mathsf{S}_{c}(x'):  ,&\quad
        \mathsf{S}_{c}(x')\mathsf{X}_{A}(x)=q_{d}^{-1}\mathscr{V}_{d}\left(q_{d}^{-1}q_{A}^{-1}x/x'\right): \mathsf{X}_{A}(x)\mathsf{S}_{c}(x'):,\\
\mathsf{W}_{\bar{a}}(x)\mathsf{S}_{a}(x')=\frac{x'}{1-q_{a}^{-1}x'/x}:\mathsf{W}_{\bar{a}}(x)\mathsf{S}_{a}(x'):,&\quad 
\mathsf{S}_{a}(x')\mathsf{W}_{\bar{a}}(x)=\frac{-q_{a}x}{1-q_{a}x/x'}:\mathsf{W}_{\bar{a}}(x)\mathsf{S}_{a}(x'):,\\
\mathsf{Z}(K,x)\mathsf{A}(x')=K^{-1}\frac{1-x'/x}{1-K^{-1}x'/x}:\mathsf{Z}(K,x)\mathsf{A}(x'):,&\quad 
\mathsf{A}(x')\mathsf{Z}(K,x)=\frac{1-x/x'}{1-Kx/x'}:\mathsf{Z}(K,x)\mathsf{A}(x'):,
\eea
where the structure functions are given in \eqref{eq:structure-funct}.
\end{proposition}
After using the definitions of the vertex operators in Def.~\ref{def:vertex-op} and the zero-modes in Appendix~\ref{app:zero-modes}, we also have the following relations:
\bea\label{eq:vertexoprelation}
\mathsf{A}(x)={:\frac{\mathsf{S}_{a}(x)}{\mathsf{S}_{a}(q_{a}x)}:},\quad  \mathsf{S}_{a}(x)=\mathsf{s}_{a,0}(x):\frac{\mathsf{X}_{ab}(x)}{\mathsf{X}_{ab}(q_{b}x)}:,\quad  \mathsf{X}_{ab}(x)={:\frac{\mathsf{W}_{abc}(x)}{\mathsf{W}_{abc}(q_{c}x)}:},\quad \mathsf{W}_{\bar{a}}(x)={:\frac{\mathsf{Z}(x)}{\mathsf{Z}(q_{a}x)}:}=\mathsf{Z}(q_{a},x)
\eea
and 
\bea\label{eq:vertexoprelation2}
\mathsf{A}(x)={\mathsf{a}_{0}(x):\frac{\mathsf{X}_{A}(x)\mathsf{X}_{A}(q_{A}x)}{\prod_{a\in A}\mathsf{X}_{A}(q_{a}x)}:}={\mathsf{a}_{0}(x):\frac{\mathsf{W}_{\bar{a}}(x)\prod_{i\in\bar{a}}\mathsf{W}_{\bar{a}}(q_{ia}^{-1}x)}{\mathsf{W}_{\bar{a}}(q_{a}^{-1}x)\prod_{i\in \bar{a}}\mathsf{W}_{\bar{a}}(q_{i}x)}:}=\mathsf{a}_{0}(x):\frac{\mathsf{Z}(x)^{2}\prod_{A\in\six}\mathsf{Z}(q_{A}x)}{\prod_{a\in\four}\mathsf{Z}(q_{a}x)\prod_{a\in\four}\mathsf{Z}(q_{a}^{-1}x)}:,
\eea
where $a,b,c,\in\four$ and $a\neq b\neq c$. For example, the relation between $\mathsf{A}(x)=:\mathsf{S}_{a}(x)/\mathsf{S}_{a}(q_{a}x):$ is obtained by using 
\bea
\mathsf{a}_{n}=(1-q_{a}^{-n})\mathsf{s}_{a,n},\quad \frac{\mathsf{s}_{a,0}(x)}{\mathsf{s}_{a,0}(q_{a}x)}=q_{a}^{-\mathsf{s}_{a,0}}=e^{\mathsf{a}_{0}}=\mathsf{a}_{0}(x).
\eea
We also will use the following relation in later sections:
\bea\label{eq:vertexoprelation3}
{:\frac{\mathsf{Z}(K,x)}{\mathsf{Z}(K,q_{a}x)}:}={:\frac{\mathsf{W}_{\bar{a}}(x)}{\mathsf{W}_{\bar{a}}(Kx)}:}.
\eea

%\begin{remark}
%It is known that the D0-D2 brane setup does not preserve supersymmetry and the D0-branes can be dissolve into the D2-branes \cite{Blumenhagen:2000eb,Sato:2005cy,Witten:2000mf,Fujii:2001wp}. Generally, this is also true for $\text{D}p-\text{D}(p+2)$ setups. From the quantum algebraic view point, this is understood as
%\bea
%\frac{\mathsf{S}_{a}(x)}{\mathsf{S}_{a}(q_{a}x)}=\mathsf{A}(x)
%\eea
%namely, if we add one-instanton to the D2-brane, we still get a D2-brane with a shift of the parameter.

%\end{remark}

The free field realizations of the contour integral formulas for the spiked instanton, tetrahedron instanton, magnificent four partition functions are obtained generally as follows.
\begin{proposition}[\cite{Kimura:2023bxy}]\label{prop:D0general}
    The $k$-instanton contribution to the partition function is given as\footnote{Just as usual partition functions, besides the non-perturbative contributions, we have classical and perturbative ones. The classical ones cannot be determined by the quantum algebraic structure and must be implemented by hand. The one-loop perturbative part may be included by using the vertex operators $\mathsf{V}_{i}(v_{i})$, but since we have the radial ordering of the vertex operators, we need to be careful of the analytic region of the spectral parameters. Since we are not interested in all of these aspects, we simply discard them in this paper.
}
    \bea
    \mathcal{Z}_{k}=\frac{\mathcal{G}^{k}}{k!} \oint \prod_{I=1}^{k}\frac{dx_{I}}{2\pi\iota x_{I}}\left\langle\prod_{I=1}^{k}\mathsf{A}(x_{I})^{-1}:\prod_{i}\mathsf{V}_{i}(v_{i}):\right\rangle
    \eea
    where $\mathsf{V}_{i}(x)$ is an operator written in $\{\mathsf{S}_{a}(x),\mathsf{X}_{A}(x),\mathsf{W}_{\bar{a}}(x),\mathsf{Z}(x)\}$ and $\langle \mathcal{O}\rangle=\bra{0}\mathcal{O}\ket{0}$. The product of the $\mathsf{A}$ operators is given in a specific order.
\end{proposition}
The vertex operators $\mathsf{V}_{i}(x)$ correspond with the framing bundles and they determine how the multi-dimensional partitions expand. Moreover, they have a one-to-one correspondence with the vacuum configuration. To obtain the contour integrals given in section~\ref{sec:D4partitionfunction}, \ref{sec:D6partitionfunction}, \ref{sec:D8partitionfunction}, we need to introduce the boundary conditions to the operators\footnote{This is called the highest weight in the context of quantum algebra.} $\mathsf{V}_{i}(x)$. In the following subsections, we classify the corresponding highest weights for each configurations.

\subsection{D4 partition functions with boundary conditions}\label{sec:D4contourfreefield}
\paragraph{One-leg}
Let us consider the one-leg boundary condition first.
The vacuum configuration when we have one boundary condition comes from the following figure and the highest weight is given as
\begin{equation}\label{eq:D4oneboundcondfigure}
        \adjustbox{valign=c}{\begin{tikzpicture}[scale=0.7]
    
        %\fill[white!60!blue] (0,0)--(0,3.8)--(0.6,3.8)--(0.6,0)--(0,0);
        \fill[white!60!red] (0,0)--(3.8,0)--(3.8,0.8)--(0,0.8)--(0,0);
        \fill[white!60!red] (0,0)--(0.6,0)--(0.6,0.8)--(0,0.8)--(0,0);
        %\draw [decorate,decoration = {brace}] (0.02,3.84)--(0.58,3.84);
        \draw [decorate,decoration = {mirror, brace}] (3.84,0.02)--(3.84,0.78);
        \node [right] at (3.84,0.45){$\,\,k_{12}$};
        %\node[above] at (0.3,3.84){$l$};
        
        \draw[->] (0,-0.5)--(0,2.8);
        \draw[->] (-0.5,0)--(4.3,0);
        \node[right] at (4.3,0){$q_{1}$};
        \node[above] at (0,2.8) {$q_{2}$};

        %\draw (0.2,3.8)--(0.2,0);
        %\draw (0.4,3.8)--(0.4,0);
        %\draw (0.6,3.8)--(0.6,0);

        \draw (0,0.2)--(3.8,0.2);
        \draw (0,0.4)--(3.8,0.4);
        \draw (0,0.6)--(3.8,0.6);
        \draw (0,0.8)--(3.8,0.8);
        \end{tikzpicture}
        }\quad =\quad :\mathsf{X}_{12}(x)\prod_{i=1}^{\infty}\prod_{j=1}^{k_{12}}\mathsf{A}^{-1}(xq_{1}^{i-1}q_{2}^{j-1}):
\end{equation}
The infinite product can be regularized properly as
\bea\label{eq:D4oneboundaryregularize}
:\mathsf{X}_{12}(x)\prod_{i=1}^{\infty}\prod_{j=1}^{k_{12}}\mathsf{A}^{-1}(xq_{1}^{i-1}q_{2}^{j-1}):&={:\mathsf{X}_{12}(x)\prod_{j=1}^{k_{12}}\mathsf{S}_{1}(xq_{2}^{j-1})^{-1}:},
\eea
where we used \eqref{eq:vertexoprelation} and \eqref{eq:vertexoprelation2} as
\bea\label{eq:D4D2vertexrelation}
\mathsf{A}(x)={:\frac{\mathsf{S}_{a}(x)}{\mathsf{S}_{a}(q_{a}x)}:},\quad {:\prod_{i=1}^{\infty}\mathsf{A}^{-1}(q_{a}^{i-1}x):}={:\prod_{i=1}^{\infty}\frac{\mathsf{S}_{a}(q_{a}^{i}x)}{\mathsf{S}_{a}(q_{a}^{i-1}x)}:}={:\mathsf{S}_{a}(x)^{-1}:}.
\eea
Actually, we further can simplify the vertex operator as
\bea
{:\prod_{j=1}^{k_{12}}\mathsf{s}_{1,0}(xq_{2}^{j-1})\frac{\mathsf{X}_{12}(x)}{\prod_{j=1}^{k}\mathsf{S}_{1}(xq_{2}^{j-1})}:}={:\mathsf{X}_{12}(q_{2}^{k_{12}}x):}
\eea
up to extra zero-modes. The zero-modes have no contractions with the $\mathsf{A}(x)$ operators and do not affect the pole structure and thus the highest weight is effectively $\mathsf{X}_{12}(q_{2}^{k_{12}}x)$. Namely, the D4-D2 highest weight simply shifts the Coulomb branch parameter.
\paragraph{Two-legs}
When there are two legs, the highest weight is given as follows:
\begin{equation}\label{eq:D4twoboundcondfigure}
        \adjustbox{valign=c}{\begin{tikzpicture}[scale=0.6]
    
        \fill[white!60!blue] (0,0)--(0,3.8)--(0.6,3.8)--(0.6,0)--(0,0);
        \fill[white!60!red] (0,0)--(3.8,0)--(3.8,0.8)--(0,0.8)--(0,0);
        \fill[white!80!gray] (0,0)--(0.6,0)--(0.6,0.8)--(0,0.8)--(0,0);
        \draw [decorate,decoration = {brace}] (0.02,3.84)--(0.58,3.84);
        \draw [decorate,decoration = {mirror, brace}] (3.84,0.02)--(3.84,0.78);
        \node [right] at (3.84,0.45){$\,k_{12}$};
        \node[above] at (0.4,3.84){$\,\,\,l_{12}$};
        
        \draw[->] (0,-0.5)--(0,4.3);
        \draw[->] (-0.5,0)--(4.3,0);
        \node[right] at (4.3,0){$q_{1}$};
        \node[above] at (0,4.3) {$q_{2}$};

        \draw (0.2,3.8)--(0.2,0);
        \draw (0.4,3.8)--(0.4,0);
        \draw (0.6,3.8)--(0.6,0);

        \draw (0,0.2)--(3.8,0.2);
        \draw (0,0.4)--(3.8,0.4);
        \draw (0,0.6)--(3.8,0.6);
        \draw (0,0.8)--(3.8,0.8);
        \end{tikzpicture}
        }={:\mathsf{X}_{12}(x)\frac{\prod_{\Abox\in l_{12}}\mathsf{
A}^{-1}(\chi_{12,x}(\Bbox))\prod_{\Abox\in k_{12}}\mathsf{
A}^{-1}(\chi_{12,x}(\Bbox))}{\prod_{\Abox\in l_{12}\cap k_{12}}\mathsf{
A}^{-1}(\chi_{12,x}(\Bbox))}:}
\end{equation}
Note that in this case, the two stacks intersect at the origin and we have to be careful of the double counting. The infinite product above can be regularized similarly as
\bea
&{:\mathsf{X}_{12}(x)\prod_{i=1}^{\infty}\prod_{j=1}^{k_{12}}\mathsf{A}^{-1}(xq_{1}^{i-1}q_{2}^{j-1})\prod_{j=1}^{\infty}\prod_{i=1}^{l_{12}}\mathsf{A}^{-1}(xq_{1}^{i-1}q_{2}^{j-1})\prod_{i=1}^{l_{12}}\prod_{j=1}^{k_{12}}\mathsf{A}(xq_{1}^{i-1}q_{2}^{j-1}):}\\
=&{:\mathsf{X}_{12}(x)\prod_{i=1}^{l_{12}}\mathsf{S}_{2}(xq_{2}^{k_{12}}q_{1}^{i-1})^{-1}\prod_{j=1}^{k_{12}}\mathsf{S}_{1}(xq_{2}^{j-1})^{-1}:}={:\mathsf{X}_{12}(x)\prod_{i=1}^{l_{12}}\mathsf{S}_{2}(xq_{1}^{i-1})^{-1}\prod_{j=1}^{k_{12}}\mathsf{S}_{1}(xq_{1}^{l_{12}}q_{2}^{j-1})^{-1}:}
\eea
When rewriting the infinite product into finite product of $\mathsf{S}_{1,2}(x)$, we need to choose an ordering to rewrite it, which corresponds to how we decompose the stack of boxes of the boundary conditions into a one-dimensional rod. Two typical examples are the following configurations:
\bea
\label{eq:D4twoboundcondfigure2}
        \adjustbox{valign=c}{\begin{tikzpicture}[scale=0.7]
    
        \fill[white!60!blue] (0,0)--(0,3.8)--(0.6,3.8)--(0.6,0)--(0,0);
        \fill[white!60!red] (0,0)--(3.8,0)--(3.8,0.8)--(0,0.8)--(0,0);
        \fill[white!60!red] (0,0)--(0.6,0)--(0.6,0.8)--(0,0.8)--(0,0);
        \draw [decorate,decoration = {brace}] (0.02,3.84)--(0.58,3.84);
        \draw [decorate,decoration = {mirror, brace}] (3.84,0.02)--(3.84,0.78);
        \node [right] at (3.84,0.5){$\,\,k_{12}$};
        \node[above] at (0.35,3.84){$\,\,\,l_{12}$};
        
        \draw[->] (0,-0.5)--(0,4.3);
        \draw[->] (-0.5,0)--(4.3,0);
        \node[right] at (4.3,0){$q_{1}$};
        \node[above] at (0,4.3) {$q_{2}$};

        \draw (0.2,3.8)--(0.2,0.8);
        \draw (0.4,3.8)--(0.4,0.8);
        \draw (0.6,3.8)--(0.6,0.8);

        \draw (0,0.2)--(3.8,0.2);
        \draw (0,0.4)--(3.8,0.4);
        \draw (0,0.6)--(3.8,0.6);
        \draw (0,0.8)--(3.8,0.8);
        \node at (1.9,-1.3){\scalebox{1.3}{${:\frac{\mathsf{X}_{12}(x)}{\prod\limits_{i=1}^{l_{12}}\mathsf{S}_{2}(xq_{2}^{k_{12}}q_{1}^{i-1})\prod\limits_{j=1}^{k_{12}}\mathsf{S}_{1}(xq_{2}^{j-1})}:}$}};
        \end{tikzpicture}
        }\qquad \qquad 
        \adjustbox{valign=c}{\begin{tikzpicture}[scale=0.7]
    
        \fill[white!60!blue] (0,0)--(0,3.8)--(0.6,3.8)--(0.6,0)--(0,0);
        \fill[white!60!red] (0,0)--(3.8,0)--(3.8,0.8)--(0,0.8)--(0,0);
        \fill[white!60!blue] (0,0)--(0.6,0)--(0.6,0.8)--(0,0.8)--(0,0);
        \draw [decorate,decoration = {brace}] (0.02,3.84)--(0.58,3.84);
        \draw [decorate,decoration = {mirror, brace}] (3.84,0.02)--(3.84,0.78);
        \node [right] at (3.84,0.5){$\,\,k_{12}$};
        \node[above] at (0.35,3.84){$\,\,\,l_{12}$};
        
        \draw[->] (0,-0.5)--(0,4.3);
        \draw[->] (-0.5,0)--(4.3,0);
        \node[right] at (4.3,0){$q_{1}$};
        \node[above] at (0,4.3) {$q_{2}$};

        \draw (0.2,3.8)--(0.2,0);
        \draw (0.4,3.8)--(0.4,0);
        \draw (0.6,3.8)--(0.6,0);

        \draw (0.6,0.2)--(3.8,0.2);
        \draw (0.6,0.4)--(3.8,0.4);
        \draw (0.6,0.6)--(3.8,0.6);
        \draw (0.6,0.8)--(3.8,0.8);
        \node at (1.9,-1.3){\scalebox{1.3}{${:\frac{\mathsf{X}_{12}(x)}{\prod_{i=1}^{l}\mathsf{S}_{2}(xq_{1}^{i-1})\prod_{j=1}^{k_{12}}\mathsf{S}_{1}(xq_{1}^{l_{12}}q_{2}^{j-1})}:}$}};
        \end{tikzpicture}}
\eea
%Other ordering of the one-dimensional rods will look like 
\begin{comment}
\bea
\begin{tikzpicture}[scale=0.7]
        \fill[white!60!blue] (0,0)--(0,3.8)--(0.6,3.8)--(0.6,0)--(0,0);
        \fill[white!60!red] (0,0)--(3.8,0)--(3.8,0.4)--(0,0.4)--(0,0);
        \fill[white!60!red] (0.2,0.4)--(3.8,0.4)--(3.8,0.6)--(0.2,0.6)--(0.2,0.4);
        \fill[white!60!red] (0.4,0.6)--(3.8,0.6)--(3.8,0.8)--(0.4,0.8)--(0.4,6);
        %\fill[white!60!red] (0,0)--(0.6,0)--(0.6,0.4)--(0,0.4)--(0,0);
        %\draw [decorate,decoration = {brace}] (0.02,3.84)--(0.58,3.84);
        %\draw [decorate,decoration = {mirror, brace}] (3.84,0.02)--(3.84,0.78);
        %\node [right] at (3.84,0.4){$\,\,k_{12}$};
        %\node[above] at (0.3,3.84){$\,\,\,l_{12}$};
        
        \draw[->] (0,-0.5)--(0,4.3);
        \draw[->] (-0.5,0)--(4.3,0);
        \node[right] at (4.3,0){$q_{1}$};
        \node[above] at (0,4.3) {$q_{2}$};

        \draw (0.2,3.8)--(0.2,0.4);
        \draw (0.4,3.8)--(0.4,0.6);
        \draw (0.6,3.8)--(0.6,0.8);

        \draw (0,0.2)--(3.8,0.2);
        \draw (0,0.4)--(3.8,0.4);
        \draw (0.2,0.6)--(3.8,0.6);
        \draw (0.4,0.8)--(3.8,0.8);
        %\node at (1.9,-1.3){\scalebox{1.3}{${:\frac{\mathsf{X}_{12}(x)}{\prod\limits_{i=1}^{l_{12}}\mathsf{S}_{2}(xq_{2}^{k_{12}}q_{1}^{i-1})\prod\limits_{j=1}^{k_{12}}\mathsf{S}_{1}(xq_{2}^{j-1})}:}$}};
        \end{tikzpicture}
\eea
\end{comment}
We further can simplify the highest weight as
\bea\label{eq:D4-D2simplication}
{:\frac{\mathsf{X}_{12}(x)}{\prod\limits_{i=1}^{l}\mathsf{S}_{2}(xq_{2}^{k_{12}}q_{1}^{i-1})\prod\limits_{j=1}^{k_{12}}\mathsf{S}_{1}(xq_{2}^{j-1})}:}\propto {:\mathsf{X}_{12}(x)\prod_{i=1}^{l_{12}}\frac{\mathsf{X}_{12}(xq_{1}^{i})}{\mathsf{X}_{12}(xq_{1}^{i-1})}\prod_{j=1}^{k_{12}}\frac{\mathsf{X}_{12}(xq_{1}^{l_{12}}q_{2}^{j})}{\mathsf{X}_{12}(xq_{1}^{l_{12}}q_{2}^{j-1})}:}=\mathsf{X}_{12}(q_{1}^{l_{12}}q_{2}^{k_{12}}x)
\eea
up to zero-modes and thus the highest weight is effectively $\mathsf{X}_{12}(q_{2}^{k_{12}}q_{1}^{l_{12}}x)$. After this simplification, the highest weight will not depend on the expressions in terms of the $\mathsf{S}_{1,2}$ operators.

\begin{proposition}[D4 two-legs]
    The highest weight of the D4 partition function with two nontrivial leg boundary conditions specified by $k_{12},l_{12}\in\mathbb{Z}_{\geq 0}$ is $\mathsf{X}_{12}(xq_{2}^{k_{12}}q_{1}^{l_{12}})$ and the free field realization of the contour integral formula is given as
    \bea
    \mathcal{Z}_{k}&=\frac{\mathcal{G}^{k}}{k!} \oint \prod_{I=1}^{k}\frac{dx_{I}}{2\pi\iota x_{I}}\left\langle \prod_{I=1}^{k}\mathsf{A}(x_{I})^{-1}\mathsf{X}_{12}(q_{2}^{k_{12}}q_{1}^{l_{12}}x)   \right\rangle\\
    &=\frac{\mathcal{G}^{k}}{k!}\oint\prod_{I=1}^{k}\frac{dx_{I}}{2\pi\iota x_{I}}\prod_{I=1}^{k}\mathscr{S}_{34}\left(\frac{q_{2}^{k_{12}}q_{1}^{l_{12}}x}{x_{I}}\right)\prod_{I<J}\mathcal{A}_{\mathbb{C}^{4}}\left(\frac{x_{I}}{x_{J}}\right)^{-1}.
    \eea
\end{proposition}

\subsection{D6 partition functions with boundary conditions}\label{sec:D6contourfreefield}
\subsubsection{Surface boundary conditions} Let us first consider the highest weight condition when we have surface boundary conditions.
For example, when we have $k_{12}$-surfaces in the 12-plane, the highest weight is given as
\bea
\adjustbox{valign=c}{
\begin{tikzpicture}[scale=0.15]
%\cube{1}{1}{1};
% horizontal y axis (the origin is at (1,0))
\draw[->] (1,0)--(13,0);
\draw[->] (1,0)--(1,9);
\draw[->] (1,0)--(-6.5,-7.5);
\node[right] at (13,0){$2$};
\node[above] at (1,9){$3$};
\node[below left] at (-6.5,-7.5){$1$};
\xysurfacevar{1}{1}{1}{6}{10};
\xysurfacevar{1}{1}{2}{6}{10};
\xysurfacevar{1}{1}{3}{6}{10};
\xysurfacevar{1}{1}{4}{6}{10};
\end{tikzpicture}}\quad =\quad {:\mathsf{W}_{\bar{4}}(x)\prod_{i,j=1}^{\infty}\prod_{k=1}^{k_{12}}\mathsf{A}^{-1}(xq_{1}^{i-1}q_{2}^{j-1}q_{3}^{k-1}):}.
\eea
Using the relations in \eqref{eq:vertexoprelation}, \eqref{eq:vertexoprelation2}, we have
\bea
%&\mathsf{A}(x)=\mathsf{a}_{0}(x):\frac{\mathsf{X}_{12}(x)\mathsf{X}_{12}(q_{12}x)}{\mathsf{X}_{12}(q_{1}x)\mathsf{X}_{12}(q_{2}x)}:,\\
&{:\prod_{i,j=1}^{\infty}\mathsf{A}(xq_{1}^{i-1}q_{2}^{j-1}):}\simeq {:\prod_{i,j=1}^{\infty}\frac{\mathsf{X}_{12}(xq_{1}^{i-1}q_{2}^{j-1})\mathsf{X}_{12}(xq_{1}^{i}q_{2}^{j})}{\mathsf
{X}_{12}(q_{1}^{i}q_{2}^{j-1}x)\mathsf{X}_{12}(xq_{1}^{i-1}q_{2}^{j})}:}=\mathsf{X}_{12}(x),
\eea
and
\bea
 {:\mathsf{W}_{\bar{4}}(x)\prod_{i,j=1}^{\infty}\prod_{k=1}^{k_{12}}\mathsf{A}^{-1}(xq_{1}^{i-1}q_{2}^{j-1}q_{3}^{k-1}):}\simeq {:\frac{\mathsf{W}_{\bar{4}}(x)}{\prod_{k=1}^{k_{12}}\mathsf{X}_{12}(xq_{3}^{k-1})}:}.
\eea
Note that the equality is up to zero-modes but since the zero-modes will not affect the contraction with $\mathsf{A}(x)$, we can effectively use the right hand side as the highest weight. The right hand side can be further simplified as
\bea
{:\frac{\mathsf{W}_{\bar{4}}(x)}{\prod_{k=1}^{k_{12}}\mathsf{X}_{12}(xq_{3}^{k-1})}:}=\mathsf{W}_{\bar{4}}(xq_{3}^{k_{12}}),
\eea
where we used \eqref{eq:vertexoprelation}. 

For the most general surface boundary conditions, a similar computation gives the highest weight 
\bea\label{eq:D6surface-vertop-simple}
\adjustbox{valign=c}{
\begin{tikzpicture}[scale=0.15]
%\cube{1}{1}{1};
% horizontal y axis (the origin is at (1,0))
\draw[->] (1,0)--(13,0);
\draw[->] (1,0)--(1,9);
\draw[->] (1,0)--(-5.5,-6.5);
\node[right] at (13,0){$2$};
\node[above] at (1,9){$3$};
\node[below left] at (-5.5,-6.5){$1$};

\zxsurfacevar{1}{1}{1}{7}{6};
\zxsurfacevar{1}{2}{1}{7}{6};
\zxsurfacevar{1}{3}{1}{7}{6};

\xysurfacevar{1}{4}{1}{6}{8};
\xysurfacevar{1}{4}{2}{6}{8};

\yzsurfacevar{1}{4}{3}{8}{5};
\yzsurfacevar{1}{4}{3}{8}{5};

\end{tikzpicture}}=\mathsf{W}_{\bar{4}}(xq_{1}^{k_{23}}q_{2}^{k_{13}}q_{3}^{k_{12}}).
\eea
\begin{proposition}
    The highest weight giving the D6 partition function with nontrivial surface boundary conditions specified by $k_{23},k_{13},k_{12}\in\mathbb{Z}_{\geq 0}$ is
    \bea
   \mathsf{W}_{\bar{4}}(xq_{1}^{k_{23}}q_{2}^{k_{13}}q_{3}^{k_{12}}).
    \eea
    The free field realization of the contour integral formulas for D6 partition functions with surface boundary conditions are given as
    \bea
\mathcal{Z}_{k}&=\frac{\mathcal{G}^{k}}{k!} \oint\prod_{I=1}^{k}\frac{dx_{I}}{2\pi\iota x_{I}}\left\langle \prod_{I=1}^{k}\mathsf{A}(x_{I})^{-1}\mathsf{W}_{\bar{4}}(xq_{1}^{k_{23}}q_{2}^{k_{13}}q_{3}^{k_{12}}) \right \rangle\\
&=\frac{\mathcal{G}^{k}}{k!}\oint \frac{dx_{I}}{2\pi\iota x_{I}}\prod_{I=1}^{k}\mathscr{V}_{4}\left(\frac{xq_{1}^{k_{23}}q_{2}^{k_{13}}q_{3}^{k_{12}}}{x_{I}}\right)\prod_{I<J}\mathcal{A}_{\mathbb{C}^{4}}\left(\frac{x_{I}}{x_{J}}\right)^{-1}
    \eea
    which reproduces \eqref{eq:D6surfacecontourintegral}. 
\end{proposition}

\subsubsection{Leg boundary conditions}
\paragraph{One-leg}Let us determine the highest weight when we have one-leg boundary condition. The highest weight is given as
\bea
\adjustbox{valign=c}{
\begin{tikzpicture}[scale=0.15]
%\cube{1}{1}{1};
% horizontal y axis (the origin is at (1,0))
\draw[->] (1,0)--(13,0);
\node[right] at (13,0){$2$};
\draw[->] (1,0)--(1,9);
\node[above] at (1,9){$3$};
\draw[->] (1,0)--(-5.5,-6.5);
\node[below left] at (-5.5,-6.5){$1$};
%\diagcube{1}{1}{1};
%\zstick{1}{1}{1};
%\ystick{1}{2}{1};
%\xstickvar{1}{1}{1}{3};

%\begin{scope}[xscale=0.5,yscale=0.5]
\zstickvar{1}{1}{1}{7};
\zstickvar{2}{1}{1}{7};
\zstickvar{3}{1}{1}{7};
\zstickvar{4}{1}{1}{7};
\zstickvar{5}{1}{1}{7};
\zstickvar{1}{2}{1}{7};
\zstickvar{2}{2}{1}{7};
\zstickvar{3}{2}{1}{7};
\zstickvar{4}{2}{1}{7};
\zstickvar{1}{3}{1}{7};
\zstickvar{2}{3}{1}{7};
\zstickvar{3}{3}{1}{7};
\zstickvar{1}{4}{1}{7};
\zstickvar{2}{4}{1}{7};
\zstickvar{1}{5}{1}{7};
\end{tikzpicture}}\quad = \quad :\mathsf{W}_{\bar{4}}(x)\prod_{(i,j)\in\nu}\prod_{k=1}^{\infty}\mathsf{A}(xq_{1}^{i-1}q_{2}^{j-1}q_{3}^{k-1})^{-1}:.
\eea
The infinite product can be regularized as
\bea
{:\mathsf{W}_{\bar{4}}(x)\prod_{(i,j)\in\nu}\prod_{k=1}^{\infty}\mathsf{A}(xq_{1}^{i-1}q_{2}^{j-1}q_{3}^{k-1})^{-1}:}={:\mathsf{W}_{4}(x)\prod_{\Abox\in\nu}\mathsf{S}_{3}(\chi_{12,x}(\Bbox))^{-1}:}
\eea
where we used \eqref{eq:D4D2vertexrelation}.
\begin{proposition}[D6 one-leg]
The free field realization of the contour integral of the D6 one-leg partition function is given as
\bea
\mathcal{Z}_{k}&=\frac{\mathcal{G}^{k}}{k!} \oint \prod_{I=1}^{k}\frac{dx_{I}}{2\pi\iota x_{I}} \left\langle \prod_{I=1}^{k}\mathsf{A}(x_{I})^{-1} :\frac{\mathsf{W}_{\bar{4}}(x)}{\prod_{\Abox \in \nu}\mathsf{S}_{3}(\chi_{12,x}(\Bbox))}:\right\rangle\\
&= \frac{\mathcal{G}^{k}}{k!} \prod_{I=1}^{k}\frac{dx_{I}}{2\pi\iota x_{I}} \prod_{I<J}\mathcal{A}_{\mathbb{C}^{4}}\left(\frac{x_{I}}{x_{J}}\right)^{-1}\prod_{I=1}^{k}\mathscr{V}_{4}\left(\frac{x}{x_{I}}\right) \prod_{I=1}^{k}\prod_{\Abox \in\nu} g_{\bar{3}}\left(\frac{\chi_{12,x}(\Bbox)}{x_{I}}\right)^{-1}.
\eea
\end{proposition}

\paragraph{Two-legs}
The highest weight when there are two nontrivial boundary conditions is given as
\bea\label{eq:D6two-legs-hw0}
\adjustbox{valign=c}{
\begin{tikzpicture}[scale=0.15]
%\cube{1}{1}{1};
% horizontal y axis (the origin is at (1,0))
\draw[->] (1,0)--(13,0);
\node[right] at (13,0){$2$};
\draw[->] (1,0)--(1,9);
\node[above] at (1,9){$3$};
\draw[->] (1,0)--(-5.5,-6.5);
\node[below left] at (-5.5,-6.5){$1$};
%\diagcube{1}{1}{1};
%\zstick{1}{1}{1};
%\ystick{1}{2}{1};
%\xstickvar{1}{1}{1}{3};

%\begin{scope}[xscale=0.5,yscale=0.5]
\ystickvar{1}{1}{1}{10};
\ystickvar{1}{1}{2}{10};
\ystickvar{1}{1}{3}{10};
\ystickvar{2}{1}{1}{10};
\ystickvar{2}{1}{2}{10};

\xstickvar{3}{1}{1}{4};
\xstickvar{3}{1}{2}{4};
\xstickvar{2}{1}{3}{4.5};
\xstickvar{3}{2}{1}{4};
\xstickvar{3}{2}{2}{4};
\xstickvar{3}{3}{1}{4};
\end{tikzpicture}}\quad =\quad {:\mathsf{W}_{\bar{4}}(x)\frac{\prod\limits_{\scube\in\mathcal{B}_{\lambda}}\mathsf{A}^{-1}(\chi_{\bar{4},x}(\cube))\prod\limits_{\scube\in\mathcal{B}_{\mu}}\mathsf{A}^{-1}(\chi_{\bar{4},x}(\cube))}{\prod\limits_{\scube\in\mathcal{B}_{\lambda\cap\mu}}\mathsf{A}^{-1}(\chi_{\bar{4},x}(\cube))}:}.
\eea
The infinite product appearing in $\mathcal{B}_{\lambda,\mu}$ can be regularized as
\bea
{:\prod_{\scube\in\mathcal{B}_{\lambda}}\mathsf{A}^{-1}(\chi_{\bar{4},x}(\cube)):}={:\prod_{\Abox\in\lambda}\mathsf{S}_{1}(\chi_{23,x}(\Bbox))^{-1}:},\quad {:\prod_{\scube\in\mathcal{B}_{\mu}}\mathsf{A}^{-1}(\chi_{\bar{4},x}(\cube)):}={:\prod_{\Abox\in\mu}\mathsf{S}_{2}(\chi_{13,x}(\Bbox))^{-1}:}.
\eea
A different realization is obtained by using \eqref{eq:D4twoboundcondfigure2} for each layer and for example, we have
\bea\label{eq:D4tensor-op}
:\frac{\mathsf{W}_{\bar{4}}(x)}{\prod_{i=1}^{\ell(\lambda^{\rmT})}\prod_{j=1}^{\lambda_{i}^{\rmT}}\mathsf{S}_{1}(q_{2}^{j-1}q_{3}^{i-1}x)\prod_{i=1}^{\ell(\mu)}\prod_{j=1}^{\mu_{i}}\mathsf{S}_{2}(q_{1}^{j-1}q_{2}^{\lambda_{i}^{\rmT}}q_{3}^{i-1}x)}:.
\eea
\begin{proposition}[D6 two-legs]
The contour integral of the D6 two-legs partition function has the free field realization:
\bea
\mathcal{Z}_{k}&=\frac{\mathcal{G}^{k}}{k!}\oint \prod_{I=1}^{k}\frac{dx_{I}}{2\pi\iota x_{I}}\left\langle \prod_{I=1}^{k}\mathsf{A}(x_{I})^{-1} {:\mathsf{W}_{\bar{4}}(x)\frac{\prod\limits_{\scube\in\mathcal{B}_{\lambda}}\mathsf{A}^{-1}(\chi_{\bar{4},x}(\cube))\prod\limits_{\scube\in\mathcal{B}_{\mu}}\mathsf{A}^{-1}(\chi_{\bar{4},x}(\cube))}{\prod\limits_{\scube\in\mathcal{B}_{\lambda\cap\mu}}\mathsf{A}^{-1}(\chi_{\bar{4},x}(\cube))}:} \right \rangle\\
&=\frac{\mathcal{G}^{k}}{k!}\oint \prod_{I=1}^{k}\frac{dx_{I}}{2\pi\iota x_{I}} \left\langle \prod_{I=1}^{k}\mathsf{A}(x_{I})^{-1} :\frac{\mathsf{W}_{\bar{4}}(x)}{\prod_{i=1}^{\ell(\lambda^{\rmT})}\prod_{j=1}^{\lambda_{i}^{\rmT}}\mathsf{S}_{1}(q_{2}^{j-1}q_{3}^{i-1}x)\prod_{i=1}^{\ell(\mu)}\prod_{j=1}^{\mu_{i}}\mathsf{S}_{2}(q_{1}^{j-1}q_{2}^{\lambda_{i}^{\rmT}}q_{3}^{i-1}x)}:\right\rangle\\
&=\frac{\mathcal{G}^{k}}{k!}\oint \prod_{I=1}^{k}\frac{dx_{I}}{2\pi\iota x_{I}}  \prod_{I<J}\mathcal{A}_{\mathbb{C}^{4}}\left(\frac{x_{I}}{x_{J}}\right)^{-1}\prod_{I=1}^{k}\mathscr{V}_{4}\left(\frac{x}{x_{I}}\right)\prod_{I=1}^{k}\prod_{l=1}^{\ell(\lambda^{\rmT})}\prod_{j=1}^{\lambda_{i}^{\rmT}}g_{\bar{1}}\left(\frac{q_{2}^{j-1}q_{3}^{i-1}x}{x_{I}}\right)^{-1} \prod_{I=1}^{k}\prod_{i=1}^{\ell(\mu)}\prod_{j=1}^{\mu_{i}}g_{\bar{2}}\left(\frac{q_{3}^{i-1}q_{1}^{j-1}q_{2}^{\lambda_{i}^{\rmT}}x}{x_{I}}\right)^{-1},
\eea
which reproduces \eqref{eq:D6twolegscontourintegral} after computation.
\end{proposition}

\paragraph{Three-legs}
The highest weight when there are three nontrivial boundary conditions is given as
\bea
\adjustbox{valign=c}{
\begin{tikzpicture}[scale=0.15]
%\cube{1}{1}{1};
% horizontal y axis (the origin is at (1,0))
\draw[->] (1,0)--(13,0);
\node[right] at (13,0){$2$};
\draw[->] (1,0)--(1,9);
\node[above] at (1,9){$3$};
\draw[->] (1,0)--(-5.5,-6.5);
\node[below left] at (-5.5,-6.5){$1$};
%\diagcube{1}{1}{1};
%\zstick{1}{1}{1};
%\ystick{1}{2}{1};
%\xstickvar{1}{1}{1}{3};

%\begin{scope}[xscale=0.5,yscale=0.5]
\zstickvar{1}{1}{1}{7};
\zstickvar{2}{1}{1}{7};

\ystickvar{1}{2}{1}{9};
\ystickvar{1}{2}{2}{9};
\ystickvar{1}{2}{3}{9};
\ystickvar{2}{2}{1}{9};
\ystickvar{2}{2}{2}{9};

\xstickvar{3}{1}{1}{4};
\xstickvar{3}{1}{2}{4};
\xstickvar{3}{1}{3}{4};
\xstickvar{3}{2}{1}{4};
\xstickvar{3}{2}{2}{4};
\xstickvar{3}{3}{1}{4};

%\end{scope}

%\diagcube{2}{1}{1};

%\diagcube{1}{2}{1};
%\diagcube{1}{1}{2};
\end{tikzpicture}}\quad =\quad {:\mathsf{W}_{\bar{4}}(x)\prod_{\scube\in\mathcal{B}_{\lambda\mu\nu}}\mathsf{A}^{-1}(\chi_{\bar{4},x}(\cube)):}.
\eea
The infinite product can be regularized as
\bea
{:\mathsf{W}_{\bar{4}}(x)\prod_{\scube\in\mathcal{B}_{\lambda\mu\nu}}\mathsf{A}^{-1}(\chi_{\bar{4},x}(\cube)):}={:\frac{\mathsf{W}_{\bar{4}}(x)}{\prod\limits_{\Abox\in\lambda}\mathsf{S}_{1}(\chi_{23,x}(\Bbox))\prod\limits_{\Abox\in\mu}\mathsf{S}_{2}(\chi_{13,x}(\Bbox))\prod\limits_{\Abox\in\nu}\mathsf{S}_{3}(\chi_{12,x}(\Bbox))}\prod_{\scube \in \mathcal{S}_{\lambda\mu\nu}}\mathsf{A}(\chi_{\bar{4},x}(\cube)):}.
\eea

\begin{proposition}[D6 three-legs]
The contour integral formula for the three-legs partition function has a free field realization:
\bea
\mathcal{Z}_{k}&=\frac{\mathcal{G}^{k}}{k!} \oint \prod_{I=1}^{k}\frac{dx_{I}}{2\pi\iota x_{I}} \left\langle \prod_{I=1}^{k}\mathsf{A}^{-1}(x_{I}):\mathsf{W}_{\bar{4}}(x)\prod_{\scube\in\mathcal{B}_{\lambda\mu\nu}}\mathsf{A}^{-1}(\chi_{\bar{4},x}(\cube)): \right\rangle\\
&=\frac{\mathcal{G}^{k}}{k!}\oint \prod_{I=1}^{k}\frac{dx_{I}}{2\pi\iota x_{I}} \left\langle \prod_{I=1}^{k}\mathsf{A}^{-1}(x_{I}):\frac{\mathsf{W}_{\bar{4}}(x)}{\prod\limits_{\Abox\in\lambda}\mathsf{S}_{1}(\chi_{23,x}(\Bbox))\prod\limits_{\Abox\in\mu}\mathsf{S}_{2}(\chi_{13,x}(\Bbox))\prod\limits_{\Abox\in\nu}\mathsf{S}_{3}(\chi_{12,x}(\Bbox))}\prod_{\scube \in \mathcal{S}_{\lambda\mu\nu}}\mathsf{A}(\chi_{\bar{4},x}(\cube)): \right\rangle\\
&=\frac{\mathcal{G}^{k}}{k!}\oint \prod_{I=1}^{k}\frac{dx_{I}}{2\pi\iota x_{I}} \prod_{I<J}\mathcal{A}_{\mathbb{C}^{4}}\left(\frac{x_{I}}{x_{J}}\right)^{-1} \prod_{I=1}^{k}\mathscr{V}_{4}\left(\frac{x}{x_{I}}\right)\prod_{I=1}^{k}\prod_{\scube \in \mathcal{S}_{\lambda\mu\nu}}\mathcal{A}_{\mathbb{C}^{4}}\left(\frac{\chi_{\bar{4},x}(\cube)}{x_{I}}\right)\\
&\quad \times \prod_{I=1}^{k}\left(\prod_{\Abox \in\lambda}g_{\bar{1}}\left(\frac{\chi_{23,x}(\Bbox)}{x_{I}}\right)^{-1}\prod_{\Abox \in\mu}g_{\bar{2}}\left(\frac{\chi_{13,x}(\Bbox)}{x_{I}}\right)^{-1}\prod_{\Abox \in\nu}g_{\bar{3}}\left(\frac{\chi_{12,x}(\Bbox)}{x_{I}}\right)^{-1}\right)
\eea
which reproduces \eqref{eq:D6threelegscontourintegral}.
\end{proposition}

\subsection{D8 partition function with boundary conditions}\label{sec:D8_w_bdry}
\subsubsection{Leg boundary conditions}
\paragraph{one-leg}The strategy to find the highest weight is the same as the previous examples. For simplicity, let us consider when we only have one-leg at the $4$-axis. The highest weight is then given as
\bea
  \adjustbox{valign=c}{  \begin{tikzpicture}[scale=0.25]
    \draw[->] (-4,-5)--(30,-5);
    \node[right] at (30,-5){$4$};
    \fill (0.5,-5) circle (5pt);
    \fill (11.5,-5) circle (5pt);
    \fill (22.5,-5) circle (5pt);
    
    \begin{scope}[scale=0.4]
        \draw[->] (1,0)--(9,0);
        \node[right] at (9,0){$2$};
\draw[->] (1,0)--(1,9);
\node[above] at (1,9){$3$};
\draw[->] (1,0)--(-3.5,-4.5);
\node[below left] at (-3.5,-4.5) {$1$};

 \foreach \z in {1,2,3,4,5,6,7} \diagcubecolor{1}{1}{\z};
 \foreach \z in {1,2,3,4,5} \diagcubecolor{2}{1}{\z};
 \foreach \z in {1,2,3,4} \diagcubecolor{3}{1}{\z};
 \foreach \z in {1,2} \diagcubecolor{4}{1}{\z};
\foreach \z in {1} \diagcubecolor{5}{1}{\z};

\foreach \z in {1,2,3,4,5} \diagcubecolor{1}{2}{\z};
 \foreach \z in {1,2,3} \diagcubecolor{2}{2}{\z};
 \foreach \z in {1,2} \diagcubecolor{3}{2}{\z};
 \foreach \z in {1} \diagcubecolor{4}{2}{\z};

\foreach \z in {1,2,3,4} \diagcubecolor{1}{3}{\z};
 \foreach \z in {1,2} \diagcubecolor{2}{3}{\z};
 \foreach \z in {1} \diagcubecolor{3}{3}{\z};

\foreach \z in {1,2,3} \diagcubecolor{1}{4}{\z};
 \foreach \z in {1,2} \diagcubecolor{2}{4}{\z};
 \foreach \z in {1} \diagcubecolor{3}{4}{\z};

\foreach \z in {1,2} \diagcubecolor{1}{5}{\z};
 \foreach \z in {1} \diagcubecolor{2}{5}{\z};

\foreach \z in {1} \diagcubecolor{1}{6}{\z};
\end{scope}
\begin{scope}[xshift=600,scale=0.4]
        \draw[->] (1,0)--(9,0);
        \node[right] at (9,0){$2$};
\draw[->] (1,0)--(1,9);
\node[above] at (1,9){$3$};
\draw[->] (1,0)--(-3.5,-4.5);
\node[below left] at (-3.5,-4.5) {$1$};

 \foreach \z in {1,2,3,4,5,6,7} \diagcubecolor{1}{1}{\z};
 \foreach \z in {1,2,3,4,5} \diagcubecolor{2}{1}{\z};
 \foreach \z in {1,2,3,4} \diagcubecolor{3}{1}{\z};
 \foreach \z in {1,2} \diagcubecolor{4}{1}{\z};
\foreach \z in {1} \diagcubecolor{5}{1}{\z};

\foreach \z in {1,2,3,4,5} \diagcubecolor{1}{2}{\z};
 \foreach \z in {1,2,3} \diagcubecolor{2}{2}{\z};
 \foreach \z in {1,2} \diagcubecolor{3}{2}{\z};
 \foreach \z in {1} \diagcubecolor{4}{2}{\z};

\foreach \z in {1,2,3,4} \diagcubecolor{1}{3}{\z};
 \foreach \z in {1,2} \diagcubecolor{2}{3}{\z};
 \foreach \z in {1} \diagcubecolor{3}{3}{\z};

\foreach \z in {1,2,3} \diagcubecolor{1}{4}{\z};
 \foreach \z in {1,2} \diagcubecolor{2}{4}{\z};
 \foreach \z in {1} \diagcubecolor{3}{4}{\z};

\foreach \z in {1,2} \diagcubecolor{1}{5}{\z};
 \foreach \z in {1} \diagcubecolor{2}{5}{\z};

\foreach \z in {1} \diagcubecolor{1}{6}{\z};
\end{scope}

\begin{scope}[xshift=300,scale=0.4]
        \draw[->] (1,0)--(9,0);
        \node[right] at (9,0){$2$};
\draw[->] (1,0)--(1,9);
\node[above] at (1,9){$3$};
\draw[->] (1,0)--(-3.5,-4.5);
\node[below left] at (-3.5,-4.5) {$1$};

 \foreach \z in {1,2,3,4,5,6,7} \diagcubecolor{1}{1}{\z};
 \foreach \z in {1,2,3,4,5} \diagcubecolor{2}{1}{\z};
 \foreach \z in {1,2,3,4} \diagcubecolor{3}{1}{\z};
 \foreach \z in {1,2} \diagcubecolor{4}{1}{\z};
\foreach \z in {1} \diagcubecolor{5}{1}{\z};

\foreach \z in {1,2,3,4,5} \diagcubecolor{1}{2}{\z};
 \foreach \z in {1,2,3} \diagcubecolor{2}{2}{\z};
 \foreach \z in {1,2} \diagcubecolor{3}{2}{\z};
 \foreach \z in {1} \diagcubecolor{4}{2}{\z};

\foreach \z in {1,2,3,4} \diagcubecolor{1}{3}{\z};
 \foreach \z in {1,2} \diagcubecolor{2}{3}{\z};
 \foreach \z in {1} \diagcubecolor{3}{3}{\z};

\foreach \z in {1,2,3} \diagcubecolor{1}{4}{\z};
 \foreach \z in {1,2} \diagcubecolor{2}{4}{\z};
 \foreach \z in {1} \diagcubecolor{3}{4}{\z};

\foreach \z in {1,2} \diagcubecolor{1}{5}{\z};
 \foreach \z in {1} \diagcubecolor{2}{5}{\z};

\foreach \z in {1} \diagcubecolor{1}{6}{\z};
\end{scope}
\end{tikzpicture}}=\quad :\mathsf{Z}(K,x)\prod_{l=1}^{\infty}\prod_{\scube\in\pi_{4}}\mathsf{A}^{-1}(q_{4}^{l-1}\chi_{\bar{4},x}(\cube)):
\eea
and the infinite product is regularized as
\bea
{:\mathsf{Z}(K,x)\prod_{l=1}^{\infty}\prod_{\scube\in\pi_{4}}\mathsf{A}^{-1}(q_{4}^{l-1}\chi_{\bar{4},x}(\cube)):}={:\frac{\mathsf{Z}(K,x)}{\prod\limits_{\scube\in\pi_{4}}\mathsf{S}_{4}(\chi_{\bar{4},x}(\cube))}:}.
\eea

\begin{proposition}[D8 one-leg]
The free field realization of the contour integral formula of the one-leg D8 partition function is 
\bea\label{eq:D8_1leg_bdy}
\mathcal{Z}_{k}&=\frac{\mathcal{G}^{k}}{k!} \oint \prod_{I=1}^{k}\frac{dx_{I}}{2\pi\iota x_{I}} \left\langle \prod_{I=1}^{k}\mathsf{A}^{-1}(x_{I}):\frac{\mathsf{Z}(K,x)}{\prod_{\scube \in\pi_{4}}\mathsf{S}_{4}(\chi_{\bar{4},x}(\cube))}: \right\rangle\\
&= \frac{\mathcal{G}^{k}}{k!}\oint \prod_{I=1}^{k}\frac{dx_{I}}{2\pi\iota x_{I}} \prod_{I<J}\mathcal{A}_{\mathbb{C}^{4}}\left(\frac{x_{I}}{x_{J}}\right)^{-1}\prod_{I=1}^{k}\frac{1-Kx/x_{I}}{1-x/x_{I}}\prod_{I=1}^{k}\prod_{\scube\in\pi_{4}}g_{\bar{4}}\left(\frac{\chi_{\bar{4},x}(\cube)}{x_{I}}\right)^{-1}.
\eea
\end{proposition}

\paragraph{Two-legs} For later use, let us also consider the situation when we have two nontrivial plane partitions at two legs of the four legs. We consider the case when we have two legs at the 1, 2 axes. Namely, we have asymptotic plane partitions $\pi_{1},\pi_{2}$. In the $(1,3)$-type description, the plane partitions $\pi_{1,2}$ are decomposed into finite Young diagrams as
\bea
\pi_{1}&=\{\lambda^{(1)},\lambda^{(2)},\cdots\},\quad \lambda^{(i)}\succeq \lambda^{(i+1)},\\
\pi_{2}&=\{\mu^{(1)},\mu^{(2)},\cdots \},\quad \mu^{(i)}\succeq \mu^{(i+1)}.
\eea
Namely, the plane partitions $\pi_{1,2}$ are stacks of Young diagrams piled up into the 4-direction. We also choose the orientation of the Young diagrams in a similar way as the two-legs situation of the plane partition:
\bea
&\lambda^{(i)}=\{\lambda^{(i)}_{1},\lambda^{(i)}_{2},\ldots\},\quad \lambda^{(i)}_{j}\geq \lambda^{(i)}_{j+1},\\
&\mu^{(i)}=\{\mu^{(i)}_{1},\mu^{(i)}_{2},\ldots\},\quad \mu^{(i)}_{j}\geq \mu^{(i)}_{j+1}.
\eea
The $\lambda^{(i)}_{j}$ extends in the 3-direction and $\mu^{(i)}_{k}$ extends in the 1-direction. Using the result of \eqref{eq:D6two-legs-hw0} and \eqref{eq:D4tensor-op}, the highest weight is then given as
\bea
:\mathsf{Z}(K,x)\prod_{k}\left(\prod_{i=1}^{\ell(\lambda^{(k)\rmT})}\prod_{j=1}^{\lambda^{(k)\rmT}_{i}}\mathsf{S}_{1}(xq_{2}^{j-1}q_{3}^{i-1}q_{4}^{k-1})^{-1} \prod_{i=1}^{\ell(\mu^{(k)})}\prod_{j=1}^{\mu^{(k)}_{i}}\mathsf{S}_{2}(xq_{1}^{j-1}q_{2}^{\lambda^{(k)\rmT}_{i}}q_{3}^{i-1}q_{4}^{k-1})^{-1}\right):.
\eea
The product of the $\mathsf{S}_{1}$ operators can be simplified as
\bea
:\prod_{\scube\in\pi_{1}}\mathsf{S}_{1}(\chi_{\bar{1},x}(\cube))^{-1}:
\eea
but the second product can not be simplified in this way. The position of the $\mathsf{S}_{2}$ will be modified because of the existence of the $\mathsf{S}_{1}$ operators.

Using this highest weight, one can explicitly write down the free field realization of the contour integral formula, but since it is too complicated, we omit the explicit formula.

\paragraph{Four-legs}
The highest weight for the general case when we have four leg boundary conditions is 
\bea
\adjustbox{valign=c}{ \begin{tikzpicture}[scale=0.25]
    \draw[->] (-4,-5)--(30,-5);
    \node[right] at (30,-5){$4$};
    \fill (0.5,-5) circle (5pt);
    \fill (11.5,-5) circle (5pt);
    \fill (22.5,-5) circle (5pt);
    
    \begin{scope}[scale=0.45]
        \draw[->] (1,0)--(12,0);
        \node[right] at (12,0){$2$};
\draw[->] (1,0)--(1,12);
\node[above] at (1,12){$3$};
\draw[->] (1,0)--(-5.5,-6.5);
\node[below left] at (-5.5,-6.5) {$1$};

%\ystickvar{1}{7}{1}{3};
%\ystickvar{1}{2}{2}{9};
%\ystickvar{1}{2}{3}{9};
%\ystickvar{2}{2}{1}{9};
%\ystickvar{2}{2}{2}{9};

%\xstickvar{3}{1}{1}{4};
%\xstickvar{3}{1}{2}{4};
%\xstickvar{3}{1}{3}{4};
%\xstickvar{3}{2}{1}{4};
%\xstickvar{3}{2}{2}{4};
%\xstickvar{3}{3}{1}{4};

 \foreach \z in {1,2,3,4,5,6,7} \diagcubecolor{1}{1}{\z};
 \zstickvar{1}{1}{8}{3};
 \foreach \z in {1,2,3,4,5} \diagcubecolor{2}{1}{\z};
\zstickvar{2}{1}{6}{5};
 
 \foreach \z in {1,2,3,4} \diagcubecolor{3}{1}{\z};
 \foreach \z in {1,2} \diagcubecolor{4}{1}{\z};
\foreach \z in {1} \diagcubecolor{5}{1}{\z};
\xstickvar{6}{1}{1}{3};
\xstickvar{5}{1}{2}{3.5};
\xstickvar{4}{1}{3}{4};

\foreach \z in {1,2,3,4,5} \diagcubecolor{1}{2}{\z};
 \foreach \z in {1,2,3} \diagcubecolor{2}{2}{\z};
 \foreach \z in {1,2} \diagcubecolor{3}{2}{\z};
 \foreach \z in {1} \diagcubecolor{4}{2}{\z};
\xstickvar{5}{2}{1}{3.5};
\xstickvar{4}{2}{2}{4};

\foreach \z in {1,2,3,4} \diagcubecolor{1}{3}{\z};
 \foreach \z in {1,2} \diagcubecolor{2}{3}{\z};
 \foreach \z in {1} \diagcubecolor{3}{3}{\z};
%\xstickvar{3}{3}{1}{4};

\foreach \z in {1,2,3} \diagcubecolor{1}{4}{\z};
 \foreach \z in {1,2} \diagcubecolor{2}{4}{\z};
 \foreach \z in {1} \diagcubecolor{3}{4}{\z};

\foreach \z in {1,2} \diagcubecolor{1}{5}{\z};
 \foreach \z in {1} \diagcubecolor{2}{5}{\z};

\foreach \z in {1} \diagcubecolor{1}{6}{\z};

\ystickvar{1}{7}{1}{3};
\ystickvar{1}{6}{2}{4};
\ystickvar{1}{5}{3}{5};
\ystickvar{2}{6}{1}{4};
\ystickvar{2}{5}{2}{5};

\end{scope}
\begin{scope}[xshift=300,scale=0.45]
             \draw[->] (1,0)--(12,0);
        \node[right] at (12,0){$2$};
\draw[->] (1,0)--(1,12);
\node[above] at (1,12){$3$};
\draw[->] (1,0)--(-5.5,-6.5);
\node[below left] at (-5.5,-6.5) {$1$};

 \foreach \z in {1,2,3,4,5,6,7} \diagcubecolor{1}{1}{\z};
 \zstickvar{1}{1}{8}{3};
 \foreach \z in {1,2,3,4,5} \diagcubecolor{2}{1}{\z};
\zstickvar{2}{1}{6}{5};
 
 \foreach \z in {1,2,3,4} \diagcubecolor{3}{1}{\z};
 \foreach \z in {1,2} \diagcubecolor{4}{1}{\z};
\foreach \z in {1} \diagcubecolor{5}{1}{\z};
\xstickvar{6}{1}{1}{3};
\xstickvar{5}{1}{2}{3.5};
%\xstickvar{4}{1}{3}{4};

\foreach \z in {1,2,3,4,5} \diagcubecolor{1}{2}{\z};
 \foreach \z in {1,2,3} \diagcubecolor{2}{2}{\z};
 \foreach \z in {1,2} \diagcubecolor{3}{2}{\z};
 \foreach \z in {1} \diagcubecolor{4}{2}{\z};
\xstickvar{5}{2}{1}{3.5};
%\xstickvar{4}{2}{2}{4};

\foreach \z in {1,2,3,4} \diagcubecolor{1}{3}{\z};
 \foreach \z in {1,2} \diagcubecolor{2}{3}{\z};
 \foreach \z in {1} \diagcubecolor{3}{3}{\z};
%\xstickvar{3}{3}{1}{4};

\foreach \z in {1,2,3} \diagcubecolor{1}{4}{\z};
 \foreach \z in {1,2} \diagcubecolor{2}{4}{\z};
 \foreach \z in {1} \diagcubecolor{3}{4}{\z};

\foreach \z in {1,2} \diagcubecolor{1}{5}{\z};
 \foreach \z in {1} \diagcubecolor{2}{5}{\z};

\foreach \z in {1} \diagcubecolor{1}{6}{\z};

\ystickvar{1}{7}{1}{3};
\ystickvar{1}{6}{2}{4};
\ystickvar{1}{5}{3}{5};
\ystickvar{2}{6}{1}{4};
%\ystickvar{2}{5}{2}{5};

\end{scope}
\begin{scope}[xshift=600,scale=0.45]
            \draw[->] (1,0)--(12,0);
        \node[right] at (12,0){$2$};
\draw[->] (1,0)--(1,12);
\node[above] at (1,12){$3$};
\draw[->] (1,0)--(-5.5,-6.5);
\node[below left] at (-5.5,-6.5) {$1$};

 \foreach \z in {1,2,3,4,5,6,7} \diagcubecolor{1}{1}{\z};
 \zstickvar{1}{1}{8}{3};
 \foreach \z in {1,2,3,4,5} \diagcubecolor{2}{1}{\z};
%\zstickvar{2}{1}{6}{5};
 
 \foreach \z in {1,2,3,4} \diagcubecolor{3}{1}{\z};
 \foreach \z in {1,2} \diagcubecolor{4}{1}{\z};
\foreach \z in {1} \diagcubecolor{5}{1}{\z};
\xstickvar{6}{1}{1}{3};
\xstickvar{5}{1}{2}{3.5};
%\xstickvar{4}{1}{3}{4};

\foreach \z in {1,2,3,4,5} \diagcubecolor{1}{2}{\z};
 \foreach \z in {1,2,3} \diagcubecolor{2}{2}{\z};
 \foreach \z in {1,2} \diagcubecolor{3}{2}{\z};
 \foreach \z in {1} \diagcubecolor{4}{2}{\z};
%\xstickvar{5}{2}{1}{3.5};
%\xstickvar{4}{2}{2}{4};

\foreach \z in {1,2,3,4} \diagcubecolor{1}{3}{\z};
 \foreach \z in {1,2} \diagcubecolor{2}{3}{\z};
 \foreach \z in {1} \diagcubecolor{3}{3}{\z};
%\xstickvar{3}{3}{1}{4};

\foreach \z in {1,2,3} \diagcubecolor{1}{4}{\z};
 \foreach \z in {1,2} \diagcubecolor{2}{4}{\z};
 \foreach \z in {1} \diagcubecolor{3}{4}{\z};

\foreach \z in {1,2} \diagcubecolor{1}{5}{\z};
 \foreach \z in {1} \diagcubecolor{2}{5}{\z};

\foreach \z in {1} \diagcubecolor{1}{6}{\z};

\ystickvar{1}{7}{1}{3};
\ystickvar{1}{6}{2}{4};
%\ystickvar{1}{5}{3}{5};
%\ystickvar{2}{6}{1}{4};
%\ystickvar{2}{5}{2}{5};

\end{scope}

\end{tikzpicture}}=\quad :\mathsf{Z}(K,x)\prod_{\shcube\in\mathcal{B}_{\pi_{1}\pi_{2}\pi_{3}\pi_{4}}}\mathsf{A}^{-1}(\chi_{\four,x}(\hcube)):
\eea
and the infinite product is regularized as
\bea
{:\mathsf{Z}(K,x)\prod_{\shcube\in\mathcal{B}_{\pi_{1}\pi_{2}\pi_{3}\pi_{4}}}\mathsf{A}^{-1}(\chi_{\four,x}(\hcube)):}={:\frac{\mathsf{Z}(K,x)}{\prod\limits_{a\in\four}\prod\limits_{\scube\in\pi_{a}}\mathsf{S}_{a}(\chi_{\bar{a},x}(\cube))}\prod_{\shcube\in\mathcal{S}_{\pi_{1}\pi_{2}\pi_{3}\pi_{4}}}\mathsf{A}(\chi_{\four,x}(\hcube)):}.
\eea
\begin{proposition}[D8 four-legs]
The free field realization of the contour integral formula \eqref{eq:D8contourintegral-leg} of the D8 partition function with nontrivial leg boundary conditions is given as
\bea
\mathcal{Z}_{k}&=\frac{\mathcal{G}^{k}}{k!} \oint \prod_{I=1}^{k}\frac{dx_{I}}{2\pi\iota x_{I}} \left\langle \prod_{I=1}^{k}\mathsf{A}^{-1}(x_{I}){:\frac{\mathsf{Z}(K,x)}{\prod\limits_{a\in\four}\prod\limits_{\scube\in\pi_{a}}\mathsf{S}_{a}(\chi_{\bar{a},x}(\cube))}\prod_{\shcube\in\mathcal{S}_{\pi_{1}\pi_{2}\pi_{3}\pi_{4}}}\mathsf{A}(\chi_{\four,x}(\hcube)):} \right\rangle\\
&=\frac{\mathcal{G}^{k}}{k!} \oint \prod_{I=1}^{k}\frac{dx_{I}}{2\pi\iota x_{I}} \prod_{I<J}\mathcal{A}_{\mathbb{C}^{4}}\left(\frac{x_{I}}{x_{J}}\right)^{-1}\prod_{I=1}^{k}\frac{1-Kx/x_{I}}{1-x/x_{I}}\prod_{I=1}^{k}\prod_{a\in\four}\prod_{\scube\in\pi_{a}}g_{\bar{a}}\left(\frac{\chi_{\bar{a},x}(\cube)}{x_{I}}\right)^{-1}\\
&\qquad\qquad  \times \prod_{I=1}^{k}\prod_{\shcube\in\mathcal{S}_{\pi_{1}\pi_{2}\pi_{3}\pi_{4}}}\mathcal{A}_{\mathbb{C}^{4}}\left(\frac{\chi_{\four,x}(\hcube)}{x_{I}}\right).
\eea
\end{proposition}

\subsubsection{Surface boundary conditions}

\paragraph{One-surface}
Let us consider the case when we only have one surface Young diagram extending in the 12-direction. The highest weight comes from the following configuration
\bea
\adjustbox{valign=c}{\begin{tikzpicture}[scale=0.3]
    \draw[->] (-4,-5)--(24,-5);
    \node[right] at (24,-5){$4$};
    \fill (0.5,-5) circle (5pt);
    \fill (11.5,-5) circle (5pt);
    \fill (22.5,-5) circle (5pt);
    
    \begin{scope}[scale=0.35]
        \draw[->] (1,0)--(13,0);
\draw[->] (1,0)--(1,9);
\draw[->] (1,0)--(-5.5,-6.5);
\node[right] at (13,0){$2$};
\node[above] at (1,9){$3$};
\node[below left] at (-5.5,-6.5){$1$};

%\zxsurfacevar{1}{1}{1}{7}{6};
%\zxsurfacevar{1}{2}{1}{7}{6};
%\zxsurfacevar{1}{3}{1}{7}{6};

\xysurfacevar{1}{1}{1}{5}{10};
\xysurfacevar{1}{1}{2}{5}{10};

%\yzsurfacevar{1}{4}{3}{8}{5};
%\yzsurfacevar{2}{4}{3}{8}{5};
\end{scope}
\begin{scope}[xshift=300,scale=0.35]
\draw[->] (1,0)--(13,0);
\draw[->] (1,0)--(1,9);
\draw[->] (1,0)--(-5.5,-6.5);
\node[right] at (13,0){$2$};
\node[above] at (1,9){$3$};
\node[below left] at (-5.5,-6.5){$1$};

%\zxsurfacevar{1}{1}{1}{7}{6};
%\zxsurfacevar{1}{2}{1}{7}{6};
%\zxsurfacevar{1}{3}{1}{7}{6};

\xysurfacevar{1}{1}{1}{5}{10};
%\xysurfacevar{1}{4}{2}{6}{8};

%\yzsurfacevar{1}{3}{2}{9}{6};
%\yzsurfacevar{2}{4}{3}{8}{5}; 
\end{scope}

\begin{scope}[xshift=600,scale=0.35]
  \draw[->] (1,0)--(13,0);
\draw[->] (1,0)--(1,9);
\draw[->] (1,0)--(-5.5,-6.5);
\node[right] at (13,0){$2$};
\node[above] at (1,9){$3$};
\node[below left] at (-5.5,-6.5){$1$};

%\zxsurfacevar{1}{1}{1}{7}{6};
%\zxsurfacevar{1}{2}{1}{7}{6};
%\zxsurfacevar{1}{3}{1}{7}{6};

\xysurfacevar{1}{1}{1}{5}{10};
%\xysurfacevar{1}{4}{2}{6}{8};

%\yzsurfacevar{1}{4}{3}{8}{5};
%\yzsurfacevar{2}{4}{3}{8}{5};
\end{scope}
\end{tikzpicture}}={:\mathsf{Z}(K,x)\prod_{i,j=1}^{\infty}\prod_{\Abox\in\lambda_{12}}\mathsf{A}^{-1}(q_{1}^{i-1}q_{2}^{j-1}\chi_{34,x}(\Bbox)):}.
\eea
In the type $(1,3)$ description, using the decomposition in \eqref{eq:D8surfacedecomp} as $\lambda_{12}=\{k_{12}^{(i)}\mid i=1,\ldots, \}$, the highest weight is expressed as
\bea
:\mathsf{Z}(K,x)\prod_{i,j=1}^{\infty}\prod_{l=1}^{\infty}\prod_{k=1}^{k_{12}^{(l)}}\mathsf{A}^{-1}(q_{1}^{i-1}q_{2}^{j-1}q_{3}^{k-1}q_{4}^{l-1}x):.
\eea

We can further regularize this highest weight as
\bea
{:\mathsf{Z}(K,x)\prod_{i,j=1}^{\infty}\prod_{\Abox\in\lambda_{12}}\mathsf{A}^{-1}(q_{1}^{i-1}q_{2}^{j-1}\chi_{34,x}(\Bbox)):}\simeq{:\frac{\mathsf{Z}(K,x)}{\prod\limits_{\Abox\in\lambda_{12}}\mathsf{X}_{12}(\chi_{34,x}(\Bbox))}:} 
\eea
where the equality is up to non-essential zero-modes. 

Although not so essential, another way to describe this setup is to use the identity in \eqref{eq:vertexoprelation3} and rewrite it as
\bea
:\prod_{i=1}^{\infty}\frac{\mathsf{W}_{\bar{4}}(xq_{3}^{k_{12}^{(i)}}q_{4}^{i-1})}{\mathsf{W}_{\bar{4}}(Kxq_{4}^{i-1})}:.
\eea

\begin{proposition}[D8 one-surface]
The free field realization of the contour integral for this case is given as
\bea\label{eq:D8_1surface_bdy}
\mathcal{Z}_{k}&=\frac{\mathcal{G}^{k}}{k!} \oint \prod_{I=1}^{k}\frac{dx_{I}}{2\pi\iota x_{I}} \left\langle \prod_{I=1}^{k}\mathsf{A}^{-1}(x_{I}):\frac{\mathsf{Z}(K,x)}{\prod\limits_{\Abox\in\lambda_{12}}\mathsf{X}_{12}(\chi_{34,x}(\Bbox))}:\right\rangle\\
&=\frac{\mathcal{G}^{k}}{k!}\oint \prod_{I=1}^{k}\frac{dx_{I}}{2\pi\iota x_{I}} \prod_{I<J}\mathcal{A}_{\mathbb{C}^{4}}\left(\frac{x_{I}}{x_{J}}\right)^{-1}\prod_{I=1}^{k}\frac{1-Kx/x_{I}}{1-x/x_{I}}\prod_{I=1}^{k}\prod_{\Abox \in\lambda_{12}}\mathscr{S}_{34}\left(\frac{\chi_{34,x}(\Bbox)}{x_{I}}\right)^{-1}.
\eea
\end{proposition}

\paragraph{Two--surfaces}For later use, let us consider the situation when we have two surfaces spanning the $12$ and $23$ planes:
\bea
\begin{tikzpicture}[scale=0.3]
    \draw[->] (-4,-5)--(30,-5);
    \node[right] at (30,-5){$4$};
    \fill (0.5,-5) circle (5pt);
    \fill (11.5,-5) circle (5pt);
    \fill (22.5,-5) circle (5pt);
    
    \begin{scope}[scale=0.4]
        \draw[->] (1,0)--(13,0);
\draw[->] (1,0)--(1,9);
\draw[->] (1,0)--(-5.5,-6.5);
\node[right] at (13,0){$2$};
\node[above] at (1,9){$3$};
\node[below left] at (-5.5,-6.5){$1$};

\xysurfacevar{1}{1}{1}{5}{10};
\xysurfacevar{1}{1}{2}{5}{10};

\yzsurfacevar{1}{1}{3}{10}{5};
\yzsurfacevar{2}{1}{3}{10}{5};
\end{scope}
\begin{scope}[xshift=300,scale=0.4]
 \draw[->] (1,0)--(13,0);
\draw[->] (1,0)--(1,9);
\draw[->] (1,0)--(-5.5,-6.5);
\node[right] at (13,0){$2$};
\node[above] at (1,9){$3$};
\node[below left] at (-5.5,-6.5){$1$};

\xysurfacevar{1}{1}{1}{5}{10};
\xysurfacevar{1}{1}{2}{5}{10};

\yzsurfacevar{1}{1}{3}{10}{5};
%\yzsurfacevar{2}{1}{3}{10}{5};
%\yzsurfacevar{2}{4}{3}{8}{5}; 
\end{scope}

\begin{scope}[xshift=600,scale=0.4]
 \draw[->] (1,0)--(13,0);
\draw[->] (1,0)--(1,9);
\draw[->] (1,0)--(-5.5,-6.5);
\node[right] at (13,0){$2$};
\node[above] at (1,9){$3$};
\node[below left] at (-5.5,-6.5){$1$};

\xysurfacevar{1}{1}{1}{5}{10};
%\xysurfacevar{1}{1}{2}{5}{10};

%\yzsurfacevar{1}{1}{3}{10}{5};
%\yzsurfacevar{2}{1}{3}{10}{5};
%\yzsurfacevar{2}{4}{3}{8}{5};
\end{scope}
\end{tikzpicture}
\eea
Using the expression in \eqref{eq:D8surfacedecomp}
\bea
\lambda_{12}=\{k_{12}^{(i)}\mid i=1,\ldots, \infty\},\quad \lambda_{23}=\{k_{23}^{(i)}\mid i=1,\ldots, \infty\}
\eea
the highest weight can be written as
\bea
:\mathsf{Z}(K,x)\prod_{i=1}^{\ell(\lambda_{12})}\prod_{j=1}^{k_{12}^{(i)}}\mathsf{X}_{12}(xq_{4}^{i-1}q_{3}^{j-1})^{-1}\prod_{i=1}^{\ell(\lambda_{23})}\prod_{j=1}^{k_{23}^{(i)}}\mathsf{X}_{23}(xq_{3}^{k_{12}^{(i)}}q_{1}^{j-1}q_{4}^{i-1})^{-1}:.
\eea
Similar to the previous case, a different way to write this will be
\bea
:\prod_{i=1}^{\infty}\frac{\mathsf{W}_{\bar{4}}(xq_{4}^{i-1}q_{1}^{k_{23}^{(i)}}q_{3}^{k_{12}^{(i)}})}{\mathsf{W}_{\bar{4}}(Kxq_{4}^{i-1})}:.
\eea
The free field realization of the contour integral formula can be written explicitly but we omit it.

\paragraph{General case}
For the generic case when we have six surfaces, the highest weight is given as
\bea
{:\mathsf{Z}(K,x)\prod_{\shcube\in\mathcal{B}_{\{\lambda_{A}\}}}\mathsf{A}^{-1}(\chi_{\four,x}(\hcube)):}\simeq {:\frac{\mathsf{Z}(K,x)}{\prod\limits_{A\in\six}\prod\limits_{\Abox\in\lambda_{A}}\mathsf{X}_{A}(\chi_{\bar{A},x}(\Bbox))}\prod_{\shcube\in\mathcal{S}_{\{\lambda_{A}\}}}\mathsf{A}(\chi_{\four,x}(\hcube)):}.
\eea

\begin{proposition}[D8 six-surfaces]
The free field realization of the contour integral of the D8 partition function with nontrivial surface boundary conditions \eqref{eq:D8contourintegral-surface} is
\bea
\mathcal{Z}_{k}&=\frac{\mathcal{G}^{k}}{k!} \oint \prod_{I=1}^{k}\frac{dx_{I}}{2\pi\iota x_{I}} \left\langle \prod_{I=1}^{k}\mathsf{A}^{-1}(x_{I})\,{:\frac{\mathsf{Z}(K,x)}{\prod\limits_{A\in\six}\prod\limits_{\Abox\in\lambda_{A}}\mathsf{X}_{A}(\chi_{\bar{A},x}(\Bbox))}\prod_{\shcube\in\mathcal{S}_{\{\lambda_{A}\}}}\mathsf{A}(\chi_{\four,x}(\hcube)):}\right\rangle\\
&= \frac{\mathcal{G}^{k}}{k!}\oint \prod_{I=1}^{k}\frac{dx_{I}}{2\pi\iota x_{I}}\prod_{I=1}^{k}\frac{1-Kx/x_{I}}{1-x/x_{I}}\prod_{I<J}\mathcal{A}_{\mathbb{C}^{4}}\left(\frac{x_{I}}{x_{J}}\right)^{-1}\prod_{I=1}^{k}\prod_{A\in\six}\prod_{\Abox\in\lambda_{A}}\mathscr{S}_{\bar{A}}\left(\frac{\chi_{A,x}(\Bbox)}{x_{I}}\right)^{-1}\prod_{I=1}^{k}\prod_{\shcube\in\mathcal{S}_{\{\lambda_{A}\}}}\mathcal{A}_{\mathbb{C}^{4}}\left(\frac{\chi_{\four,x}(\hcube)}{x_{I}}\right).
\eea
\end{proposition}

\subsubsection{Hypersurface boundary conditions}
When we have only one type of hypersurfaces spanning the $124$-plane, the highest weight is given as
\bea
\adjustbox{valign=c}{\begin{tikzpicture}[scale=0.35]
    \draw[->] (-4,-5)--(18,-5);
    \node[right] at (18,-5){$4$};
    \fill (0.5,-5) circle (5pt);
    \fill (7.5,-5) circle (5pt);
    \fill (14.5,-5) circle (5pt);
    
    \begin{scope}[scale=0.3]
        \draw[->] (1,0)--(13,0);
\draw[->] (1,0)--(1,9);
\draw[->] (1,0)--(-5.5,-6.5);
\node[right] at (13,0){$2$};
\node[above] at (1,9){$3$};
\node[below left] at (-5.5,-6.5){$1$};

%\zxsurfacevar{1}{1}{1}{7}{6};
%\zxsurfacevar{1}{2}{1}{7}{6};
%\zxsurfacevar{1}{3}{1}{7}{6};

\xysurfacevar{1}{1}{1}{5}{10};
\xysurfacevar{1}{1}{2}{5}{10};

%\yzsurfacevar{1}{4}{3}{8}{5};
%\yzsurfacevar{2}{4}{3}{8}{5};
\end{scope}
\begin{scope}[xshift=200,scale=0.3]
 \draw[->] (1,0)--(13,0);
\draw[->] (1,0)--(1,9);
\draw[->] (1,0)--(-5.5,-6.5);
\node[right] at (13,0){$2$};
\node[above] at (1,9){$3$};
\node[below left] at (-5.5,-6.5){$1$};

%\zxsurfacevar{1}{1}{1}{7}{6};
%\zxsurfacevar{1}{2}{1}{7}{6};
%\zxsurfacevar{1}{3}{1}{7}{6};

\xysurfacevar{1}{1}{1}{5}{10};
\xysurfacevar{1}{1}{2}{5}{10};

%\yzsurfacevar{1}{4}{3}{8}{5};
%\yzsurfacevar{2}{4}{3}{8}{5};
%\yzsurfacevar{2}{4}{3}{8}{5}; 
\end{scope}

\begin{scope}[xshift=400,scale=0.3]
 \draw[->] (1,0)--(13,0);
\draw[->] (1,0)--(1,9);
\draw[->] (1,0)--(-5.5,-6.5);
\node[right] at (13,0){$2$};
\node[above] at (1,9){$3$};
\node[below left] at (-5.5,-6.5){$1$};

%\zxsurfacevar{1}{1}{1}{7}{6};
%\zxsurfacevar{1}{2}{1}{7}{6};
%\zxsurfacevar{1}{3}{1}{7}{6};

\xysurfacevar{1}{1}{1}{5}{10};
\xysurfacevar{1}{1}{2}{5}{10};

%\yzsurfacevar{1}{4}{3}{8}{5};
%\yzsurfacevar{2}{4}{3}{8}{5};
%\yzsurfacevar{2}{4}{3}{8}{5};
\end{scope}
\end{tikzpicture}}={:\mathsf{Z}(K,x)\prod_{i,j,l=1}^{\infty}\prod_{k=1}^{k_{\bar{3}}}\mathsf{A}^{-1}(q_{1}^{i-1}q_{2}^{j-1}q_{3}^{k-1}q_{4}^{l-1}x):}.
\eea
The infinite part is regularized as
\bea\label{eq:D81hypersurfacevertex}
:\mathsf{Z}(K,x)\prod_{i,j,l=1}^{\infty}\prod_{k=1}^{k_{\bar{3}}}\mathsf{A}^{-1}(q_{1}^{i-1}q_{2}^{j-1}q_{3}^{k-1}q_{4}^{l-1}x):\,\simeq {:\frac{\mathsf{Z}(K,x)}{\prod_{k=1}^{k_{\bar{3}}}\mathsf{W}_{\bar{3}}(xq_{3}^{k-1})}:}
\eea
where we used \eqref{eq:vertexoprelation2} and the equality of \eqref{eq:D81hypersurfacevertex} is up to non-essential zero-modes. We actually can further simplify the highest weight as
\bea
{:\frac{\mathsf{Z}(K,x)}{\prod_{k=1}^{k_{\bar{3}}}\mathsf{W}_{\bar{3}}(xq_{3}^{k-1})}:}={:\frac{\mathsf{Z}(q_{3}^{k_{\bar{3}}}x)}{\mathsf{Z}(Kx)}:}=\mathsf{Z}(Kq_{3}^{-k_{\bar{3}}},q_{3}^{k_{\bar{3}}}x)
\eea
where we used \eqref{eq:vertexoprelation}. 
Note that this time the equality is exact.

We can do a similar computation for the generic case when we have four hypersurfaces and the highest weight is simply given as
\bea
\mathsf{Z}(Kq_{1}^{-k_{\bar{1}}}q_{2}^{-k_{\bar{2}}}q_{3}^{-k_{\bar{3}}}q_{4}^{-k_{\bar{4}}},q_{1}^{k_{\bar{1}}}q_{2}^{k_{\bar{2}}}q_{3}^{k_{\bar{3}}}q_{4}^{k_{\bar{4}}}x)={:\frac{\mathsf{Z}(q_{1}^{k_{\bar{1}}}q_{2}^{k_{\bar{2}}}q_{3}^{k_{\bar{3}}}q_{4}^{k_{\bar{4}}}x)}{\mathsf{Z}(Kx)}:}.
\eea

\begin{proposition}[D8 four-hypersurfaces]
The free field realization of the contour integral formula of the D8 partition function with four nontrivial hypersurface boundary conditions \eqref{eq:D8contourintegral-hypersurface} is
\bea
\mathcal{Z}_{k}&=\frac{\mathcal{G}^{k}}{k!}\oint \prod_{I=1}^{k}\frac{dx_{I}}{2\pi\iota x_{I}} \left\langle \prod_{I=1}^{k}\mathsf{A}^{-1}(x_{I})\mathsf{Z}(Kq_{1}^{-k_{\bar{1}}}q_{2}^{-k_{\bar{2}}}q_{3}^{-k_{\bar{3}}}q_{4}^{-k_{\bar{4}}},q_{1}^{k_{\bar{1}}}q_{2}^{k_{\bar{2}}}q_{3}^{k_{\bar{3}}}q_{4}^{k_{\bar{4}}}x)\right\rangle \\
&=\frac{\mathcal{G}^{k}}{k!} \oint \prod_{I=1}^{k}\frac{dx_{I}}{2\pi\iota x_{I}}\prod_{I=1}^{k}\frac{1-Kx/x_{I}}{1-q_{1}^{k_{\bar{1}}}q_{2}^{k_{\bar{2}}}q_{3}^{k_{\bar{3}}}q_{4}^{k_{\bar{4}}}x/x_{I}}\prod_{I<J}\mathcal{A}_{\mathbb{C}^{4}}\left(\frac{x_{I}}{x_{J}}\right)^{-1}.
\eea
\end{proposition}

\subsection{Dynamical generation of boundary conditions}

We have discussed the D0 brane counting with the fixed boundary conditions associated with the D2, D4, D6 branes extending in the non-compact directions.
From this point of view, we have focused only on the dynamics of D0 branes, while the remaining branes are treated as non-dynamical objects.
On the other hand, in order to apply the vertex formalism to construct generic toric geometries, it is indispensable to glue the building blocks by summing up all the possible boundary conditions, and hence we should incorporate dynamics of the non-compact branes as well \cite{Kimura-Noshita1}.
In this part, we explain that the contour integral formula arising from the vertex operator formalism naturally describes the dynamics of non-compact branes, which also gives rise to the edge and the face contribution obtained in~\cite{Nekrasov:2023nai}.

As discussed above, the boundary conditions are concisely organized by the vertex operators corresponding to D2, D4 branes.
We have the following operator product factors.
\begin{lemma}\label{lem:SS_XX_OPE}
    Let $x_{ij} = x_i/x_j$.
    Then, we have {\allowdisplaybreaks
   \begin{subequations}
        \begin{align}
            \prod_{i=1}^k \mathsf{S}_a(x_i) & = {: \prod_{i=1}^k \mathsf{S}_a(x_i) :} \times 
            \begin{cases}
                \displaystyle 
                \prod_{1 \le i < j \le k} \frac{(x_{ij};q_a^{-1})_\infty \prod_{b \in \bar{a}} (q_{ab}^{-1} x_{ij};q_a^{-1})_\infty}{(q_a^{-1} x_{ij};q_a^{-1})_\infty \prod_{b \in \bar{a}} (q_b x_{ij};q_a^{-1})_\infty}
                & (|q_a| > 1) \\
                \displaystyle 
                \prod_{1 \le i < j \le k} \frac{(x_{ij};q_a)_\infty \prod_{b \in \bar{a}} (q_{ab} x_{ij};q_a)_\infty}{(q_a x_{ij};q_a)_\infty \prod_{b \in \bar{a}} (q_b^{-1} x_{ij};q_a)_\infty}
                & (|q_a| < 1)
            \end{cases}\\
            \prod_{i=1}^k \mathsf{X}_{ab}(x_i) & = {: \prod_{i=1}^k \mathsf{X}_{ab}(x_i) :} \times 
            \begin{cases}
                \displaystyle 
                \prod_{1 \le i < j \le k} \frac{(x_{ij};q_{a,b}^{-1})_\infty (q_{ab}^{-1} x_{ij};q_{a,b}^{-1})_\infty}{\prod_{c \in \overline{ab}} (q_{c} x_{ij};q_{a,b}^{-1})_\infty}
                & (|q_{a,b}| > 1) \\
                \displaystyle 
                \prod_{1 \le i < j \le k} \frac{(x_{ij};q_{a,b})_\infty (q_{ab} x_{ij};q_{a,b})_\infty}{\prod_{c \in \overline{ab}} (q_{c}^{-1} x_{ij};q_{a,b})_\infty}
                & (|q_{a,b}| < 1) \\
                \displaystyle 
                \prod_{1 \le i < j \le k} \frac{\prod_{c \in \overline{ab}}(q_{bc}^{-1} x_{ij};q_a,q_b^{-1})_\infty}{ (q_{a} x_{ij};q_a,q_b^{-1})_\infty (q_{b}^{-1} x_{ij};q_a,q_b^{-1})_\infty}
                & (|q_{a}| < 1, |q_b| > 1)
            \end{cases}
        \end{align}
    \end{subequations}
    and
    \begin{subequations}
        \begin{align}
            \mathsf{S}_a(x)^{-1} \mathsf{Z}(K,x') & = {:\frac{\mathsf{Z}(K,x')}{\mathsf{S}_a(x)}:} \times
            \begin{cases}
                \displaystyle \frac{(Kx'/x;q_a^{-1})_\infty}{(x'/x;q_a^{-1})_\infty} & (|q_a|>1) \\[1em]
                \displaystyle \frac{(x/Kx';q_a)_\infty}{(x/x';q_a)_\infty} \frac{\theta(x/x';q_a)}{\theta(x/Kx';q_a)} & (|q_a|<1) 
            \end{cases}\\
            \mathsf{X}_{ab}(x)^{-1} \mathsf{Z}(K,x') & = {: \frac{\mathsf{Z}(K,x')}{\mathsf{X}_{ab}(x)} :} \times
            \begin{cases}
                \displaystyle \frac{(Kx'/x;q_{a,b}^{-1})_\infty}{(x'/x;q_{a,b}^{-1})_\infty} & (|q_{a,b}|>1) \\[1em]
                \displaystyle \frac{(x/Kx';q_{a,b})_\infty}{(x/x';q_{a,b})_\infty} \frac{\Gamma(x/Kx';q_{a,b})}{\Gamma(x/x';q_{a,b})} & (|q_{a,b}|>1) \\[1em]
                \displaystyle \frac{(x/Kx';q_{a,b})_\infty}{(x/x';q_{a,b})_\infty} \frac{\theta(x/x';q_a^{-1},q_b)_\infty}{\theta(x/Kx';q_a^{-1},q_b)_\infty} & (|q_{a}|>1,|q_b|<1)
            \end{cases}
        \end{align}
    \end{subequations}}
\end{lemma}
Based on these factors, we have a generalized version of Prop.~\ref{prop:D0general} which provides the contour integral formula incorporating the boundary conditions.
\begin{conjecture}
    Let $\underline{d}_{\underline{\mathbf{4}}} = (d_a)_{a \in \underline{\mathbf{4}}}$, $\underline{d}_{\underline{\mathbf{6}}} = (d_A)_{A \in \underline{\mathbf{6}}}$ the numbers of the non-compact D2, D4 branes, and $(y_{a,i})_{a \in \underline{{\mathbf{4}}},i=1,\ldots,d_a}$, $(z_{A,i})_{A \in \underline{{\mathbf{6}}},i=1,\ldots,d_A}$ the positions of these branes.
    The free field realization of the contour integral formula for the single D8 brane system is given as follows, 
    \begin{align}
    \mathcal{Z}_{k,\underline{d}_{\underline{\mathbf{4}}},\underline{d}_{\underline{\mathbf{6}}}} &= \frac{1}{k!\underline{d}_{\underline{\mathbf{4}}}!\underline{d}_{\underline{\mathbf{6}}}!} \oint \prod_{I=1}^k \frac{dx_I}{2 \pi \iota x_I} \oint \prod_{\substack{a \in \underline{\mathbf{4}} \\ i = 1, \ldots, d_a}} \frac{dy_{a,i}}{2 \pi \iota y_{a,i}} \oint \prod_{\substack{A \in \underline{\mathbf{6}} \\ i = 1, \ldots, d_A}} \frac{dz_{A,i}}{2 \pi \iota z_{A,i}} \nonumber \\ 
    & \hspace{8em} \left\langle \prod_{I=1,\ldots,k}\mathsf{A}^{-1}(x_{I}) \prod_{\substack{a \in \underline{\mathbf{4}} \\ i = 1, \ldots, d_a}}\mathsf{S}_{a}^{-1}(y_{a,i}) \prod_{\substack{A \in \underline{\mathbf{6}} \\ i = 1, \ldots, d_A}}\mathsf{X}_{A}^{-1}(z_{A,i}) \mathsf{Z}(K,x) \right\rangle
    \end{align}
    where we use the notation, $\underline{d}_{\underline{\mathbf{4}}}! = \prod_{a \in \underline{\mathbf{4}}} d_a!$, $\underline{d}_{\underline{\mathbf{6}}}! = \prod_{A \in \underline{\mathbf{6}}} d_A!$.
    The choice of integration contours of the $y$- and $z$-variables controls the boundary conditions.
    This integral also reproduces the edge and face contributions.
\end{conjecture}
From the point of view of the vertex operators, generalization to the multiple D8 brane system is straightforward.
We can in the same way consider the D4 and D6 brane systems by replacing the $\mathsf{Z}$-operator with the corresponding D4 and D6 vertex operators.
We examine this conjecture with several examples in the following.

%\subsubsection{D8 partition function}
%
%We first consider the D8 partition function with boundary conditions discussed in Sec.~\ref{sec:D8_w_bdry}. 

\paragraph{One-leg boundary conditions}
The contour integral formula with a one-leg boundary condition~\eqref{eq:D8_1leg_bdy} may be obtained from the following vertex operator correlator.
\begin{proposition}
Let $|q_4| > 1$.
For the one-leg boundary condition such that $|\pi_4| = d$, namely $d_a = d \delta_{a,4}$ for $a \in \underline{\mathbf{4}}$, we have
\begin{align}
\mathcal{Z}_{k,d,0}&= \frac{1}{k!d!} \oint \prod_{I=1}^{k}\frac{dx_{I}}{2\pi\iota x_{I}} \oint \prod_{i=1}^{d}\frac{dy_{i}}{2\pi\iota y_{i}} \left\langle \prod_{I=1}^{k}\mathsf{A}^{-1}(x_{I}) \prod_{i=1}^{d}\mathsf{S}_{4}^{-1}(y_i) \mathsf{Z}(K,x) \right\rangle \nonumber \\
& = \frac{{\mathcal{Z}}_\text{edge}}{k!} \oint \prod_{I=1}^{k}\frac{dx_{I}}{2\pi\iota x_{I}} \left\langle \prod_{I=1}^{k}\mathsf{A}^{-1}(x_{I}):\frac{\mathsf{Z}(K,x)}{\prod_{\scube \in\pi_4}\mathsf{S}_{4}(\chi_{\bar{4},x}(\cube))}: \right\rangle
\end{align}
where the integration contour of the $y$-variables is taken to pick up the poles at $y_i = \chi_{\bar{4},x}(\cube)$ for $\cube \in \pi_4$, and the edge factor is given by
\begin{align}
    {\mathcal{Z}}_\text{edge} = \left( q_{1,2,3};q_4^{-1} \right)^d \mathbb{I}\left[ \frac{1}{\mathbf{P}_4} \left( (1 -\mathbf{Q}_{1,2,3}^\vee) \mathbf{S}^\vee \mathbf{S} - \mathbf{N}^\vee \mathbf{S}  \right) \right] \left.\prod_{1 \le i < j \le d} \frac{\theta(q_{1,2,3}y_{ji};q_4^{-1})}{\theta(y_{ji};q_4^{-1})}\right|_{y_i = \chi_{\bar{4},x}(\cube)}
\end{align}
with $\mathbf{Q}_{1,2,3} = \sum_{i=1}^3 \mathbf{Q}_i$ and $\mathbf{S} = \sum_{\scube \in \pi_4} \chi_{\bar{4},x}(\cube)$.
\end{proposition}
\begin{proof}
By Lemma~\ref{lem:SS_XX_OPE}, we have
\begin{align}
\mathcal{Z}_{k,d,0} &= \frac{1}{k!d!} \oint \prod_{I=1}^{k}\frac{dx_{I}}{2\pi\iota x_{I}} \oint \prod_{i=1}^{d}\frac{dy_{i}}{2\pi\iota y_{i}} \prod_{I<J}^k\mathcal{A}_{\mathbb{C}^{4}}\left(x_{IJ}\right)^{-1}\prod_{I=1}^{k}\prod_{i=1}^{d}g_{\bar{4}}\left(\frac{y_i}{x_{I}}\right)^{-1} \prod_{I=1}^{k}\frac{1-Kx/x_{I}}{1-x/x_{I}} \nonumber \\
& \hspace{10em} \times \prod_{i<j}^d \frac{(y_{ij};q_4^{-1})_\infty(q_{12,23,31}y_{ij};q_4^{-1})_\infty}{(q_{1,2,3}y_{ij};q_4^{-1})_\infty(q_4^{-1}y_{ij};q_4^{-1})_\infty} \prod_{i=1}^d \frac{(Kx/y_i;q_4^{-1})_\infty}{(x/y_i;q_4^{-1})_\infty} \nonumber \\
&= \frac{1}{k!d!} \oint \prod_{I=1}^{k}\frac{dx_{I}}{2\pi\iota x_{I}} \oint \prod_{i=1}^{d}\frac{dy_{i}}{2\pi\iota y_{i}} \prod_{I<J}^k\mathcal{A}_{\mathbb{C}^{4}}\left(x_{IJ}\right)^{-1}\prod_{I=1}^{k}\prod_{i=1}^{d}g_{\bar{4}}\left(\frac{y_i}{x_{I}}\right)^{-1} \prod_{I=1}^{k}\frac{1-Kx/x_{I}}{1-x/x_{I}} \nonumber \\
& \hspace{10em} \times \prod_{i \neq j}^d \frac{(y_{ij};q_4^{-1})_\infty}{(q_{1,2,3}y_{ij};q_4^{-1})_\infty}\prod_{i=1}^d \frac{(Kx/y_i;q_4^{-1})_\infty}{(x/y_i;q_4^{-1})_\infty} \prod_{i<j}^d \frac{\theta(q_{1,2,3}y_{ji};q_4^{-1})}{\theta(y_{ji};q_4^{-1})} 
\, .
\end{align}
Evaluating the residue at $y_i = \chi_{\bar{4},x}(\cube)$ for $\cube \in \pi_4$, we obtain the result.
\end{proof}
The edge factor $\mathcal{Z}_\text{edge}$ agrees with that given in~\cite{Nekrasov:2023nai} up to the boundary contribution of the vertex function~\cite{Kimura:2024xpr}.
We also remark that the poles at $y_i \mapsto x q_1^{i_1 - 1} q_2^{i_2 - 1} q_3^{i_3 - 1} q_4^{\mathbb{Z}_{<0}}$ for $(i_1,i_2,i_3) = \cube \in \pi_4$, that we do not take into account in this case, give rise to the one-leg PT4 vertex~\cite{Cao:2023lon,Piazzalunga:2023qik,Kimura:2024xpr}.

\paragraph{Two-leg boundary conditions}
We can similarly consider the two-leg boundary condition. 
Let us examine a concrete example with $d_{\underline{\mathbf{4}}}=(1,2,0,0)$ and $|q_{1,2}|>1$,
\begin{align}
    & Z_{k,d_{\underline{\mathbf{4}}},0} \nonumber \\
    & = \frac{1}{2 \cdot k!} \oint \prod_{I=1}^k \frac{dx_I}{2\pi\iota x_I} \oint \frac{dy_{1}}{2\pi\iota y_{1}} \oint \prod_{i=1,2}\frac{dy_{2,i}}{2\pi\iota y_{2,i}} \left< \prod_{I=1}^k \mathsf{A}^{-1}(x_I) \mathsf{S}_1^{-1}(y_{1}) \prod_{i=1,2} \mathsf{S}_2^{-1}(y_{2,i}) \mathsf{Z}(K,x) \right> \nonumber \\
    & = \frac{1}{2 \cdot k!} \oint \prod_{I=1}^k \frac{dx_I}{2\pi\iota x_I} \oint \frac{dy_{1}}{2\pi\iota y_{1}} \oint \prod_{i=1,2}\frac{dy_{2,i}}{2\pi\iota y_{2,i}} \prod_{I<J}^k\mathcal{A}_{\mathbb{C}^{4}}\left(x_{IJ}\right)^{-1}\prod_{I=1}^{k}\left[g_{\bar{1}}\left(\frac{y_1}{x_{I}}\right)^{-1}\prod_{i=1,2}g_{\bar{2}}\left(\frac{y_{2,i}}{x_{I}}\right)^{-1}\right] 
    \nonumber \\
    & \hspace{6em} \times \prod_{I=1}^{k}\frac{1-Kx/x_{I}}{1-x/x_{I}} \frac{(y_{2,1}/y_{2,2};q_2^{-1})_\infty(q_{12,23,24}^{-1}y_{2,1}/y_{2,2};q_2^{-1})_\infty}{(q_{1,3,4}y_{2,1}/y_{2,2};q_2^{-1})_\infty(q_{2}^{-1}y_{2,1}/y_{2,2};q_2^{-1})_\infty} \nonumber \\
    & \hspace{6em} \times \frac{(Kx/y_1;q_1^{-1})_\infty}{(x/y_1;q_1^{-1})_\infty} \prod_{i=1,2} \frac{(Kx/y_{2,i};q_2^{-1})_\infty}{(x/y_{2,i};q_{2}^{-1})_\infty} \prod_{i=1,2} \mathscr{S}_{34}(q_1 y_{2,i}/y_1) \, .
\end{align}
For this contour integral, we first take the pole at $y_1 = x$.
Then, we take the pole at $y_{2,1} = x q_2$ in the $\mathscr{S}$-function.
There are three possible poles for the remaining $y_{2,2}$-variable, $y_{2,2} = q_{12}x$, $q_3x$, and $q_4x$, which correspond to the following configurations,
\if0
\begin{align}
    & Z_{k,d_{\underline{\mathbf{4}}},0} 
    \nonumber \\
    & = \frac{1}{2 \cdot k!} \frac{(Kq_1^{-1};q_1^{-1})_\infty}{(q_1^{-1};q_1^{-1})_\infty} \oint \prod_{I=1}^k \frac{dx_I}{2\pi\iota x_I} \oint \prod_{i=1,2}\frac{dy_{2,i}}{2\pi\iota y_{2,i}} \prod_{I<J}^k\mathcal{A}_{\mathbb{C}^{4}}\left(x_{IJ}\right)^{-1} \prod_{I=1}^{k}\left[g_{\bar{1}}\left(\frac{x}{x_{I}}\right)^{-1}\prod_{i=1,2}g_{\bar{2}}\left(\frac{y_{2,i}}{x_{I}}\right)^{-1}\right] \nonumber \\
    & \hspace{6em} \times \prod_{I=1}^{k}\frac{1-Kx/x_{I}}{1-x/x_{I}} \frac{(y_{2,1}/y_{2,2};q_2^{-1})_\infty(q_{12,23,24}^{-1}y_{2,1}/y_{2,2};q_2^{-1})_\infty}{(q_{1,3,4}y_{2,1}/y_{2,2};q_2^{-1})_\infty(q_{2}^{-1}y_{2,1}/y_{2,2};q_2^{-1})_\infty} \nonumber \\
    & \hspace{6em} \times \prod_{i=1,2} \frac{(Kx/y_{2,i};q_2^{-1})_\infty}{(x/y_{2,i};q_{2}^{-1})_\infty} \prod_{i=1,2} \mathscr{S}_{34}(q_1 y_{2,i}/x) \, .
\end{align}
\fi
\begin{equation}
\begin{tikzpicture}[scale=0.35,baseline=(current bounding  box.center)]
    \begin{scope}[scale=.5]
        \draw[->] (1,0)--(13,0);
        \node[right] at (13,0){$2$};
        \draw[->] (1,0)--(1,9);
        \node[above] at (1,9){$3, 4$};
        \draw[->] (1,0)--(-5.5,-6.5);
        \node[below left] at (-5.5,-6.5){$1$};
        \xstickvar{1}{1}{1}{5};
        \ystickvar{1}{2}{1}{9};
        \ystickvar{2}{2}{1}{9};
    \end{scope}
    \begin{scope}[shift={(15,0)},scale=.5]
        \draw[->] (1,0)--(13,0);
        \node[right] at (13,0){$2$};
        \draw[->] (1,0)--(1,9);
        \node[above] at (1,9){$3, 4$};
        \draw[->] (1,0)--(-5.5,-6.5);
        \node[below left] at (-5.5,-6.5){$1$};
        \xstickvar{1}{1}{1}{5};
        \ystickvar{1}{2}{1}{9};
        \ystickvar{1}{1}{2}{10};
    \end{scope}
\end{tikzpicture}
\end{equation}
In this way, we can dynamically generate the leg boundary conditions from the contour integral of the $y$-variables, which are identified with the D2 brane positions.
The higher-rank generalization, i.e., the multiple D8 brane system, is straightforward.
\paragraph{Surface boundary conditions}
We then discuss the contour integral formula with the surface boundary condition~\eqref{eq:D8_1surface_bdy}.
\begin{proposition}
Let $|q_{1,2}| > 1$.
For the one surface boundary condition such that $|\lambda_{12}| = d$, namely $d_A = d \delta_{A,12}$ for $A \in \underline{\mathbf{6}}$, we have
\begin{align}
\mathcal{Z}_{k,0,d}&= \frac{1}{k!d!} \oint \prod_{I=1}^{k}\frac{dx_{I}}{2\pi\iota x_{I}} \oint \prod_{i=1}^{d}\frac{dy_{i}}{2\pi\iota y_{i}} \left\langle \prod_{I=1}^{k}\mathsf{A}^{-1}(x_{I}) \prod_{i=1}^{d}\mathsf{X}_{12}^{-1}(y_i) \mathsf{Z}(K,x) \right\rangle \nonumber \\
& = \frac{{\mathcal{Z}}_\text{face}}{k!} \oint \prod_{I=1}^{k}\frac{dx_{I}}{2\pi\iota x_{I}} \left\langle \prod_{I=1}^{k}\mathsf{A}^{-1}(x_{I}):\frac{\mathsf{Z}(K,x)}{\prod_{\Abox \in\lambda_{12}}\mathsf{X}_{12}(\chi_{34,x}(\Bbox))}: \right\rangle
\end{align}
where the integration contour of the $y$-variables is taken to pick up the poles at $\{y_i\} = \{\chi_{34,x}(\Bbox)\}_{\Bbox \in \lambda_{12}}$, and the surface factor is given by 
\beq
    {\mathcal{Z}}_\text{face} = \Gamma(q_3;q_{1,2}^{-1})^d \, \mathbb{I}\left[ \frac{1}{\mathbf{P}_{12}} \left( \mathbf{P}_3^\vee \mathbf{X}^\vee \mathbf{X} - \mathbf{N}^\vee \mathbf{X}  \right) \right] \left.\prod_{1 \le i < j \le d} \frac{\Gamma(y_{ji};q_{1,2}^{-1})}{\Gamma(q_3 y_{ji};q_{1,2}^{-1})}\right|_{y_i = \chi_{34,x}(\Bbox)}
\eeq
with $\mathbf{X} = \sum_{\Abox \in \lambda_{12}} \chi_{34,x}(\Bbox)\,.$
\end{proposition}
\begin{proof}
By Lemma~\ref{lem:SS_XX_OPE}, we have
\begin{align}
\mathcal{Z}_{k,0,d} &= \frac{1}{k!d!} \oint \prod_{I=1}^{k}\frac{dx_{I}}{2\pi\iota x_{I}} \oint \prod_{i=1}^{d}\frac{dy_{i}}{2\pi\iota y_{i}} \prod_{I<J}^k\mathcal{A}_{\mathbb{C}^{4}}\left(x_{IJ}\right)^{-1}\prod_{I=1}^{k}\prod_{i=1}^{d}\mathscr{S}_{34}\left(\frac{y_i}{x_{I}}\right)^{-1} \prod_{I=1}^{k}\frac{1-Kx/x_{I}}{1-x/x_{I}} \nonumber \\
& \hspace{10em} \times \prod_{i<j}^d \frac{(y_{ij};q_{1,2}^{-1})_\infty(q_{34}y_{ij};q_{1,2}^{-1})_\infty}{(q_{3,4}y_{ij};q_{1,2}^{-1})_\infty} \prod_{i=1}^d \frac{(Kx/y_i;q_{1,2}^{-1})_\infty}{(x/y_i;q_{1,2}^{-1})_\infty} \nonumber \\
&= \frac{1}{k!d!} \oint \prod_{I=1}^{k}\frac{dx_{I}}{2\pi\iota x_{I}} \oint \prod_{i=1}^{d}\frac{dy_{i}}{2\pi\iota y_{i}} \prod_{I<J}^k\mathcal{A}_{\mathbb{C}^{4}}\left(x_{IJ}\right)^{-1}\prod_{I=1}^{k}\prod_{i=1}^{d}\mathscr{S}_{34}\left(\frac{y_i}{x_{I}}\right)^{-1} \prod_{I=1}^{k}\frac{1-Kx/x_{I}}{1-x/x_{I}} \nonumber \\
& \hspace{10em} \times \prod_{i \neq j}^d \frac{(y_{ij};q_{1,2}^{-1})_\infty}{(q_{3}y_{ij};q_{1,2}^{-1})_\infty}\prod_{i=1}^d \frac{(Kx/y_i;q_{1,2}^{-1})_\infty}{(x/y_i;q_{1,2}^{-1})_\infty} \prod_{i<j}^d \frac{\Gamma(y_{ji};q_{1,2}^{-1})}{\Gamma(q_3 y_{ji};q_{1,2}^{-1})} 
\, .
\end{align}
Evaluating the residue at $y_i = \chi_{34,x}(\Bbox)$ for $\Bbox \in \lambda_{12}$, we obtain the result.
\end{proof}
The multi-surface boundary condition can be similarly discussed as in the case of the multi-leg configuration.

\subsection{Donaldson--Thomas \texorpdfstring{$qq$}{qq}-characters}\label{sec:general_qq} 
The above free field realizations imply the existence of an underlying quantum algebraic structure \cite{Kimura:2015rgi,Kimura:2019hnw,Kimura:2022zsm}. When we say we have a quantum algebraic structure, we are meaning that there is a quantum algebraic operator whose expectation value gives the partition functions. Eventually, this means that we have the BPS/CFT correspondence.

%To make the discussion concrete, we will focus on the setup discussed in the previous sections. 
The free field realization given in previous sections means that we have an operator lift up of the partition function:
\bea
\mathsf{T}_{k}(v_{i},q_{a})&=\frac{\mathcal{G}^{k}}{k!}\oint_{\text{JK}} \prod_{I=1}^{k}\frac{dx_{I}}{2\pi\iota x_{I}}\prod_{I=1}^{k}\mathsf{A}(x_{I})^{-1}:\prod_{i}\mathsf{V}_{i}(v_{i}):\\
&=\frac{\mathcal{G}^{k}}{k!}\oint_{\text{JK}} \prod_{I=1}^{k}\frac{dx_{I}}{2\pi\iota x_{I}}\mathcal{Z}(v_{i},x_{I}):\prod_{I=1}^{k}\mathsf{A}(x_{I})^{-1}\prod_{i}\mathsf{V}_{i}(v_{i}):
\eea
where $\mathcal{Z}(v_{i},x_{I})$ is some rational function, and $\{v_{i}\},\{q_{a}\}$ are the flavor symmetry fugacities of the underlying supersymmetric quantum mechanics. Generally, one can also consider rational and elliptic analogues corresponding to matrix models and elliptic genera. Then, using the Jeffrey--Kirwan prescription \cite{Jeffrey1993LocalizationFN} (see \cite{Benini:2013xpa,Benini:2013nda,Hwang:2014uwa,Hori:2014tda} for applications to physics) and evaluating the poles, schematically we have
\bea
\mathsf{T}_{k}(v_{i},q_{a})=\frac{1}{k!}\sum_{x_{\ast}}\operatorname*{JK-Res}_{\eta,x_{\ast}}\left(\mathcal{
Z}(v_{i},x_{I})\prod_{I=1}^{k}\frac{dx_{I}}{2\pi\iota x_{I}}:\prod_{I=1}^{k}\mathsf{A}(x_{I})^{-1}\prod_{i}\mathsf{V}_{i}(v_{i}):\right),
\eea
where we denote the poles collectively as $x_{\ast}$ and the JK-residue means we are taking the residue there. Note that after taking the contraction, the operator part is a regular function and thus the poles are simply classified by $\mathcal{Z}(v_{i},x_{I})$ as how it is done for normal partition functions. The operator lift up of the instanton partition function is then given as
\bea
\mathsf{T}(v_{i})=\sum_{k=0}^{\infty}\mathfrak{q}^{k}
\mathsf{T}_{k}(v_{i},q_{a})
\eea
whose vacuum expectation value is just the instanton partition function:
\bea
\mathcal{Z}_{\text{inst.}}=\bra{0}\mathsf{T}(v_{i})\ket{0}.
\eea

This operator $\mathsf{T}(v_{i})$ is actually called the $qq$-character and identified with the generator of quiver W-algebras \cite{Frenkel:1998ojj, Frenkel:1997CMP, Shiraishi:1995rp,Awata:1996dx,Awata:1995zk, Kimura:2015rgi,Kimura:2016dys,Kimura:2017hez}. We expect that as long as we have a nice vertex operator representation of the rational function $\mathcal{Z}(v_{i},x_{I})$, this JK-residue procedure gives the $qq$-characters in a generic way. Moreover, we expect that it is still applicable to instanton partition functions for other gauge theories including theories with SO, Sp groups \cite{Marino:2004cn,Nekrasov:2004vw, Shadchin:2005mx,Hollands:2010xa}. In the context of quantum algebras, for the moment, such direction is still left for future work. See \cite{Chen:2023smd, Nawata:2021dlk, Nawata:2023wnk, Hayashi:2020hhb} for recent attempts on this direction. In this paper, we will not make an attempt to discuss the most general setup to obtain the $qq$-characters nor give a physical explanation of the existence of such vertex operator representations, but rather derive the $qq$-characters for concrete examples.

A different way to derive the $qq$-characters is to use the properties of the commutativity with the screening charges. To get the $qq$-characters, one will first define the screening charge. Starting from a highest weight $\mathsf{V}_{i}(v_{i})$ and imposing the commutativity with the screening charge, we obtain the expanded version of the $qq$-characters. The advantage of this method is that once the explicit zero-modes are fixed properly, the commutativity with the screening charge determines all the extra factors uniquely. In the following sections, we will derive the $qq$-characters by using this method.

Although we expect that the $qq$-characters obtained by using the JK-residue method and the screening charge method always give the same result, there are still some points unclear for the moment. On one hand, in the JK-residue method, we have the $\eta$ vector which determines the pole structure of the contour integral formula. Depending on $\eta$, the pole structure and the partition function might change, eventually giving the wall crossing phenomenon \cite{Hori:2014tda}. On the other hand, in the screening charge method, the commutativity with the screening charge determines the pole structure. Thus, one would expect that the definition of the screening charge corresponds to the choice of $\eta$. For the moment, we do not know the explicit correspondence of them and how to define different screening charges. The screening charge introduced in the following sections seems to correspond to the typical choice $\eta=(1,1,\ldots,1)$ and we will always use this.

The $qq$-characters we will introduce explicitly in the following sections comes from the contour integrals discussed in sections~\ref{sec:D4partitionfunction}, \ref{sec:D6partitionfunction}, and \ref{sec:D8partitionfunction} and thus they are operator versions of the equivariant DT vertex used to compute Donaldson--Thomas invariants of toric Calabi--Yau three-folds and four-folds (see for example \cite{Nekrasov:2014nea,Nekrasov:2023nai,Cao:2019tvv,Monavari:2022rtf} and references therein). We note that using the elliptic formulas in \cite[Sec.~12]{Kimura:2023bxy}, one can easily generalize the discussion in this paper and they correspond to elliptic DT invariants \cite{Benini:2018hjy,Fasola:2020hqa}. In this sense, it is natural to call the $qq$-characters introduced in this paper \textit{Donaldson--Thomas $qq$-characters}. 

In the context of algebraic geometry, we have the famous DT/PT correspondence~\cite{Pandharipande:2007sq,Nekrasov:2014nea,Okounkov:2015spn}. From this viewpoint, it is natural to ask if we can define a \textit{PT $qq$-character} and if we have an \textit{operator version of the DT/PT correspondence}. Using the fact that the commutativity with screening charges uniquely determine the full DT $qq$-characters, the discovery of the PT $qq$-characters would help understanding the combinatorial aspects of PT counting. All of these interesting aspects are left for future work.

\section{D4 \texorpdfstring{$qq$}{qq}-characters}\label{sec:D4_qq}
Before moving on to explicit constructions on the\footnote{When it is obvious, we will omit the terminology ``DT'' from now on.} D4 DT $qq$-characters, let us first review how to derive the D4 $qq$-character which is the generator of the affine quiver W-algebra.
\begin{definition}
    The screening charges are defined as
    \bea
\mathscr{Q}_{a}(x)=\sum_{k\in\mathbb{Z}}\mathsf{S}_{a}(q_{a}^{k}x),\quad a\in\four.
    \eea 
\end{definition}
The D4 $qq$-character is a $qq$-character whose highest weight is $\mathsf{X}_{A}(x)\,(A\in\six)$ that commutes with the screening charge in the transverse directions. Young diagrams label the monomial terms of the D4 $qq$-character. The operator part of the monomial terms of the $qq$-character is obtained by the iWeyl reflection as
\bea
\mathsf{X}_{A}(x)\rightarrow :\mathsf{X}_{A}(x)\mathsf{A}^{-1}(x):
\eea
recursively. Similar to \cite{Kimura:2023bxy}, we can also rescale the root current and include topological terms as
\bea
\mathsf{A}(x)\rightarrow \mathfrak{q}^{-1}\mathsf{A}(x).
\eea
\begin{proposition}
The D4 $qq$-characters are defined as
\bea
\,&\mathsf{T}_{A}(x)=\sum_{\lambda}\mathfrak{q}^{|\lambda|}\widetilde{\mathcal{Z}}_{A}^{\D4}[\lambda]\Lambda_{A,\lambda}(x),\quad\Lambda_{A,\lambda}(x)={:\mathsf{X}_{A}(x)\prod_{\Abox\in\lambda}\mathsf{A}^{-1}(\chi_{A,x}(\Bbox)):},\quad A\in\six,\\
\eea
and they obey
\bea
\relax[\mathsf{T}_{A}(x),\mathscr{Q}_{a}(x')]=0,\quad a\in\bar{A}.
\eea
\end{proposition}
To obtain the D4 $qq$-characters with nontrivial boundary conditions, we simply need to start from the highest weight given as in section~\ref{sec:D4contourfreefield}. We denote such $qq$-character as
\bea
\mathsf{T}_{12,k_{12}\,l_{12}}(x)&=\Lambda^{k_{12}\,l_{12}}_{12,\emptyset}(x)+\cdots ,\quad
\Lambda^{k_{12}\,l_{12}}_{12,\emptyset}(x)= {:\mathsf{X}_{12}(x)\frac{\prod_{\Abox\in l_{12}}\mathsf{
A}^{-1}(\chi_{12,x}(\Bbox))\prod_{\Abox\in k_{12}}\mathsf{
A}^{-1}(\chi_{12,x}(\Bbox))}{\prod_{\Abox\in l_{12}\cap k_{12}}\mathsf{
A}^{-1}(\chi_{12,x}(\Bbox))}:}
\eea
and impose the condition
\bea
\relax[\mathsf{T}_{12,k_{12}\,l_{12}},\mathscr{Q}_{3,4}(x')]=0.
\eea
Other D4 $qq$-characters can be obtained by using the quadrality symmetry.

\paragraph{One-leg $qq$-character}
The one-leg $\D4_{12}$ heighest weight is
\bea
:\mathsf{X}_{12}(x)\prod_{j=1}^{k_{12}}\mathsf{S}_{1}(xq_{2}^{j-1})^{-1}:.
\eea
As mentioned in section~\ref{sec:D4contourfreefield}, the highest weight is simply $\mathsf{X}_{12}(q_{2}^{k_{12}}x)$ up to zero-modes coming from $\mathsf{s}_{1,0}(x)$. Moreover, the zero-modes $\mathsf{s}_{1,0}(x)$ commute with the screening charges $\mathscr{Q}_{3,4}(x)$. Therefore, the terms coming from the iWeyl reflection are classified by a finite Young diagram with the origin at $q_{2}^{k_{12}}x$.
\begin{proposition}%[D4 one-leg $qq$-character]
    The one-leg D4 $qq$-character is given as
    \bea
\mathsf{T}_{12,k_{12}\emptyset}(x)&={:\mathsf{X}_{12}(x)\prod_{j=1}^{k_{12}}\mathsf{S}_{1}(xq_{2}^{j-1})^{-1}:}+\cdots \\
&=\sum_{\lambda}\mathfrak{q}^{|\lambda|}\widetilde{\mathcal{Z}}^{\D4}_{12}[\lambda]:\frac{\mathsf{X}_{12}(x)}{\prod_{j=1}^{k_{12}}\mathsf{S}_{1}(xq_{2}^{j-1})}\prod_{\Abox\in\lambda}\mathsf{A}^{-1}(\chi_{12,q_{2}^{k_{12}}x}(\Bbox)):
\eea
where $\widetilde{\mathcal{Z}}_{12}^{\D4}[\lambda]$ is \eqref{eq:D4-partition} with 
\bea
\relax[\mathsf{T}_{12,k_{12}\emptyset}(x),\mathscr{Q}_{3,4}(x')]=0.
\eea
We can rewrite it as
\bea
\mathsf{T}_{12,k_{12}\emptyset}(x)={:\prod_{j=1}^{k_{12}}\mathsf{s}_{1,0}(xq_{2}^{j-1})^{-1}\mathsf{T}_{12}(q_{2}^{k_{12}}x):}.
\eea

\end{proposition}
Instead of the above highest weight, we may dress the highest weight with the extra zero-modes as
\bea
{:\prod_{j=1}^{k_{12}}\mathsf{s}_{1,0}(xq_{2}^{j-1})\frac{\mathsf{X}_{12}(x)}{\prod_{j=1}^{k_{12}}\mathsf{S}_{1}(xq_{2}^{j-1})}:}=\mathsf{X}_{12}(q_{2}^{k_{12}}x)
\eea
and then the $qq$-character is just $\mathsf{T}_{12}(q_{2}^{k}x)$. Such zero-modes of the highest weight only affects the perturbative part without modifying the instanton part and thus we may effectively use $\mathsf{T}_{12}(q_{12}^{k_{12}}x)$ as the one-leg D4$_{12}$ $qq$-character.

\paragraph{Two-legs $qq$-character}
Similar to the previous case, the highest weight is proportional to $\mathsf{X}_{12}(xq_{1}^{l_{12}}q_{2}^{k_{12}})$ up to zero-modes (see \eqref{eq:D4-D2simplication}) and thus we obtain the following.
\begin{proposition}
The two-legs D4 $qq$-character is
\bea
\mathsf{T}_{12,k_{12}\,l_{12}}(x)&=\sum_{\lambda}\mathfrak{q}^{|\lambda|}\widetilde{\mathcal{Z}}_{12}^{\D4}[\lambda]\Lambda_{12,\lambda}^{k_{12}\,l_{12}}(x),\\
\Lambda_{12,\lambda}^{k_{12}\,l_{12}}(x)&={:\Lambda_{12,\emptyset}^{k_{12}\,l_{12}}(x)\prod_{\Abox\in\lambda}\mathsf{A}^{-1}(\chi_{12,xq_{1}^{l_{12}}q_{2}^{k_{12}}}(\Bbox)):}
\eea
which is proportional to $\mathsf{T}_{12}(xq_{1}^{l_{12}}q_{2}^{k_{12}})$ up to zero-modes non-essential for instanton computations.
\end{proposition}

\section{D6 \texorpdfstring{$qq$}{qq}-characters}\label{sec:D6_qq}

The D6 $qq$-character is a $qq$-character whose highest weight is $\mathsf{W}_{\bar{a}}(x)\,(a\in\four)$ and that commutes with the screening charge. Plane partitions label the monomial terms of the D6 $qq$-character. The operator part of the monomial terms of the $qq$-character is obtained by the iWeyl reflection:
\bea
\mathsf{W}_{\bar{a}}(x)\rightarrow {:\mathsf{W}_{\bar{a}}(x)\mathsf{A}^{-1}(x):}.
\eea
Doing this iWeyl reflection recursively, the operator part of the monomial terms is given by
\bea
\Lambda_{\bar{a},\pi}(x)\coloneqq{:\mathsf{W}_{\bar{a}}(x)\prod_{\scube\in\pi}\mathsf{A}^{-1}(\chi_{\bar{a},x}(\cube)):}.
\eea
The coefficients are the $\U(1)$ partition functions of the tetrahedron instanton system obtained by one D6-brane spanning $\mathbb{C}^{3}_{\bar{a}}\times \mathbb{S}^{1}$:
\bea\label{eq:D6U1partitionfunction}
\widetilde{\mathcal{Z}}_{\bar{a}}^{\D6}[\pi]=\prod_{\scube\in\pi}\frac{1-q_{a}x/\chi_{\bar{a},x}(\cube)}{1-x/\chi_{\bar{a},x}(\cube)}\prod_{\substack{\scube\in\pi\\\scubeF\in\pi}}g_{\bar{a}}\left(\frac{\chi_{\bar{a},x}(\cube)}{\chi_{\bar{a},x}(\cubeF)}\right)^{-1}.
\eea

\begin{definition}%[MacMahon $qq$-character]
The D6 $qq$-character is given as
\bea
\mathsf{T}_{\bar{a}}(x)=\sum_{\pi\in\mathcal{PP}}\mathfrak{q}^{|\pi|}\widetilde{\mathcal{Z}}^{\D6}_{\bar{a}}[\pi]\Lambda_{\bar{a},\pi}(x),\quad a\in\four,
\eea
where $\mathcal{PP}$ denotes the set of possible plane partitions. Moreover, the D6 $qq$-character $\mathsf{T}_{\bar{a}}(x)$ commutes with the screening charge $\mathscr{Q}_{a}(x')$:
    \bea
    \left[\mathsf{T}_{\bar{a}}(x),\mathscr{Q}_{a}(x')\right]=0,\quad a\in\four.
    \eea
\end{definition}

\begin{remark}
    There is another way to derive the D6 $qq$-character above using the intertwiner formalism of quantum toroidal $\mathfrak{gl}_{1}$ \cite{Awata:2011ce,Bourgine:2017jsi,Awata:2016riz} (see \cite{DIMreview} for a review). In this formalism, one first consider a physical theory and its corresponding brane web and then assign representations to branes and intertwiners or $R$-matrices to brane junctions. Compositions of the assigned operators will then automatically give the partition function of the theory considered. In \cite{Zenkevich:2023cza} (see also \cite{Zenkevich:2022dju}), the author considered a setup where branes \textit{spiraling} with each other appears and derived the K-theoretic vertex (see Thm.~1 there). The compositions of the intertwiners there is interpreted as an operator lift up of the K-theoretic vertex and thus it is a $qq$-character. Actually, when there are no nontrivial boundary conditions for the plane partitions, they are just the D6 $qq$-character given above, while when there are nontrivial boundary conditions, they are the ones that will be introduced in the following subsections. We also note that combining their discussions and the derivation of the D8 $qq$-character in \cite{Kimura:2023bxy} and section~\ref{sec:D8_qq}, it is almost obvious that the magnificent four partition function should appear by considering D7-NS5 branes and the corresponding MacMahon intertwiners in their setup.
\end{remark}

\subsection{Surface boundary conditions}
We derive the D6 $qq$-characters whose highest weight is associated with the surface boundary conditions. Let us focus on the case when we have $k_{12}$-surfaces in the 12-plane, where the highest weight is given as
\bea
\adjustbox{valign=c}{
\begin{tikzpicture}[scale=0.15]
%\cube{1}{1}{1};
% horizontal y axis (the origin is at (1,0))
\draw[->] (1,0)--(13,0);
\draw[->] (1,0)--(1,9);
\draw[->] (1,0)--(-6.5,-7.5);
\node[right] at (13,0){$2$};
\node[above] at (1,9){$3$};
\node[below left] at (-6.5,-7.5){$1$};
\xysurfacevar{1}{1}{1}{6}{10};
\xysurfacevar{1}{1}{2}{6}{10};
\xysurfacevar{1}{1}{3}{6}{10};
\xysurfacevar{1}{1}{4}{6}{10};
\end{tikzpicture}}\quad =\quad 
{:\frac{\mathsf{W}_{\bar{4}}(x)}{\prod_{k=1}^{k_{12}}\mathsf{X}_{12}(xq_{3}^{k-1})}:}=\mathsf{W}_{\bar{4}}(xq_{3}^{k_{12}}).
\eea
Here, instead of using the formula coming from the infinite product of $\mathsf{A}^{-1}(x)$, we used the regularized formula coming from $\mathsf{X}_{12}(x)^{-1}$. Since, the highest weight is simply $\mathsf{W}_{\bar{4}}(xq_{3}^{k_{12}})$, the $qq$-character is given as
\bea
\mathsf{T}_{\bar{4}}^{k_{23}k_{13}k_{12}}(x)&={:\frac{\mathsf{W}_{\bar{4}}(x)}{\prod_{k=1}^{k_{12}}\mathsf{X}_{12}(xq_{3}^{k-1})}:}+\cdots\\
&=\sum_{\pi\in\mathcal{PP}}\mathfrak{q}^{|\pi|}\widetilde{\mathcal{Z}}^{\D6}_{123}[\pi]:\mathsf{W}_{\bar{4}}(xq_{3}^{k_{12}})\prod_{\scube\in\pi}\mathsf{A}^{-1}(\chi_{\bar{4},q_{3}^{k_{12}}x}(\cube)):\\
&=\mathsf{T}_{\bar{4}}(q_{3}^{k_{12}}x).
\eea
For general surface boundary conditions, we have the following.
\begin{proposition}\label{prop:D6surface_qq}
    Let $\mathsf{T}_{abc}^{k_{bc}k_{ac}k_{ab}}(x)\,(a<b<c)$ be the D6 $qq$-character with the highest weight associated with the following configuration
    \bea
\adjustbox{valign=c}{
\begin{tikzpicture}[scale=0.17]
%\cube{1}{1}{1};
% horizontal y axis (the origin is at (1,0))
\draw[->] (1,0)--(13,0);
\draw[->] (1,0)--(1,9);
\draw[->] (1,0)--(-5.5,-6.5);
\node[right] at (13,0){$b$};
\node[above] at (1,9){$c$};
\node[below left] at (-5.5,-6.5){$a$};
\zxsurfacevar{1}{1}{1}{7}{6};
\zxsurfacevar{1}{2}{1}{7}{6};
\zxsurfacevar{1}{3}{1}{7}{6};
\xysurfacevar{1}{4}{1}{6}{8};
\xysurfacevar{1}{4}{2}{6}{8};
\yzsurfacevar{1}{4}{3}{8}{5};
\yzsurfacevar{1}{4}{3}{8}{5};
\end{tikzpicture}}\quad =\quad :\frac{\mathsf{W}_{abc}(x)}{\prod\limits_{j=1}^{k_{ac}}\mathsf{X}_{ac}(xq_{b}^{j-1})\prod\limits_{k=1}^{k_{ab}}\mathsf{X}_{ab}(xq_{b}^{k_{ac}}q_{c}^{k-1})\prod\limits_{i=1}^{k_{bc}}\mathsf{X}_{bc}(xq_{b}^{k_{ac}}q_{c}^{k_{ab}}q_{a}^{i-1})}:
\eea
Note that the above presentation in $\mathsf{X}_{A}(x)^{-1}$ depends on how we order the surfaces but after computation it is equal to $\mathsf{W}_{abc}(xq_{a}^{k_{bc}}q_{b}^{k_{ac}}q_{c}^{k_{ab}})$ (see \eqref{eq:D6surface-vertop-simple}), which does not depend on the ordering. We then have
\bea
\mathsf{T}_{abc}^{k_{bc}k_{ac}k_{ab}}(x)=\mathsf{T}_{abc}(xq_{a}^{k_{bc}}q_{b}^{k_{ac}}q_{c}^{k_{ab}}).
\eea
\end{proposition}

\subsection{Leg boundary conditions}
Although, the D6 $qq$-characters associated with surface boundary conditions will not give new D6 $qq$-characters, the D6 $qq$-characters associated with leg boundary conditions give new D6 $qq$-characters.

The highest weight corresponds to the vacuum configuration of the plane partition with nontrivial boundary conditions:
\bea
:\mathsf{W}_{\bar{a}}(x)\prod_{\scube \in \mathcal{B}_{\lambda\mu\nu}}\mathsf{A}^{-1}(\chi_{\bar{a},x}(\cube)):
\eea
where $\mathcal{B}_{\lambda\mu\nu}$ is the vacuum configuration with nontrivial boundary conditions. For $\bar{a}=bcd\,(b<c<d)$, $\mathcal{B}_{\lambda\mu\nu}$ is the vacuum configuration whose nontrivial boundary conditions are $\lambda,\mu,\nu$ at the axes $b,c,d$ respectively. We denote these $qq$-characters as
\bea
\mathsf{T}_{\bar{a},\lambda\mu\nu}(x)&={:\mathsf{W}_{\bar{a}}(x)\prod_{\scube \in \mathcal{B}_{\lambda\mu\nu}}\mathsf{A}^{-1}(\chi_{\bar{a},x}(\cube)):}+\cdots
%&=\sum_{\pi\in\mathcal{PP}_{\lambda\mu\nu}}\widetilde{\mathcal{Z}}^{\D6}_{\bar{a};\lambda\mu\nu}[\pi]\Lambda^{\lambda\mu\nu}_{\bar{a},\pi}(x)
\eea
Let us derive the complete formula of this $qq$-character.

%where
%\bea
%\Lambda^{\lambda\mu\nu}_{\bar{a},\pi}(x)={:\mathsf{W}_{\bar{a}}(x)\prod_{\scube\in\mathcal{B}_{\lambda\mu\nu}}\mathsf{A}^{-1}(\chi_{\bar{a},x}(\cube))\prod_{\scube \in \pi} \mathsf{A}^{-1}(\chi_{\bar{a},x}(\cube)) :}.
%\eea
%The set $\mathcal{PP}_{\lambda\mu\nu}$ is the set of possible plane partition configurations with nontrivial boundary conditions $\lambda\mu\nu$.

\paragraph{Structure functions}
We introduce the following structure functions as
\bea
\mathscr{W}^{\bar{a},\lambda\mu\nu}_{\pi,x}(x')=\mathbb{I}\left[\mathbf{Y}_{\bar{a}}^{\vee}x'\right],\quad \mathscr{W}^{\bar{a},\lambda\mu\nu\vee}_{\pi,x}(x')&=\mathbb{I}\left[\mathbf{Y}_{\bar{a}}x'^{\vee}\right]\\
\mathbf{Y}_{\bar{a}}|_{\pi}=\bfN-\bfP_{\bar{a}}\left(\bfK_{\pi}+\bfN_{\lambda\mu\nu}\right)
\eea
which is explicitly given as
\bea
\mathscr{W}^{abc,\lambda\mu\nu}_{\pi,x}(x')&=\left(1-\frac{x}{x'}\right) \prod_{\scube \in \mathcal{B}_{\lambda\mu\nu}}g_{abc}\left(\frac{\chi_{abc,x}(\cube)}{x'}\right)\prod_{\scube \in \pi}g_{abc}\left(\frac{\chi_{abc,x}(\cube)}{x'}\right),\\
{\mathscr{W}^{abc,\lambda\mu\nu}_{\pi,x}}^{\vee}(x')&=\left(1-\frac{x'}{x}\right)\prod_{\scube\in\mathcal{B}_{\lambda\mu\nu}}g_{abc}\left(q_{abc}^{-1}\frac{x'}{\chi_{abc,x}(\cube)}\right)^{-1}\prod_{\scube\in\pi}g_{abc}\left(q_{abc}^{-1}\frac{x'}{\chi_{abc,x}(\cube)}\right)^{-1}
\eea
where $a<b<c$ and $\pi\in\mathcal{PP}_{\lambda\mu\nu}$. The structure function is an infinite product but after nontrivial cancellations of numerators and denominators they will be a finite product. When there is only one nontrivial boundary condition, we can write this finite product explicitly. For example, we have
\bea
\mathscr{W}^{abc,\lambda\emptyset\emptyset}_{\pi,x}(x')&=\left(1-\frac{x}{x'}\right)\prod_{\Abox\in\lambda}\mathscr{S}_{bc}\left(\frac{\chi_{bc,x}(\Bbox)}{x'}\right)\prod_{\scube \in \pi}g_{abc}\left(\frac{\chi_{abc,x}(\cube)}{x'}\right),\\
%\mathscr{W}^{abc,\emptyset\mu\emptyset}_{\pi,x}(x')&=\left(1-\frac{x}{x'}\right)\prod_{\Abox\in\mu}\mathscr{S}_{ac}\left(\frac{\chi_{ac,x}(\Bbox)}{x'}\right) \prod_{\scube \in \pi}g_{abc}\left(\frac{\chi_{abc,x}(\cube)}{x'}\right),\\
%\mathscr{W}^{abc,\emptyset\emptyset\nu}_{\pi,x}(x')&=\left(1-\frac{x}{x'}\right)\prod_{\Abox\in\nu}\mathscr{S}_{ab}\left(\frac{\chi_{ab,x}(\Bbox)}{x'}\right)\prod_{\scube \in \pi}g_{abc}\left(\frac{\chi_{abc,x}(\cube)}{x'}\right),\\
{\mathscr{W}^{abc,\lambda\emptyset\emptyset}_{\pi,x}}^{\vee}(x')&=\left(1-\frac{x'}{x}\right)\prod_{\Abox\in\lambda}\mathscr{S}_{bc}\left(\frac{q_{bc}^{-1}x'}{\chi_{bc,x}(\Bbox)}\right)\prod_{\scube\in\pi}g_{abc}\left(q_{abc}^{-1}\frac{x'}{\chi_{abc,x}(\cube)}\right)^{-1}.%\\
%{\mathscr{W}^{abc,\emptyset\mu\emptyset}_{\pi,x}}^{\vee}(x')&=\left(1-\frac{x'}{x}\right)\prod_{\Abox\in\mu}\mathscr{S}_{ac}\left(\frac{q_{ac}^{-1}x'}{\chi_{ac,x}(\Bbox)}\right)\prod_{\scube\in\pi}g_{abc}\left(q_{abc}^{-1}\frac{x'}{\chi_{abc,x}(\cube)}\right)^{-1}\\
%{\mathscr{W}^{abc,\emptyset\emptyset\nu}_{\pi,x}}^{\vee}(x')&=\left(1-\frac{x'}{x}\right)\prod_{\Abox\in\nu}\mathscr{S}_{ab}\left(\frac{q_{ab}^{-1}x'}{\chi_{ab,x}(\Bbox)}\right)\prod_{\scube\in\pi}g_{abc}\left(q_{abc}^{-1}\frac{x'}{\chi_{abc,x}(\cube)}\right)^{-1}
\eea

\begin{proposition}
The zeros of the structure function are determined by the addable and removable boxes of the plane partition $\pi$:
\bea
\mathscr{W}^{\bar{a},\lambda\mu\nu}_{\pi,x}(x')\propto \prod_{\scube\in A(\pi)}\left(1-\frac{\chi_{\bar{a},x}(\cube)}{x'}\right)\prod_{\scube\in R(\pi)}\left(1-q_{a}^{-1}\frac{\chi_{\bar{a},x}(\cube)}{x'}\right)
\eea
\end{proposition}
%\begin{proof}
 Though we do not have a complete proof of this proposition for the moment, we have checked it for the situations when we have one leg $\lambda\neq \emptyset,\mu=\nu=\emptyset$ with $|\lambda|\leq 2$, when the two legs are $\lambda=\mu=\Bbox,\nu=\emptyset$, and when we have three legs $\lambda=\mu=\nu=\Bbox$. For the case when there is no boundary conditions, see \cite{Kimura:2023bxy}. 
%\end{proof}

\begin{proposition}
The partition function in \eqref{eq:D6legpartitionfunction} has the following recursive relation
\bea
\frac{\widetilde{\mathcal{Z}}_{\bar{a};\lambda\mu\nu}^{\D6}[\pi+\cube]}{\widetilde{\mathcal{Z}}_{\bar{a};\lambda\mu\nu}^{\D6}[\pi]}&=-\frac{\underset{x'=q_{a}\chi_{\bar{a},x}(\scube)}{\Res}\mathscr{W}_{\pi,x}^{\bar{a},\lambda\mu\nu}(q_{a}^{-1}x')^{-1}}{\underset{x'=\chi_{\bar{a},x}(\scube)}{\Res}\mathscr{W}_{\pi+\scube,x}^{\bar{a},\lambda\mu\nu}(q_{a}^{-1}x')^{-1}}
\eea

\end{proposition}
\begin{proof}
    Let us study the recursive relation of $\mathbf{v}$:
    \bea
    \mathbf{v}=-\bfP_{4}^{\vee}\bfN^{\vee}\bfK+\bfP_{\four}\bfN_{\lambda\mu\nu}^{\vee}\bfK+\bfP_{123}^{\vee}\bfK^{\vee}\bfK
    \eea
    We focus on $a=4$. Assume that $x'=\chi_{\bar{4},x}(\cube)$, where $\cube \in A(\pi)$ and $\cube \in R(\pi+\cube)$. Then, the recursive relation is given as
\bea
\delta \mathbf{v}&=\mathbf{v}|_{\pi+\scube}-\mathbf{v}|_{\pi}\\
%&=-\bfP_{4}^{\vee}\bfN^{\vee} x'+\bfP_{123}^{\vee}{x'}^{-1}\bfK_{\pi}+\bfP_{123}^{\vee}\bfK_{\pi}^{\vee}x'+\bfP_{123}^{\vee}+\bfP_{\four}\bfN_{\lambda\mu\nu}^{\vee}x'\\
%&=-\bfN^{\vee}x'+\bfP_{123}^{\vee}\bfK_{\pi}^{\vee}x'+\bfP_{123}^{\vee}\bfN^{\vee}_{\lambda\mu\nu}x'\\
%&+q_{4}^{-1}\bfN^{\vee}x'+\bfP_{123}^{\vee}\bfK_{\pi}{x'}^{-1}-q_{4}^{-1}\bfP_{123}\bfN^{\vee}_{\lambda\mu\nu}x'+\bfP_{123}^{\vee}\\
&=-\left(\bfN-\bfP_{123}\bfK_{\pi}-\bfP_{123}\bfN_{\lambda\mu\nu}\right)^{\vee}x'\\
&+q_{4}^{-1}\bfN^{\vee}x'+\textcolor{red}{\bfP_{123}^{\vee}\bfK_{\pi+\scube}{x'}^{-1}}-q_{4}^{-1}\bfP_{123}^{\vee}\bfN_{\lambda\mu\nu}^{\vee}x'.
\eea
When taking the index, the red term will give an extra sign factor coming from Prop.~\ref{app-prop:reflection_sign} and we obtain
\bea
\mathbb{I}\left[\delta\mathbf{v}\right]&=-\frac{\mathbb{I}\left[(\bfN-\bfP_{123}(\bfK_{\pi+\scube}+\bfN_{\lambda\mu\nu}))^{\vee}(q_{4}^{-1}x')\right]}{\mathbb{I}\left[(\bfN-\bfP_{123}(\bfK_{\pi}+\bfN_{\lambda\mu\nu}))^{\vee}x'\right]}\\
&=-\lim_{x'\rightarrow \chi_{\bar{4},x}(\scube)}\frac{\mathscr{W}^{\bar{4},\lambda\mu\nu}_{\pi,x}(x')^{-1}}{\mathscr{W}^{\bar{4},\lambda\mu\nu}_{\pi+\scube,x}(q_{4}^{-1}x')^{-1}}\\
&=-\frac{\underset{x'=q_{4}\chi_{\bar{4},x}(\scube)}{\Res}\mathscr{W}_{\pi,x}^{\bar{4},\lambda\mu\nu}(q_{4}^{-1}x')^{-1}}{\underset{x'=\chi_{\bar{4},x}(\scube)}{\Res}\mathscr{W}_{\pi+\scube,x}^{\bar{4},\lambda\mu\nu}(q_{4}^{-1}x')^{-1}}
%&=-\frac{\widetilde{\mathcal{Z}}_{\bar{4};\lambda\mu\nu}^{\D6}[\pi+\cube]}{\widetilde{\mathcal{Z}}_{\bar{4};\lambda\mu\nu}^{\D6}[\pi]}.
\eea
Note that after setting the initial condition, the partition function is determined uniquely. In our setup, we set the vacuum configuration (the configuration when we have no yellow boxes) to have trivial partition function, which is $1$.
\end{proof}

\begin{proposition}\label{prop:D6leg_qq}
    The D6 $qq$-characters with nontrivial leg boundary conditions are given as
    \bea
    \mathsf{T}_{\bar{a},\lambda\mu\nu}(x)=\sum_{\pi\in\mathcal{PP}_{\lambda\mu\nu}}\mathfrak{q}^{|\pi|}\widetilde{\mathcal{Z}}^{\D6}_{\bar{a};\lambda\mu\nu}[\pi]\Lambda^{\lambda\mu\nu}_{\bar{a},\pi}(x)
    \eea
    where
    \bea
    \Lambda^{\lambda\mu\nu}_{\bar{a},\pi}(x)={:\mathsf{W}_{\bar{a}}(x)\prod_{\scube\in\mathcal{B}_{\lambda\mu\nu}}\mathsf{A}^{-1}(\chi_{\bar{a},x}(\cube))\prod_{\scube \in \pi} \mathsf{A}^{-1}(\chi_{\bar{a},x}(\cube)) :}.
    \eea
    We have
    \bea
    \relax[\mathsf{T}_{\bar{a},\lambda\mu\nu}(x),\mathscr{Q}_{a}(x')]=0.
    \eea
\end{proposition}
\begin{proof}
The contraction of $\Lambda^{\lambda\mu\nu}_{\bar{a},\pi}(x)$ and the screening current $\mathsf{S}_{a}(x)$ is 
\bea
\Lambda^{\lambda\mu\nu}_{\bar{a},\pi}(x)\mathsf{S}_{a}(x')&=(-q_{a}x)\left[\mathscr{W}^{\bar{a},\lambda\mu\nu}_{\pi,x}(q_{a}^{-1}x')^{-1}\right]_{-}:\Lambda^{\lambda\mu\nu}_{\bar{a},\pi}(x)\mathsf{S}_{a}(x'):,\\
\mathsf{S}_{a}(x')\Lambda^{\lambda\mu\nu}_{\bar{a},\pi}(x)&=(-q_{a}x)\left[\mathscr{W}^{\bar{a},\lambda\mu\nu}_{\pi,x}(q_{a}^{-1}x')^{-1}\right]_{+}:\Lambda^{\lambda\mu\nu}_{\bar{a},\pi}(x)\mathsf{S}_{a}(x'):
\eea
The commutation of the $qq$-characters are then given as
\bea
\relax &[\mathsf{T}_{\bar{a},\lambda\mu\nu}(x),\mathsf{S}_{a}(x')]\\
=&q_{a}x\sum_{\pi\in\mathcal{PP}_{\lambda\mu\nu}}\widetilde{\mathcal{Z}}_{\bar{a};\lambda\mu\nu}^{\D6}[\pi]\left(\sum_{\scube \in A(\pi)}\underset{x'=q_{a}\chi_{\bar{a},x}(\scube)}{\Res}\mathscr{W}_{\pi,x}^{\bar{a},\lambda\mu\nu}(q_{a}^{-1}x')^{-1}\delta\left(\frac{x'}{q_{a}\chi_{\bar{a},x}(\cube)}\right):\Lambda^{\lambda\mu\nu}_{\bar{a},\pi}(x)\mathsf{S}_{a}(q_{a}\chi_{\bar{a},x}(\cube)):\right.\\
&\left.+\sum_{\scube \in R(\pi)}\underset{x'=\chi_{\bar{a},x}(\scube)}{\Res}\mathscr{W}_{\pi,x}^{\bar{a},\lambda\mu\nu}(q_{a}^{-1}x')^{-1}\delta\left(\frac{x'}{\chi_{\bar{a},x}(\cube)}\right):\Lambda^{\lambda\mu\nu}_{\bar{a},\pi}(x)\mathsf{S}_{a}(\chi_{\bar{a},x}(\cube)):\right)
\eea
Shifting the second term as $\pi'=\pi-\scube$, the second term can be rewritten in terms of summation of $\sum_{\scube\in A(\pi')}$. Since we have the recursion relation of the coefficients $\widetilde{\mathcal{Z}}_{\bar{a};\lambda\mu\nu}^{\D6}[\pi]$:
\bea
\frac{\widetilde{\mathcal{Z}}_{\bar{a};\lambda\mu\nu}^{\D6}[\pi+\cube]}{\widetilde{\mathcal{Z}}_{\bar{a};\lambda\mu\nu}^{\D6}[\pi]}&=-\frac{\underset{x'=q_{a}\chi_{\bar{a},x}(\scube)}{\Res}\mathscr{W}_{\pi,x}^{\bar{a},\lambda\mu\nu}(q_{a}^{-1}x')^{-1}}{\underset{x'=\chi_{\bar{a},x}(\scube)}{\Res}\mathscr{W}_{\pi+\scube,x}^{\bar{a},\lambda\mu\nu}(q_{a}^{-1}x')^{-1}},\\
&=-\lim_{x'\rightarrow \chi_{\bar{a},x}(\scube)}\frac{\mathscr{W}^{\bar{a},\lambda\mu\nu}_{\pi,x}(x')^{-1}}{\mathscr{W}^{\bar{a},\lambda\mu\nu}_{\pi+\scube,x}(q_{a}^{-1}x')^{-1}},
\eea
we obtain the result.\footnote{\label{foot:initial}Note that the initial condition is chosen to be $\widetilde{\mathcal{Z}}^{\D6}_{\bar{a};\lambda\mu\nu}[\pi]=1$ in this paper. One may impose different initial conditions depending on the boundary conditions $\lambda,\mu,\nu$. Even if one does so, the recursion relation does not change and we still have the commutativity.}
%\todo{The proof updated}
% Old version
\if0
Then, the commutativity is translated into the recursive relation of the coefficients $\widetilde{\mathcal{Z}}_{\bar{a};\lambda\mu\nu}^{\D6}[\pi]$:
\bea
\frac{\widetilde{\mathcal{Z}}_{\bar{a};\lambda\mu\nu}^{\D6}[\pi+\cube]}{\widetilde{\mathcal{Z}}_{\bar{a};\lambda\mu\nu}^{\D6}[\pi]}&=-\frac{\underset{x'=q_{a}\chi_{\bar{a},x}(\scube)}{\Res}\mathscr{W}_{\pi,x}^{\bar{a},\lambda\mu\nu}(q_{a}^{-1}x')^{-1}}{\underset{x'=\chi_{\bar{a},x}(\scube)}{\Res}\mathscr{W}_{\pi+\scube,x}^{\bar{a},\lambda\mu\nu}(q_{a}^{-1}x')^{-1}},\\
&=-\lim_{x'\rightarrow \chi_{\bar{a},x}(\scube)}\frac{\mathscr{W}^{\bar{a},\lambda\mu\nu}_{\pi,x}(x')^{-1}}{\mathscr{W}^{\bar{a},\lambda\mu\nu}_{\pi+\scube,x}(q_{a}^{-1}x')^{-1}}
\eea
Since this recursive relation is true, we obtain the commutativity.
\fi
\end{proof}

\subsection{Fusion of D4 \texorpdfstring{$qq$}{qq}-characters}
The $\D6$ $qq$-characters can be obtained by fusion of infinite number of $\D4$ $qq$-characters \cite{Kimura:2023bxy}. Let us review this process in detail including the sign factor aspects not discussed in \cite{Kimura:2023bxy}.
\begin{lemma}[\cite{Kimura:2023bxy}]\label{lem:D4contraction}
    The contraction of the operators $\Lambda_{12,\lambda}(x)$ are given as 
    \bea
    \Lambda_{12,\lambda^{(2)}}(x_{2})\Lambda_{12,\lambda^{(1)}}(x_{1})&=\mathcal{Z}_{\text{1-loop}}^{\D4\tbar\D4}(x_{1},12\mid x_{2},12)\mathcal{Z}_{12|12}^{\D4\tbar\D4}(x_{1},\lambda^{(1)}\mid x_{2},\lambda^{(2)}):\Lambda_{12,\lambda^{(2)}}(x_{2})\Lambda_{12,\lambda^{(1)}}(x_{1}):
    \eea
    where
    \bea
\mathcal{Z}_{A|B}^{\D4\tbar\D4}(v_{1},\lambda^{(1)}\,|\,v_{2},\lambda^{(2)})&=\prod_{\Abox\in\lambda^{(1)}}\mathscr{S}_{\bar{B}}\left(q_{B}\frac{\chi_{A,v_{1}}(\Bbox)}{v_{2}}\right)\prod_{\AboxF\in\lambda^{(2)}}\mathscr{S}_{\bar{A}}\left(\frac{v_{1}}{\chi_{B,v_{2}}(\BboxF)}\right)\prod_{\substack{\Abox\in\lambda^{(1)}\\\AboxF\in\lambda^{(2)}}}\mathcal{A}_{\mathbb{C}^{4}}\left(\frac{\chi_{A,v_{1}}(\Bbox)}{\chi_{B,v_{2}}(\BboxF)}\right)^{-1}\\
\mathcal{Z}_{\text{1-loop}}^{\D4\tbar\D4}(x_{1},A\,|\,x_{2},B)&=\exp\left(-\sum_{n=1}^{\infty}\frac{1}{n}\frac{\bfP_{\bar{A}}^{[n]}\bfP_{\bar{B}}^{[-n]}}{\bfP_{\four}^{[n]}}\left(\frac{x_{1}}{x_{2}}\right)^{n}\right)
    \eea
\end{lemma}
\begin{lemma}[\cite{Kimura:2023bxy}]\label{lem:D4tuning}
    Given two Young diagrams $\lambda^{(1)},\lambda^{(2)}$ and the parameters $x_{2}=q_{a}x_{1}\,\,(a\in\bar{A})$, we have 
    \bea
    \mathcal{Z}_{A|A}^{\D4\tbar\D4}(x_{1},\lambda^{(1)}\mid q_{a}x_{1},\lambda^{(2)})=0
    \eea
    for $\lambda^{(2)}\succ \lambda^{(1)}$.
\end{lemma}

\begin{proposition}\label{lem:finiteplanedecomposition}
    Given a finite plane partition $\pi$ spanning the 123-plane, we can decompose it into layers of non-increasing finite Young diagrams:
    \bea
    \pi=\{\lambda^{(1)},\lambda^{(2)},\ldots \lambda^{(k)},\ldots\},\quad \lambda^{(1)}\succeq \lambda^{(2)} \succeq \cdots
    \eea
    and then we have
    \bea
    \prod_{k=1}^{\infty}\widetilde{\mathcal{Z}}^{\D4}_{12}[\lambda^{(k)}]\,\,\overleftarrow{\prod_{k=1}^{\infty}}\Lambda_{12,\lambda^{(k)}}(xq_{3}^{k-1})\simeq \widetilde{\mathcal{Z}}^{\D6}_{\bar{4}}[\pi]:\Lambda_{\bar{4},\pi}(x):
    \eea
    up to one-loop perturbative factors.
\end{proposition}

\begin{proof}
Using Lemma~\ref{lem:D4contraction}, the left hand side gives
\bea
\prod_{k=1}^{\infty}\widetilde{\mathcal{Z}}^{\D4}_{12}[\lambda^{(k)}]\,\,\overleftarrow{\prod_{k=1}^{\infty}}\Lambda_{12,\lambda^{(k)}}(xq_{3}^{k-1})\simeq \prod_{i=1}^{\infty}\widetilde{\mathcal{Z}}^{\D4}_{12}[\lambda^{(i)}]\prod_{i<j}\mathcal{Z}^{\D4\tbar\D4}_{12|12}(x_{i},\lambda^{(i)}\mid x_{j},\lambda^{(j)}):\prod_{k=1}^{\infty}\Lambda_{12,\lambda^{(k)}}(xq_{3}^{k-1}):
\eea
where $x_{i}=q_{3}^{i-1}x$ and the equality is up to one-loop perturbative factors. The operator part is obvious (see \cite{Kimura:2023bxy}) so we focus only on the coefficient part. Introducing 
\bea
\bm{\lambda}^{(i)}=\sum_{\Abox\in\lambda^{(i)}}\chi_{12,x_{i}}(\Bbox)
\eea
the factor is rewritten as
\bea
\prod_{i=1}^{\infty}\widetilde{\mathcal{Z}}_{12}^{\D4}[\lambda^{(i)}]&=\mathbb{I}\left[\sum_{i=1}^{\infty}\left(-\bfP_{34}^{\vee}x_{i}^{-1}\bm{\lambda}^{(i)}+\bfP_{123}^{\vee}\bm{\lambda}^{(i)\vee}\bm{\lambda}^{(i)}\right)\right],\quad \\
\prod_{i<j}\mathcal{Z}^{\D4\tbar\D4}_{12|12}(x_{i},\lambda^{(i)}\mid x_{j},\lambda^{(j)})&=\mathbb{I}\left[\sum_{i<j}\left(-\bfP_{34}^{\vee}x_{j}q_{12}^{-1}{\bm{\lambda}^{(i)\vee}}-\bfP_{34}^{\vee}\bm{\lambda}^{(j)}x_{i}^{-1}+\bfP_{\four}\bm{\lambda}^{(j)}\bm{\lambda}^{(i)\vee}\right)\right].
\eea
We denote the total character as $\mathbf{v}_{\D4\rightarrow\D6}$:
\bea
\mathbf{v}_{\D4\rightarrow\D6}&=\sum_{i=1}^{\infty}\left(-\bfP_{34}^{\vee}x_{i}^{-1}\bm{\lambda}^{(i)}+\bfP_{123}^{\vee}\bm{\lambda}^{(i)\vee}\bm{\lambda}^{(i)}\right)\\
&+\sum_{i<j}\left(-\bfP_{34}^{\vee}x_{j}q_{12}^{-1}{\bm{\lambda}^{(i)\vee}}-\bfP_{34}^{\vee}\bm{\lambda}^{(j)}x_{i}^{-1}+\bfP_{\four}\bm{\lambda}^{(j)}\bm{\lambda}^{(i)\vee}\right).
\eea
By direct computation, one can show that $\mathbf{v}_{\D4\rightarrow \D6}$ is \textit{movable} (see Def.~\ref{app-def:movable}). To compare with $\widetilde{\mathcal{Z}}_{\bar{4}}^{\D6}[\pi]$, we need to apply the reflection properties in \eqref{eq:reflectionprop0}, Prop.~\ref{app-prop:reflection_sign}, \ref{app-prop:reflection-mod} to some terms of $\mathbf{v}_{\D4\rightarrow \D6}$.
Let us focus on the following term:
\bea
-\bfP_{34}^{\vee}x_{j}q_{12}^{-1}{\bm{\lambda}^{(i)\vee}}=-\bfP_{34}^{\vee}q_{3}^{j-i}q_{12}^{-1}\sum_{(x,y)\in\lambda^{(i)}}q_{1}^{-x+1}q_{2}^{-y+1},\quad j>i.
\eea
and each term is expanded as
\bea
&(1-q_{3}^{-1})(1-q_{4}^{-1})q_{3}^{j-i}q_{1}^{-x}q_{2}^{-y}\\
=&q_{3}^{j-i}q_{1}^{-x}q_{2}^{-y}-q_{3}^{j-i-1}q_{1}^{-x}q_{2}^{-y}-q_{3}^{j-i+1}q_{1}^{-x+1}q_{2}^{-y+1}+q_{3}^{j-i}q_{1}^{-x+1}q_{2}^{-y+1}.
\eea
Since $j>i$ and $x,y\geq 1$, no term will be unmovable and we have
\bea
\mathbb{I}\left[-\bfP_{34}^{\vee}x_{j}q_{12}^{-1}{\bm{\lambda}^{(i)\vee}}\right]=\mathbb{I}\left[-\bfP_{34}^{\vee}x_{j}^{-1}{\bm{\lambda}^{(i)}}\right].
\eea
By doing a similar analysis, we also have
\bea
\mathbb{I}\left[ \bfP_{\four}\bm{\lambda}^{(j)}\bm{\lambda}^{(i)\vee}  \right]&=\mathbb{I}\left[( \bfP_{123}^{\vee}+\bfP_{123})\bm{\lambda}^{(i)\vee}\bm{\lambda}^{(j)} \right]=\mathbb{I}\left[ \bfP_{123}^{\vee}\left(\bm{\lambda}^{(j)\vee}\bm{\lambda}^{(i)}+\bm{\lambda}^{(i)\vee}\bm{\lambda}^{(j)} \right) \right].
\eea
To see this, we need to check that there are no unmovable terms in $\bfP_{123}\bm{\lambda}^{(i)\vee}\bm{\lambda}^{(j)}$.
For $\Bbox=(A_{1},B_{1})\in\lambda^{(i)},\,\,\BboxF=(A_{2},B_{2})\in\lambda^{(j)}$, each term takes the form
\bea
\bfP_{123}q_{3}^{j-i}q_{1}^{A_{2}-A_{1}}q_{2}^{B_{2}-B_{1}}.
\eea
The unmovable terms appear only when all of the powers of $q_{1,2,3}$ are zero. The factor $\bfP_{123}$ will be a Laurent polynomial with $q_{3}^{\geq 0}$ while the remaining factor is $q_{3}^{>0}$ due to $j>i$. The character $\bfP_{123}\bm{\lambda}^{(i)\vee}\bm{\lambda}^{(j)}$ is then in powers of $q_{3}^{>0}$ and thus it is movable. 

Therefore, the index of $\mathbf{v}_{\D4\rightarrow \D6}$ is
\bea
\mathbb{I}\left[\mathbf{v}_{\D4\rightarrow \D6}\right]&=\mathbb{I}\left[-\bfP_{34}^{\vee}\sum_{i}x_{i}^{\vee}\sum_{j}\bm{\lambda}^{(j)}+\bfP_{123}^{\vee}\sum_{i}\bm{\lambda}^{(i)\vee}\sum_{j}\bm{\lambda}^{(j)}\right]\\
&=\mathbb{I}\left[-\bfP_{4}^{\vee}x^{-1}\sum_{j}\bm{\lambda}^{(j)}+\bfP_{123}^{\vee}\sum_{i}\bm{\lambda}^{(i)\vee}\sum_{j}\bm{\lambda}^{(j)}\right]=\widetilde{\mathcal{Z}}^{\D6}_{\bar{4}}[\pi].
\eea
\end{proof}

\begin{theorem}[\cite{Kimura:2023bxy}]\label{thm:D4tensortoD6}
The D6 $qq$-characters are obtained by infinite products of the D4 $qq$-characters:
\bea
\overleftarrow{\prod_{i=1}^{\infty}}\mathsf{T}_{ab}(xq_{c}^{i-1})\simeq \mathsf{T}_{abc}(x)
\eea
where the equality is up to one-loop perturbative factors.
\end{theorem}

\paragraph{$qq$-characters with nontrivial boundary conditions}
Lemma~\ref{lem:finiteplanedecomposition} can be generalized to infinite size plane partitions by doing the same computation as in section~\ref{sec:D6partitionfunction}. Let $\tilde{\pi}$ be an infinite size plane partition decomposed as in $\pi_{\bd},\pi_{\reg}$. We can decompose the plane partition $\tilde{\pi}$ into sequence of infinite size Young diagrams $\tilde{\lambda}=\{\tilde{\lambda}^{(k)}\mid k=1,\ldots, \infty\}$ with the condition
\bea
\tilde{\lambda}^{(1)}\succeq \tilde{\lambda}^{(2)} \succeq \cdots.
\eea
We then have
\bea
\prod_{k=1}^{\infty}\widetilde{\mathcal{Z}}^{\D4}_{12}[\tilde{\lambda}^{(k)}]\overleftarrow{\prod_{k=1}^{\infty}}\Lambda_{12,\tilde{\lambda}^{(k)}}(xq_{3}^{k-1})\simeq \widetilde{\mathcal{Z}}^{\D6}_{\bar{4}}[\tilde{\pi}]:\Lambda_{\bar{4},\tilde{\pi}}(x):.
\eea
Generally, $\widetilde{\mathcal{Z}}_{\bar{4}}^{\D6}[\tilde{\pi}]$ can be decomposed into contributions from $\pi_{\bd}$ and $\pi_{\reg}$ as $\widetilde{\mathcal{Z}}^{\D6}_{\bar{4}}[\widetilde{\pi}]=\widetilde{\mathcal{Z}}^{\D6}_{\bar{4}}[\pi_{\bd}]\widetilde{\mathcal{Z}}^{\D6}_{\bar{4}}[\pi_{\reg}]$. Moreover, depending on the boundary conditions, we have
\bea
\Lambda_{\bar{4},\tilde{\pi}}(x)=\begin{dcases}
    \Lambda_{\bar{4},\pi_{\reg}}^{\lambda\mu\nu}(x),\quad \text{leg bd. cond.},\\
    \Lambda_{\bar{4},\pi_{\reg}}(q_{1}^{k_{23}}q_{2}^{k_{13}}q_{3}^{k_{12}}x),\quad \text{surface bd. cond.}
\end{dcases}
\eea
For example, let us consider the plane partition spanning $123$-plane with $\lambda,\mu$ at the 1,2 axes, respectively. The highest weight in \eqref{eq:D4tensor-op} is represented as
\bea
:\frac{\mathsf{W}_{\bar{4}}(x)}{\prod_{i=1}^{\ell(\lambda^{\rmT})}\prod_{j=1}^{\lambda_{i}^{\rmT}}\mathsf{S}_{1}(q_{2}^{j-1}q_{3}^{i-1}x)\prod_{i=1}^{\ell(\mu)}\prod_{j=1}^{\mu_{i}}\mathsf{S}_{2}(q_{1}^{j-1}q_{2}^{\lambda_{i}^{\rmT}}q_{3}^{i-1}x)}:&={:\prod_{k=1}^{\infty}\frac{\mathsf{X}_{12}(xq_{3}^{k-1})}{\prod_{j=1}^{\lambda_{k}^{\rmT}}\mathsf{S}_{1}(q_{2}^{j-1}q_{3}^{k-1})\prod_{i=1}^{\mu_{k}}\mathsf{S}_{2}(q_{1}^{i-1}q_{2}^{\lambda_{k}^{\rmT}})}:}\\
&={:\prod_{k=1}^{\infty}\mathsf{X}_{12}(xq_{3}^{k-1}q_{1}^{\mu_{k}}q_{2}^{\lambda^{\rmT}_{k}}):}.
\eea
Therefore, up to one-loop perturbative factors and boundary contributions, each monomial term of the D6 $qq$-characters with boundary conditions can be obtained by infinite products of the monomial terms of the D4 $qq$-characters with boundary conditions.

Whether we can decompose the entire D6 $qq$-character to infinite products of D4 $qq$-characters is nontrivial. For this to happen, each layer needs to be a consistent $qq$-character. For the surface boundary conditions, after shifting the spectral parameters properly, the $qq$-character is just $\mathsf{T}_{\bar{4}}(x)$ and thus we have a nice decomposition in the D4 $qq$-characters.

For the leg boundary conditions, this is possible when we have up to two nontrivial legs.
\begin{theorem}
The D6 $qq$-character $\mathsf{T}_{\bar{4},\lambda\mu\emptyset}(x)$ can be decomposed into infinite number of D4 $qq$-characters $\mathsf{T}_{12}(x)$ as
    \bea
\overleftarrow{\prod_{k=1}^{\infty}}\mathsf{T}_{12}(xq_{3}^{k-1}q_{1}^{\mu_{k}}q_{2}^{\lambda_{k}^{\rmT}})\simeq \mathsf{T}_{\bar{4},\lambda\mu\emptyset}(x)
\eea
where the equality is understood up to one-loop perturbative factors.\footnote{In the infinite product process, the one-loop perturbative sector part appears because of the contractions between the vertex operators of the highest weight vertex operators. Such one-loop perturbative part can be included in the $qq$-characters by modifying the initial condition to the recursion relation of the partition functions (see for example footnote \ref{foot:initial}). However, such contributions are just overall factors so we excluded them in all the computations. }
\end{theorem}

The situation is different when we have three nontrivial legs $\lambda,\mu,\nu$. In this case, we can not decompose the D6 $qq$-character $\mathsf{T}_{\bar{4},\lambda\mu\nu}(x)$ into infinite products of D4 $qq$-characters. For layers at $k\gg 1$, the boundary contribution is illustrated as
\bea
\adjustbox{valign=c}{\begin{tikzpicture}[scale=0.8]
 %\fill[white] (0.9,1.4)--(1.6,1.4)--(1.6,2.1)--(0.9,2.1)--(0.9,1.4);
        \draw[->] (-1,0)--(3,0);
        \node[above] at (-0.5,3){$2$};
        \node [right] at (3,0){$1$};
         \draw[->]   (-0.5,-0.5)--(-0.5,3);
         %%%%%%%%%%%%%%%%%%%%%%%%%%%%%%%%%%%%%%
          \fill[gray!70!white] (-0.5,0)--(0,0)--(0,0.5)--(-0.5,0.5)--(-0.5,0);
         \fill[gray!70!white] (0,0)--(0.5,0)--(0.5,0.5)--(-0,0.5)--(-0,0);
         \fill[gray!70!white] (0.5,0)--(1,0)--(1,0.5)--(0.5,0.5)--(0.5,0);
         \fill[gray!70!white] (1,0)--(1.5,0)--(1.5,0.5)--(1,0.5)--(1,0);
         \draw (-0.5,0)--(0,0)--(0,0.5)--(-0.5,0.5)--(-0.5,0);
         \draw (0,0)--(0.5,0)--(0.5,0.5)--(-0,0.5)--(-0,0);
         \draw (0.5,0)--(1,0)--(1,0.5)--(0.5,0.5)--(0.5,0);
         \draw (1,0)--(1.5,0)--(1.5,0.5)--(1,0.5)--(1,0);
        \fill[gray!70!white] (-0.5,0.5)--(0,0.5)--(0,1)--(-0.5,1)--(-0.5,0.5);
         \fill[gray!70!white] (0,0.5)--(0.5,0.5)--(0.5,1)--(-0,1)--(-0,0.5);
         \fill[gray!70!white] (0.5,0.5)--(1,0.5)--(1,1)--(0.5,1)--(0.5,0.5);
         \draw (-0.5,0.5)--(0,0.5)--(0,1)--(-0.5,1)--(-0.5,0.5);
         \draw (0,0.5)--(0.5,0.5)--(0.5,1)--(-0,1)--(-0,0.5);
         \draw (0.5,0.5)--(1,0.5)--(1,1)--(0.5,1)--(0.5,0.5);
         \fill[gray!70!white] (-0.5,1)--(0,1)--(0,1.5)--(-0.5,1.5)--(-0.5,1);
         \fill[gray!70!white] (0,1)--(0.5,1)--(0.5,1.5)--(-0,1.5)--(-0,1);
         \draw (-0.5,1)--(0,1)--(0,1.5)--(-0.5,1.5)--(-0.5,1);
         \draw (0,1)--(0.5,1)--(0.5,1.5)--(-0,1.5)--(-0,1);
         \fill[gray!70!white] (-0.5,1.5)--(0,1.5)--(0,2)--(-0.5,2)--(-0.5,1.5);
         \draw (-0.5,1.5)--(0,1.5)--(0,2)--(-0.5,2)--(-0.5,1.5);
          \end{tikzpicture}
        }\quad =\quad :\mathsf{X}_{12}(xq_{3}^{k-1})\prod_{\Abox\in\nu}\mathsf{A}^{-1}(q_{3}^{k-1}\chi_{12,x}(\Bbox)): .
\eea
We cannot construct any $qq$-character with the highest weight above and thus the D6 $qq$-character with three nontrivial legs is not represented as an infinite product of D4 $qq$-characters.

\subsection{General D6 \texorpdfstring{$qq$}{qq}-characters}
The D6 $qq$-characters introduced in the previous sections are the $qq$-characters whose highest weights come only from one D6-brane. We may generalize the situation when we have multiple D6-branes giving higher rank generalizations and tetrahedron instantons generalizations. The discussion is similar to \cite[Sec.~7.5]{Kimura:2023bxy}, so we omit the generalizations. Another generalization is the D6 $qq$-characters with a negative highest weight appearing in the denominator, the so-called supergroup generalization. Physically, the negative highest weight corresponds to the anti D6-branes and they give extra anti-fundamental contributions to the partition function. Since we will use them to construct D8 $qq$-characters in section~\ref{sec:D8_qq}, let us list down the explicit formulas.

\paragraph{Supergroup generalization} Instead of including only the $\mathsf{W}_{\bar{a}}(x)$ as the highest weight, we can also include a negative weight as
\bea
:\frac{\mathsf{W}_{\bar{a}}(x)}{\mathsf{W}_{\bar{a}}(Kx)}\prod_{\scube\in\mathcal{B}}\mathsf{A}^{-1}(\chi_{\bar{a},x}(\cube)):
\eea
where $\mathcal{B}$ is the contributions coming from the boundaries. When $\mathcal{B}=\emptyset$, this will give the $qq$-character of the 7d $\U(1|1)$ theory (see \cite{Kimura:2023bxy} for details):
\bea
\mathsf{T}_{\bar{a}}(x\mid Kx)&\coloneqq\sum_{\pi\in\mathcal{PP}}\mathfrak{q}^{|\pi|}\widetilde{\mathcal{Z}}_{\bar{a}}[\pi,K]\Lambda_{\bar{a},\pi}^{K}(x),\quad \Lambda^{K}_{\bar{a}}(x)=:\frac{\mathsf{W}_{\bar{a}}(x)}{\mathsf{W}_{\bar{a}}(Kx)}\prod_{\scube\in\pi}\mathsf{A}^{-1}(\chi_{\bar{a},x}(\cube)):\\
\widetilde{\mathcal{Z}}^{\D6}_{\bar{a}}[\pi,K]&=\prod_{\scube\in\pi}\frac{(1-Kx/\chi_{\bar{a},x}(\cube))(1-q_{a}x/\chi_{\bar{a},x}(\cube))}{(1-Kq_{a}x/\chi_{\bar{a},x}(\cube))(1-x/\chi_{\bar{a},x}(\cube))}\prod_{\substack{\scube\in\pi\\\scubeF\in\pi}}g_{\bar{a}}\left(\frac{\chi_{\bar{a},x}(\cube)}{\chi_{\bar{a},x}(\cubeF)}\right)^{-1}.
\eea
The extra parameter $K$ here physically corresponds to the distance between the D6 and $\overline{\D6}$ branes. As discussed in \cite[Sec.~7.5]{Kimura:2023bxy}, tuning the parameter $K$ properly, one can obtain lower dimensional such as the D4 $qq$-characters or spiked instanton $qq$-characters. This situation will be the same when we have nontrivial boundary conditions. However, how we can tune $K$ depends on the boundaries. In this paper, to keep the discussion simple, we will always keep $K$ to be generic.

\paragraph{Surface boundary conditions}
The supergroup analogue of the D6 $qq$-character with surface boundary conditions given as in Prop.~\ref{prop:D6surface_qq} is
\bea
\mathsf{T}_{abc}^{k_{bc}k_{ac}k_{ab}}(x\mid Kx)&={:\frac{\mathsf{W}_{abc}(x)}{\mathsf{W}_{abc}(Kx)\prod\limits_{j=1}^{k_{ac}}\mathsf{X}_{ac}(xq_{b}^{j-1})\prod\limits_{k=1}^{k_{ab}}\mathsf{X}_{ab}(xq_{b}^{k_{ac}}q_{c}^{k-1})\prod\limits_{i=1}^{k_{bc}}\mathsf{X}_{bc}(xq_{b}^{k_{ac}}q_{c}^{k_{ab}}q_{a}^{i-1})}:}+\cdots\\
&=\mathsf{T}_{abc}(xq_{a}^{k_{bc}}q_{b}^{k_{ac}}q_{c}^{k_{ab}}\mid Kx)
\eea
and
\bea
\left\langle \mathsf{T}_{abc}^{k_{bc}k_{ac}k_{ab}}(x\mid Kx)\right\rangle= \sum_{\pi\in\mathcal{PP}}\mathfrak{q}^{|\pi|}\widetilde{\mathcal{Z}}_{abc}[\pi,Kq_{a}^{-k_{bc}}q_{b}^{-k_{ac}}q_{c}^{-k_{ab}}].
\eea
Namely, the surface boundary condition effectively shifts the $K$ parameter.

\paragraph{Leg boundary conditions}
The supergroup analogue of the D6 $qq$-characters with nontrivial leg boundary conditions as given in Prop.~\ref{prop:D6leg_qq} are
\bea
\mathsf{T}_{\bar{a},\lambda\mu\nu}(x\mid Kx)&={:\frac{\mathsf{W}_{\bar{a}}x}{\mathsf{W}_{\bar{a}}(Kx)}\prod_{\scube\in\mathcal{B}_{\lambda\mu\nu}}\mathsf{A}^{-1}(\chi_{\bar{a},x}(\cube)):}+\cdots\\
&=\sum_{\pi\in\mathcal{PP}_{\lambda\mu\nu}}\mathfrak{q}^{|\pi|}\widetilde{\mathcal{Z}}_{\bar{a};\lambda\mu\nu}^{\D6}[\pi,K]\Lambda_{\bar{a},\pi}^{K,\,\lambda\mu\nu}(x),\\
\eea
where
\bea
 \Lambda^{K,\,\lambda\mu\nu}_{\bar{a},\pi}(x)&={:\frac{\mathsf{W}_{\bar{a}}(x)}{\mathsf{W}_{\bar{a}}(Kx)}\prod_{\scube\in\mathcal{B}_{\lambda\mu\nu}}\mathsf{A}^{-1}(\chi_{\bar{a},x}(\cube))\prod_{\scube \in \pi} \mathsf{A}^{-1}(\chi_{\bar{a},x}(\cube)) :},\\
 \widetilde{\mathcal{Z}}^{\D6}_{\bar{a};\lambda\mu\nu}[\pi,K]&=\prod_{\scube\in\pi}\frac{1-Kx/\chi_{\bar{a},x}(\cube)}{1-Kq_{a}x/\chi_{\bar{a,x}}(\cube)}\widetilde{\mathcal{Z}}^{\D6}_{\bar{a};\lambda\mu\nu}[\pi].
\eea

\section{D8 \texorpdfstring{$qq$}{qq}-characters}\label{sec:D8_qq}
In this section, we generalize what we have done in previous sections and construct D8 $qq$-characters. Compared to the D6 $qq$-character case, we do not have screening charges for the D8 case. Instead, in our previous paper \cite{Kimura:2023bxy}, we made an attempt to construct the D8 $qq$-character by fusing infinite numbers of D6 $qq$-characters. In \cite{Kimura:2023bxy}, the sign issue that is crucial for the magnificent four partition function was not dealt in detail. We will show in section~\ref{sec:fusionD6toD8} that after taking care of the \textit{unmovable terms} carefully, we can reproduce the correct D8 $qq$-character \textit{including} the sign factor. Namely, we will give a quantum algebraic proof of the sign rule given in \eqref{eq:signfactor}. In section~\ref{sec:D8qqsignrule}, we will show that the D8 $qq$-characters commute with each other once the sign factor is fixed. Oppositely, the commutativity uniquely determines the sign factor. We will also discuss the relation with the plethystic exponential formula of the partition functions. In the following sections~\ref{sec:D8qqlegboundary} and \ref{sec:D8qqsurfaceboundary}, we construct D8 $qq$-characters with nontrivial boundary conditions using infinite products of D6 vertex operators.

\subsection{Fusion of D6 \texorpdfstring{$qq$}{qq}-characters and sign rules}\label{sec:fusionD6toD8}

The D8 $qq$-character is a $qq$-character whose highest weight is $\mathsf{Z}(K,x)$ and whose monomial terms are labeled by solid partitions. The operator part of the monomial terms of the D8 $qq$-character is given by
\bea
\Lambda^{K}_{\four,\rho}(x)={:\mathsf{Z}(K,x)\prod_{\shcube\in\rho}\mathsf{A}^{-1}(\chi_{\four,x}(\hcube)):},
\eea
where $\rho$ is a solid partition. The coefficients of the D8 $qq$-character are the $\U(1)$ partition functions of the magnificent four system obtained by a set of D8-$\overline{\D8}$ branes spanning $\mathbb{C}^{4}\times \mathbb{S}^{1}$ (see section~\ref{sec:D8partitionfunction}). Depending on how we take the infinite products, there are four possibilities of the U(1) partition function:
\bea%\label{eq:D8U1partitionfunction}
\mathcal{Z}^{\D8}_{\four;a}[\rho,\,K]&=\prod_{\shcube\in\rho}\frac{(1-Kx/\chi_{\four,x}(\shcube))}{(1-x/\chi_{\four,x}(\shcube))}\prod_{\shcube,\shcube'\in\rho}g_{\bar{a}}\left(\frac{\chi_{\four,x}(\hcube)}{\chi_{\four,x}(\hcube')}\right)^{-1},\quad a\in\four.
\eea
\begin{definition}[D8 $qq$-character]
The D8 $qq$-character is defined as
\bea
\mathsf{T}^{K}_{\four}(x)=\sum_{\rho\in\mathcal{SP}}\mathfrak{q}^{|\rho|}(-1)^{\sigma_{a}(\rho)}\mathcal{Z}^{\D8}_{\four;a}[\rho,K]\Lambda^{K}_{\four,\rho}(x),\quad a\in\four,
\eea
where $\mathcal{SP}$ denotes the set of arbitrary solid partitions extending in the four directions $1,2,3,4$. 
The explicit formula for the sign factor $\sigma_{a}(\rho)$ when $a=4$ is given in \eqref{eq:signfactor} and other formulas are obtained by using the quadrality symmetry.
\end{definition}
We note that actually the total coefficient $(-1)^{\sigma_{a}(\rho)}\mathcal{Z}^{\D8}_{\four;a}[\rho,K]$ does not depend on the choice of $a\in\four$ and thus the above modified D8 $qq$-character is a unique $qq$-character. 
\begin{theorem}[\cite{Monavari:2022rtf}]\label{thm:M4invariance}
For any $a,b\in\four$, we have
\bea
(-1)^{\sigma_{a}(\rho)}\mathcal{Z}^{\D8}_{\four;a}[\rho,K]=
(-1)^{\sigma_{b}(\rho)}\mathcal{Z}^{\D8}_{\four;b}[\rho,K].
\eea
\end{theorem}

The main claim of this section is the following theorem.
\begin{theorem}\label{thm:D6toD8}
The D8 $qq$-character is obtained as
\bea
\mathsf{T}^{K}_{\four}(x)\simeq \overleftarrow{\prod_{i=1}^{\infty}}\mathsf{T}_{\bar{a}}(xq_{a}^{i-1}\mid xq_{a}^{i-1}K),
\eea
where the equality is up to one-loop perturbative factors.
\end{theorem}
Let us focus on $a=4$ and give a proof of this theorem step by step. 
\begin{lemma}
    The contraction of the operators $\Lambda^{K}_{\bar{a},\Pi}(x)$ are
    \bea
    \Lambda^{K_{2}}_{\bar{b},\Pi^{(2)}}(x_{2})\Lambda^{K_{1}}_{\bar{a},\Pi^{(1)}}(x_{1})&=\mathcal{Z}^{\D6\tbar\D6}_{\text{1-loop}}(x_{1},\bar{a},K_{1}\mid x_{2},\bar{b},K_{2})\mathcal{Z}^{\D6\tbar\D6}_{\bar{a};K_{1}\mid \bar{b};K_{2}}(x_{1},\Pi^{(1)}\mid x_{2},\Pi^{(2)})\\ &\quad \times :\Lambda^{K_{2}}_{\bar{b},\Pi^{(2)}}(x_{2})\Lambda^{K_{1}}_{\bar{a},\Pi^{(1)}}(x_{1}):
    \eea
    where
    \bea
    \mathcal{Z}^{\D6\tbar\D6}_{\text{1-loop}}(x_{1},\bar{a},K_{1}\mid x_{2},\bar{b},K_{2})&=\exp\left(-\sum_{n>0}\frac{1}{n}\frac{\bfP_{a}^{[n]}\bfP_{b}^{[-n]}}{\bfP_{\four}^{[n]}}(1-K_{1}^{n})(1-K_{2}^{-n})\left(\frac{x_{1}}{x_{2}}\right)^{n}\right),\\
    \mathcal{Z}^{\D6\tbar\D6}_{\bar{a};K_{1}\mid \bar{b};K_{2}}(x_{1},\Pi^{(1)}\mid x_{2},\Pi^{(2)})&=\mathbb{I}\left[-\bfP_{a}^{\vee}(1-K_{1}^{-1})x_{1}^{-1}\bm{\Pi}^{(2)}-\bfP_{4}(1-K_{2})x_{2}\bm{\Pi}^{(i)\vee}+\bfP_{\four}\bm{\Pi}^{(i)\vee}\bm{\Pi}^{(j)}\right].
    \eea
    We introduced $\bm{\Pi}^{(1,2)}=\sum_{\scube\in\Pi^{(i)}}\chi_{\bar{a},\bar{b}}(\cube)$ for plane partitions $\Pi^{(1,2)}$ and for the moment $x_{1,2}$ are generic here.
\end{lemma}

\begin{lemma}[\cite{Kimura:2023bxy}]
Given two plane partitions $\Pi^{(1)},\Pi^{(2)}$ and parameters $x_{2}=q_{4}x_{1}$, we have
\bea
\mathcal{Z}^{\D6\tbar\D6}_{\bar{4};K\mid \bar{4};K}(x_{1},\Pi^{(1)}\mid x_{2},\Pi^{(2)})=0
\eea
for $\Pi^{(2)}\npreceq \Pi^{(1)}$.
\end{lemma}

Combining these lemmas, the infinite products of D6 $qq$-characters are expanded as
\bea
\overleftarrow{\prod_{i=1}^{\infty}}\mathsf{T}_{\bar{4}}(x_{i}\mid Kx_{i})&\simeq \sum_{\cdots \Pi^{(i+1)}\preceq \Pi^{(i)}\cdots}\mathfrak{q}^{\sum_{i}|\Pi^{(i)}|}\prod_{i=1}^{\infty}\widetilde{\mathcal{Z}}^{\D6}_{\bar{4}}[K,\Pi^{(i)}]\prod_{i<j}\mathcal{Z}^{\D6\tbar\D6}_{\bar{4};K\mid \bar{4};K}(x_{i},\Pi^{(i)}\mid x_{j},\Pi^{(j)}) :\prod_{i=1}^{\infty}\Lambda^{K}_{\bar{4},\Pi^{(i)}}(x_{i}):
\eea
where $x_{i}=xq_{4}^{i-1}$ and the perturbative factors are all omitted. Given a \textit{finite} solid partition $\rho$, we can decompose it into non-increasing finite plane partitions: $\rho=(\Pi^{(1)},\Pi^{(2)},\Pi^{(3)},\cdots)$ with $\Pi^{(i)}\succeq \Pi^{(i+1)}$. Since only finite numbers of $\Pi^{(i)}$ will be nonempty, the topological term is $\mathfrak{q}^{|\rho|}=\mathfrak{q}^{\sum_{i}|\Pi^{(i)}|}$. Moreover, by direct computation, one can easily show that the operator part obeys
\bea
{:\prod_{i=1}^{\infty}\Lambda^{K}_{\bar{4},\Pi^{(i)}}(x_{i}):}=\Lambda^{K}_{\four,\rho}(x).
\eea
Thus, the nontrivial part is how to obtain the coefficient part.

\begin{proposition}\label{prop:D6toD8vacuumproof}
    The coefficient part obeys:
    \bea
\prod_{i=1}^{\infty}\widetilde{\mathcal{Z}}^{\D6}_{\bar{4}}[K,\Pi^{(i)}]\prod_{i<j}\mathcal{Z}^{\D6\tbar\D6}_{\bar{4};K\mid \bar{4};K}(x_{i},\Pi^{(i)}\mid x_{j},\Pi^{(j)})=(-1)^{\sigma_{4}(\rho)}\mathcal{Z}^{\D8}_{\four;4}[\rho,K].
    \eea
\end{proposition}
The left hand side comes from the following character
\bea\label{eq:D8vacuum-proof1}
\mathbf{v}_{\D6\rightarrow \D8}&=\sum_{k=1}^{\infty}\left(-\bfP_{4}^{\vee}(1-K^{-1})x_{k}^{-1}\bm{\Pi}^{(k)} +\bfP_{123}^{\vee} \bm{\Pi}^{(k)\vee}\bm{\Pi}^{(k)}\right)\\
&+\sum_{i<j}\left(-\bfP_{4}^{\vee}(1-K^{-1})x_{i}^{-1}\bm{\Pi}^{(j)}-\bfP_{4}(1-K)x_{j}\bm{\Pi}^{(i)\vee}+\bfP_{\four}\bm{\Pi}^{(i)\vee}\bm{\Pi}^{(j)} \right).
\eea
First of all, one can show that this character is movable (see \cite{Nekrasov:2018xsb} and \cite{Monavari:2022rtf} for example). The term $-\bfP_{4}(1-K)x_{j}\bm{\Pi}^{(i)\vee}$ is also movable. Let $\chi=x_{i}q_{1}^{A-1}q_{2}^{B-1}q_{3}^{C-1}\,\,(A,B,C\in\mathbb{Z}_{\geq 1})$ be a term included in $\bm{\Pi}^{(i)}$, i.e. $(A,B,C)\in\Pi^{(i)}$. The unmovable part is  
\bea
\left[-\bfP_{4}(1-K)x_{j}\chi^{-1}\right]^{(0)}&=\left[-\bfP_{4}q_{4}^{j-i}q_{1}^{-A+1}q_{2}^{-B+1}q_{3}^{-C+1}\right]^{(0)}\\
&=-\delta_{A=B=C=i-j+1}+\delta_{A=B=C=i-j}
\eea
but since $j-i>0$ and $A,B,C\in\mathbb{Z}_{\geq 1}$ there is no unmovable part and we can safely use the reflection property \eqref{eq:reflectionprop0}, \eqref{eq:reflectionprop1}, \eqref{eq:reflectionprop2}, which eventually gives
\bea
\mathbb{I}\left[-\bfP_{4}(1-K)x_{j}\bm{\Pi}^{(i)\vee}\right]=\mathbb{I}\left[-\bfP^{\vee}_{4}(1-K^{-1})x^{-1}_{j}\bm{\Pi}^{(i)}\right].
\eea
Thus, the character $\mathbf{v}_{\D6\rightarrow \D8}$ above is equivalent to the following character
\bea\label{eq:D8vacuum-proof2}
\mathbf{v}'_{\D6\rightarrow \D8}&=-\bfP_{4}^{\vee}(1-K^{-1})\sum_{i}x_{i}^{-1}\sum_{j}\bm{\Pi}^{(j)}+\sum_{k=1}^{\infty}\bfP_{123}^{\vee} \bm{\Pi}^{(k)\vee}\bm{\Pi}^{(k)}+\sum_{i<j}\bfP_{\four}\bm{\Pi}^{(i)\vee}\bm{\Pi}^{(j)}\\
&=-(1-K^{-1})\sum_{j}\bm{\Pi}^{(j)}+\sum_{k=1}^{\infty}\bfP_{123}^{\vee} \bm{\Pi}^{(k)\vee}\bm{\Pi}^{(k)}+\sum_{i<j}\bfP_{\four}\bm{\Pi}^{(i)\vee}\bm{\Pi}^{(j)}
\eea
where in the second line, we formally regularized $\sum_{i}x_{i}=1/\bfP_{4}$. Note that we have $\mathbb{I}[\mathbf{v}_{\D6\rightarrow \D8}]=\mathbb{I}\left[\mathbf{v}'_{\D6\rightarrow \D8}\right]$ and no sign factor appears at this part. The nontrivial sign factor actually comes from the $\sum_{j>i}\bfP_{\four}\bm{\Pi}^{(i)\vee}\bm{\Pi}^{(j)}$ part.
\begin{definition}\label{def:signfactorrgeneral}
Let $\bm{\Upsilon}^{(i)}$ be a character whose terms take the form as $xq_{4}^{i-1}q_{1}^{\geq0}q_{2}^{\geq 0}q_{3}^{\geq 0}$. We define the sign factor as
\bea
s\left(\Upsilon\right)\coloneqq\left[\sum_{i<j}\bfP_{123}\bm{\Upsilon}^{(i)\vee}\bm{\Upsilon}^{(j)}\right]^{(0)}=\left[\sum_{i<j}\bfP^{\vee}_{123}\bm{\Upsilon}^{(i)}\bm{\Upsilon}^{(j)\vee}\right]^{(0)}
\eea
where the second equality is a consequence of the fact that the unmovable terms are always equal to its dual.
\end{definition}
Using $\bfP_{\four}=\bfP_{123}^{\vee}+\bfP_{123}$ and Prop.~\ref{app-prop:reflection_sign}, we have
\bea
\mathbb{I}\,[\,\mathbf{v}_{\D6\rightarrow \D8}]=
\mathbb{I}\,[\,\mathbf{v}'_{\D6\rightarrow \D8}]&=\mathbb{I}\left[\cdots +\sum_{i<j}\bfP_{\four}\bm{\Pi}^{(i)\vee}\bm{\Pi}^{(j)}\right]=\mathbb{I}\left[\cdots +\sum_{i<j}\bfP_{123}^{\vee}\bm{\Pi}^{(i)\vee}\bm{\Pi}^{(j)}+\sum_{i<j}\bfP_{123}\bm{\Pi}^{(i)\vee}\bm{\Pi}^{(j)}\right]\\
&=(-1)^{s(\Pi)}\mathbb{I}\left[\cdots +\sum_{i<j}\bfP_{123}^{\vee}\bm{\Pi}^{(i)\vee}\bm{\Pi}^{(j)}+\sum_{i<j}\bfP_{123}^{\vee}\bm{\Pi}^{(i)}\bm{\Pi}^{(j)\vee}\right]\\
&=(-1)^{s(\Pi)}\mathbb{I}\left[-(1-K^{-1})\sum_{j}\bm{\Pi}^{(j)}+\bfP_{123}^{\vee} \sum_{i}\bm{\Pi}^{(i)\vee}\sum_{j}\bm{\Pi}^{(j)}\right].
\eea
The index part indeed gives $\mathcal{Z}^{\D8}_{\four;4}[\rho,K]$ of \eqref{eq:D8U1partitionfunction}. Combining with the following proposition, we get Thm.~\ref{thm:D6toD8}.

\begin{proposition}\label{prop:sign_proof}
The sign factor $s(\Pi)$ is
\bea
s(\Pi)=\sigma_{4}(\rho) \mod2
\eea
where $\sigma_{4}(\rho)=\#\{(i,i,i,j)\in\rho\mid i<j\}$ and we have
\bea
(-1)^{s(\Pi)}=(-1)^{\sigma_{4}(\rho)}.
\eea

\end{proposition}

\begin{proof}
    Let us consider the unmovable part of
    \bea\label{eq:proofdef}
    \sum_{i=1}^{j-1}\bfP_{123}\bm{\Pi}^{(i)\vee}\bm{\Pi}^{(j)},\quad \bm{\Pi}^{(i)}=\sum_{\scube\in\Pi^{(i)}}q_{4}^{i-1}\chi_{\bar{4},x}(\cube).
    \eea
    We fix a term $\eta=x_{j}q_{1}^{A-1}q_{2}^{B-1}q_{3}^{C-1},\,\,(A,B,C)\in\Pi^{(j)}$ and $A,B,C\geq 1$. Since the plane partitions obey $\Pi^{(i)}\succeq \Pi^{(j)}$ for $j>i$, we have $(A,B,C)\in\Pi^{(i)}$. From the plane partition condition, the character $\bm{\Pi}^{(i)}$ can be decomposed as
    \bea
    \bm{\Pi}^{(i)}=\Delta_{i}(\eta)+\bm{\Pi}^{(i)}(\eta),\quad \Delta_{i}(\eta)=\sum_{a=1}^{A}\sum_{b=1}^{B}\sum_{c=1}^{C}xq_{4}^{i-1}q_{1}^{a-1}q_{2}^{b-1}q_{3}^{c-1}
    \eea
where $\bm{\Pi}^{i}(\eta)$ contains terms expressed as $xq_{4}^{i-1}q_{1}^{a-1}q_{2}^{b-1}q_{3}^{c-1}$ with $a>A$ or $b>B$ or $c>C$. 

First of all, the term $\bfP_{123}\bm{\Pi}^{(i)\vee}(\eta)\eta$ is movable. Focusing on
\bea
\bfP_{123}\xi^{\vee}\eta=\bfP_{123}q_{4}^{j-i}q_{1}^{A-a}q_{2}^{B-b}q_{3}^{C-c}=\bfP_{123}q_{1}^{A-a+i-j}q_{2}^{B-b+i-j}q_{3}^{C-c+i-j}
\eea
for $\forall\xi=xq_{4}^{i-1}q_{1}^{a-1}q_{2}^{b-1}q_{3}^{c-1}\in\bm{\Pi}^{(i)}(\eta)$, the term $q_{1}^{A-a+i-j}q_{2}^{B-b+i-j}q_{3}^{C-c+i-j}$ is strictly negative in either $q_{1}^{\leq -2},q_{2}^{\leq -2},q_{3}^{\leq -2}$ because $a-A,b-B,c-C\geq 1,\,\, j-i\geq 1$. Since $\bfP_{123}$ only contains terms where the degrees with respect to $q_{1,2,3}$ are $0$ or $1$, there are no unmovable terms.

It is then enough to focus on $\bfP_{123}\Delta^{\vee}_{i}(\eta)\eta$:
\bea
\bfP_{123}\Delta^{\vee}_{i}(\eta)\eta&=\bfP_{123}\frac{(1-q_{1}^{-A})(1-q_{2}^{-B})(1-q_{3}^{-C})}{\bfP^{\vee}_{123}}q_{4}^{j-i}q_{1}^{A-1}q_{2}^{B-1}q_{3}^{C-1}\\
&=(1-q_{1}^{A})(1-q_{2}^{B})(1-q_{3}^{C})(q_{1}q_{2}q_{3})^{i-j}.
\eea
The unmovable term only comes from $-q_{1}^{A}q_{2}^{B}q_{3}^{C}(q_{123})^{i-j}$ with $A=B=C=j-i>0$.

Combining all of these, we finally have
\bea
s(\Pi)=\left[\sum_{j}\sum_{i=1}^{j-1}\bfP_{123}\bm{\Pi}^{(i)\vee}\bm{\Pi}^{(j)}\right]^{(0)}&=\left[\sum_{j}\sum_{i=1}^{j-1}\sum_{\eta\in\Pi^{(j)}} \bfP_{123}\Delta^{\vee}_{i}(\eta)\eta  \right]^{(0)}\\
&=-\sum_{j}\sum_{i=1}^{j-1}\sum_{\eta=(A,B,C)\in\Pi^{(j)}}\left[q_{1}^{A}q_{2}^{B}q_{3}^{C}(q_{123})^{i-j}\right]^{(0)}\\
&=-\sum_{j}\sum_{i=1}^{j-1}\sum_{\eta=(A,B,C)\in\Pi^{(j)}}\delta_{A=B=C=j-i}\\
&=-\sum_{j}\sum_{\eta=(A,B,C)\in\Pi^{(j)}}\delta_{A=B=C<j}\\
&=\sigma_{4}(\rho)\quad \mod2.
\eea

\end{proof}

\begin{remark}
The analysis of the unmovable terms given in the proof above is similar to the original proof in \cite{Nekrasov:2018xsb}. We note again that the difference of the sign factors there comes from the definition of the square root part of the total character. In this paper, we are using $\bfP_{123}^{\vee}$, while in theirs they are using $\bfP_{123}$. The sign factor defined in Def.~\ref{def:signfactorrgeneral} here also resembles the one defined in \cite[Thm.~5.16]{Nekrasov:2023nai}. After changing the square root part there to $\bfP_{123}^{\vee}$, the sign factor there is defined as
\bea
s^{\text{NP}}(\chi)=\left[\bfP_{123}^{\vee}\sum_{i<j}\chi_{i}\chi_{j}^{\vee}\right]^{(0)}
\eea
where $\{\chi_{i}\}_{i=1}^{k}$ are the poles of the JK-residue. Since when taking the JK-residue, the poles will be ordered in the lexicographic order on monomials in the four variables $q_{1,2,3,4}$ (see \cite[Sec.~2.2.2]{Nekrasov:2018xsb}), after denoting the four-dimensional coordinates of $\chi_{i}$ as $(a_{i},b_{i},c_{i},d_{i})$, we obtain
\bea
i<j\Leftrightarrow (d_{i}<d_{j})\vee(d_{i}=d_{j}\wedge c_{i}<c_{j})\vee(d_{i}=d_{j},c_{i}=c_{j},b_{i}<b_{j})\vee(d_{i}=d_{j},c_{i}=c_{j},b_{i}=b_{j},a_{i}<a_{j}).
\eea
The unmovable terms are computed as
\bea
\left[\bfP_{123}^{\vee}\chi_{i}\chi_{j}^{\vee}\right]^{(0)}=0
\eea
for $(d_{i}=d_{j}\wedge c_{i}<c_{j})\vee(d_{i}=d_{j},c_{i}=c_{j},b_{i}<b_{j})\vee(d_{i}=d_{j},c_{i}=c_{j},b_{i}=b_{j},a_{i}<a_{j})$, because $\chi_{i}\chi_{j}^{\vee}$ will be either in powers $q_{1}^{<0},q_{2}^{< 0},q_{3}^{<0}$. Therefore, the contribution of the unmovable terms only comes from $d_{i}<d_{j}$. We may then collect all of the terms with the same coordinate $d_{i}=k$ and get the equality
\bea
\bm{\Pi}^{(k)}=\sum_{d_{i}=k}\chi_{i}
\eea
where $\bm{\Pi}^{(k)}$ is defined in \eqref{eq:proofdef}. The sign factor $s^{\text{NP}}(\chi)$ is then rewritten as
\bea
s^{\text{NP}}(\chi)=\left[\bfP_{123}^{\vee}\sum_{i<j}\chi_{i}\chi_{j}^{\vee}\right]^{(0)}=s(\Pi)
\eea
which matches with our definition.

\end{remark}

\subsection{Commutativity of D8 \texorpdfstring{$qq$}{qq}-characters and sign rules}\label{sec:D8qqsignrule}
In the previous section, we derived the complete D8 $qq$-character including proper sign rules using the infinite products of lower dimensional D6 $qq$-characters. Given such $qq$-characters, one would like to find the quadratic relations of them and determine the quantum algebraic relations. In \cite{Kimura:2023bxy}, we gave a set of quadratic relations of the D2, D4, D6 $qq$-characters and showed that when the $qq$-characters are related with D-branes spanning transverse subspaces, they commute with each other. However, since the sign rules for the D8 $qq$-characters were not fixed and we do not know any screening charges for them, we could not discuss the quadratic relations of the D8 $qq$-characters. In this section, we give a conjecture of the quadratic relations of the D8 $qq$-characters and also the D6 $qq$-characters and discuss the physical implication of them.

We focus on the following D8 $qq$-character 
\bea
\mathsf{T}^{K}_{\four}(x)=\sum_{\rho\in\mathcal{SP}}\mathfrak{q}^{|\rho|}(-1)^{\sigma_{4}(\rho)}\mathcal{Z}^{\D8}_{\four;4}[\rho,K]\Lambda^{K}_{\four,\rho}(x).
\eea
The composition of the operators $\Lambda_{\four,\rho}^{K}(x)$ are
\bea
\Lambda_{\four,\rho^{(2)}}^{K_{2}}(x_{2})\Lambda_{\four,\rho^{(1)}}^{K_{1}}(x_{1})&=\mathcal{Z}^{\D8\tbar\D8}_{\text{1-loop}}(x_{1},K_{1}\,|\,x_{2},K_{2})\mathcal{Z}^{\D8\tbar\D8}_{K_{1}|K_{2}}(x_{1},\rho^{(1)}\,|\,x_{2},\rho^{(2)}):\Lambda_{\four,\rho^{(2)}}^{K_{2}}(x_{2})\Lambda_{\four,\rho^{(1)}}^{K_{1}}(x_{1}):,
\eea
where
\bea
\mathcal{Z}^{\D8\tbar\D8}_{\text{1-loop}}(x_{1},K_{1}\,|\,x_{2},K_{2})&=\exp\left(-\sum_{n>0}\frac{1}{n}\frac{(1-K_{2}^{-n})(1-K_{1}^{n})}{\bfP_{\four}^{[n]}}\left(\frac{x_{1}}{x_{2}}\right)^{n}\right),\\
\mathcal{Z}^{\D8\tbar\D8}_{K_{1}|K_{2}}(x_{1},\rho^{(1)}\,|\,x_{2},\rho^{(2)})&=\prod_{\shcube\in\rho^{(2)}}\left(\frac{1-K_{1}x_{1}/\chi_{\four,x_{2}}(\hcube)}{1-x_{1}/\chi_{\four,x_{2}}(\hcube)}\right)\prod_{\shcube\in\rho^{(1)}}\left(K_{2}\frac{1-K_{2}^{-1}\chi_{\four,x_{1}}(\hcube)/x_{2}}{1-\chi_{\four,x_{1}}(\hcube)/x_{2}}\right)\\
&\qquad \qquad \times\prod_{\substack{\shcube\in\rho^{(1)}\\\shcube'\in\rho^{(2)}}}\mathcal{A}_{\mathbb{C}^{4}}\left(\frac{\chi_{\four,x_{1}}(\hcube)}{\chi_{\four,x_{2}}(\hcube')}\right)^{-1}.
\eea
For later use, we define
\bea
\mathsf{f}^{K_{1}|K_{2}}_{\four\four}\left(x_{1}/x_{2}\right)=\mathcal{Z}^{\D8\tbar\D8}_{\text{1-loop}}(x_{1},K_{1}\,|\,x_{2},K_{2})^{-1}
\eea
and then the composition of the $qq$-characters are given as
\bea
\mathsf{f}^{K_{1}|K_{2}}_{\four\four}(x_{1}/x_{2})\mathsf{T}^{K_{2}}_{\four}(x_{2})\mathsf{T}^{K_{1}}_{\four}(x_{1})
%=&\sum_{\rho^{(1)},\rho^{(2)}\in\mathcal{SP}}\mathfrak{q}^{|\rho^{(1)}|+|\rho^{(2)}|}(-1)^{\sigma_{4}(\rho^{(1)})+\sigma_{4}(\rho^{(2)})}\mathcal{Z}^{\D8}_{\four;4}[\rho^{(1)};K_{1}]\mathcal{Z}^{\D8}_{\four;4}[\rho^{(2)};K_{2}]\\
%&\qquad \times \mathcal{Z}^{\D8\tbar\D8}_{K_{1}|K_{2}}(x_{1},\rho^{(1)}\,|\,x_{2},\rho^{(2)}):\Lambda_{\four,\rho^{(2)}}^{K_{2}}(x_{2})\Lambda_{\four,\rho^{(1)}}^{K_{1}}(x_{1}):\\
=&\sum_{k=0}^{\infty}\mathfrak{q}^{k}\sum_{|\rho^{(1)}|+|\rho^{(2)}|=k}(-1)^{\sigma_{4}(\rho^{(1)})+\sigma_{4}(\rho^{(2)})}\mathcal{Z}^{\D8}_{\four;4}[\rho^{(1)},K_{1}]\mathcal{Z}^{\D8}_{\four;4}[\rho^{(2)},K_{2}]\\
&\qquad \times \mathcal{Z}^{\D8\tbar\D8}_{K_{1}|K_{2}}(x_{1},\rho^{(1)}\,|\,x_{2},\rho^{(2)}):\Lambda_{\four,\rho^{(2)}}^{K_{2}}(x_{2})\Lambda_{\four,\rho^{(1)}}^{K_{1}}(x_{1}):\\
\eqqcolon &\sum_{k=0}^{\infty}\mathfrak{q}^{k}\,\mathsf{F}_{k}(x_{1},K_{1}\,|\, x_{2}, K_{2})
\eea
where we denoted the $k$ instanton contribution of the quadratic relation as $\mathsf{F}_{k}(x_{1},K_{1}\,|\, x_{2},K_{2})$. Note that this coefficient is a sum of operators with coefficients. Our main claim of this section is the following conjecture.
\begin{conjecture}\label{conj:D8commutativity}
The quadratic relations of the D8 $qq$-characters are
    \bea
\mathsf{f}^{K_{2}|K_{1}}_{\four\four}\left(x_{2}/x_{1}\right)\mathsf{T}_{\four}^{K_{1}}(x_{1})\mathsf{T}^{K_{2}}_{\four}(x_{2})-\mathsf{f}^{K_{1}|K_{2}}_{\four\four}\left(x_{1}/x_{2}\right)\mathsf{T}_{\four}^{K_{2}}(x_{2})\mathsf{T}^{K_{1}}_{\four}(x_{1})=0,
    \eea
    where $K_{1},K_{2}$ are arbitrary. Moreover, this commutativity uniquely determines the sign factor $(-1)^{\sigma_{4}(\rho)}$ up to a global $\mathbb{Z}_{2}$ symmetry. 
\end{conjecture}
To fix this global $\mathbb{Z}_{2}$ symmetry, we can impose the sign of the first instanton contribution by hand. Imposing the sign factor of the one-instanton to be $(-1)^{\sigma_{4}(\rho)}=1$, we have
    \bea
 \sigma_{4}(\rho)=\#\left\{(i,j)\mid (i,i,i,j)\in\rho,\quad i<j\right\},
    \eea
    which is the same as \eqref{eq:signfactor}. Note that the commutativity of the quadratic relation implies
    \bea
\mathsf{F}_{k}(x_{1},K_{1}\mid x_{2},K_{2})=\mathsf{F}_{k}(x_{2},K_{2}\mid x_{1},K_{2}).
    \eea
    Note also that this identity should be understood as a relation between operators of different analytic region.
    
For the moment, we do not have a proof of this conjecture, but we have checked it up to five instantons $(k\leq 5)$ by using a computer program. Since we have $1,1,4,10,26,59$ possible configurations for $|\rho|=0,1,2,3,4,5$, respectively, the number of terms for each $k=0,1,2,3,4,5$ are
\bea
\begin{tabular}{|c|c|c|} \hline
$k$ &$(|\rho^{(1)}|,|\rho^{(2)}|)$ & total number of terms \\ \hline\hline
     $0$ & $(0,0)$ & $1$   \\\hline
     $1$ & $(1,0),(0,1)$ & 2  \\ \hline
     $2$   & $(2,0),(1,1),(0,2)$   &  9 \\\hline
      $3$  &  $(3,0),(2,1),(1,2),(0,3)$  &  28  \\\hline
      $4$  &  $(4,0),(3,1),(2,2),(1,3),(0,4)$ &  88  \\\hline
       $5$ &  $(5,0),(4,1),(3,2),(2,3),(1,4),(0,5)$  & 250    \\\hline   
\end{tabular}
\eea
The confirmation up to five instanton is already a nontrivial confirmation. 

At higher instanton levels $k\geq 4$, actually the factor $\mathcal{Z}^{\D8\tbar\D8}_{K_{1}|K_{2}}(x_{1},\rho^{(1)}\,|\,x_{2},\rho^{(2)})$ has poles with higher orders.
\begin{conjecture}
    The factor $\mathcal{Z}^{\D8\tbar\D8}_{K_{1}|K_{2}}(x_{1},\rho^{(1)}\,|\,x_{2},\rho^{(2)})$ has higher order poles only at $x_{1}=x_{2}$.
\end{conjecture}
Strictly speaking, to consider the quadratic relations, one needs to deal with these poles. However, to confirm the conjecture, we excluded such poles and only focus on the contribution from poles which are single order.  To include higher order poles, we expect that we need to deal with differentiated operators. Somehow, in the context of quantum algebra, discussions on higher order poles are poorly studied. Well-known examples where this kind of phenomena appears are the $qq$-characters associated with geometries with D, E-type quivers or higher rank $qq$-characters with non-trivial limits of spectral parameters. In such cases, the iWeyl reflection procedure needs to be modified and derivatives of the vertex operators will appear (see \cite{Nekrasov:2015wsu, Kimura:2015rgi} for example). In this paper, we will not make an attempt to discuss all of these aspects and leave it for future work.

\paragraph{Low instanton computations}
To see the commutativity, focusing on the $k=0,1$ sectors is already intuitive. For the zero--instanton case, the commutativity is trivial because
\bea
\mathsf{f}^{K_{2}|K_{1}}_{\four\four}(x_{2}/x_{1})\mathsf{Z}(K_{1},x_{1})\mathsf{Z}(K_{2},x_{2})&={:\mathsf{Z}(K_{2},x_{2})\mathsf{Z}(K_{1},x_{1}):}\\
\mathsf{f}^{K_{1}|K_{2}}_{\four\four}(x_{1}/x_{2})\mathsf{Z}(K_{2},x_{2})\mathsf{Z}(K_{1},x_{1})&={:\mathsf{Z}(K_{2},x_{2})\mathsf{Z}(K_{1},x_{1}):}.
\eea

For the one--instanton case, the sign factor $(-1)^{\sigma_{4}(\{\{\{1\}\}\})}=1,-1$ is arbitrary because it is an overall factor when considering the quadratic relation. The possible configurations are $(\rho^{(1)},\rho^{(2)})=(\{\{\{1\}\}\},\emptyset),(\emptyset,\{\{\{1\}\}\})$. The contribution from $(\rho^{(1)},\rho^{(2)})=(\{\{\{1\}\}\},\emptyset)$ is
\bea
\mathsf{f}^{K_{1}| K_{2}}_{\four\four}(x_{1}/x_{2})\mathsf{Z}(K_{2},x_{2}):\mathsf{Z}(K_{1},x_{1})\mathsf{A}^{-1}(x_{1}):&=K_{2}\frac{1-K_{2}^{-1}x_{1}/x_{2}}{1-x_{1}/x_{2}}:\mathsf{Z}(K_{2},x_{2})\mathsf{Z}(K_{1},x_{1})\mathsf{A}^{-1}(x_{1}):,\\
\mathsf{f}_{\four\four}^{K_{2}|K_{1}}(x_{2}/x_{1}):\mathsf{Z}(K_{1},x_{1})\mathsf{A}^{-1}(x_{1}):\mathsf{Z}(K_{2},x_{2})&=\frac{1-K_{2}\,x_{2}/x_{1}}{1-x_{2}/x_{1}}:\mathsf{Z}(K_{2},x_{2})\mathsf{Z}(K_{1},x_{1})\mathsf{A}^{-1}(x_{1}):
\eea
which gives
\bea
&\mathsf{f}^{K_{1}| K_{2}}_{\four\four}(x_{1}/x_{2})\mathsf{Z}(K_{2},x_{2}):\mathsf{Z}(K_{1},x_{1})\mathsf{A}^{-1}(x_{1}):-\mathsf{f}_{\four\four}^{K_{2}|K_{1}}(x_{2}/x_{1}):\mathsf{Z}(K_{1},x_{1})\mathsf{A}^{-1}(x_{1}):\mathsf{Z}(K_{2},x_{2})\\
=&-(1-K_{2})\delta\left(x_{1}/x_{2}\right):\mathsf{Z}(K_{1},x_{1})\mathsf{Z}(K_{2},x_{1})\mathsf{A}^{-1}(x_{1}):.
\eea
The contribution from $(\rho^{(1)},\rho^{(2)})=(\emptyset,\{\{\{1\}\}\})$ is obtained by switching the parameters as $K_{1}\leftrightarrow K_{2},\,x_{1}\leftrightarrow x_{2}$:
\bea
\mathsf{f}_{\four\four}^{K_{1}|K_{2}}(x_{1}/x_{2}):\mathsf{Z}(K_{2},x_{2})\mathsf{A}^{-1}(x_{2}):\mathsf{Z}(K_{1},x_{1})&=\frac{1-K_{1}\,x_{1}/x_{2}}{1-x_{1}/x_{2}}:\mathsf{Z}(K_{1},x_{1})\mathsf{Z}(K_{2},x_{2})\mathsf{A}^{-1}(x_{2}):,\\
\mathsf{f}^{K_{2}| K_{1}}_{\four\four}(x_{2}/x_{1})\mathsf{Z}(K_{1},x_{1}):\mathsf{Z}(K_{2},x_{2})\mathsf{A}^{-1}(x_{2}):&=K_{1}\frac{1-K_{1}^{-1}x_{2}/x_{1}}{1-x_{2}/x_{1}}:\mathsf{Z}(K_{1},x_{1})\mathsf{Z}(K_{2},x_{2})\mathsf{A}^{-1}(x_{2}):
\eea
which gives
\bea
&\mathsf{f}_{\four\four}^{K_{1}|K_{2}}(x_{1}/x_{2}):\mathsf{Z}(K_{2},x_{2})\mathsf{A}^{-1}(x_{2}):\mathsf{Z}(K_{1},x_{1})-\mathsf{f}^{K_{2}| K_{1}}_{\four\four}(x_{2}/x_{1})\mathsf{Z}(K_{1},x_{1}):\mathsf{Z}(K_{2},x_{2})\mathsf{A}^{-1}(x_{2}):\\
=&(1-K_{1})\delta(x_{1}/x_{2}):\mathsf{Z}(K_{1},x_{1})\mathsf{Z}(K_{2},x_{1})\mathsf{A}^{-1}(x_{1}):.
\eea
Using
\bea
\mathcal{Z}^{\D8}_{\four;4}[\{\{\{1\}\}\},K]=\frac{(1-K)(1-q_{12})(1-q_{13})(1-q_{23})}{(1-q_{1})(1-q_{2})(1-q_{3})(1-q_{123})}%=(1-K_{1})\widetilde{\mathcal{Z}}_{\four;4}[\{\{\{1\}\}\}]
\eea
we have the following one-instanton contribution
\bea
&\mathsf{F}_{1}(x_{1},K_{1}\,|\,x_{2}, K_{2})-\mathsf{F}_{1}(x_{2},K_{2}\,|\, x_{1}, K_{1})\\
=&-\mathcal{Z}^{\D8}_{\four;4}[{\{\{\{1\}\}\}}\,, K_{1}](1-K_{2})\delta\left(x_{1}/x_{2}\right):\mathsf{Z}(K_{1},x_{1})\mathsf{Z}(K_{2},x_{1})\mathsf{A}^{-1}(x_{1}):\\
&+\mathcal{Z}^{\D8}_{\four;4}[{\{\{\{1\}\}\}},K_{2}](1-K_{1})\delta(x_{1}/x_{2}):\mathsf{Z}(K_{1},x_{1})\mathsf{Z}(K_{2},x_{1})\mathsf{A}^{-1}(x_{1}):\\
=&\frac{(1-K_{1})(1-K_{2})(1-q_{12})(1-q_{13})(1-q_{23})}{(1-q_{1})(1-q_{2})(1-q_{3})(1-q_{123})}\\
&\qquad \times \left(:\mathsf{Z}(K_{1},x_{1})\mathsf{Z}(K_{2},x_{1})\mathsf{A}^{-1}(x_{1}):-:\mathsf{Z}(K_{1},x_{1})\mathsf{Z}(K_{2},x_{1})\mathsf{A}^{-1}(x_{1}):\right)\delta\left(x_{1}/x_{2}\right)\\
=&0.
\eea
Therefore, at the one-instanton level, the D8 $qq$-characters commute with each other.

\paragraph{BPS/CFT correspondence} Using the D8 $qq$-characters, we can indeed reproduce the correct $\U(1|1)$ magnificent four partition function including the sign factors, following the logic in \cite{Kimura:2023bxy}. The BPS/CFT correspondence for higher rank magnificent four theory with $\U(n|n)$ is
\bea\label{eq:D8qqBPSCFT}
\bra{0}\mathsf{T}^{K_{n}}_{\four}(x_{n})\cdots \mathsf{T}_{\four}^{K_{1}}(x_{1})\ket{0}&=\prod_{\beta>\alpha}\mathcal{Z}_{\text{1-loop}}^{\D8\tbar\D8}(x_{\alpha},K_{\alpha}\mid x_{\beta},K_{\beta})\\
&\times \sum_{\rho^{(1)},\cdots ,\rho^{(n)}}\mathfrak{q}^{|\vec{\rho}|}\prod_{\alpha=1}^{n}(-1)^{\sigma_{4}(\rho^{(\alpha)})}\mathcal{Z}_{\four;4}^{\D8}[\rho^{(\alpha)},K_{\alpha}]\prod_{\beta>\alpha}\mathcal{Z}^{\D8\tbar\D8}_{K_{\alpha},K_{\beta}}(x_{\alpha},\rho^{(\alpha)}\mid x_{\beta},\rho^{(\beta)}).
\eea

\paragraph{Reduction to D6 \texorpdfstring{$qq$}{qq}-characters}
The D6 $qq$-characters can be obtained by reductions of the D8 $qq$-characters. The relation \eqref{eq:vertexoprelation} gives $\mathsf{Z}(q_{a},x)=\mathsf{W}_{\bar{a}}(x)$. Under this specialization, we actually can show that the solid partition stops its growth in one of the four directions.
\begin{lemma}\label{lemm:D8toD6reduce}
    If we specialize $K=q_{a}\,(a\in\four)$, then the coefficient $\mathcal{Z}^{\D8}_{\four;4}[\rho;K]\,\,(\rho\in\mathcal{SP})$ disappears when the solid partition extends to the direction $q_{a}$:
    \bea
    \mathcal{Z}^{\D8}_{\four;4}[\rho;q_{a}]=0,\quad \rho\notin\mathcal{PP}_{a}
    \eea
    where we denoted $\mathcal{PP}_{a}$ as the set of plane partitions not extending in the $q_{a}$ direction.
\end{lemma}

\begin{lemma}\label{lemm:signfactorreduce}
Note that when the solid partition $\rho$ is reduced to a plane partition as $\pi\in\mathcal{PP}_{a}\,(a\in\four)$, the sign factor can be written in a simple form:
\bea
\sigma_{4}(\pi)=\begin{dcases}
    \min\{h_{4}(\pi)-1,0\},\quad \pi\in\mathcal{PP}_{1,2,3}\\
    0,\qquad \pi\in\mathcal{PP}_{4}
\end{dcases}
\eea
and
\bea
(-1)^{\sigma_{4}(\pi)}=\begin{dcases}
    (-1)^{\min\{-1+h_{4}(\pi),0\}},\quad \pi\in\mathcal{PP}_{1,2,3}\\
    1,\qquad \pi\in\mathcal{PP}_{4}
    \end{dcases}
\eea
where $h_{4}(\pi)$ here is the height of the plane partition in the $4$-direction. Similar formulas for $\sigma_{a}(\pi)$ can be obtained by using the quadrality symmetry.
\end{lemma}
\begin{proof}
When $\pi\in\mathcal{PP}_{4}$, we have
\bea
\left\{(i,j)\mid (i,i,i,j)\in\pi,\,\, i< j\right\}&=\left\{(i,1)\mid (i,i,i,1)\in\pi,\,\, i< 1\right\}=\emptyset
\eea
which gives $(-1)^{\sigma_{4}(\pi)}=(-1)^{2}=1$.

For other cases, when $\pi\in\mathcal{PP}_{a}\,(a=1,2,3)$:
\bea
\left\{(i,j)\mid (i,i,i,j)\in\pi,\,\, i< j\right\}&=\left\{(1,j)\mid (1,1,1,j)\in\pi,\,\, 1<j\right\}.
\eea
When $\pi=\emptyset$, then this set is also empty, which gives $\sigma_{4}(\pi)=0$. Focusing on the three-dimensional part $(1,1,j)$ where $(1,1)$ is a two-dimensional part of $\four\setminus 4$, the number of the elements of this set is rewritten using the height of the plane partition in the fourth direction:
\bea
\sigma_{4}(\pi)=h_{4}(\pi)-1,\quad h_{4}(\pi)\coloneqq\min\left\{j\geq 1\mid (1,1,j+1)\neq \pi\right\}.
\eea

\end{proof}

\begin{lemma}\label{lemm:D8-D6reduce-sign}
After specializing $K=q_{a}\,(a\in\four)$, we have the following identity:
\bea
(-1)^{\sigma_{4}(\pi)}\mathcal{Z}_{\four;4}^{\D8}[\pi,q_{a}]=\widetilde{\mathcal{Z}}^{\D6}_{\bar{a}}[\pi], \quad a\in\four
\eea
where the plane partition $\pi\in\mathcal{PP}_{a}$.% as $\pi\hookrightarrow \rho$ in the left-hand side \remred{(refer to notation part and refer to it)}.
\end{lemma}

\begin{proof}
For $a=4$, it is obvious. Other cases are simply obtained by using Lemma~\ref{lemm:D8toD6reduce} and Lemma~\ref{lemm:signfactorreduce}. A rather direct proof is given in Appendix~\ref{app-sec:D8-D6reduce-sign-proof}.
\end{proof}

\begin{proposition}
The $\D8$ $qq$-character reduces to an expansion of plane partitions of $\mathcal{PP}_{a}$. Using 
    \bea
    \Lambda^{q_{a}}_{\four,\pi}(x)=\Lambda_{\bar{a},\pi}(x),\quad \pi\in\mathcal
{PP}_{a},
    \eea
    we have
    \bea
    \mathsf{T}_{\four}^{q_{a}}(x)&=\sum_{\pi\in\mathcal{PP}_{a}}\mathfrak{q}^{|\pi|}(-1)^{\sigma_{4}(\pi)}\mathcal{Z}^{\D8}_{\four;4}[\pi,q_{a}]\Lambda_{\bar{a},\pi}(x)\\
    &=\mathsf{T}_{\bar{a}}(x).
    \eea
\end{proposition}
Combining with the Conj.~\ref{conj:D8commutativity}, we obtain the commutativity of the D6 $qq$-characters.

\begin{corollary}\label{cor:D6commutativity}
    The D6 $qq$-characters all commute with each other:
    \bea
\mathsf{f}^{q_{b}|q_{a}}_{\four\four}\left(x_{2}/x_{1}\right)\mathsf{T}_{\bar{a}}(x_{1})\mathsf{T}_{\bar{b}}(x_{2})-\mathsf{f}^{q_{a}|q_{b}}_{\four\four}\left(x_{1}/x_{2}\right)\mathsf{T}_{\bar{b}}(x_{2})\mathsf{T}_{\bar{a}}(x_{1})=0,\quad a,b\in\four.
    \eea
    Moreover, they also commute with the D8 $qq$-characters:
    \bea
    \mathsf{f}^{q_{a}|K}_{\four\four}\left(x_{2}/x_{1}\right)\mathsf{T}_{\four}^{K}(x_{1})\mathsf{T}_{\bar{a}}(x_{2})-\mathsf{f}^{K|q_{a}}_{\four\four}\left(x_{1}/x_{2}\right)\mathsf{T}_{\bar{a}}(x_{2})\mathsf{T}^{K}_{\four}(x_{1})=0,\quad a\in\four.
    \eea
\end{corollary}
Conj.~\ref{conj:D8commutativity} and Cor.~\ref{cor:D6commutativity} show that all the D6 and D8 $qq$-characters commute with each other. This property is surprising from the quantum algebraic viewpoint because usually when the generators of the deformed W-algebra commute with each other, one would expect it to be a trivial algebra, though they do reproduce the tetrahedron instantons and magnificent four partition functions. For the moment, we do not know how to understand this phenomenon from the representation theoretic viewpoint and details are left for future work.

\paragraph{Relation with the plethystic exponential formulas} Although, the quantum algebraic meaning of the commutativity of the D6 and D8 $qq$-characters is not so clear for the moment, it has a physical meaning. The commutativity of the D6 and D8 $qq$-characters is actually a consequence of the fact that the partition functions of the tetrahedron instantons and magnificent four do \emph{not} depend on the Coulomb branch parameters (spectral parameters) and have a beautiful plethystic exponential formula \cite{Awata:2009dd,Nekrasov:2017cih,Nekrasov:2009JJM,Pomoni:2023nlf, Fasola:2023ypx}. Recall that the plethystic exponential formula for the $\U(n|n)$ magnificent four instanton partition function is given as
\bea
\mathcal{Z}^{\D8}_{\text{PE}}\,\left[\{K_{\alpha}\}_{\alpha=1}^{n}\right]\coloneqq\PE\left[\frac{-q_{4}\prod_{i=1}^{3}(1-q_{i4}^{-1})}{\prod_{a\in\four}(1-q_{a})}\frac{1-\prod_{\alpha=1}^{n}K_{\alpha}^{-1}}{(1-\mathfrak{q})(1-\prod_{\alpha=1}^{n}K_{\alpha}^{-1}\mathfrak{q}^{-1})}\right]
\eea
in our notation. Namely, the instanton partition function only depends on the product $\prod_{\alpha=1}^{n}K_{\alpha}$ which is algebraically the central charge of this system. Moreover by setting $K_{\alpha}$ to $q_{1,2,3,4}$, we will get the plethystic formula for the tetrahedron instantons \cite{Pomoni:2023nlf, Fasola:2023ypx}. Physically, this corresponds to the tachyon condensation of the pairs of $\D8\tbar\overline{\D8}$ branes. 

Using \eqref{eq:D8qqBPSCFT}, we have
\bea
\prod_{\beta>\alpha}\mathsf{f}_{\four\four}^{K_{\alpha}|K_{\beta}}(x_{\alpha}/x_{\beta})\bra{0}\mathsf{T}^{K_{n}}_{\four}(x_{n})\cdots \mathsf{T}_{\four}^{K_{1}}(x_{1})\ket{0}=\mathcal{Z}^{\D8}_{\text{PE}}\,\left[\{K_{\alpha}\}_{\alpha=1}^{n}\right].
\eea
Since the right hand side does not depend on $\{x_{\alpha}\}_{\alpha=1}^{n}$, the left hand side should also not depend on them. Focusing on $n=2$, this gives exactly the commutativity in Conj.~\ref{conj:D8commutativity}. Similarly, this discussion is applicable to the tetrahedron instanton case by tuning the parameters $\{K_{\alpha}\}_{\alpha=1}^{n}$.

We stress that we are \emph{not} saying that the commutativity of the D6, D8 $qq$-characters proves the independence of the Coulomb branch parameters nor the existence of a plethystic formula. We are saying that if such kind of properties exist, then we should have the commutativity of the $qq$-characters and indeed for the D6, D8 $qq$-characters, it is true. 

We also note that this commutativity is \emph{not} satisfied for D4 and D6 $qq$-characters. As mentioned in \cite{Kimura:2023bxy}, we can introduce a $\D6\tbar\overline{\D6}$ $qq$-character which after tachyon condensation, we can reproduce the spiked instanton D4 $qq$-characters. The commutation relation of such $\D6\tbar\overline{\D6}$ $qq$-characters actually reproduces extra terms which imply nontrivial quadratic relations for the affine quiver W-algebra. Moreover, the commutation relations of the D6 $qq$-characters with nontrivial boundary conditions discussed in the following sections also seem to produce extra terms. Discussion of all of these cases are left for future work.

\subsection{Leg boundary conditions}\label{sec:D8qqlegboundary}
The D8 $qq$-characters with nontrivial leg boundary conditions are defined as the following.
\begin{definition}
Let $\pi_{a},\,(a\in\four)$ be finite plane partitions and $\mathcal{SP}_{\pi_{1}\pi_{2}\pi_{3}\pi_{4}}$ be the set of possible solid plane partitions with the boundary plane partitions $\pi_{1,2,3,4}$. The D8 $qq$-character with leg boundary conditions is defined as 
    \bea
    \mathsf{T}_{\four,\pi_{1}\pi_{2}\pi_{3}\pi_{4}}^{K}(x)=\sum_{\rho\in\mathcal{SP}_{\pi_{1}\pi_{2}\pi_{3}\pi_{4}}}\mathfrak{q}^{|\rho|}(-1)^{\sigma_{4}(\rho)}\mathcal{Z}^{\D8}_{\four;4;\pi_{1}\pi_{2}\pi_{3}\pi_{4}}[\rho,K]\Lambda_{\four,\rho}^{K;\pi_{1}\pi_{2}\pi_{3}\pi_{4}}(x),
    \eea
    where
    \bea
    \Lambda_{\four,\rho}^{K;\pi_{1}\pi_{2}\pi_{3}\pi_{4}}(x)={:\frac{\mathsf{Z}(K,x)}{\prod_{a\in\four}\prod_{\scube\in\pi_{a}}\mathsf{S}_{a}(\chi_{\bar{a},x}(\cube))}\prod_{\shcube\in\mathcal{S}_{\pi_{1}\pi_{2}\pi_{3}\pi_{4}}}\mathsf{A}(\chi_{\four,x}(\hcube))\prod_{\shcube\in\rho}\mathsf{A}^{-1}(\chi_{\four,x}(\hcube)):}
    \eea
    and $\mathcal{Z}^{\D8}_{\four;4;\pi_{1}\pi_{2}\pi_{3}\pi_{4}}[\rho,K]$, $\sigma_{4}(\rho)$ are defined in \eqref{eq:D8legpartition} and \eqref{eq:D8boundary_signrule}, respectively.
\end{definition}
We give a derivation of this D8 $qq$-character using the fusion process of the D6 $qq$-characters. To make the discussion simple, we only focus on the case when there is one nontrivial leg plane partition.

\paragraph{One-leg D8 $qq$-character: part 1}
Let us derive the one-leg D8 $qq$-character by taking infinite products of lower dimensional $qq$-characters. In particular, we will use the D6 $qq$-characters with boundary conditions spanning the 123-plane and choose the 4-direction to be a special direction. This decomposition is the one explained in section~\ref{sec:D8partitionfunction}. Since this decomposition breaks the quadrality symmetry and we only have the triality symmetry with respect with the 1,2,3 directions, we have two cases that needs to be treated differently. The D8 $qq$-character with a plane partition extending at one of the 123-directions and the D8 $qq$-character with a plane partition extending at the 4-direction.

Let us consider the case when the solid partition $\widetilde{\rho}$ has a boundary plane partition $\pi_{3}$ extending semi-infinitely in the 3-direction as \eqref{eq:fig-solidpartitionleg123}.
We can use the $(1,3)$ decomposition of the infinite size solid partition $\tilde{\rho}$ (whose finite part is denoted by $\rho$) and decompose it into infinite size plane partitions $\tilde{\rho}=(\widetilde{\Pi}^{(1)},\widetilde{\Pi}^{(2)},\widetilde{\Pi}^{(3)},\ldots)$. Under this decomposition, the boundary plane partition $\pi_{3}$ will be also decomposed into non-increasing sequences of Young diagrams:
\bea
\pi_{3}=\left(\nu^{(1)},\nu^{(2)},\nu^{(3)},\ldots \right),\quad \nu^{(i)}\succeq \nu^{(i+1)}
\eea
where $\nu^{(k)}$ are finite Young diagrams. Therefore, for the $k$-th layer, we will have an infinite size plane partition $\widetilde{\Pi}^{(k)}$ with an asymptotic Young diagram $\nu^{(k)}$. We denote $\Pi^{(k)}$ as the set of boxes not included in the boundary Young diagram $\nu^{(k)}$. The D6 $qq$-character for each layer comes from the following ingredients:
\bea
\text{operator part:}&\qquad \Lambda^{K,\,\emptyset\emptyset\nu^{(k)}}_{\bar{4},\Pi^{(k)}}(q_{4}^{k-1}x)={:\frac{\mathsf{W}_{\bar{4}}(q_{4}^{k-1}x)\prod_{\scube\in\Pi^{(k)}}\mathsf{A}^{-1}(\chi_{\bar{4},q_{4}^{k-1}x}(\cube))}{\mathsf{W}_{\bar{4}}(Kq_{4}^{k-1}x)\prod_{\Abox\in\nu^{(k)}}\mathsf{S}_{3}(\chi_{12,q_{4}^{k-1}x}(\Bbox))}:},\\
\text{coefficient part:}&\qquad \widetilde{\mathcal{Z}}_{\bar{4};\emptyset\emptyset\Lambda^{(k)}}^{\D6}[\,K,\Pi^{(k)}\,].
\eea
The coefficient part and the operator part of the D8 $qq$-character comes from the infinite product
\bea
\prod_{k=1}^{\infty}\widetilde{\mathcal{Z}}_{\bar{4};\emptyset\emptyset\nu^{(k)}}^{\D6}[\,K,\Pi^{(k)}\,]\overleftarrow{\prod_{k=1}^{\infty}}\Lambda^{K,\,\emptyset\emptyset\nu^{(k)}}_{\bar{4},\Pi^{(k)}}(q_{4}^{k-1}x).
\eea
\begin{proposition}\label{prop:D6toD8-oneleg1}
    We have
    \bea\label{eq:D6toD8-oneleg1}
\prod_{k=1}^{\infty}\widetilde{\mathcal{Z}}_{\bar{4};\emptyset\emptyset\nu^{(k)}}^{\D6}[\,K,\Pi^{(k)}\,]\overleftarrow{\prod_{k=1}^{\infty}}\Lambda^{K,\,\emptyset\emptyset\nu^{(k)}}_{\bar{4},\Pi^{(k)}}(q_{4}^{k-1}x)\simeq (-1)^{\widetilde{\sigma}_{4}(\rho)}\mathcal{Z}^{\D8}_{\four;4;\emptyset\emptyset\pi_{3}\emptyset}[\rho,K] \Lambda^{K;\emptyset\emptyset\pi_{3}\emptyset}_{\four;\rho}(x)
    \eea
    where the equality is up to one-loop perturbative factors and contributions coming only from the boundaries. Extra sign factors $\widetilde{\sigma}_{4}(\rho)$ depending on $\rho$ will also appear in this equality.
\end{proposition}
\begin{proof}
We first have
\bea
\widetilde{\mathcal{Z}}_{\bar{4};\emptyset\emptyset\nu^{(k)}}^{\D6}[\,K,\Pi^{(k)}\,]=\mathbb{I}\left[ -\bfP_{4}^{\vee}(1-K^{-1})x_{k}^{-1}\bm{\Pi}^{(k)}+\bfP_{\four}\frac{\bm{\nu}^{(k)\vee}}{\bfP_{3}^{\vee}}\bm{\Pi}^{(k)}+\bfP_{123}^{\vee}\bm{\Pi}^{(k)\vee}\bm{\Pi}^{(k)}  \right].
\eea
The contractions of the vertex operators are given as
\bea
\frac{\Lambda^{K,\,\emptyset\emptyset\nu^{(j)}}_{\bar{4},\Pi^{(j)}}(x_{j})\Lambda^{K,\,\emptyset\emptyset\nu^{(i)}}_{\bar{4},\Pi^{(i)}}(x_{i})}{:\Lambda^{K,\,\emptyset\emptyset\nu^{(j)}}_{\bar{4},\Pi^{(j)}}(x_{j})\Lambda^{K,\,\emptyset\emptyset\nu^{(i)}}_{\bar{4},\Pi^{(i)}}(x_{i}):}&\simeq\mathbb{I}\left[-\bfP_{4}^{\vee}(1-K^{-1})x_{i}^{-1}\bm{\Pi}^{(j)}-\bfP_{4}(1-K)x_{j}\bm{\Pi}^{(i)\vee}\right.\\
&\quad \left.+\bfP_{\four}\left(\frac{\bm{\nu}^{(j)}}{\bfP_{3}}\bm{\Pi}^{(i)\vee}+\frac{\bm{\nu}^{(i)\vee}}{\bfP_{3}^{\vee}}\bm{\Pi}^{(j)}\right)+\bfP_{\four}\mathbf{\Pi}^{(i)\vee}\mathbf{\Pi}^{(j)}\right]\\
&=\mathbb{I}\left[-\bfP_{4}^{\vee}(1-K^{-1})x_{i}^{-1}\bm{\Pi}^{(j)}-\bfP_{4}^{\vee}(1-K^{-1})x_{j}^{-1}\bm{\Pi}^{(i)}\right.\\
&\quad \left.+\bfP_{\four}\left(\frac{\bm{\nu}^{(j)}}{\bfP_{3}}\bm{\Pi}^{(i)\vee}+\frac{\bm{\nu}^{(i)\vee}}{\bfP_{3}^{\vee}}\bm{\Pi}^{(j)}\right)+\bfP_{\four}\mathbf{\Pi}^{(i)\vee}\mathbf{\Pi}^{(j)}\right]
\eea
for $j>i$ and $x_{i}=xq_{4}^{i-1}$, where we omitted contributions non-essential to instanton computations. We identified the plane partitions with the characters as
\bea
\mathbf{\Pi}^{(i)}=\sum_{\scube\in\Pi^{(i)}}\chi_{\bar{4},x_{i}}(\cube),\quad \boldsymbol{\nu}^{(i)}=\sum_{\Abox\in\nu^{(i)}}\chi_{12,x_{i}}(\Bbox)
\eea
and also used the reflection property in the second line. The computation is done similar as \eqref{eq:D8vacuum-proof1} and the second term of the first line is movable and so the reflection property can be safely used.

The coefficients of the left hand side of \eqref{eq:D6toD8-oneleg1} comes from the following character $\mathbf{v}_{\D6\rightarrow \D8}^{\emptyset\emptyset\pi_{3}\emptyset}$:
\bea
\mathbf{v}_{\D6\rightarrow \D8}^{\emptyset\emptyset\pi_{3}\emptyset}&=-(1-K^{-1})x^{-1}\sum_{i}\bm{\Pi}^{(i)}+\sum_{k}\left(\bfP_{123}^{\vee}\bm{\Pi}^{(k)\vee}\bm{\Pi}^{(k)}+\bfP_{\four}\frac{\bm{\nu}^{(k)\vee}}{\bfP_{3}^{\vee}}\bm{\Pi}^{(k)}  \right)\\
&+\sum_{j>i}\left(\bfP_{\four}\left(\frac{\bm{\nu}^{(j)}}{\bfP_{3}}\bm{\Pi}^{(i)\vee}+\frac{\bm{\nu}^{(i)\vee}}{\bfP_{3}^{\vee}}\bm{\Pi}^{(j)}\right)+\bfP_{\four}\mathbf{\Pi}^{(i)\vee}\mathbf{\Pi}^{(j)}\right)\\
\eea
where we used
\bea
-\sum_{j>i}\left(\bfP_{4}^{\vee}x_{i}^{-1}\bm{\Pi}^{(j)}+\bfP_{4}^{\vee}x_{j}^{-1}\bm{\Pi}^{(i)}\right)-\bfP_{4}^{\vee}\sum_{k=1}^{\infty}x_{k}^{-1}\bm{\Pi}^{(k)}&=-x^{-1}\sum_{i=1}^{\infty}\bm{\Pi}^{(i)}.
\eea
One can show by direct computation that this total character itself is movable \cite{Monavari:2022rtf, Nekrasov:2023nai}.

The nontrivial sign factors come from the second line:
\bea
\mathbb{I}\left[\mathbf{v}_{\D6\rightarrow \D8}^{\emptyset\emptyset\pi_{3}\emptyset}\right]&=\mathbb{I}\left[\cdots+\sum_{j>i}\bfP_{\four}\frac{\bm{\nu}^{(j)}}{\bfP_{3}}\bm{\Pi}^{(i)\vee}+\cdots +\sum_{j>i}\bfP_{\four}\bm{\Pi}^{(i)\vee}\bm{\Pi}^{(j)}\right]\\
&=(-1)^{\widetilde{\sigma}_{4}(\rho)}\mathbb{I}\left[\cdots+\sum_{j>i}\bfP_{\four}\frac{\bm{\nu}^{(j)\vee}}{\bfP^{\vee}_{3}}\bm{\Pi}^{(i)}+\cdots +\sum_{j>i}\bfP_{123}^{\vee}\left(\bm{\Pi}^{(i)\vee}\bm{\Pi}^{(j)}+\bm{\Pi}^{(i)}\bm{\Pi}^{(j)\vee}\right)\right]
\eea
where extra sign contributions appear due to Prop.~\ref{app-prop:reflection_sign}, \ref{app-prop:reflection-mod}. The signs are defined as
\bea
\widetilde{\sigma}_{4}(\rho)=\left[\sum_{j>i}\left(\bfP_{\four}\frac{\bm{\nu}^{(j)}}{\bfP_{3}}\bm{\Pi}^{(i)\vee}+\bfP_{123}\bm{\Pi}^{(i)\vee}\bm{\Pi}^{(j)}\right)\right]^{(0)}.
\eea
Using
\bea
\sum_{j>i}\bfP_{\four}\left(\frac{\bm{\nu}^{(j)\vee}}{\bfP_{3}^{\vee}}\bm{\Pi}^{(i)}+\frac{\bm{\nu}^{(i)\vee}}{\bfP_{3}^{\vee}}\bm{\Pi}^{(j)}\right)+\sum_{k=1}^{\infty}\bfP_{\four}\frac{\bm{\nu}^{(k)\vee}}{\bfP_{3}^{\vee}}\bm{\Pi}^{(k)}&=\bfP_{\four}\frac{\sum_{i}\bm{\nu}^{(i)\vee}}{\bfP_{3}^{\vee}}\sum_{j}\bm{\Pi}^{(j)},
\eea
we then have
\bea
\mathbb{I}\left[\mathbf{v}_{\D6\rightarrow \D8}^{\emptyset\emptyset\pi_{3}\emptyset}\right]&=(-1)^{\widetilde{\sigma}_{4}(\rho)}\, \mathbb{I}\left[-(1-K^{-1})x^{-1}\sum_{i}\bm{\Pi}^{(i)}+\bfP_{\four}\frac{\sum_{i}\bm{\nu}^{(i)\vee}}{\bfP_{3}^{\vee}}\sum_{j}\bm{\Pi}^{(j)}+\bfP_{123}^{\vee}\left(\sum_{i}\bm{\Pi}^{(i)}\right)^{\vee}\left(\sum_{j}\bm{\Pi}^{(j)}\right)\right]\\
   =&(-1)^{\widetilde{\sigma}_{4}(\rho)}\mathcal{Z}^{\D8}_{\four;4;\emptyset\emptyset\pi_{3}\emptyset}[\rho,K].
\eea
The vertex operator part comes from 
\bea
{:\prod_{k=1}^{\infty}\Lambda^{K,\,\emptyset\emptyset\nu^{(k)}}_{\bar{4},\Pi^{(k)}}(q_{4}^{k-1}x):}={:\frac{\mathsf{Z}(K,x)}{\prod_{\scube\in\pi_{3}}\mathsf{S}_{3}(\chi_{\bar{3},x}(\cube))}\prod_{\shcube\in\rho}\mathsf{A}^{-1}(\chi_{\four,x}(\hcube)):}
\eea
and we obtain the claim.
\end{proof}

\begin{conjecture}
    Up to sign factors coming from the boundary contributions, we have 
    \bea
\widetilde{\sigma}_{4}(\rho)=\left[\sum_{j>i}\left(\bfP_{\four}\frac{\bm{\nu}^{(j)}}{\bfP_{3}}\bm{\Pi}^{(i)\vee}+\bfP_{123}\bm{\Pi}^{(i)\vee}\bm{\Pi}^{(j)}\right)\right]^{(0)}\simeq \sigma_{4}(\rho)\quad \mod2
    \eea
    where $\sigma_{4}(\rho)=\#\{(i,i,i,j)\in\rho\mid i<j\}$. 
\end{conjecture}

\paragraph{One-leg D8 $qq$-character: part 2}
Let us next consider the situation when the infinite size solid partition $\tilde{\rho}$ has a boundary plane partition $\pi_{4}$ in the 4-direction (see \eqref{eq:fig-solidpartitionleg4}). In the $(1,3)$ decomposition, we have a finite size plane partition\footnote{Note that this case, each plane partition is a \textit{finite} plane partition, but not \textit{infinite} plane partition.} for each layer $\tilde{\rho}=(\widetilde{\Pi}^{(1)},\widetilde{\Pi}^{(2)},\ldots)$ where we can decompose it into two parts $\Pi^{(k)}, \pi_{4}$. The contribution for each layer is
\bea
\text{operator part:}&\quad \Lambda^{K,\emptyset\emptyset\emptyset}_{\bar{4},\widetilde{\Pi}^{(k)}}(q_{4}^{k-1}x)={:\frac{\mathsf{W}_{\bar{4}}(q_{4}^{k-1}x)}{\mathsf{W}_{\bar{4}}(Kq_{4}^{k-1}x)}\prod_{\scube\in\pi_{4}}\mathsf{A}^{-1}(\chi_{\bar{4},q_{4}^{k-1}x}(\cube))\prod_{\scube\in\Pi^{(k)}}\mathsf{A}^{-1}(\chi_{\bar{4},q_{4}^{k-1}x}(\cube)):},\\
\text{coefficient part:}&\quad \widetilde{\mathcal{Z}}^{\D6}_{\bar{4};\emptyset\emptyset\emptyset}[K,\widetilde{\Pi}^{(k)}].
\eea

\begin{proposition}\label{prop:D6toD8-oneleg2}
    The coefficient part and vertex operator part of the D8 $qq$-character with leg boundary condition at the 4-direction comes from the following infinite product
    \bea
    \prod_{k=1}^{\infty}\widetilde{\mathcal{Z}}^{\D6}_{\bar{4};\emptyset\emptyset\emptyset}[K,\widetilde{\Pi}^{(k)}]\overleftarrow{\prod_{k=1}^{\infty}}\Lambda^{K,\emptyset\emptyset\emptyset}_{\bar{4},\widetilde{\Pi}^{(k)}}(q_{4}^{k-1}x)\simeq (-1)^{\widetilde{\sigma}_{4}(\rho)}\mathcal{Z}^{\D8}_{\four;4;\emptyset\emptyset\emptyset\pi_{4}}[\rho,K]\Lambda_{\four;\rho}^{K;\emptyset\emptyset\emptyset\pi_{4}}(x)
    \eea
    where the equality is understood up to one-loop perturbative and boundary contributions and $\widetilde{\sigma}_{4}(\rho)$ is some sign factor.
\end{proposition}
\begin{proof}
    The contractions of the vertex operator part and the coefficient part are
    \bea
\frac{\Lambda^{K,\,\emptyset\emptyset\emptyset}_{\bar{4},\widetilde{\Pi}^{(j)}}(x_{j})\Lambda^{K,\,\emptyset\emptyset\emptyset}_{\bar{4},\widetilde{\Pi}^{(i)}}(x_{i})}{:\Lambda^{K,\,\emptyset\emptyset\emptyset}_{\bar{4},\widetilde{\Pi}^{(j)}}(x_{j})\Lambda^{K,\,\emptyset\emptyset\emptyset}_{\bar{4},\widetilde{\Pi}^{(i)}}(x_{i}):}&=\mathbb{I}\left[-\bfP_{4}^{\vee}(1-K^{-1})x_{i}^{-1}(\bm{\Pi}^{(j)}+q_{4}^{j-1}\bm{\pi}_{4})-\bfP_{4}(1-K)x_{j}(\bm{\Pi}^{(i)}+q_{4}^{i-1}\bm{\pi}_{4})^{\vee}\right.\\
&\qquad \left.  +\bfP_{\four}(\bm{\Pi}^{(i)}+q_{4}^{i-1}\bm{\pi}_{4})^{\vee}(\bm{\Pi}^{(j)}+q_{4}^{j-1}\bm{\pi}_{4}) \right] ,\\
\widetilde{\mathcal{Z}}^{\D6}_{\bar{4};\emptyset\emptyset\emptyset}[K,\widetilde{\Pi}^{(k)}]&=\mathbb{I}\left[ -\bfP_{4}^{\vee}(1-K^{-1})x_{k}^{-1}(\bm{\Pi}^{(k)}+q_{4}^{k-1}\bm{\pi}_{4}) +\bfP_{123}^{\vee} (\bm{\Pi}^{(k)}+q_{4}^{k-1}\bm{\pi}_{4})^{\vee}(\bm{\Pi}^{(k)}+q_{4}^{k-1}\bm{\pi}_{4})\right],
    \eea
    where $x_{i}=q_{4}^{i-1}x$. Omitting the boundary contributions and the one loop contributions such as $\bm{\pi}_{4}^{\vee}\bm{\pi}_{4},x_{i}^{-1}\bm{\pi}_{4}$, and using the reflection property of the index, we eventually obtain
    \bea
   \mathbb{I}\left[ -(1-K^{-1})x^{-1}\sum_{i}\bm{\Pi}^{(i)}+\bfP_{\four}\frac{\bm{\pi}_{4}^{\vee}}{\bfP_{4}^{\vee}}\sum_{i}\bm{\Pi}^{(i)}+\bfP_{123}^{\vee}\left(\sum_{i}\bm{\Pi}^{(i)}\right)^{\vee}\left(\sum_{j}\bm{\Pi}^{(j)}\right)\right]=\mathcal{Z}^{\D8}_{\four;4;\emptyset\emptyset\emptyset\pi_{4}}[\rho,K].
    \eea
The extra sign factor $\widetilde{\sigma}_{4}(\rho)$ comes from
\bea
\widetilde{\sigma}_{4}(\rho)=\left[\sum_{j>i}\left(\bfP_{\four}q_{4}^{j-1}\bm{\pi}_{4}\bm{\Pi}^{(i)\vee}+\bfP_{123}\bm{\Pi}^{(i)\vee}\bm{\Pi}^{(j)}\right)\right]^{(0)}.
\eea
The vertex operator part comes from
    \bea
    :\prod_{k=1}^{\infty}\Lambda^{K,\emptyset\emptyset\emptyset}_{\bar{4},\widetilde{\Pi}^{(k)}}(q_{4}^{k-1}x):=:\frac{\mathsf{Z}(K,x)}{\prod_{\scube\in\pi_{4}}\mathsf{S}_{4}(\chi_{\bar{4},x}(\cube))}\prod_{\shcube\in\rho}\mathsf{A}^{-1}(\chi_{\four,x}(\hcube)):.
    \eea
\end{proof}

\begin{conjecture}
    Up to contributions coming from the boundaries, we have
    \bea
    \widetilde{\sigma}_{4}(\rho)\simeq \sigma_{4}(\rho)\quad \mod 2
    \eea
\end{conjecture}

\begin{remark}
    The equality of the sign factors above is expected to be true up to extra sign factors appearing from boundary contributions which are non-essential for DT vertex computations. Technically, deriving the combinatorial formula from the above expression is difficult because we need to deal with infinite number of boxes and regularize them properly (see Appendix~\ref{app-sec:boundary-signrule} for some related formulas). Moreover, the sign factors coming from the boundary contributions are also essential when we want to study Donaldson--Thomas partition functions of toric Calabi--Yau 4-folds and thus we can not simply throw them away. However, we still expect that the procedure to derive the sign factors introduced in section~\ref{sec:D8qqsignrule} gives the sign factors in a systematical way once the concept of \textit{gluings} of $qq$-characters is established. We hope to come back to this problem in a near future.  
\end{remark}

\subsection{Surface boundary conditions}\label{sec:D8qqsurfaceboundary}
The D8 $qq$-character with nontrivial surface boundary conditions is defined as follows.
\begin{definition}
Let $\lambda_{A},\,(A\in\six)$ be finite Young diagrams and $\mathcal{SP}_{\{\lambda_{A}\}_{A\in\six}}$ be the set of possible solid plane partitions with the boundary Young diagrams $\lambda_{A\in\six}$. The D8 $qq$-character with surface boundary conditions is defined as 
    \bea
    \mathsf{T}_{\four,\{\lambda_{A}\}}^{K}(x)=\sum_{\rho\in\mathcal{SP}_{\{\lambda_{A}\}}}\mathfrak{q}^{|\rho|}(-1)^{\sigma_{4}(\rho)}\mathcal{Z}^{\D8}_{\four;4;\{\lambda_{A}\}}[\rho,K]\Lambda_{\four,\rho}^{K;\{\lambda_{A}\}}(x),
    \eea
    where
    \bea
    \Lambda_{\four,\rho}^{K;\{\lambda_{A}\}}(x)={:\frac{\mathsf{Z}(K,x)}{\prod_{A\in\six}\prod_{\Abox\in\lambda_{A}}\mathsf{X}_{A}(\chi_{\bar{A},x}(\Bbox))}\prod_{\shcube\in\mathcal{S}_{\{\lambda_{A}\}}}\mathsf{A}(\chi_{\four,x}(\hcube))\prod_{\shcube\in\rho}\mathsf{A}^{-1}(\chi_{\four,x}(\hcube)):}
    \eea
    and $\mathcal{Z}^{\D8}_{\four;4;\{\lambda_{A}\}}[\rho,K]$, $\sigma_{4}(\rho)$ are defined in \eqref{eq:D8surfacepartition} and \eqref{eq:D8boundary_signrule}, respectively.
\end{definition}

Similar to the leg boundary conditions, we derive the $qq$-characters with surface boundary conditions by taking infinite products of the D6 $qq$-characters. Similar to the leg boundary conditions, the $(1,3)$-type decomposition breaks the symmetry between the six elements $A\in\six$ and only triality symmetry remains. We will explicitly derive the D8 $qq$-characters when there is only one surface boundary condition.
\paragraph{One-face D8 $qq$-character: part 1}Let us derive the one-face D8 $qq$-character by taking the infinite products of lower dimensional $qq$-characters. Let $\widetilde{\rho}$ denote the solid partition with an asymptotic Young diagram $\lambda_{34}$ extending semi-infinitely in the 34-surface as in \eqref{eq:fig-solidpartitionsurface1}. Under the $(1,3)$ decomposition, the solid partition $\widetilde{\rho}$ and the surface boundary Young diagram will be decomposed as 
\bea
\widetilde{\rho}=(\widetilde{\Pi}^{(1)},\widetilde{\Pi}^{(2)},\ldots),\qquad (\lambda_{34},\lambda_{34},\ldots).
\eea
Namely, for each layer, we will have an infinite size plane partition $\widetilde{\Pi}^{(k)}$ with one leg $\lambda_{34}$. Similar to the previous cases, we denote the boxes not included in the leg boundary as $\Pi^{(k)}$. The D6 $qq$-character for each layer comes from
\bea
\text{operator part:}&\qquad \Lambda^{K,\,\emptyset\emptyset\lambda_{34}}_{\bar{4},\Pi^{(k)}}(q_{4}^{k-1}x),\\
\text{coefficient part:}&\qquad \widetilde{\mathcal{Z}}_{\bar{4};\emptyset\emptyset\lambda_{34}}^{\D6}[\,K,\Pi^{(k)}\,].
\eea
The difference with the situation in Prop.~\ref{prop:D6toD8-oneleg1} is that now the non-increasing sequence of Young diagrams are all the same $\nu^{(k)}=\lambda_{34}$.
\begin{proposition}\label{prop:D6toD8-oneface1}
    The coefficient part and vertex operator part of the D8 $qq$-character with surface boundary condition at the 34-surface comes from the following infinite product
    \bea
    \prod_{k=1}^{\infty}\widetilde{\mathcal{Z}}^{\D6}_{\bar{4};\emptyset\emptyset\lambda_{34}}[K,\Pi^{(k)}]\overleftarrow{\prod_{k=1}^{\infty}}\Lambda^{K,\emptyset\emptyset\lambda_{34}}_{\bar{4},\Pi^{(k)}}(q_{4}^{k-1}x)\simeq (-1)^{\widetilde{\sigma}_{4}(\rho)} \mathcal{Z}^{\D8}_{\four;4;\{\lambda_{34}\}}[\rho,K]\Lambda_{\four;\rho}^{K;\{\lambda_{34}\}}(x)
    \eea
    where the equality is understood up to one-loop perturbative and boundary contributions and $\widetilde{\sigma}_{4}(\rho)$ is some sign factor.
\end{proposition}
\begin{proof}
    The derivation is similar to Prop.~\ref{prop:D6toD8-oneleg1}. The character $\bm{\nu}^{(k)}$ will be now modified to 
    \bea
    \bm{\nu}^{(k)}=q_{4}^{k-1}\bm{\lambda}_{34},\quad \bm{\lambda}_{34}=\sum_{\Abox\in\lambda_{34}}\chi_{12,x}(\Bbox).
    \eea
    We first have
    \bea
\prod_{k=1}^{\infty}\widetilde{\mathcal{Z}}_{\bar{4};\emptyset\emptyset\lambda_{34}}^{\D6}[\,K,\Pi^{(k)}\,]=\mathbb{I}\left[\sum_{k=1}^{\infty}\left( -(1-K^{-1})x_{k}^{-1}\bm{\Pi}^{(k)}+\bfP_{\four}\frac{q_{4}^{-k+1}\bm{\lambda}_{34}^{\vee}}{\bfP_{3}^{\vee}}\bm{\Pi}^{(k)}+\bfP_{123}^{\vee}\bm{\Pi}^{(k)\vee}\bm{\Pi}^{(k)}  \right)\right],
\eea
and the contraction of the vertex operator operators gives
\bea
\left\langle\overleftarrow{\prod_{k=1}^{\infty}}\Lambda^{K,\emptyset\emptyset\lambda_{34}}_{\bar{4},\Pi^{(k)}}(q_{4}^{k-1}x)\right\rangle &\simeq \mathbb{I}\left[\sum_{j>i} \left(-\bfP_{4}^{\vee}(1-K^{-1})x_{i}^{-1}\bm{\Pi}^{(j)} -\bfP_{4}^{\vee}(1-K^{-1})x_{j}^{-1}\bm{\Pi}^{(i)}    \right) \right.\\
&\left.+ \sum_{j>i}\left(\bfP_{\four}\left(\frac{\bm{\lambda}_{34}q_{4}^{j-1}}{\bfP_{3}}\right)\bm{\Pi}^{(i)\vee} +\bfP_{\four}\left(\frac{q_{4}^{i-1}\bm{\lambda}_{34}}{\bfP_{3}}\right)^{\vee}\bm{\Pi}^{(j)} +\bfP_{\four}\bm{\Pi}^{(i)\vee}\bm{\Pi}^{(j)} \right)  \right],
\eea
where $x_{i}=q_{4}^{i-1}x$ and contributions from the one-loop part and boundaries are excluded.

We define the character $\mathbf{v}_{\D6\rightarrow \D8}^{\lambda_{34}}$ as
\bea
\mathbf{v}_{\D6\rightarrow \D8}^{\lambda_{34}}&=-(1-K^{-1})x^{-1}\sum_{i}\bm{\Pi}^{(i)}+\sum_{k}\left(\bfP_{\four}\frac{q_{4}^{-k+1}\bm{\lambda}_{34}^{\vee}}{\bfP_{3}^{\vee}}\bm{\Pi}^{(k)}+\bfP_{123}^{\vee}\bm{\Pi}^{(k)\vee}\bm{\Pi}^{(k)}
\right)\\
&+\sum_{j>i}\left(\bfP_{\four}\left(\frac{\bm{\lambda}_{34}q_{4}^{j-1}}{\bfP_{3}}\right)\bm{\Pi}^{(i)\vee} +\bfP_{\four}\left(\frac{q_{4}^{i-1}\bm{\lambda}_{34}}{\bfP_{3}}\right)^{\vee}\bm{\Pi}^{(j)} +\bfP_{\four}\bm{\Pi}^{(i)\vee}\bm{\Pi}^{(j)} \right)
\eea
and then using the reflection property Prop.~\ref{app-prop:reflection_sign} and \ref{app-prop:reflection-mod}, we have
\bea
\mathbb{I}\left[ 
\mathbf{v}_{\D6\rightarrow \D8}^{\lambda_{34}}\right]&=(-1)^{\widetilde{\sigma}_{4}(\rho)}\mathbb{I}\left[ -(1-K^{-1})x^{-1}\sum_{i}\bm{\Pi}^{(i)}+\bfP_{\four}\frac{\bm{\lambda}_{34}^{\vee}}{\bfP_{34}^{\vee}}\sum_{i}\bm{\Pi}^{(i)}+\bfP_{123}^{\vee}\left(\sum_{i}\bm{\Pi}^{(i)}\right)^{\vee}\left(\sum_{j}\bm{\Pi}^{(j)}\right)  \right]\\
&=(-1)^{\widetilde{\sigma}_{4}(\rho)}\mathcal{Z}^{\D8}_{\four;4;\{\lambda_{34}\}}[\rho,K]
\eea
where
\bea
\widetilde{\sigma}_{4}(\rho)=\left[\sum_{j>i}\left(\bfP_{\four}\left(\frac{\bm{\lambda}_{34}q_{4}^{j-1}}{\bfP_{3}}\right)\bm{\Pi}^{(i)\vee}+\bfP_{123}\bm{\Pi}^{(i)\vee}\bm{\Pi}^{(j)}\right)\right]^{(0)}.
\eea
The vertex operator part comes from
\bea
{:\prod_{k=1}^{\infty}\Lambda^{K,\emptyset\emptyset\lambda_{34}}_{\bar{4},\Pi^{(k)}}(q_{4}^{k-1}x):}={:\frac{\mathsf{Z}(K,x)}{\prod_{\Abox\in\lambda_{34}}\mathsf{X}_{34}(\chi_{12,x}(\Bbox))}\prod_{\shcube\in\rho}\mathsf{A}^{-1}(\chi_{\four,x}(\hcube)):}.
\eea
\end{proof}

\begin{conjecture}
    Up to sign factors coming from the boundary contributions, we have
    \bea
    \widetilde{\sigma}_{4}(\rho)\simeq \sigma_{4}(\rho)\quad \mod2.
    \eea
\end{conjecture}

\paragraph{One-face D8 $qq$-character: part 2}
Let us next consider the case when the surface boundary spans the $12$-surface as in \eqref{eq:fig-solidpartitionsurface2}. For this situation, in the $(1,3)$ decomposition, each layer will have a finite plane partition $(\Pi^{(1)},\Pi^{(2)},\ldots)$ whose origin is shifted in the 3-direction. The asymptotic Young diagram $\lambda_{12}$ in the surface 12 will be decomposed into non-increasing 1d partitions as
\bea
\lambda_{12}=\{k_{12}^{(1)},k_{12}^{(2)},\ldots\},\quad k_{12}^{(i)}\geq k_{12}^{(i+1)}.
\eea
The origin of the plane partition for layer $k$ will be $xq_{4}^{k-1}q_{3}^{k_{12}^{(k)}}$ in the multiplicative language. The following contributes to the D8 $qq$-character:
\bea
\text{operator part:}&\qquad \Lambda^{K,\,\emptyset\emptyset\emptyset}_{\bar{4},\Pi^{(k)}}(q_{3}^{k_{12}^{(k)}}q_{4}^{k-1}x),\\
\text{coefficient part:}&\qquad \widetilde{\mathcal{Z}}_{\bar{4};\emptyset\emptyset\emptyset}^{\D6}[\,Kq_{3}^{-k_{12}^{(k)}},\Pi^{(k)}\,].
\eea
\begin{proposition}\label{prop:D6toD8-oneface2}
    The coefficient part and vertex operator part of the D8 $qq$-character with surface boundary condition at the 34-surface comes from the following infinite product
    \bea
    \prod_{k=1}^{\infty}\widetilde{\mathcal{Z}}^{\D6}_{\bar{4};\emptyset\emptyset\emptyset}[q_{3}^{-k_{12}^{(k)}}K,\Pi^{(k)}]\overleftarrow{\prod_{k=1}^{\infty}}\Lambda^{K,\emptyset\emptyset\emptyset}_{\bar{4},\Pi^{(k)}}(q_{4}^{k-1}q_{3}^{k_{12}^{(k)}}x)\simeq (-1)^{\widetilde{\sigma}_{4}(\rho)} \mathcal{Z}^{\D8}_{\four;4;\{\lambda_{12}\}}[\rho,K]\Lambda_{\four;\rho}^{K;\{\lambda_{12}\}}(x)
    \eea
    where the equality is understood up to one-loop perturbative and boundary contributions and $\widetilde{\sigma}_{4}(\rho)$ is some sign factor.
\end{proposition}
\begin{proof}
    We first introduce the following characters
    \bea
    \bm{\Pi}^{(i)}=\sum_{\scube\in\Pi^{(i)}}\chi_{\bar{4},x_{i}q_{3}^{k_{12}^{(i)}}}(\cube),\quad \bm{\lambda}_{12}=\sum_{i=1}^{\infty}\sum_{j=1}^{k_{12}^{(i)}}xq_{4}^{i-1}q_{3}^{j-1},\quad x_{i}=xq_{4}^{i-1}
    \eea
    where note that the first character is different from the previous case since the origin of the plane partition is shifted. We then have
    \bea
     \prod_{k=1}^{\infty}\widetilde{\mathcal{Z}}^{\D6}_{\bar{4};\emptyset\emptyset\emptyset}[K,\Pi^{(k)}]&=\mathbb{I}\left[\sum_{k=1}^{\infty}\left(-(1-q_{3}^{k_{12}^{(k)}}K^{-1})(x_{k}q_{3}^{k_{12}^{(k)}})^{-1}\bm{\Pi}^{(k)}+\bfP_{123}^{\vee}\bm{\Pi}^{(k)\vee}\bm{\Pi}^{(k)}\right)\right],\\
     \left\langle \overleftarrow{\prod_{k=1}^{\infty}}\Lambda^{K,\emptyset\emptyset\emptyset}_{\bar{4},\Pi^{(k)}}(q_{4}^{k-1}q_{3}^{k_{12}^{(k)}}x) \right\rangle &\simeq \mathbb{I}\left[ \sum_{j>i}\left(-\bfP_{4}^{\vee}(1-q_{3}^{k_{12}^{(k)}}K^{-1})x_{i}^{-1}q_{3}^{-k_{12}^{(i)}}\bm{\Pi}^{(j)}\right)\right.\\
     &\qquad \left.\sum_{j>i}\left(-\bfP_{4}(1-Kq_{3}^{-k_{12}^{(k)}})x_{j}q_{3}^{k_{12}^{(j)}}\bm{\Pi}^{(i)\vee}+\bfP_{\four}\bm{\Pi}^{(i)\vee}\bm{\Pi}^{(j)}  \right) \right].
    \eea
We define
\bea
\mathbf{v}_{\D6\rightarrow \D8}^{\lambda_{12}}&=\sum_{k=1}^{\infty}\left(-(1-q_{3}^{k_{12}^{(k)}}K^{-1})(x_{k}q_{3}^{k_{12}^{(k)}})^{-1}\bm{\Pi}^{(k)}+\bfP_{123}^{\vee}\bm{\Pi}^{(k)\vee}\bm{\Pi}^{(k)}\right)\\
&+\sum_{j>i}\left(-\bfP_{4}^{\vee}(1-q_{3}^{k_{12}^{(i)}}K^{-1})x_{i}^{-1}q_{3}^{-k_{12}^{(i)}}\bm{\Pi}^{(j)}-\bfP_{4}(1-Kq_{3}^{-k_{12}^{(j)}})x_{j}q_{3}^{k_{12}^{(j)}}\bm{\Pi}^{(i)\vee}+\bfP_{\four}\bm{\Pi}^{(i)\vee}\bm{\Pi}^{(j)}  \right).
\eea
Using the reflection property in Prop.~\ref{app-prop:reflection_sign} and \ref{app-prop:reflection-mod}, we have
\bea
\mathbb{I}\left[\mathbf{v}_{\D6\rightarrow \D8}^{\lambda_{12}}\right]&=(-1)^{\widetilde{\sigma}_{4}(\rho)}\mathbb{I}\left[-(1-K^{-1})x^{-1}\sum_{i}\bm{\Pi}^{(i)}+\bfP_{\four}\frac{\bm{\lambda}_{12}^{\vee}}{\bfP_{12}^{\vee}}\sum_{i}\bm{\Pi}^{(i)}+\bfP_{123}^{\vee}\left(\sum_{i}\bm{\Pi}^{(i)}\right)^{\vee}\left(\sum_{j}\bm{\Pi}^{(j)}\right)\right]\\
&=(-1)^{\widetilde{\sigma}_{4}(\rho)}\mathcal{Z}^{\D8}_{\four;4;\{\lambda_{12}\}}[\rho,K]
\eea
where we used
\bea
\bfP_{4}\sum_{i=1}xq_{4}^{i-1}q_{3}^{k_{12}^{(i)}}=x-\bfP_{34}\bm{\lambda}_{12}
\eea
and defined the sign factor
\bea
\widetilde{\sigma}_{4}(\rho)=\left[\sum_{j>i}\left(-\bfP_{4}x_{j}q_{3}^{k_{12}^{(j)}}\bm{\Pi}^{(i)\vee} +\bfP_{123}\bm{\Pi}^{(i)\vee} \bm{\Pi}^{(j)} \right)\right]^{(0)}.
\eea
The vertex operator part is given as
\bea
{:\prod_{k=1}^{\infty}\Lambda^{K,\emptyset\emptyset\emptyset}_{\bar{4},\Pi^{(k)}}(q_{4}^{k-1}q_{3}^{k_{12}^{(k)}}x):}={:\frac{\mathsf{Z}(K,x)}{\prod_{\Abox\in\lambda_{12}}\mathsf{X}_{12}(\chi_{34,x}(\Bbox))}\prod_{\shcube\in\rho}\mathsf{A}^{-1}(\chi_{\four,x}(\hcube)):}.
\eea
\end{proof}

\begin{conjecture}
    Up to boundary contributions, we have
    \bea
    \widetilde{\sigma}_{4}(\rho)\simeq \sigma_{4}(\rho)\quad \mod2.
    \eea
\end{conjecture}

\subsection{Hypersurface boundary conditions}\label{sec:D8qqhypersurfaceboundary}
The D8 $qq$-character with hypersurface boundary conditions is defined as follows.
\begin{definition}
    Let $k_{\bar{1},\bar{2},\bar{3},\bar{4}}\in\mathbb{Z}_{\geq 0}$ be 1d partitions and $\mathcal{SP}_{\bar{1}\bar{2}\bar{3}\bar{4}}=\mathcal{SP}$ be the set of possible solid partitions with the boundary 1d partitions. The D8 $qq$-character with these boundary conditions is defined as
    \bea
    \mathsf{T}^{K}_{\four,k_{\bar{1}}k_{\bar{2}}k_{\bar{3}}k_{\bar{4}}}(x)&=\sum_{\rho\in\mathcal{SP}}\mathfrak{q}^{|\rho|}(-1)^{\sigma_{4}(\rho)}\mathcal{Z}^{\D8}_{\four;4}[\rho, \overline{K}]\Lambda_{\four,\rho}^{\overline{K}}(q_{1}^{k_{\bar{1}}}q_{2}^{k_{\bar{2}}}q_{3}^{k_{\bar{3}}}q_{4}^{k_{\bar{4}}}x)\\
    &=\mathsf{T}^{\overline{K}}_{\four}(q_{1}^{k_{\bar{1}}}q_{2}^{k_{\bar{2}}}q_{3}^{k_{\bar{3}}}q_{4}^{k_{\bar{4}}}x)
    \eea
    where $\overline{K}=q_{1}^{-k_{\bar{1}}}q_{2}^{-k_{\bar{2}}}q_{3}^{-k_{\bar{3}}}q_{4}^{-k_{\bar{4}}}K$. Namely, it is the normal D8 $qq$-character with the parameters $K,x$ shifted to $\overline{K}, q_{1}^{k_{\bar{1}}}q_{2}^{k_{\bar{2}}}q_{3}^{k_{\bar{3}}}q_{4}^{k_{\bar{4}}}x$.
\end{definition}

Let us derive this by using the infinite products of the D6 vertex operators. For simplicity, we consider the case when there are only 1d partitions $k_{234,134,124}$ in the boundaries. To obtain the case for the 1d partition $k_{123}$, we just need to shift the origin to $q_{4}^{k_{\bar{4}}}x$ so it is not difficult. Let $\widetilde{\rho}$ be the solid partition with the hypersurface boundary conditions and let $\widetilde{\Pi}^{(i)}$ be its decomposition, i.e. infinite size plane partitions. We denote the finite part of the solid partition and its decomposition as $\rho,\Pi^{(i)}$, respectively. $\Pi^{(i)}$ are non-increasing finite plane partitions obeying $\Pi^{(i)}\succeq \Pi^{(i+1)}$ where the origin is shifted by $(k_{\bar{1}},k_{\bar{2}},k_{\bar{3}})$. Using the $(1,3)$ decomposition in \eqref{eq:fig-solidpartitionhypersurface123}, the D6 vertex operators for each layer, will be
\bea
\text{operator part:}&\quad  \Lambda^{K,\emptyset\emptyset\emptyset}_{\bar{4},\Pi^{(i)}}(q_{1}^{k_{\bar{1}}}q_{2}^{k_{\bar{2}}}q_{3}^{k_{\bar{3}}}q_{4}^{i-1}x)   \\
\text{coefficient part:}&\quad  \widetilde{\mathcal{Z}}_{\bar{4};\emptyset\emptyset\emptyset}^{\D6}[q_{1}^{-k_{\bar{1}}}q_{2}^{-k_{\bar{2}}}q_{3}^{-k_{\bar{3}}}K,\Pi^{(i)}].
\eea

\begin{proposition}
    The coefficient part and vertex operator part of the D8 $qq$-character with hypersurface boundary conditions $k_{\bar{1},\bar{2},\bar{3}}$ comes from the following infinite product
    \bea
    \prod_{i=1}^{\infty}\widetilde{\mathcal{Z}}_{\bar{4};\emptyset\emptyset\emptyset}^{\D6}[q_{1}^{-k_{\bar{1}}}q_{2}^{-k_{\bar{2}}}q_{3}^{-k_{\bar{3}}}K,\Pi^{(i)}]\overleftarrow{\prod_{i=1}^{\infty}}\Lambda^{K,\emptyset\emptyset\emptyset}_{\bar{4},\Pi^{(i)}}(q_{1}^{k_{\bar{1}}}q_{2}^{k_{\bar{2}}}q_{3}^{k_{\bar{3}}}q_{4}^{i-1}x)\simeq (-1)^{\sigma_{4}(\rho)}\mathcal{Z}^{\D8}_{\four;4}[\rho, \overline{K}]\Lambda_{\four,\rho}^{\overline{K}}(q_{1}^{k_{\bar{1}}}q_{2}^{k_{\bar{2}}}q_{3}^{k_{\bar{3}}}x)
    \eea
    where the one-loop perturbative part is omitted and $\overline{K}=q_{1}^{-k_{\bar{1}}}q_{2}^{-k_{\bar{2}}}q_{3}^{-k_{\bar{3}}}K$. Note also that the sign factor $\sigma_{4}(\rho)$ is exact this time and we do not need to introduce $\widetilde{\sigma}_{4}(\rho)$.
\end{proposition}
\begin{proof}
    The proof is essentially the same with Prop.~\ref{prop:D6toD8vacuumproof} and \ref{prop:sign_proof}. The vertex operator part is trivial. The coefficient of the left hand side comes from the following character
    \bea
    \mathbf{v}^{k_{\bar{1}}k_{\bar{2}}k_{\bar{3}}}_{\D6\rightarrow \D8}&=-(1-\overline{K}^{-1})\sum_{j}\bm{\Pi}^{(j)}+\sum_{k=1}^{\infty}\bfP_{123}^{\vee} \bm{\Pi}^{(k)\vee}\bm{\Pi}^{(k)}+\sum_{i<j}\bfP_{\four}\bm{\Pi}^{(i)\vee}\bm{\Pi}^{(j)}
    \eea
    where
    \bea
    \bm{\Pi}^{(i)}=\sum_{\scube\in\Pi^{(i)}}\chi_{\bar{4},\bar{x}_{i}}(\cube),\quad \bar{x}_{i}=q_{1}^{k_{\bar{1}}}q_{2}^{k_{\bar{2}}}q_{3}^{k_{\bar{3}}}q_{4}^{i-1}x.
    \eea
    Effectively $K$ of \eqref{eq:D8vacuum-proof1} and \eqref{eq:D8vacuum-proof2} is transformed to $\overline{K}$. Using
    \bea
    \left[\sum_{i<j}\bfP_{123}\bm{\Pi}^{(i)\vee}\bm{\Pi}^{(j)}\right]^{(0)}=\sigma_{4}(\rho)
    \eea
    we obtain
    \bea
    \mathbb{I}\left[\mathbf{v}^{k_{\bar{1}}k_{\bar{2}}k_{\bar{3}}}_{\D6\rightarrow \D8}\right]=(-1)^{\sigma_{4}(\rho)}\mathcal{Z}^{\D8}_{\four;4}[\rho, \overline{K}].
    \eea
    This is because the character $\bm{\Pi}^{(i)}$ differs from the character appearing in the proof of Prop.~\ref{prop:D6toD8vacuumproof} and \ref{prop:sign_proof} just by an overall constant factor $q_{1}^{k_{\bar{1}}}q_{2}^{k_{\bar{2}}}q_{3}^{k_{\bar{3}}}$. When considering $\bm{\Pi}^{(i)\vee}\bm{\Pi}^{(j)}$, they will not appear and thus the sign factor is just $\sigma_{4}(\rho)$. Therefore, we obtain the claim.
\end{proof}

\section{Conclusion}\label{sec:conclusion}
Following our previous paper \cite{Kimura:2023bxy}, we generalized our analysis to cases when we have multiple  $\D(2p)$ $(p=1,2,3)$-branes extending in the non-compact directions and introduced free field realizations of the contour integral formulas. These free field realizations lead to $qq$-characters which we call the Donaldson--Thomas $qq$-characters. Namely, we have shown the BPS/CFT correspondence of the partition function of D0-brane counting with fixed boundary conditions associated with D2, D4, D6-branes. Combinatorially, introducing D2, D4, D6-branes as boundary conditions correspond to adding one-dimensional rods, two-dimensional surfaces, three-dimensional hypersurfaces to the multi-dimensional partitions, respectively (see \eqref{eq:rod-surface-figure}). These boundaries should obey the multi-dimensional partitions condition by themselves and adding D0-branes correspond to adding boxes to the setup. 

Besides the DT $qq$-characters, we also revisited the D8 $qq$-character first introduced in \cite{Kimura:2023bxy} and gave a complete proof on the sign rules for the case when there are no boundaries, which was not discussed in \cite{Kimura:2023bxy}. We found that the fusion process of D6 $qq$-characters at the end will give the complete magnificent four partition function \textit{including} the sign rules. We also showed that the D6 and D8 $qq$-characters without boundaries all commute with each other, which is compatible with the fact that we have a plethystic exponential formula for the partition functions of them.

Let us list some possible directions we hope to address in near future.

\paragraph{Webs of BPS $qq$-characters and 4G network}

As mentioned in the main text, the DT $qq$-characters we introduced correspond to the operator versions of the equivariant DT vertex. To compute partition functions for toric Calabi--Yau 4-folds, we need to glue these DT vertices by introducing extra terms for \textit{edges} and \textit{faces}. From the physical viewpoint, we need to take the sum over not only the D0-branes but also the D2, D4, D6-branes at the boundaries. Namely, we will have possibly multiple D8-branes and the sum will be over all possible D0, D2, D4, D6-branes which eventually give the D8-D6-D4-D2-D0 partition function. This formalism is called the \textbf{4G network} \cite{Nekrasov:2023nai,Cao:2019tnw,Cao:2019tvv,Monavari:2022rtf,Bae:2022pif,Bae:2024bpx}. Believing in the BPS/CFT correspondence, we should have a $qq$-character reproducing this 4G network. The concept of this framework was dubbed as \textbf{webs of BPS $qq$-characters} in \cite{Kimura:2023bxy}, though not explicitly established yet. We will discuss this framework in our future work \cite{Kimura-Noshita1}.

\paragraph{Sign rules for 4G network}
The derivation of the $qq$-characters is different depending on the non-compact dimensions of the D-branes. For D6-branes (and lower dimensions), we have transverse directions and using them we can derive the corresponding $qq$-characters by using the commutativity with the screening charges. On the other hand, we do not have any transverse direction for the D8-case and thus we can not use the commutativity with the screening charge to derive them. In section~\ref{sec:fusionD6toD8}, we derived the D8 $qq$-characters and gave a proof for the sign rules by studying the infinite products of the D6 $qq$-characters. By studying the unmovable terms carefully, the infinite products automatically reproduces the sign rules in a natural way. We expect this is true even for the 4G network/webs of BPS $qq$-characters. Namely, we expect we can introduce screening charges of the network and derive the $qq$-characters associated with D6-branes. Moreover, the infinite products of these D6 $qq$-characters will automatically produce the D8 $qq$-character with correct sign rules and complete the discussion in~\cite{Nekrasov:2023nai}. In this process, we expect that the conjectures in section~\ref{sec:D8qqlegboundary} and \ref{sec:D8qqsurfaceboundary} will be proven.

\paragraph{PT $qq$-characters}
Deriving the PT $qq$-characters as mentioned in section~\ref{sec:general_qq} is also interesting. One strategy to study this is to start by studying the contour integral formula giving the PT invariants \cite{Cao:2023lon,Piazzalunga:2023qik,Kimura:2024xpr}. Finding the free field realizations of them is the starting point. The difficult part might be how to define the \textit{screening charge}, if it exists. If such kind of screening charge could be defined, one may obtain the PT $qq$-character. Studying relation of the PT $qq$-characters and the refined topological vertices is also another interesting topic.

\section*{Acknowledgements}
The work of TK was in part supported by CNRS through MITI interdisciplinary programs, EIPHI Graduate School (No.~ANR-17-EURE-0002) and Bourgogne-Franche-Comté region. 
GN is supported by JSPS KAKENHI Grant-in-Aid for JSPS fellows Grant No.~JP22J20944, JSR Fellowship, and FoPM (WINGS Program), the University of Tokyo. 
A part of this work was presented in \href{https://www.birs.ca/events/2024/5-day-workshops/24w5501}{Geometry and Physics of Quantum Toroidal Algebra} (see also the \href{https://www.birs.ca/events/2024/5-day-workshops/24w5501/videos}{website} for the \href{https://www.birs.ca/iasm-workshops/2024/24w5501/files/9.4-02%20Noshita%20Go.pdf}{slides} and the \href{https://www.birs.ca/events/2024/5-day-workshops/24w5501/videos/watch/202409041015-Go.html}{video}). GN is grateful to the organizers for the invitation and hospitality.

%%%%%%%%%%%%%%%%%%%%%%%%%%%%%%%%%%%%%%%%%%%%%%%%%%%%%%%%%%%%%%%%%%%%%%%%%%%%%%%%%%%%%%%%%%%%%%%%%%%%
\appendix
\section{Notations and special functions}\label{app:notations}
Let us summarize the notations and special functions we use in the main text. For details and motivations of the notations, see \cite[Sec.~3.1, App.~A,B]{Kimura:2023bxy}. 
\paragraph{Finite subsets} We introduce the following sets of non-negative integers:
\bea
\four=\{1,2,3,4\},\quad \six=\{12,13,14,23,24,34\},\quad \four^{\vee}=\{123,124,134,234\},
\eea
where $\six$ and $\four^{\vee}$ are the 2,3-element subsets of $\four$ respectively. We also denote the complement of $A\in\four,\six,\four^{\vee}$ under $\four$ as $\bar{A}$. For example, we have $\bar{A}=124$ for $A=3$. Using this map, it is obvious to see that $\four\simeq \four^{\vee}: a\leftrightarrow \bar{a}$

\paragraph{Special functions}The $q$-shifted factorial is 
\bea
(x;q)_{n}=\prod_{m=0}^{n-1}(1-xq^{m}),\quad (x;q)_{\infty}=\prod_{m=0}^{\infty}(1-xq^{m})=\exp\left(-\sum_{m=1}^{\infty}\frac{x^{m}}{m(1-q^{m})}\right)
\eea
for $|q|<1$. We have the analytic continuation
\bea
(x;q)_{\infty}=(xq^{-1};q^{-1})_{\infty}^{-1}.
\eea
We similarly denote the multiple infinite product by
\bea
(x;q_1,\ldots,q_m)_\infty=\prod_{0 \le n_1,\ldots,n_m \le \infty}(1 - x q_1^{n_1} \cdots q_m^{n_m}).
\eea
The theta function is defined as
\bea
\theta (x;p)=(x;p)_{\infty}(px^{-1};p)_{\infty},\quad |p|<1
\eea
and we have
\bea
\theta(x^{-1};p)=-x^{-1}\theta(x;p).
\eea
We denote the elliptic shifted factorial by%
\footnote{We apply this non-standard notation to avoid a confusion with the multiple infinite product function.}
\bea
\theta(z;q,p)_n = \prod_{k=0}^{n-1} \theta(zq^n;p).
\eea
The elliptic gamma function is defined as
\begin{align}
    \Gamma(x;p,q) = \frac{(pq/x;p,q)_\infty}{(x;p,q)_\infty},\quad |p|,|q|<1.
\end{align}

\paragraph{Equivariant index} For a vector bundle $\mathbf{X}$ with the character 
\bea
\operatorname{ch}\mathbf{X}=\sum_{i}n_{i}x_{i},
\eea
we define the index as
\bea
\mathbb{I}[\mathbf{X}]=\prod_{i}(1-x_{i})^{n_{i}}.
\eea
For example, we have
\bea
\mathbb{I}[x]=1-x^{-1}=\exp\left(-\sum_{n=1}^{\infty}\frac{1}{n}x^{-n}\right).
\eea
The dual of $\mathbf{X}$ is defined as
\bea
\operatorname{ch}\mathbf{X}^{\vee}=\sum_{i}n_{i}x_{i}^{-1}
\eea
and the index obeys the reflection property
\bea\label{eq:reflectionprop0}
\mathbb{I}[\mathbf{X}^{\vee}]=(-1)^{\operatorname{rk}\mathbf{X}}\det\mathbf{X}\,\,\mathbb{I}[\mathbf{X}]
\eea
where $\operatorname{rk}\mathbf{X}=\sum_{i}n_{i}$ and $\det \mathbf{X}=\prod_{i}x_{i}^{n_{i}}$.

We denote the $p$-th Adams operation on $\mathbf{X}$ as
\bea
\operatorname{ch}\mathbf{X}^{[p]}=\sum_{i}n_{i}x_{i}^{p}.
\eea
From now on, the vector bundle and the character will be identified.

\paragraph{$q$-deformation parameters}
The $q$-deformation parameters $q_{1,2,3,4}$ obey the condition $q_{1}q_{2}q_{3}q_{4}=1$ and are identified with the $\Omega$-background parameters. We define
\bea
\bfP_{a}=1-q_{a},\quad \bfP_{a}^{\vee}=1-q_{a}^{-1},\quad a\in\four
\eea
and for any subset $S\subseteq\four$
\bea
q_{S}=\prod_{a\in S}q_{a},\quad \bfP_{S}=\prod_{a\in S}\bfP_{a}.
\eea
Let $\bar{S}$ be the complement of the subset $S$ and then we define
\bea
\bfP_{\bar{S}}=\prod_{a\in\bar{S}}\bfP_{a}.
\eea

\paragraph{Reflection property}
We have the following properties which are obtained from the reflection property \eqref{eq:reflectionprop0}:
\bea\label{eq:reflectionprop1}
\mathbb{I}[\bfP_{a}x]=q_{a}^{-1}\mathbb{I}[\bfP_{a}^{\vee}x^{\vee}],\quad \mathbb{I}[\bfP_{ab}x]=\mathbb{I}[\bfP_{ab}^{\vee}x^{\vee}],\quad \mathbb{I}[\bfP_{abc}x]=\mathbb{I}[\bfP_{abc}^{\vee}x^{\vee}],\quad \mathbb{I}[\bfP_{\four}x]=\mathbb{I}[\bfP_{\four}^{\vee}x^{\vee}].
\eea
This can be generalized straightforward for a character whose rank is not infinite. Let $\mathbf{X}=\sum_{i\in I}x_{i}$ where $I$ is a finite set ($|I|<\infty$). Then,
\bea\label{eq:reflectionprop2}
\mathbb{I}[\bfP_{a}\mathbf{X}]=q_{a}^{-|I|}\mathbb{I}[\bfP_{a}^{\vee}\mathbf{X}^{\vee}],\quad \mathbb{I}[\bfP_{ab}\mathbf{X}]=\mathbb{I}[\bfP_{ab}^{\vee}\mathbf{X}^{\vee}],\quad \mathbb{I}[\bfP_{abc}\mathbf{X}]=\mathbb{I}[\bfP_{abc}^{\vee}\mathbf{X}^{\vee}],\quad \mathbb{I}[\bfP_{\four}\mathbf{X}]=\mathbb{I}[\bfP_{\four}^{\vee}\mathbf{X}^{\vee}].
\eea
This is because for example
\bea
\mathbb{I}[\bfP_{ab}\mathbf{X}]&=\prod_{i\in I}\frac{(1-x_{i}^{-1})(1-q_{ab}^{-1}x_{i}^{-1})}{(1-q_{a}^{-1}x_{i}^{-1})(1-q_{b}^{-1}x_{i}^{-1})}\\
&=\prod_{i\in I}\frac{(1-x_{i})(1-q_{ab}x_{i})}{(1-q_{a}x_{i})(1-q_{b}x_{i})}=\mathbb{I}[\bfP_{ab}^{\vee}\mathbf{X}^{\vee}].
\eea
When the character is an infinite sum, we need to regularize the infinite product properly and the above identities will not hold. For example, assume 
\bea
\mathbf{X}=\frac{x}{1-p}=x\sum_{i=0}p^{i},\quad |p|<1,
\eea
where $p$ is some generic parameter. The index is then 
\bea
\mathbb{I}[\mathbf{X}]&=\exp\left(-\sum_{n=1}^{\infty}\frac{1}{n}\frac{x^{-n}}{1-p^{-n}}\right)=(px^{-1};p)_{\infty}^{-1},\\
\mathbb{I}[\mathbf{X}^{\vee}]&=\exp\left(-\sum_{n=1}^{\infty}\frac{1}{n}\frac{x^{n}}{1-p^{n}}\right)=(x;p)_{\infty}
\eea
and therefore we have
\bea
\mathbb{I}[\mathbf{X}^{\vee}]&=\theta (x;p)\mathbb{I}[\mathbf{X}],\\
\mathbb{I}[\bfP_{ab}^{\vee}\mathbf{X}^{\vee}]&=\frac{\theta(x;p)\theta(q_{ab}x;p)}{\theta (q_{a}x;p)\theta(q_{b}x;p)}\mathbb{I}[\bfP_{ab}\mathbf{X}].
\eea

\paragraph{Structure functions}
We define the structure functions as
\bea\label{eq:structure-funct}
\mathscr{V}_{a}(x)&=\mathbb{I}[-\bfP_{a}^{\vee}x^{\vee}]=\frac{1-q_{a}x}{1-x},\\
\mathscr{S}_{ab}(x)&=\mathbb{I}[-\bfP_{ab}^{\vee}x^{\vee}]=\frac{(1-q_{a}x)(1-q_{b}x)}{(1-x)(1-q_{a}q_{b}x)},\\
g_{\bar{a}}(x)&=\mathbb{I}[-\bfP_{\bar{a}}^{\vee}x^{\vee}]=\frac{\prod_{i\neq a}(1-q_{i}x)(1-q_{a}^{-1}x)}{(1-x)\prod_{i\neq a}(1-q_{a}^{-1}q_{i}^{-1}x)},\\
\mathcal{A}_{\mathbb{C}^{4}}(x)&=\mathbb{I}[-\bfP_{\four}^{\vee}x^{\vee}]=\frac{\prod_{a\in\four}(1-q_{a}x)\prod_{a\in\four}(1-q_{a}^{-1}x)}{(1-x)^{2}\prod_{i\neq j}(1-q_{i}q_{j}x)}.
\eea
Note that we have the following properties:
\bea
\mathscr{S}_{ab}(x)=\frac{\mathscr{V}_{a}(x)}{\mathscr{V}_{a}(q_{b}x)},\quad g_{abc}(x)=\frac{\mathscr{S}_{ab}(x)}{\mathscr{S}_{ab}(q_{c}x)},\quad \mathcal{A}_{\mathbb{C}^{4}}(x)=\frac{g_{\bar{a}}(x)}{g_{\bar{a}}(q_{a}x)}.
\eea

\paragraph{Sign rules}
The reflection properties mentioned above are true only when the roots $\{x_{i}\}$ are generic. When $\{x_{i}\}$ is not generic, the reflection property will give extra sign factors. For example, let us consider the character $\bfP_{123}x$. When $x$ is generic, we simply have
\bea
\mathbb{I}[\bfP_{123}x]=\mathbb{I}[\bfP_{123}^{\vee}x^{-1}]
\eea
because
\bea
\bfP_{123}x=x-(q_{1}+q_{2}+q_{3})x+(q_{12}+q_{13}+q_{23})x-q_{123}x
\eea
and
\bea
\operatorname{rk}(\bfP_{123}x)=0,\quad \det(\bfP_{123}x)=1.
\eea
However, when $x\rightarrow 1$, the character will contain a $1$ term which give zeros after taking the index. Therefore, the reflection property should be modified as
\bea
\mathbb{I}[\bfP_{123}-1]=(-1)\mathbb{I}[\bfP_{123}^{\vee}-1].
\eea
%In the main text, we will only use this property when $\bfP_{123}$ appears in the character, so let us summarize its property.

\begin{definition}\label{app-def:movable}
    Let $\mathbf{A}$ be a Laurent polynomial
    \bea
    \mathbf{A}=\sum_{(n_{1},\ldots ,n_{p})}A_{\vec{n}}x_{1}^{n_{1}}\cdots x_{p}^{n_{p}},
    \eea
    where the sum $(n_{1},\ldots,n_{p})$ is taken over some subset in $\mathbb{Z}^{p}$. If there is no constant term $A_{\vec{n}}=0$ for $n_{1}=n_{2}=\cdots n_{p}=0$, then $\mathbf{A}$ is \textbf{movable}. The constant term is called the \textbf{unmovable} part. We denote the movable part and unmovable part as
    \bea
\left[\mathbf{A}\right]^{(\neq0)},\quad \left[\mathbf{A}\right]^{(0)},
    \eea
    respectively. Note that this means
    \bea
    \mathbf{A}=\left[\mathbf{A}\right]^{(\neq0)}+\left[\mathbf{A}\right]^{(0)}.
    \eea
\end{definition}

\begin{proposition}\label{app-prop:reflection_sign}
    Let $\mathbf{X}=\sum_{i\in I}x_{i}$ be a character, where $I$ is a finite set and $\{x_{i}\}_{i\in I}$ are generic. Then, the reflection property is
    \bea
    \mathbb{I}\left[\bfP_{123}\mathbf{X}\right]=\mathbb{I}[\bfP_{123}^{\vee}\mathbf{X}^{\vee}].
    \eea
    Suppose that for example when $\{x_{i}\}$ are specialized the character $\bfP_{123}\mathbf{X}$ is decomposed as
    \bea
    \bfP_{123}\mathbf{X}=\widetilde{\mathbf{X}}+\sum_{i\in I_{+}} 1-\sum_{i\in I_{-}}1
    \eea
    where $\widetilde{\mathbf{X}}$ does not contain any $\pm 1$ term and $I_{\pm}$ are finite sets giving $\pm1$ terms. Namely, 
    \bea
    \widetilde{\mathbf{X}}=\left[\bfP_{123}\bfX\right]^{(\neq 0)},\quad \sum_{i\in I_{+}} 1-\sum_{i\in I_{-}}1=\left[\bfP_{123}\bfX\right]^{(0)}.
    \eea
    The reflection property is then given as
    \bea
    \mathbb{I}[\bfP_{123}\mathbf{X}-|I_{+}|+|I_{-}|]=(-1)^{|I_{+}|-|I_{-}|}\mathbb{I}[\bfP_{123}^{\vee}\mathbf{X}^{\vee}-|I_{+}|+|I_{-}|].
    \eea
    Namely, the sign factor is determined by the unmovable terms (or the number of unmovable terms). Similar formulas can be obtained for others: $\bfP_{a}\bfX,\bfP_{A}\bfX,\bfP_{\four}\bfX$.
    
\end{proposition}

\begin{proposition}\label{app-prop:reflection-mod}
    Let $\mathbf{A},\mathbf{B}$ be a character where $\mathbf{A}+\mathbf{B}^{\vee}$ is movable. Then
    \bea
    \mathbb{I}\left[\mathbf{A}+\mathbf{B}^{\vee}\right]=(-1)^{\operatorname{rk}([\mathbf{B}]^{(\neq 0)})}\det [\mathbf{B}]^{(\neq 0)} \mathbb{I}\left[\mathbf{A}+\mathbf{B}\right]
    \eea
\end{proposition}
\begin{proof}
    Let $[\mathbf{B}]^{(0)}=m\in\mathbb{Z}$. Since $\mathbf{A}+\mathbf{B}^{\vee}$ is movable, using $\mathbf{A}=[\mathbf{A}]^{(0)}+[\mathbf{A}]^{(\neq0)}$, we have $[\mathbf{A}]^{(0)}=-m$. We then have
    \bea
    \mathbf{A}+\mathbf{B}^{\vee}=[\mathbf{A}]^{(\neq 0)}+[\mathbf{B}]^{(\neq 0)\vee}.
    \eea
    Since both term are movable, we can safely use the reflection property \eqref{eq:reflectionprop0} and we obtain the statement.
\end{proof}

\section{Sign rules}
\subsection{Direct proof of Lemma~\ref{lemm:D8-D6reduce-sign}}\label{app-sec:D8-D6reduce-sign-proof}
\begin{lemma}
After specializing $K=q_{a}\,(a\in\four)$, we have the following identity:
\bea
(-1)^{\sigma_{4}(\pi)}\mathcal{Z}_{\four;4}^{\D8}[\pi,q_{a}]=\widetilde{\mathcal{Z}}^{\D6}_{\bar{a}}[\pi], \quad a\in\four
\eea
where the plane partition $\pi\in\mathcal{PP}_{a}$.% as $\pi\hookrightarrow \rho$ in the left-hand side \remred{(refer to notation part and refer to it)}.
\end{lemma}

\begin{proof}
For $a=4$ it is trivial. Let us focus on $a=1$ and the other $a=2,3$ are obtained from the triality symmetry between $a=1,2,3$.

The partition functions are written as
\bea
\mathcal{Z}^{\D8}_{\four;4}[\pi,q_{1}]&=\mathbb{I}\left[-\bfP_{1}^{\vee}x^{-1}\bm{\pi}+\bfP_{123}^{\vee}\bm{\pi}^{\vee}\bm{\pi}\right],\quad \widetilde{\mathcal{Z}}^{\D6}_{\bar{1}}[\pi]=\mathbb{I}\left[-\bfP_{1}^{\vee}x^{-1}\bm{\pi}+\bfP_{234}^{\vee}\bm{\pi}^{\vee}\bm{\pi}\right],
\eea
where $\bm{\pi}=\sum_{\scube\in\pi}\chi_{\bar{1},x}(\cube)$. Using $\bfP_{234}^{\vee}=\bfP_{23}^{\vee}-q_{1}\bfP_{23}$, we have
\bea
\widetilde{\mathcal{Z}}^{\D6}_{\bar{1}}[\pi]&=\mathbb{I}\left[-\bfP_{1}^{\vee}x^{-1}\bm{\pi}+\bfP_{234}^{\vee}\bm{\pi}^{\vee}\bm{\pi}\right]\\
&=\mathbb{I}\left[-\bfP_{1}^{\vee}x^{-1}\bm{\pi}+\bfP_{23}^{\vee}\bm{\pi}^{\vee}\bm{\pi}-q_{1}\bfP_{23}\bm{\pi}^{\vee}\bm{\pi}\right]\\
&=(-1)^{\left[q_{4}\bfP_{23}\bm{\pi}^{\vee}\bm{\pi}\right]^{(0)}}\mathbb{I}\left[-\bfP_{1}^{\vee}x^{-1}\bm{\pi}+\bfP_{123}^{\vee}\bm{\pi}^{\vee}\bm{\pi}\right]\\
&=(-1)^{\left[q_{4}\bfP_{23}\bm{\pi}^{\vee}\bm{\pi}\right]^{(0)}}\mathcal{Z}^{\D8}_{\four;4}[\pi,q_{1}],
\eea
where in the third and last line, we used Prop.~\ref{app-prop:reflection_sign}, \ref{app-prop:reflection-mod}.

Let us show that the sign factors are
\bea
\left[q_{4}\bfP_{23}\bm{\pi}^{\vee}\bm{\pi}\right]^{(0)}=\min\{h_{4}(\pi)-1,0\}.
\eea
When $\pi=\emptyset$ the sign is trivial. We assume $\pi\neq 0$ which is equivalent to $h_{4}(\pi)\geq 1$. We decompose the plane partitions to sequence of non-increasing Young diagrams in the 4-direction $\pi=\{\lambda^{(1)},\lambda^{(2)},\cdots\}$ with $\lambda^{(i)}\succeq \lambda^{(i+1)}$ and then the character is
\bea
\bm{\pi}=\sum_{i=1}^{h_{4}(\pi)-1}\bm{\lambda}^{(i)},\quad \bm{\lambda}^{(i)}=\sum_{\Abox\in\lambda^{(i)}}\chi_{23,xq_{4}^{i-1}}(\Bbox).
\eea
We first have
\bea
\left[q_{4}\bfP_{23}\bm{\pi}^{\vee}\bm{\pi}\right]^{(0)}&=\left[q_{4}\bfP_{23}\left(\sum_{i\leq j}\bm{\lambda}^{(i)\vee}\bm{\lambda}^{(j)}+\sum_{i>j} \bm{\lambda}^{(i)\vee}\bm{\lambda}^{(j)}   \right)\right]^{(0)}\\
&=\left[q_{4}\bfP_{23}\sum_{i>j} \bm{\lambda}^{(i)\vee}\bm{\lambda}^{(j)}   \right]^{(0)}.
\eea
This is because $\bm{\lambda}^{(i)\vee}\bm{\lambda}^{(j)},\,\,(i\leq j)$ contains terms with degree $q_{4}^{j-i\geq 0}$ which will not cancel the overall factor $q_{4}$ and thus it is movable. The remaining term $\bm{\lambda}^{(i)\vee}\bm{\lambda}^{(j)},\,\,(i>j)$ contains terms with $q_{4}^{j-i}$ and for them to be unmovable, we need $j-i+1=0$:
\bea
\left[q_{4}\bfP_{23}\sum_{i>j} \bm{\lambda}^{(i)\vee}\bm{\lambda}^{(j)}   \right]^{(0)}=
\left[q_{4}\bfP_{23}\sum_{i=1}^{h_{4}(\pi)-1} \bm{\lambda}^{(i+1)\vee}\bm{\lambda}^{(i)}   \right]^{(0)}.
\eea
\begin{figure}[t]
    \centering
    \includegraphics[width=0.5\linewidth]{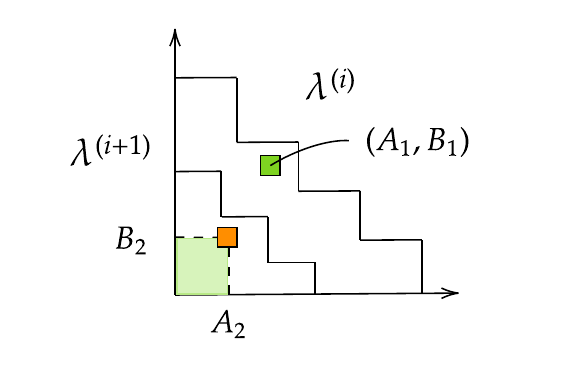}
    \caption{Positions of boxes}
    \label{fig:app-D8-D6-reduce-sign}
\end{figure}
We then can compute it as
\bea
\left[q_{4}\bfP_{23}\sum_{i=1}^{h_{4}(\pi)-1} \bm{\lambda}^{(i+1)\vee}\bm{\lambda}^{(i)}   \right]^{(0)}&=\left[\bfP_{23}\sum_{i=1}^{h_{4}(\pi)-1} \sum_{\eta=(A_{1},B_{1})\in\lambda^{(i)}}\sum_{\xi=(A_{2},B_{2})\in\lambda^{(i+1)}} \xi^{\vee}\eta\right]^{(0)}\\
&=\left[\bfP_{23}\sum_{i=1}^{h_{4}(\pi)-1} \sum_{(A_{1},B_{1})\in\lambda^{(i)}}\sum_{(A_{2},B_{2})\in\lambda^{(i+1)}} q_{2}^{A_{1}-A_{2}}q_{3}^{B_{1}-B_{2}}\right]^{(0)}
\eea
where $\eta=q_{2}^{A_{1}-1}q_{3}^{B_{1}-1},\,\,\xi=q_{2}^{A_{2}-1}q_{3}^{B_{2}-1}$ and $(A_{2},B_{2})\in\lambda^{(i+1)}\preceq \lambda^{(i)}$. Since $\bfP_{23}$ is a polynomial in powers of $q_{2}^{\geq 0}q_{3}^{\geq 0}$, the unmovable terms only come from $1\leq A_{1}\leq A_{2},\,\, 1\leq B_{1}\leq B_{2}$ which are all included in $\lambda^{(i)}$ (see Figure~\ref{fig:app-D8-D6-reduce-sign}):
\bea
&\left[\bfP_{23}\sum_{i=1}^{h_{4}(\pi)-1} \sum_{A_{1}=1}^{A_{2}}\sum_{B_{1}=1}^{B_{2}}\sum_{(A_{2},B_{2})\in\lambda^{(i+1)}} q_{2}^{A_{1}-A_{2}}q_{3}^{B_{1}-B_{2}}\right]^{(0)}\\
=&\left[\sum_{i=1}^{h_{4}(\pi)-1} \sum_{(A_{2},B_{2})\in\lambda^{(i+1)}}  q_{2}^{-A_{2}+1}q_{3}^{-B_{2}+1}(1-q_{2}^{A_{2}})(1-q_{3}^{B_{2}})\right]^{(0)}\\
=&\sum_{i=1}^{h_{4}(\pi)-1} \sum_{(A_{2},B_{2})\in\lambda^{(i+1)}}\delta_{A_{2}=1}\delta_{B_{2}=1}\\
=&h_{4}(\pi)-1.
\eea
Therefore, we get the claim.

\end{proof}

\subsection{Sign rules for solid partitions with nontrivial boundary conditions}\label{app-sec:boundary-signrule}
Let $\widetilde{\rho}$ be an infinite solid partition with nontrivial boundary conditions. We denote $\rho$ as the set of boxes not including the boxes in the boundaries and $\rho_{\bd}$ the boundary boxes. Let $\widetilde{\Pi}^{(i)}$ be the possibly infinite plane partitions corresponding to the $(1,3)$ decompositions of the solid partition $\widetilde{\rho}$ and $\Pi^{(i)},\Pi^{(i)}_{\bd}$ the decomposition of $\rho,\rho_{\bd}$. The corresponding characters are denoted as $\widetilde{\bm{\Pi}}^{(i)},\,\, \bm{\Pi}^{(i)},\,\,\bm{\Pi}^{(i)}_{\bd}$ and we have $\widetilde{\bm{\Pi}}^{(i)}=\bm{\Pi}^{(i)}+\bm{\Pi}^{(i)}_{\bd}$.

Using Prop.~\ref{prop:sign_proof}, we have the following propositions.
\begin{proposition}
    The sign factor $s(\widetilde{\Pi})$ which is defined using Def.~\ref{def:signfactorrgeneral} is
    \bea
    s(\widetilde{\Pi})=\#\{(i,i,i,j)\in\widetilde{\rho}\mid i<j \}.
    \eea
\end{proposition}
\begin{proposition}
    The sign factor $s(\Pi_{\bd})$ which is defined using Def.~\ref{def:signfactorrgeneral} is
    \bea
    s(\Pi_{\bd})=\#\{(i,i,i,j)\in\rho_{\bd}\mid i<j \}.
    \eea
\end{proposition}
These are obtained by following the proof of Prop.~\ref{prop:sign_proof} and noticing that since $\widetilde{\Pi}^{(i)},\Pi_{\bd}^{(i)}$ are still non-increasing plane partitions, they obey $\widetilde{\Pi}^{(i)}\succeq \widetilde{\Pi}^{(j)}$ for $j>i$. Note that this non-increasing condition seems to be crucial for the proof given in Prop.~\ref{prop:sign_proof}.

Strictly speaking, the sign factor $s(\widetilde{\Pi}), s(\Pi_{\bd})$ may be ill-defined since there might be infinite number of terms and the right hand side might diverge. However, we may regularize them as the following:
\bea
s(\widetilde{\Pi})-s(\Pi_{\bd})=\#\{(i,i,i,j)\in \rho \mid i<j\}\eqqcolon \sigma_{4}(\rho)
\eea
Since $\rho$ is just a finite set of boxes, this is well-defined. Using Def.~\ref{def:signfactorrgeneral} and $\widetilde{\bm{\Pi}}^{(i)}=\bm{\Pi}^{(i)}+\bm{\Pi}^{(i)}_{\bd}$ the left hand side is written as
\bea
s(\widetilde{\Pi})-s(\Pi_{\bd})&=\left[\sum_{i<j}\bfP_{123}\left(\widetilde{\bm{\Pi}}^{(i)\vee}\widetilde{\bm{\Pi}}^{(j)}-\bm{\Pi}_{\bd}^{(i)\vee}\bm{\Pi}_{\bd}^{(j)}   \right)\right]^{(0)}\\
&=\left[\sum_{i<j}\bfP_{123}\left(\bm{\Pi}^{(i)\vee}\bm{\Pi}^{(j)}+\bm{\Pi}_{\bd}^{(i)\vee}\bm{\Pi}^{(j)}+\bm{\Pi}^{(i)\vee}\bm{\Pi}_{\bd}^{(j)}  \right)\right]^{(0)}.
\eea

\begin{proposition}
We have
    \bea
\left[\sum_{j>i}\left(\bfP_{\four}\bm{\Pi}^{(j)}_{\bd}\bm{\Pi}^{(i)\vee}+\bfP_{123}\bm{\Pi}^{(i)\vee}\bm{\Pi}^{(j)}\right)\right]^{(0)}=\sigma_{4}(\rho)-\left[\sum_{i<j}\bfP_{123}\bm{\Pi}^{(i)\vee}_{\bd}\bm{\Pi}^{(j)}\right]^{(0)}.
    \eea
   
\end{proposition}

\begin{proof}  
The left hand side is rewritten as
\bea
\,&\sum_{i<j}\left(\bfP_{\four}\bm{\Pi}^{(j)}_{\bd}\bm{\Pi}^{(i)\vee}+\bfP_{123}\bm{\Pi}^{(i)\vee}\bm{\Pi}^{(j)}\right)\\
=&\sum_{i<j}\bfP_{123}\left(\bm{\Pi}^{(i)\vee}\bm{\Pi}^{(j)}+\bm{\Pi}_{\bd}^{(i)\vee}\bm{\Pi}^{(j)}+\bm{\Pi}^{(i)\vee}\bm{\Pi}_{\bd}^{(j)}  \right)\\
&+\sum_{i<j}\left(\bfP_{123}^{\vee}\bm{\Pi}^{(j)}_{\bd}\bm{\Pi}^{(i)\vee}-\bfP_{123}\bm{\Pi}_{\bd}^{(i)\vee}\bm{\Pi}^{(j)}\right).
\eea
Using the property $[\mathbf{X}]^{(0)}=[\mathbf{X}^{\vee}]^{(0)}$, the unmovable terms are
\bea
\sigma_{4}(\rho)+\left[\sum_{i<j}\bfP_{123}\left(\bm{\Pi}^{(j)\vee}_{\bd}\bm{\Pi}^{(i)} - \bm{\Pi}_{\bd}^{(i)\vee}\bm{\Pi}^{(j)}\right)\right]^{(0)}.
\eea
The first term $\bm{\Pi}^{(j)\vee}\bm{\Pi}^{(i)}$ contains terms with $q_{4}^{i-j}=(q_{123})^{j-i>0}$ and thus it is movable. It is enough to focus on the second term and we obtain the claim.

\end{proof}

\section{Vertex operators and zero-modes}\label{app:zero-modes}
\begin{definition}[\cite{Kimura:2023bxy}]
The zero modes of the vertex operators are defined as 
\bea
%\begin{split}
    &\mathsf{a}_{0}(x)=e^{\mathsf{a}_{0}},\quad \mathsf{s}_{a,0}(x)=x^{\mathsf{s}_{a,0}}e^{\widetilde{\mathsf{s}}_{a,0}},\quad \mathsf{w}_{\bar{a}}(x)=x^{\mathsf{w}_{\bar{a},0}}e^{\widetilde{\mathsf{w}}_{\bar{a},0}}e^{\widetilde{\widetilde{\mathsf{w}}}_{\bar{a},0}},\\
    &\mathsf{x}_{A,0}(x)=e^{\mathsf{x}_{A,0}}, \quad \widetilde{\mathsf{z}}_{0}^{K}(x)=x^{\mathsf{z}^{K}_{0}}e^{\widetilde{\mathsf{z}}^{K}_{0}}e^{\widetilde{\widetilde{\mathsf{z}}}^{K}_{0}}
%\end{split}
\eea
with
\begin{align}
\begin{split}
& \mathsf{a}_{0}=\mathsf{t}_{0},\quad \mathsf{s}_{a,0}=-(\log q_{a})^{-1}\mathsf{t}_{0},\quad \widetilde{\mathsf{s}}_{a,0}=-(\log q_{a})^{-1}\widetilde{\partial}_{\mathsf{t}},\\
    &\mathsf{w}_{\bar{a},0}=-\log q_{a}\,\widetilde{\mathsf{t}}_{0},\quad \widetilde{\mathsf{w}}_{\bar{a},0}=-\log q_{a}\log (-q_{a})\,\widetilde{\mathsf{t}}_{0},\quad \widetilde{\widetilde{\mathsf{w}}}_{\bar{a},0}=-\log q_{a}\partial_{\mathsf{t}},\\
    &\mathsf{x}_{A,0}=\log q_{c}\log q_{d}\,\widetilde{\mathsf{t}}_{0},\,\,\quad (\bar{A}=cd),\\
    &\mathsf{z}^{K}_{0}=-\log K \tilde{\mathsf{t}}_{0},\quad \widetilde{\mathsf{z}}^{K}_{0}=-\log K\log(-K)\widetilde{\mathsf{t}}_{0},\quad \tilde{\tilde{\mathsf{z}}}^{K}_{0}=-\log K\partial_{\mathsf{t}}
\end{split}
\end{align}
where we introduced two independent sets of zero modes
\beq
    [\partial_{\mathsf{t}},\mathsf{t}_{0}]=[\widetilde{\partial}_{\mathsf{t}},\tilde{\mathsf{t}}_{0}]=1,\quad [\mathsf{t}_{0},\tilde{\mathsf{t}}_{0}]=[\partial_{\mathsf{t}},\widetilde{\partial}]=[\mathsf{t}_{0},\widetilde{\partial}_{\mathsf{t}}]=[\tilde{\mathsf{t}}_{0},\partial_{\mathsf{t}}]=0.\label{eq:zeromodesdef}
\eeq
The normal ordering is defined as 
\beq
{:\partial_{\mathsf{t}}\,\mathsf{t}_{0}:}=\mathsf{t}_{0}\partial_{\mathsf{t}},\quad {:\tilde{\partial}_{\mathsf{t}}\,\tilde{\mathsf{t}}_{0}:}=\tilde{\mathsf{t}}_{0}\tilde{\partial}_{\mathsf{t}}.
\eeq
\end{definition}

%%%%%%%%%%%%%%%%%%%%%%%%%%%%%%%%%%%%%%%%%%%%%%%%%%%%%%%%%%%%%%%
\bibliographystyle{ytamsalpha.bst}
\bibliography{Worigami}
\end{document}